\documentclass[acmsmall,authorversion]{acmart}

\newif\iflongversion%
\longversiontrue%

\newif\ifextended%
\extendedfalse%

\usepackage{graphicx}
\usepackage{bm}
\usepackage{amsmath}
\usepackage{dsfont}
\usepackage{subcaption}
\usepackage{cleveref}
\usepackage{thmtools, thm-restate}
\usepackage[ruled,vlined,linesnumbered]{algorithm2e}
\usepackage[normalem]{ulem}
\usepackage{enumitem}

\newcommand\numberthis{\addtocounter{equation}{1}\tag{\theequation}}

\newcommand{\bc}{betweenness centrality}

\newcommand{\algname}{\textsc{SILVAN}}
\newcommand{\algnametopk}{\textsc{SILVAN}-\textsc{TopK}}

\newcommand{\kadabra}{\texttt{KADABRA}}
\newcommand{\bavarian}{\texttt{BAVARIAN}}
\newcommand{\bavarianp}{\bavarian-\texttt{P}}

\newcommand{\pars}[1]{\left( #1 \right)}
\newcommand{\sqpars}[1]{\left[ #1 \right]}
\newcommand{\brpars}[1]{\left\lbrace #1 \right\rbrace}
\newcommand{\qt}[1]{\lq\lq#1\rq\rq}
\newcommand{\qtm}[1]{\text{\lq\lq}#1\text{\rq\rq}}
\newcommand{\ind}[1]{\mathds{1} \left[ #1 \right] }

\newcommand{\sample}{\mathcal{S}}

\newcommand{\BO}[1]{\mathcal{O}\left(#1\right)}
\newcommand{\BOi}[1]{\mathcal{O}( #1 )}

\newcommand{\E}{\mathbb{E}}
\newcommand{\F}{\mathcal{F}}
\newcommand{\X}{\mathcal{X}}

\newcommand{\N}{\mathbb{N}}
\newcommand{\Q}{\mathbb{Q}}

\DeclareMathOperator*{\argmax}{arg\,max}

\newcommand{\mcera}{\text{MCERA}}
\newcommand{\nmcera}{$c$\text{-\mcera}}

\newcommand{\abs}[1]{\lvert#1 \rvert}

\newcommand{\rc}{\rade(\F, m)}

\newcommand{\era}{\erade\left(\F, \sample\right)}

\newcommand{\sd}{\mathsf{D}}
\newcommand{\supdev}{\sd(\F, \sample)}
\newcommand{\supdevj}{\sd(\F_j, \sample)}

\newcommand{\bagsp}{\tau}
\newcommand{\probdist}{\gamma}

\newcommand{\varx}{ g(x) }

\newcommand{\vsigma}{{\bm{\sigma}}}

\newcommand{\mera}{\erade^{c}_{m}(\F, \sample, \vsigma)}

\newcommand{\maxabsf}{\hat{\xi}}
\newcommand{\ewvar}{w}

\newcommand{\R}{\mathbb{R}}
\newcommand{\rade}{\mathsf{R}}
\newcommand{\erade}{\hat{\rade}}

\newcommand{\ebc}{\tilde{b}}

\newcommand{\unionbdelta}{\ln\pars{\frac{4t}{\delta}}}

\newcommand{\abra}{ABRA}

\setcopyright{acmlicensed}
\acmJournal{TKDD}

\begin{document}

\title{\algname: Estimating  Betweenness Centralities with Progressive Sampling and Non-uniform Rademacher Bounds}



\author{Leonardo Pellegrina}
\affiliation{%
  \institution{Dept. of Information Engineering, University of Padova}
  \streetaddress{Via Gradenigo 6b}
  \city{Padova}
  \country{Italy}
  \postcode{35131}
}
\email{pellegri@dei.unipd.it}

\author{Fabio Vandin}
\affiliation{%
  \institution{Dept. of Information Engineering, University of Padova}
  \streetaddress{Via Gradenigo 6b}
  \city{Padova}
  \country{Italy}
  \postcode{35131}
}
\email{fabio.vandin@unipd.it}

\renewcommand{\shortauthors}{Pellegrina and Vandin.}

\begin{abstract}
Betweenness centrality is a popular centrality measure with applications in several domains, and whose exact computation is impractical for modern-sized networks. We present \algname, a novel, efficient algorithm to compute, with high probability, accurate estimates of the betweenness centrality of all nodes of a graph and a high-quality approximation of the top-$k$ betweenness centralities. \algname\ follows a progressive sampling approach, and builds on novel bounds based on Monte-Carlo Empirical Rademacher Averages, a powerful and flexible tool from statistical learning theory. \algname\ relies on a novel estimation scheme providing \emph{non-uniform} bounds on the deviation of the estimates of the betweenness centrality of all the nodes from their true values, 
and a refined characterisation of the number of samples required to obtain a high-quality approximation. 
Our extensive experimental evaluation shows that \algname\ extracts high-quality approximations while outperforming, in terms of number of samples and accuracy, the state-of-the-art approximation algorithm with comparable quality guarantees.
\end{abstract}

\ccsdesc[500]{Information systems~Data mining}
\ccsdesc[500]{Mathematics of computing~Probabilistic algorithms}
\ccsdesc[500]{Theory of computation~Sketching and sampling}

\keywords{Betweenness Centrality, Rademacher Averages, Random Sampling}

\begin{teaserfigure}
  \noindent\textit{``Sim Sala Bim!''} -- Silvan\\\url{https://en.wikipedia.org/wiki/Silvan\_(illusionist)}
  \Description{This is not a figure, it is just a quote before the abstract. The
    quote is: ``Sim Sala Bim!'', and it is from the illusionist Silvan.}
\end{teaserfigure}

\maketitle


\section{Introduction}

The computation of node centrality measures, which are scores quantifying the importance of nodes, is a fundamental task in graph analytics~\cite{newman2018networks}. \emph{Betweenness centrality} is a popular centrality measure, defined first in sociology~\cite{anthonisse1971rush,freeman1977set}, that quantifies the importance of a node as the fraction of shortest paths in the graph that go through the node. 

The computation of the \emph{exact} betweenness centrality for all nodes in a graph $G=(V,E)$ can be obtained with Brandes' algorithm~\cite{brandes2001faster} in time $\BO{|V| |E|}$ for unweighted graphs and in time $\BO{|V||E|+|V|^2 \log |V|}$ for graphs with positive weights, which is impractical for modern networks with up to hundreds of millions of nodes and edges. Several works (e.g., \cite{erdHos2015divide,sariyuce2013shattering}) proposed heuristics to improve Brandes' algorithm, but they do not improve on its worst-case complexity. In fact, for unweighted graphs a corresponding lower bound (based on the Strong Exponential Time Hypothesis) was proved in~\cite{borassi2016into}.
The impracticality of the exact computation for modern networks, and the use of betweenness centrality mostly in exploratory analyses of the data, have motivated the study of efficient algorithms to compute approximations of the betweenness centrality, trading precision for efficiency. 

Several works~\cite{RiondatoK15,RiondatoU18,borassi2019kadabra,cousins2021bavarian} have recently proposed sampling approaches to approximate the betweenness centrality of all nodes in a graph. The main idea is to sample shortest paths uniformly at random, and use such paths to estimate the betweenness centrality of the nodes. 
As for all sampling approaches, the main difficulty is then to relate the estimates obtained from the samples with the corresponding exact quantities, providing tight trade-offs between guarantees on the quality of the estimates and the required computational work. 
\cite{RiondatoK15,RiondatoU18,borassi2019kadabra,cousins2021bavarian} all provide rigorous approximations of the betweenness centrality, and~\cite{RiondatoK15,RiondatoU18,cousins2021bavarian} rely on tools from statistical learning theory, such as the \emph{VC-dimension}~\cite{Vapnik:1971aa}, the \emph{pseudodimension}~\cite{pollard2012convergence}, or \emph{Rademacher Averages}~\cite{KoltchinskiiP00}, which have been successfully used to obtain rigorous approximations for other data mining tasks (e.g., pattern mining~\cite{RiondatoU14,RiondatoU15,RiondatoV18}). 
For pattern mining, recent work~\cite{pellegrina2020mcrapper} has shown that a more advanced tool from statistical learning theory, namely \emph{Monte-Carlo Empirical Rademacher Averages (MCERA)}~\cite{BartlettM02} (see Sect.~\ref{sec:rad_ave}), leads to improved results, mostly thanks to it \emph{data-dependent} nature (in contrast to \emph{distribution-free} tools such as the VC-dimension and the pseudodimension). 
Indeed, the MCERA was recently used in \bavarian~\cite{cousins2021bavarian} to obtain an unifying framework compatible with different estimators of the \bc.

\paragraph{Our contributions} 
In this work we study the problem of approximating the betweenness centralities of nodes in a graph. We propose \algname~(\emph{e\underline{S}timat\underline{I}ng betweenness centra\underline{L}ities with progressi\underline{V}e s\underline{A}mpling and \underline{N}on-uniform rademacher bounds}), a novel, efficient, progressive sampling algorithm to approximate betweenness centralities while providing rigorous guarantees on the quality of various approximations. 
\begin{itemize}[leftmargin=8pt]

\item 
Our first contribution is \emph{empirical peeling}, a novel technique that we introduce to obtain sharp \emph{non-uniform data-dependent bounds} on the maximum deviation of families of functions (\Cref{sec:empiricalpeeling}). 
Empirical peeling is based on the MCERA and relies on an effective \emph{data-dependent} approach to \emph{partition} a family  of functions according to their empirically estimated variance; this allows to fully exploit \emph{variance-dependent} bounds at the core of the technique. 
Our algorithm \algname\ (\Cref{sec:algunif}) relies on such novel bounds to provide guarantees on the approximation of the \bc\ that are much sharper than the ones obtained by previous works; 
these new contributions make \algname\ a practical algorithm for obtaining different approximations of the \bc. 
In fact, we show that combining the MCERA with empirical peeling allows us to design flexible algorithms with different guarantees (e.g., additive or relative) and for different tasks (e.g., estimating all betweenness centralities or, in \Cref{alg:topk}, the top-$k$ ones). This is the first work that obtains different types of approximation guarantees based on the MCERA. 
Most importantly, our approach is general and of independent interest, as it may apply to other problems, even outside of data mining applications.

\item 
We derive a new bound on the sufficient number of samples to approximate the \bc\ for all nodes (\Cref{sec:upperboundsamplesabs}), that naturally combines with the progressive sampling strategy of \algname\ by introducing an upper limit to the number of samples required to converge. 
Our new bound is governed by key quantities of the underlying graph, not considered by previous works, such as the \emph{average shortest path length}, and the \emph{maximum variance} of \bc\ estimators, significantly improving the state-of-the-art bounds for the task. 
Our proof combines techniques from combinatorial optimization and key results from theory of concentration inequalities. 
While previous results were tailored to analyse a specific estimator of the \bc, our result is general, since it applies to all available estimators of the \bc. 
Furthermore, we extend this result to obtain sharper \emph{relative} deviation bounds from a random sample.

\item We perform an extensive experimental evaluation~(\Cref{sec:experiments}), showing that \algname\ improves the state-of-the-art by 
requiring a fraction of the sample sizes and running times to achieve a given approximation quality or, equivalently, sharper guarantees for the same amount of work. 
Our experimental evaluations shows that \algname's guarantees, provided by our theoretical analysis, hold with a true approximation error close to its probabilistic upper bound, confirming the sharpness of our analysis.
For the extraction of the top-$k$ betweenness centralities, our algorithm provides faster approximations, using less samples, and with fewer false positives. 
\end{itemize}

\section{Related work}
\label{sec:relwork}
We now review the works on approximating the betweenness centralities that are most relevant to our contributions. In particular, we focus on approaches that provide guarantees on the quality of the approximation, an often necessary requirement. 

The first practical sampling algorithm to approximate the betweenness centrality of all nodes with guarantees on the quality of the approximation is presented in~\cite{RiondatoK15}. Studying the VC-dimension of shortest paths, \citet{RiondatoK15} proved that when $\BOi{\log(D/\delta)/\varepsilon^{2}}$ shortest paths are sampled uniformly at random, the approximations are within an additive error $\varepsilon$ of the exact centralities with probability $\ge 1 - \delta$, 
where $D$ is (an upper bound to) the vertex diameter of the graph. While interesting, this result is characterized by the \emph{worst-case} and \emph{distribution-free} nature of the VC-dimension, and thus provides an overly conservative bound 
and cannot be used to design a \emph{progressive} sampling approach.

The first rigorous progressive sampling algorithm is \texttt{ABRA}~\cite{RiondatoU18}, which builds on the theory of Rademacher averages and pseudodimension and does not require an estimate of the vertex diameter. 
However, 
\texttt{ABRA} leverages a \emph{deterministic} and \emph{worst-case} upper bound to the Rademacher complexity (based on Massart's Lemma~\cite{Massart00,ShalevSBD14,mitzenmacher2017probability}); in a different scenario, \citet{pellegrina2020mcrapper} show it provides conservative results in most cases compared to its Monte Carlo approximation given by the MCERA. 
In addition, 
similarly to~\cite{RiondatoK15}, \abra\ obtains \emph{uniform and variance-agnostic bounds} that hold for all nodes in the graph $G$. 
The most recent approach is \bavarian~\cite{cousins2021bavarian}, which addresses some of the limitations of \abra\ using the MCERA and variance-aware tail bounds. 
Leveraging the flexibility and generality of the MCERA, \bavarian\ is compatible with different estimators of the \bc, but still obtains \emph{uniform} approximation bounds not sensible to the heterogeneity of the centrality of different nodes of the graph. 
In contrast, our algorithm \algname\ uses \emph{efficiently computable}, \emph{non-uniform, and variance-dependent bounds} for different subsets of the nodes, which lead to a significant reduction of the number of samples and running times required to obtain rigorous guarantees w.r.t. all methods mentioned above.

A different approach has been proposed by \kadabra~\cite{borassi2019kadabra}, a progressive sampling algorithm based on adaptive sampling. 
\kadabra\ is not based on tools from statistical learning theory, and our experimental evaluation shows that \kadabra\ is the state-of-the-art solution for the task.
\kadabra\ 
is based on a weighted union bound, 
using a 
{data-dependent} scheme to 
assign different probabilistic confidences
on the estimates of the betweenness centrality of each individual node, achieving improved approximations compared to the algorithm of~\cite{RiondatoK15} and to \abra~\cite{RiondatoU18}. 
Our algorithm \algname\ uses a similar intuition and data-dependent approach, but with crucial differences. In particular, \kadabra\ assigns confidence parameters $\delta_v$ for each $v\in V$, such that the probability that the approximation of the betweenness centrality for node $v$ not being accurate is at most $\delta_v$ \emph{for each node} $v \in V$, with $\sum_{v \in V} \delta_v \leq \delta$. In contrast, our approach uses Rademacher averages and empirical peeling to obtain variance-dependent approximations that are valid for \emph{sets of nodes}, exploiting correlations among nodes instead of considering each node individually. 
As we show in our experimental evaluation, 
this leads to significant improvements on the approximation quality compared to \kadabra. 

Most of existing methods~\cite{RiondatoK15,RiondatoU18,borassi2019kadabra} propose variants of their algorithms for the approximation of the top-$k$ betweenness centralities.  Our algorithm  \algname\  achieves better results for this task as well, thanks 
to the non-uniform bounding scheme from the use of empirical peeling.

Other papers study different, but related, problems and notions of centralities. 
\citet{de2020estimating} use an approach similar to~\cite{RiondatoK15}, and based on pseudodimension, to estimate percolation centrality. 
\citet{Chechik2015} propose an algorithm based on probability proportional to size sampling to estimate closeness centralities, which is the inverse of the average distance of a node to all other nodes in a graph, a popular importance measure in the study of social networks. 
\citet{boldi2013core} consider closeness and harmonic centralities though HyperLogLog counters \cite{flajolet1983probabilistic,boldi2011hyperanf}. 
\citet{bergamini2019computing} propose a new algorithm for selecting the $k$ nodes with the highest closeness centralities in a graph. 
\citet{mahmoody2016scalable} study the problem of finding a set of at most $k$ nodes of maximum betweenness centrality, where the betweenness centrality of a set is defined analogously to the betweenness centrality of nodes~\cite{ishakian2012framework,yoshida2014almost}, and present efficient randomized algorithms to obtain rigorous approximations with high probability. 
\citet{bergamini2018improving} study the problem of increasing the \bc\ of a node by adding new links. 
Other recent works considered extending the computation of the \bc\ to dynamic graphs~\cite{bergamini2014approximating,bergamini2015fully,hayashi2015fully}, uncertain graphs~\cite{saha2021shortest}, and temporal networks~\cite{santoro2022onbra}. 
\algname\ uses
estimates of the vertex diameter of graphs, for which several approximation approaches have been proposed 
(e.g., \cite{magnien2009fast,crescenzi2012computing,crescenzi2013computing}), including some for distributed frameworks~\cite{ceccarello2020distributed}.

\section{Preliminaries}

In this section we introduce the basic notions used in the remaining of the paper.

\subsection{Graphs and Betweenness Centrality}
Let $G=(V,E)$ be a graph. For ease of exposition, we focus on unweighted graphs, however our algorithms can be easily adapted to weighted graphs. For any pair $(s,t)$ of different nodes $(s\neq t)$, let $\sigma_{st}$ be the number of shortest paths between $s$ and $t$, and let $\sigma_{st}(v)$ be the number of shortest paths between $s$ and $t$ that \emph{pass through} (i.e., contain) $v$, with $s \neq v \neq t$. The (normalized) \emph{betweenness centrality} $b(v)$ of a node $v\in V$ is defined as 
\begin{equation*}
b(v) = \frac{1}{|V|(|V|-1)} \sum_{s\neq v \neq t} \frac{\sigma_{st}(v)}{\sigma_{st}}.
\end{equation*}

Some of our algorithms rely on the knowledge of the \emph{vertex diameter} $D$ of a graph $G$, defined as the maximum number of nodes in any shortest path of $G$. (If $G$ is unweighted, the vertex diameter of $G$ is equal to the diameter of $G$ plus $1$.)

\subsection{Rademacher Averages}
\label{sec:rad_ave}

Rademacher averages are a core concept in statistical learning theory~\cite{KoltchinskiiP00} and in the study of empirical processes~\cite{boucheron2013concentration}. 
We now present the main notions and results used in our work and defer additional details to \cite{boucheron2013concentration,ShalevSBD14,mitzenmacher2017probability}. 
Let $\X$ be a finite domain and consider 
a probability distribution $\probdist$ 
over the ele\-ments of $\X$. 
Let $\F$ be a family of functions from $\X$ to $[0,1]$, and let $\sample = \{ \bagsp_1,\dots, \bagsp_m\}$ be a collection of $m$ independent and identically distributed samples from $\X$ taken according to $\probdist$. 
For each function $f \in \F$, define its average value 
over the sample $\sample$ as $\mu_{\sample}(f) = \frac{1}{m}\sum_{i=1}^m f(s_i)$
and its expectation, taken w.r.t. $\sample$, as 
$\mu_{\probdist}(f) = \E_\sample [ \mu_{\sample}(f) ]$.
Note that, by definition, 
$\mu_{\sample}(f)$ is an \emph{unbiased} estimator of $\mu_{\probdist}(f)$.
Given $\sample$, we are interested in bounding the \emph{supremum deviation} $\supdev$ of $\mu_{\sample}(f)$ from $\mu_{\probdist}(f)$ among all $f \in \F$, that is
\begin{equation*}
\supdev = \sup_{f\in \F} | \mu_{\sample}(f) - \mu_{\probdist}(f) |.
\end{equation*}

The \emph{Empirical Rademacher Average} (ERA) $\era$  of $\F$ on $\sample$ is a key quantity to obtain a data-dependent upper bound to the supremum deviation $\supdev$. Let $\vsigma=\left< \vsigma_1,\dots, \vsigma_m \right>$ be a collection of $m$ i.i.d. Rademacher random variables (r.v.'s), each taking value in $\{-1,1\}$ with equal probability. The ERA $\era$  of $\F$ on $\sample$ is
\begin{equation*}
 \era = \E_\vsigma \left[ \sup_{f \in \F } \frac{1}{m}
 \sum_{i=1}^m \vsigma_i f(s_i) \right].
\end{equation*}
Computing the ERA $\era$ is usually intractable, since there are $2^m$ possible assignments for $\vsigma$ and for each such assignment a supremum over the functions in $\F$ must be computed. 
A useful approach to obtain sharp probabilistic bounds on the ERA is given by Monte-Carlo estimation~\cite{BartlettM02}. 
For $c \ge 1$, let $\vsigma \in \{-1,1\}^{c\times m}$ be a $c\times m$ matrix of i.i.d. Rademacher r.v.'s. The \emph{$c$-samples Monte-Carlo Empirical Rademacher average (\nmcera)} $\mera$ of $\F$ on $\sample$ using $\vsigma$ is:
\begin{equation*}
\mera = \frac{1}{c} \sum_{j=1}^c \sup_{f\in\F} \frac{1}{m} \sum_{i=1}^m\vsigma_{j,i}f(s_i).
\end{equation*}
The \nmcera\ allows to obtain sharp \emph{data-dependent} probabilistic upper bounds to the supremum deviation, as they directly estimate the expected supremum deviation of sets of functions by taking into account their correlation. 
For this reason, they are often significantly more accurate than other methods~\citep{pellegrina2020mcrapper}, such as the ones based on often loose \emph{deterministic upper bounds} to Rademacher averages (e.g., Massart's Lemma~\cite{Massart00}), or other \emph{distribution-free} notions of complexity, such as the VC-dimension. 
In general, the \nmcera\ may be hard to compute, due to the supremums over $\F$ \cite{BartlettM02}. However, for the case of betweenness centralities, we show in \Cref{sec:algunif} that all quantities relevant to the \nmcera\ can be efficiently and incrementally updated as shortest paths are randomly sampled. 

\section{\algname: Efficient Progressive Estimation of Betweenness Centralities}
\label{sec:mainalgorithm}

In this section we introduce \algname~(\emph{e\underline{S}timat\underline{I}ng betweenness centra\underline{L}ities with progressi\underline{V}e s\underline{A}mpling and \underline{N}on-uniform rademacher bounds}) 
and the techniques at its core. 

We start, in Section~\ref{sec:empiricalpeeling}, by presenting the empirical peeling technique and the related main technical results, which provide sharp data-dependent non-uniform approximation bounds supporting our algorithms.
We then describe, in Section~\ref{sec:algunif}, our algorithm \algname\ that builds on such improved bounds to obtain an approximation within \emph{additive} error $\varepsilon$ of the betweenness centrality for all nodes via progressive sampling. 
We then present, in Section~\ref{sec:upperboundsamplesabs}, improved bounds on the number of sufficient samples to achieve absolute approximations with high probability. These bounds are naturally combined with the progressive sampling scheme of \algname.  
Finally, in Section~\ref{alg:topk} we introduce \algnametopk, an extension of \algname\ to obtain a relative approximation of the $k$ nodes with highest betweenness centrality.

\subsection{Non-uniform Bounds via Empirical Peeling}
\label{sec:empiricalpeeling}

In this section we introduce \emph{empirical peeling}, a new data-dependent scheme based on the \nmcera\ to obtain sharp non-uniform bounds to the supremum deviation.
The main idea behind empirical peeling is to \emph{partition} the set of functions $\F$ in order to obtain 
the best possible bounds for different subsets of $\F$. 
Classical concentration inequalities, such as Bernstein's and Bennet's~\cite{boucheron2013concentration}, are well suited to control the deviation $\sd(\{f\}, \sample)$ of a single function $f$, and to derive an approximation whose accuracy depends on its variance $Var(f)$. 
Instead, when \emph{simultaneously} bounding the deviation of multiple functions belonging to a set of functions $\F$, 
the accuracy of the probabilistic bound on the supremum deviation $\supdev$ has a strong but natural dependence on the \emph{maximum} variance $\sup_{f \in \F} Var(f)$.
However, when the variances of the members of $\F$ are highly heterogenous, this leads to a significant loss of accuracy in the approximation of functions with variance much smaller than the maximum (i.e., we obtain a ``blurred'' approximation of functions $f^{\prime}$ with $Var(f^{\prime}) \ll \sup_{f \in \F} Var(f)$). 
We propose an intuitive solution to achieve a higher granularity in the approximation: we partition $\F$ into $t \geq 1$ subsets $\{ \F_{j} , j \in [1,t] \}$ with $\bigcup_{j} \F_j = \F$, such that functions with similar variance belong to the same subset $\F_{j}$; 
this allows to control the supremum deviations $\sd(\F_{j},\sample)$ for each $\F_{j}$ separately, exploiting the fact that the maximum variance is now computed on each subset $\F_{j}$ instead on the entire set $\F$. 
This idea leads to sharp \emph{non-uniform bounds} that are locally valid for each subset $\F_{j}$ of $\F$, and it is the main motivation and intuition behind empirical peeling. 

The idea of computing a stratification of the set of functions under consideration is at the core of peeling, an important technique in the study of fine properties of empirical processes, extensively studied in statistical learning theory~\cite{van1996weak,geer2000empirical,bartlett2005local,boucheron2013concentration}. 
However, the issue with existing peeling techniques is that the partition $\{ \F_{j} \}$ either relies on strong assumptions about $\F$, or depends on information computed from the sample $\sample$; this latter approach incurs in non-trivial issues due to the dependency between the bounds to $\sd(\F_{j},\sample)$ and $\{ \F_{j} \}$, since both are estimated on the same data $\sample$. 
For this reason, available methods are often loose (e.g., with bounds featuring very large constants) and thus received scant application in practical scenarios. 
Instead, the main idea behind \emph{empirical peeling} is to use an \emph{independent} sample $\sample^{\prime}$ to partition the set $\F$. 
This simple but effective idea significantly simplifies the analysis as it resolves the above mentioned dependency issues, with minimal additional work (as $\sample^{\prime}$ can be taken to be much smaller than $\sample$). 

Before discussing efficient and effective procedures to partition $\F$, we present a result providing the probabilistic guarantees at the core of empirical peeling, in which we assume the partitioning is fixed \emph{before} drawing the sample $\sample$. 
This improved bound (Theorem~\ref{thm:bound_dev} below) on the supremum deviations holds for families $\F$ of functions $f$ with value in $[0,1]$, building on Monte Carlo Rademacher Averages introduced in Section~\ref{sec:rad_ave}.
We use this result in \algname\ to obtain sharp data-dependent non-uniform bounds on the approximation of \bc. 

Before stating Theorem~\ref{thm:bound_dev}, we remark that it is based on  
a novel tight probabilistic bound for the concentration of the \nmcera\ for general families of functions (Theorem~\ref{thm:neweraboundhypercube}) that scales with the \emph{empirical wimpy variance} of $\F$ (defined below). 
Our novel bound proves that the \nmcera\ is a \emph{sub-gaussian} random variable~\cite{boucheron2013concentration}, therefore it satisfies concentration bounds that are \emph{uniformly sharper} than state-of-the-art sub-gamma results (i.e., Theorem~11~of~\cite{pellegrina2020sharper} and Theorem~5~of~\cite{cousins2020sharp}). 

The \emph{empirical wimpy variance} $\ewvar_{\F}(\sample)$ of $\F$ on $\sample$ is defined~\cite{boucheron2013concentration} as
\begin{align*}
\ewvar_{\F}(\sample) = \sup_{f \in \F}  \frac{1}{m}\sum_{i=1}^{m} \pars{f(s_i)}^2 . 
\end{align*} 

\begin{restatable}{theorem}{neweraboundhypercube}
\label{thm:neweraboundhypercube}
For $c,m \geq 1$, let $\vsigma \in \{-1 , 1 \}^{c \times m} $ be an $c \times m$ matrix of Rademacher random variables, such that $\vsigma_{j,i} \in \{-1 , 1 \}$ independently and with equal probability. 
Then, 
with probability $\geq 1- \delta$ over $\vsigma$, it holds
\begin{align*}
\era \leq \mera + \sqrt{\frac{4 \ewvar_{\F}(\sample) \ln \pars{ \frac{1}{ \delta } } }{ c m }} . 
\end{align*}
\end{restatable}

The proof of Theorem~\ref{thm:neweraboundhypercube}, 
deferred to the Appendix, 
leverages a concentration inequality for functions uniformly distributed on the binary hypercube (see Section~5.2 of~\cite{boucheron2013concentration}).

We are now ready to state the main technical result of this Section. (The proof is in Section~\ref{sec:proofs}.) 

\begin{restatable}{theorem}{thmboundsupdev}
\label{thm:bound_dev}
Let $\F = \bigcup_{j=1}^t \F_j$ be a family of functions with codomain in $[0,1]$. 
Let $\sample$ be a sample of size $m$ taken i.i.d. from a distribution $\probdist$. 
Denote $\nu_{\F_j}$ such that $\sup_{f \in \F_j} Var(f) \leq \nu_{\F_j}$.
For any $\delta \in (0,1)$, define 
\begin{align*}
& \tilde{\rade}_j \doteq \erade^{c}_{m}(\F_j, \sample, \vsigma) + \sqrt{ \frac{4  \ewvar_{\F_j}(\sample) \unionbdelta }{c m} } , \\
& \rade_j \doteq \tilde{\rade}_j + \frac{\unionbdelta}{m} + \sqrt{ \pars{\frac{\unionbdelta}{m}}^2 + \frac{2\unionbdelta \tilde{\rade}_j }{ m } } , \\
& \varepsilon_{\F_j} \doteq 2\rade_j + \sqrt{\frac{2 \unionbdelta
    \left( \nu_{\F_j} + 4\rade_j \right)}{m}}
        + \frac{ \unionbdelta}{3m} \numberthis \label{eq:epsrade} .
\end{align*}
With probability at least $1-\delta$ over the choice of $\sample$ and $\vsigma$, it holds
$\sd(\F_j, \sample) \le \varepsilon_{\F_j}$ for all $j\in[1,t]$.
\end{restatable}

From Theorem~\ref{thm:bound_dev} we observe that, since each $\nu_{\F_j}$ strongly affects $\varepsilon_{\F_j}$, as it typically dominates~\eqref{eq:epsrade}, partitioning $\F$ according to different stratifications of $\nu_{\F_j}$ is very beneficial to obtain sharp \emph{non-uniform} bounds. 
We remark that recent works based on Monte Carlo Rademacher Averages \cite{pellegrina2020mcrapper,cousins2021bavarian} used bounds that apply to the particular case $t=1$ (without any partitioning of $\F$) to obtain a \emph{uniform} variance-dependent bound that can be very loose for most functions as it ignores any heterogeneity of variances within~$\F$ (see Thereom~3.2~of~\cite{pellegrina2020mcrapper} and Theorem~3.1~of~\cite{cousins2021bavarian}). 

We note that in many cases, appropriate values for variance upper bounds $\nu_{\F_j}$ are not known. 
The following result upper bounds every supremum variance $\sup_{f \in \F_j} Var(f)$ of all sets of functions $\{ \F_j \}$ using the empirical wimpy variances $\ewvar_{\F_j}(\sample)$. This bound conveniently defines sharp data-dependent values of $\nu_{\F_j}$ that we plug in \eqref{eq:epsrade}. 
Our proof is based on the self-bounding properties of the function $\ewvar_{\F_j}(\sample)$, proved by~\cite{leogrinaphdthesis}.

\begin{restatable}{proposition}{varianceupperbounds}
\label{prop:varianceupperbounds}
With probability at least $1-\delta$ it holds, for all $j\in[1,t]$, 
\begin{align}
\sup_{f \in \F_j} Var(f) \leq \nu_{\F_j} \doteq \ewvar_{\F_j}(\sample) + \frac{ \ln\pars{\frac{t}{\delta}}}{ m } + \sqrt{ \pars{\frac{ \ln\pars{\frac{t}{\delta}} }{ m}}^2 + \frac{2 \ewvar_{\F_j}(\sample) \ln\pars{\frac{t}{\delta}} }{ m} } .
\label{eq:valuevarsupperbounds}
\end{align}
\end{restatable}

Theorem~\ref{thm:bound_dev} and Proposition \ref{prop:varianceupperbounds} are easily combined by replacing $4/\delta$ by $5/\delta$ in Theorem~\ref{thm:bound_dev}, and 
$1/\delta$ by $5/\delta$ in \eqref{eq:valuevarsupperbounds}; 
with this adjustment we obtain that both statements hold simultaneously with probability at least $1-\delta$.

\subsection{\algname}
\label{sec:algunif}

In this Section we introduce \algname, our algorithm, based on the contributions of Section~\ref{sec:empiricalpeeling}, 
to compute rigorous approximations of the \bc\ of all nodes in a graph. 

We first describe, in Section~\ref{sec:samplingshortestpaths}, the algorithm to efficiently sample shortest paths, that is at the core of \algname\ to approximate the \bc. We then present \algname\ in Section~\ref{sec:subsilvanalg}. 

\subsubsection{Sampling Shortest Paths}
\label{sec:samplingshortestpaths}
\algname\ works by sampling shortest paths in $G$ uniformly at random and using the fraction of shortest paths containing $v$ as an unbiased estimator of its \bc\ $b(v)$. 
The first estimator following this idea was introduced by \citet{RiondatoK15} (the \texttt{rk} estimator). The idea is to first samples two uniformly random nodes $s,t$, and then a uniformly distributed shortest path $\pi$ between $s$ and $t$. 
With this procedure the probability $\Pr[v \in \pi]$ that a node $v$ is internal to $\pi$ is $\Pr[v \in \pi]=b(v)$. 
A more refined approach was proposed by \cite{RiondatoU14} (the \texttt{ab} estimator), which considers \emph{all} shortest paths between $s$ and $t$ instead of only one, approximating the \bc\ $b(v)$ as the \emph{fraction} of such shortest paths passing through $v$. 
The \texttt{ab} estimator has been shown to provide higher quality approximations than the \texttt{rk} in practice~\cite{cousins2021bavarian}; this is because, intuitively, it updates estimations among all nodes involved in shortest paths between $s$ and $t$, and thus, informally, provides \qt{more information per sample}. 
Computationally, the set $\Pi_{st}$ of shortest paths between $s$ and $t$, required by both the \texttt{rk} and \texttt{ab} estimators, can be obtained in time $\BOi{|E|}$ using a (truncated) BFS, initialized from $s$ and expanded until $t$ is found. 
For the \texttt{rk} estimator, a faster approach based on a \emph{balanced bidirection BFS} was proposed and analysed by~\citet{borassi2019kadabra}: they show that all information required to sample one shortest path between two vertices $s$ and $t$ can be obtained in time $\BOi{|E|^{\frac{1}{2}}}$ with high probability on several random graph models, and experimentally on real-world instances. 
While this approach drastically speeds-up \bc\ approximations via the \texttt{rk} estimator~\cite{borassi2019kadabra}, an analogous extension of this technique to the \texttt{ab} estimator is currently lacking.

Our sampling algorithm extends the balanced bidirection BFS to the \texttt{ab} estimator; this allows to combine superior statistical properties of \texttt{ab} with the much faster balanced bidirection BFS enjoyed by \texttt{rk}. 
Our main idea is that, once the set of all shortest paths $\Pi_{st}$ between $s$ and $t$ is \emph{implicitly} computed by the two BFSs, then it is very efficient to sample multiple shortest paths uniformly at random from $\Pi_{st}$ (while~\cite{borassi2019kadabra} only sampled one shortest path).

\algname\ samples shortest paths with the following procedure: 
\begin{enumerate}
\item sample two uniformly random nodes $s,t$; 
\item performs a balanced bidirection BFS starting from $s$ and $t$, until the two BFSs \qt{meet}; 
\item sample uniformly at random $\lceil \alpha \sigma_{st} \rceil$ shortest paths from the set $\Pi_{st}$ of shortest paths between $s$ and $t$, where $\sigma_{st} = |\Pi_{st}|$ is the number of shortest paths between $s$ and $t$ and $\alpha \ge 1$ a positive constant. 
\end{enumerate}
It is easy to see that the expected fraction of shortest paths sampled using this procedure containing $v$ is equal to the \bc\ $b(v)$ of $v$. 
In particular, for each node $v\in V$ and a bag of shortest paths $\bagsp$ obtained from this sampling procedure, define the function $f_v(\bagsp)$, with 
$f_v(\bagsp) = |\bagsp|^{-1} \sum_{ \pi \in \bagsp } \ind{ v \in \pi } $ where $\ind{v \in \pi} = 1$ if $v$ is internal to the shortest path $\pi \in \bagsp$, $0$ otherwise. 
Consequently, the set of functions we use for \bc\ approximation contains all $f_{v}$ with $v \in V$, so that $\F = \{ f_{v} , v \in V \}$. 
By considering a sample $\sample$ of size $m$ taken as described above, we define the estimate $\ebc(v)$ of the \bc\ $b(v)$ of $v$ as 
\begin{align*}
\ebc(v) = \mu_{\sample}(f_v) = \frac{1}{m} \sum_{\bagsp \in \sample} f_v(\bagsp). 
\end{align*}
We have that $\ebc(v)$ is an unbiased estimator of $\mu_{\probdist}(f_v) = b(v)$, 
so that $\E_{\sample}[\ebc(v)] = b(v)$:
\begin{align*}
\E_{\sample}[\ebc(v)] 
= \E_{\bagsp} \sqpars{ f_v(\bagsp)} 
= \E \sqpars{ \frac{1}{|\bagsp|} \sum_{ \pi \in \bagsp } \ind{ v \in \pi } } 
=  \E \sqpars{ \ind{ v \in \pi } } 
=  \Pr \pars{  v \in \pi } = b(v) . 
\end{align*}
Regarding $\alpha$, from standard Poisson approximation to the balls and bins model~\cite{mitzenmacher2017probability}, we obtain that the expected fraction of shortest paths that are not sampled from the set $\Pi_{st}$ in step (3) is $\sigma_{st} (1-1/\sigma_{st})^{\alpha \sigma_{st}} \approx e^{-\alpha}$. 
Consequently, to ensure that the set of sampled shortest paths well represents $\Pi_{st}$, we set $\alpha$ to $\ln \frac{1}{\lambda}$ where $\lambda$ is a small value (e.g., in practice we use $\lambda = 0.1$). 

\subsubsection{\algname\ algorithm}
\label{sec:subsilvanalg}
\algname\ is based on a \emph{progressive sampling} approach. At a high level, the algorithm works in iterations, and in iteration $i$ \algname\ extracts an approximation of the values $b(v)$ for all $v \in V$ from a sample $\sample_i$, which is a collection of $m_i = |\sample_i|$ randomly sampled bags of shortest paths. Then \algname\ checks whether a suitable \emph{stopping condition} is satisfied. If the \emph{stopping condition} is satisfied, the algorithm reports the achieved approximation. It is important that the stopping condition is satisfied as soon as possible, as each sample is expensive to compute, in particular for large graphs. 

In this section we show how to achieve an \emph{$\varepsilon$-approximation} of the set $BC(G) = \{ b(v) : v \in V \}$, defined as follows. 
\begin{definition}
A set $\tilde{BC}(G) = \{ \tilde{b}(v) : v \in V \}$ is an \emph{$\varepsilon$-approximation} of $BC(G) = \{ b(v) : v \in V \}$ if it holds, for all $v \in V$, that
$\abs{ b(v) - \tilde{b}(v) } \leq \varepsilon$.
\end{definition}

Algorithm~\ref{algo:main} describes the algorithm \algname\ to compute an $\varepsilon$-approximation of $BC(G)$ by employing the techniques introduced in 
Section~\ref{sec:empiricalpeeling}. 
\begin{algorithm}[htb]
\SetNoFillComment%
  \KwIn{Graph $G=(V,E)$; $c , m^{\prime} \geq1$; $\varepsilon,\delta \in (0,1)$.}
  \KwOut{$\varepsilon$-approximation of BC(G) with probability $\ge 1 - \delta$}
  
  $\sample^{\prime} \gets $ \texttt{sampleSPs($m^{\prime}$)}\; \label{alg:firstfirstline}
  $\{ \F_{j} , j \in [1,t] \} \gets $ \texttt{empiricalPeeling($\F , \sample^{\prime} $)}\;  \label{alg:emppeeling}
  $\hat{m} \gets $ \texttt{sufficientSamples($\F , \sample^{\prime} , \delta/2 $)}\; \label{alg:maxsamples}
  $\{ m_i \} \gets $ \texttt{samplingSchedule($\F ,\sample^{\prime} $)}\; \label{alg:schedule}
  \lForAll{$j \in [1,t]$}{  $\varepsilon_{\F_{j}} \gets 1$} \label{alg:initial_eps}
  $i \gets 0$; 
  $\sample_{0} \gets \emptyset$;
  $\vsigma \gets$ empty matrix\; \label{alg:initvars}
  \While{not \texttt{stoppingCond($\varepsilon, \{ \varepsilon_{\F_{j}} \}, \hat{m} ,m_i $) } } { \label{alg:iterationsstart}
	$i \gets i+1$; 
	$d_{m} \gets m_{i} - m_{i-1}$\;
	$\sample_{i} \gets \sample_{i-1} \cup $ \texttt{sampleSPs($d_{m}$)}\; \label{alg:newsamples}
	$\vsigma \gets $ add columns $\{$\texttt{sampleRrvs($d_{m} , c$)}$\}$ to $\vsigma$\; \label{alg:newsamplessigma}
	$\tilde{b} , \tilde{r} , \{ \nu_{\F_j} \} \gets $ \texttt{updateEstimates($\sample_{i} , \vsigma , \{ \F_{j}\}$)}\; \label{alg:updateestimates}
	\ForAll{$j \in [1,t]$}{
	$\erade^{c}_{m_{i}}(\F_{j}, \sample_{i}, \vsigma) \gets \frac{1}{c} \sum_{ x=1 }^{c} \max_{v \in V , f_{v} \in \F_{j}} \{ \tilde{r}(v , x) \} $\; \label{alg:computemcera}
	$\varepsilon_{\F_{j}} \gets $ \texttt{epsBound($\erade^{c}_{m_{i}}(\F_{j}, \sample_{i}, \vsigma) , \nu_{\F_j} , \delta / 2^{i+1}$)}\; \label{alg:epsilons}
	}
  }
  \textbf{return} $\tilde{b}$\; \label{alg:secondlastline}
  \caption{\algname}\label{algo:main}
\end{algorithm}
\algname\ can be logically divided into two phases: in the first phase (lines~\ref{alg:firstfirstline}-\ref{alg:schedule}), \algname\ generates a sample $\sample^{\prime}$ that is used for empirical peeling (Section~\ref{sec:empiricalpeeling}) to partition $\F$ into $t$ subsets $\{\F_{j} , j \in [1,t]\}$. 
The second phase (lines~\ref{alg:initial_eps}-\ref{alg:secondlastline}) describes the main operations of the algorithm to approximate the \bc. 

We start by describing the first phase.
In line~\ref{alg:firstfirstline}, the sample $\sample^{\prime}$ is generated using the procedure \texttt{sampleSPs($m^{\prime}$)}, which samples uniformly at random $m^{\prime}$ bags of shortest paths (following the procedure described in \Cref{sec:samplingshortestpaths}), where $m^{\prime} \geq 1$ is given in input. 
After obtaining $\sample^{\prime}$, \algname\ uses the procedure \texttt{empiricalPeeling} (line~\ref{alg:emppeeling}) to partition $\F$ into $t$ subsets. \texttt{empiricalPeeling} splits $\F$ according to an estimate of the variance of its members (we describe a simple but very effective  partitioning scheme in more details at the end of this section). 
Using $\sample^\prime$, \algname\ computes an upper bound $\hat{m}$ to the number of samples required to obtain an $\varepsilon$ approximation with the procedure \texttt{sufficientSamples} (line~\ref{alg:maxsamples}), which is described in Section~\ref{sec:upperboundsamplesabs}. 
This bound $\hat{m}$ guarantees that any sample $\sample$ with size $\geq \hat{m}$ provides an $\varepsilon$ approximation with probability $\geq 1 - \delta/2$. 
Then, the algorithm fixes a sampling schedule given by the values of $m_i$ using the function \texttt{samplingSchedule} (line~\ref{alg:schedule}). We note that arbitrary schedules can be used, but we discuss in Section~\ref{sec:experiments} a data-dependent scheme that leverages $\sample^\prime$. 

We now describe the second phase of the algorithm. 
First, in line~\ref{alg:initial_eps} it initializes all values of $\varepsilon_{\F_{j}}$ to 1.
Then, in line~\ref{alg:initvars}, the other variables used by \algname\ are initialized: $i$ is the index of the iteration, while $m_{i}$ is the size of the sample $\sample_{i}$ considered at the $i$-th iteration. The algorithm initializes $\vsigma$ as an empty matrix, needed at every iteration to compute the \nmcera\ (Section~\ref{sec:rad_ave}). 
In line~\ref{alg:iterationsstart}, the iterations of \algname\ begin, which terminate according to the stopping condition \texttt{stoppingCond} (defined below). 
In every iteration of the while loop, \algname\ performs the following operations: first, it increments $i$; at
the $i$-th iteration, it generates $d_{m} = m_{i} - m_{i-1}$ new samples (line~\ref{alg:newsamples}) using \texttt{sampleSPs($d_{m}$)}, adding them to the set $\sample_{i-1}$ to obtain $\sample_{i}$. 
Then, it extends $\vsigma$ (line~\ref{alg:newsamplessigma}) adding $d_{m}$ columns (each composed by $c$ rows), so that $\vsigma \in \{0 , 1\}^{c \times m_i}$ in order to consider a sample $\sample_{i}$ of size $m_i$. Such columns are generated with the procedure \texttt{sampleRrvs($c , d_{m}$)} that samples a $c \times d_{m}$ matrix in which each entry is a Rademacher r.v. (Section~\ref{sec:rad_ave}). 
\algname\ updates all estimates needed for the approximation (line~\ref{alg:updateestimates}) using the procedure \texttt{updateEstimates}. 
This procedure uses the sample $\sample_{i}$, the matrix $\vsigma$, and the partition $\{ \F_{j} \}$ to compute three quantities: 
$\tilde{b}$ is a vector of $|V|$ components containing the estimates $\tilde{b}(v)$ of $b(v)$ for all $v\in V$; 
$\tilde{r}$ is a matrix of $|V| \times c$ components, in which each entry $\tilde{r}(v , x)$ is defined as the estimated \nmcera\ for the function $f_{v}$ using the $x$-th row $\vsigma_{x,\cdot}$ of $\vsigma$, such that $\tilde{r}(v , x) = \erade^{c}_{m_{i}}(\{f_{v}\}, \sample_{i}, \vsigma_{x,\cdot})$; these values are required to compute the \nmcera\ of each set $\F_{i}$. 
Then, the set $\{ \nu_{\F_j} \}$ contains 
probabilistic upper bounds to supremum variances $\sup_{f \in \F_j} Var(f)$, such that $\sup_{f \in \F_j} Var(f) \leq \nu_{\F_j}$;
we take each $\nu_{\F_j}$ as in \eqref{eq:valuevarsupperbounds} (replacing $1/\delta$ by $2^{i+1}5/\delta$ for reasons discussed below). 
While we describe \texttt{updateEstimates} as a separate procedure executed after the creation of $\sample_{i}$ for ease the presentation, in practice all of these quantities can be updated \emph{incrementally} as every new sample is added to $\sample_{i}$. 
More precisely, the algorithm increases $\tilde{b}(v)$ for all nodes $v \in \pi$ and for all $\pi$ in each sample $\bagsp$, and similarly $\tilde{r}(v,j)$ all $j\in [1,c]$; 
analogously, it updates $\ewvar_{\F_{j}}(\sample_i)$ for all $j$ as each new sample is obtained. 
All these operations are done in $\BOi{c |\bagsp| D}$ time per sample, where $D$ is the vertex diameter of the graph, since $|\pi| \leq D , \forall \pi \in \bagsp$. 
Furthermore, every sample generated within \texttt{sampleSPs} can be sampled and processed in parallel with minimal synchronization. 
After $\sample_{i}$ is created and processed, the algorithm computes (line~\ref{alg:computemcera}), for all partitions $\F_{j}$ of $\F$, the \nmcera\ $\erade^{c}_{m_{i}}(\F_{j}, \sample_{i}, \vsigma)$ using the values stored in $\tilde{r}$. 
Then, it computes (line~\ref{alg:epsilons}) a probabilistic upper bound $\varepsilon_{\F_{j}}$ to the supremum deviation $\sd(\F_{j} , \sample)$ for each partition $\F_{j}$
using the function \texttt{epsBound}. This function returns $\varepsilon_{\F_{j}}$ from~\eqref{eq:epsrade} replacing $4/\delta$ by $2^{i+1}5/\delta$. 
This value of $\delta$ takes into account the fact that we want simultaneous guarantees for all iterations of the algorithm and for all the probabilistic estimates of $\nu_{\F_j}$, as we formally prove with \Cref{prop:main}. 
\algname\ continues to iterate until the stopping condition \texttt{stoppingCond} is true: since we are interested in an $\varepsilon$-approximation, \texttt{stoppingCond} checks that $\varepsilon \geq \varepsilon_{\F_{j}} , \forall j \in [1,t]$, or that $m_i \geq \hat{m}$. 
When \texttt{stoppingCond} is true, \algname\ returns the approximation $\tilde{b}$ (line~\ref{alg:secondlastline}).

The following result establishes the guarantees provided by \algname. 
The proof (Section~\ref{secappx:mainalgproofs}) follows from contributions of Section~\ref{sec:empiricalpeeling} and the samples bound we formally describe in Section~\ref{sec:upperboundsamplesabs}.

\begin{restatable}{proposition}{silvancorrect}
\label{prop:main}
With probability $\ge 1 - \delta$, the output $\tilde{b}$ of \algname\ is a $\varepsilon$-approximation of $BC(G)$.
\end{restatable}

We now describe a simple but effective criteria to partition $\F$, implementing the \texttt{empiricalPeeling} method. 
First, we denote with $\tilde{w}_v$ the estimated wimpy variance of the function $f_v$ on sample $\sample^\prime$ as 
\begin{align*}
\tilde{w}_v = \frac{1}{\abs{\sample^\prime}} \sum_{ \bagsp \in \sample^\prime } \pars{ f_v \pars{\bagsp} }^2 .
\end{align*}
We assign each function $f_v$ for each node $v \in V$ to the set $\F_j$ with index $j =  \lceil \log_{a}( \min \{\tilde{w}_v^{-1} , |\sample^\prime|\} )  \rceil $ for a constant $a > 1$. 
Intuitively, this allows to split $\F$ into (at most) $t = \lceil \log_{a}(|\sample^\prime| )  \rceil$ partitions, such that each set $\F_j$ groups functions with variances in $[ 1/a^{j+1} , 1/a^{j} ]$, therefore within a multiplicative factor $a$. 
Our main intuition is that the empirical wimpy variances $\ewvar_{\F_j}(\sample)$ control the accuracy of the bounds on the supremum deviations $\supdevj$ (through $\hat{\nu}_j$ in \Cref{thm:bound_dev} and \Cref{prop:varianceupperbounds}); 
this partitioning scheme 
fully exploits the non-uniform variance-dependent bounds at the core of \algname\ since the empirical 
wimpy variances 
$\ewvar_{\F_j}(\sample)$ are approximated by $\ewvar_{\F_j}(\sample^\prime)$ and are $\ewvar_{\F_j}(\sample^\prime) \leq 1/a^{j}$, which decrease exponentially with~$j$. 

\subsection{Upper Bound to the Number of Samples}
\label{sec:upperboundsamplesabs}

In this section we prove new upper bounds to the sufficient number of samples to obtain accurate approximations of \bc.
In Section \ref{secsub:boundabsolutesamples} we prove a new bound on the number of samples required to obtian an absolute $\varepsilon$ approximation.
Our proof is based on a novel connection between key results from combinatorial optimization~\citep{toth1990knapsack} and fundamental concentration inequalities~\cite{boucheron2013concentration}. 
Then, in Section \ref{secsub:reldevbounds} we show that the same technique can be applied to guarantee that, from a sample of a given size, the estimates of the \bc\ satisfy tight relative deviation bounds.
In Section \ref{secsub:empiricalboundsavgspl} we provide empirical bounds on the average shortest path length, a key parameter governing our novel bounds.
For ease of presentation, we defer the proofs to the Appendix (Section \ref{secappx:upperbounds}). 

\subsubsection{Absolute Approximation}
\label{secsub:boundabsolutesamples}

The following result shows an improved bound to the number of samples to obtain an absolute $\varepsilon$ approximation. 
We obtain a \emph{distribution-dependent} bound, since it takes into account the \emph{maximum} variance of the \bc\ estimators.
In addition, our bound scales with the \emph{average shortest path length} $\rho$, a key graph characteristic not considered by previous results.
A first observation is that the average shortest path length is equal to the sum of \bc\ $b(v)$ over all $v \in V$. 
If we denote $\Pi$ as the set of shortest paths of the graph $G$, and $|\pi|$ the number of internal nodes to the path $\pi \in \Pi$, then it holds
\begin{align*}
\rho = \frac{1}{|\Pi|} \sum_{\pi \in \Pi} |\pi| = \sum_{v \in V} b(v) .
\end{align*}
The key intuition behind our new bounds (Theorems~\ref{thm:lowerboundsamplesaboslute} and ~\ref{thm:reldeviationbounds}, formally stated below) is
that the \bc\ measure satisfies a form of negative correlation among vertices: 
the existance of a node $v$ with high \bc\ $b(v)$ constraints the sum of the \bc\ of all other nodes to be at most $\rho - b(v)$; intuitively, this means that the number of vertices of $G$ with high \bc\ cannot be arbitrarily large. 
Moreover, we assume that the maximum variance $\max_{v \in V} Var_\probdist(f_v)$ of the \bc\ estimators is at most $\hat{\nu}$, rather than using the worst-case bound of $\max_{v \in V} Var_\probdist(f_v) \leq 1/4$. 
Consequently, the estimates $\ebc(v)$ which can incur in large deviations w.r.t. to their expected values $b(v)$ may not only be (na\"ively) bounded by the number $n = |V|$ of vertices of $G$, but are tightly constrained by the parameters $\rho$ and $\hat{\nu}$. 
Building on this idea, we are able to characterize an upper bound to the probability of not obtaining an absolute $\varepsilon$ approximation from a sample of size $m$, where this probability is taken over the space of graphs with average shortest path length at most $\rho$, and such that the maximum estimator variance is at most $\hat{\nu}$. 
The key technical tool we use to achieve this is to express this probability as an instance of a Bounded Knapsack Problem~\citep{toth1990knapsack}; we explicitly optimize this combinatorial problem through different relaxations, leading to sharp upper bounds. 
We remark that taking advantage of these additional constraints on the space of possible graphs is in strong contrast with the best available results, based on worse-case analyses leading to more conservative guarantees. 

\begin{restatable}{theorem}{boundonnumberofsamples}
\label{thm:lowerboundsamplesaboslute}
Let $\F = \{ f_v , v \in V \}$ be a set of functions from a domain $\X$ to $[0,1]$. 
Let a distribution $\probdist$ such that $\E_{ \bagsp \sim \probdist }[f_v(\bagsp)] = b(v)$.  
Define $\hat{\nu}\in (0,1/4]$, $\rho \geq 0$
such that 
\begin{align*}
\max_{v \in V} Var_\probdist(f_v) \leq \hat{\nu} \text{ , and } \sum_{v \in V} b(v) \leq \rho . 
\end{align*}
Fix $\delta \in (0,1)$, $\varepsilon \in (0 , 1)$, and define the functions $\varx = x(1-x)$ and $h(x) = (1+x)\ln(1+x)-x$ for $x \geq 0$. 
Let $\hat{x}_1$, $\hat{x}_2$, and $\hat{x}$ be 
\begin{align*}
& \hat{x}_1 = \inf \brpars{ x : \frac{1}{2} - \sqrt{\frac{\varepsilon}{3} - \frac{\varepsilon^2}{9}} \leq x \leq \frac{1}{2} , \varx h \pars{ \frac{\varepsilon}{\varx} } \leq 2 \varepsilon^2 }  , \enspace
 \hat{x}_2 = \frac{1}{2} - \sqrt{\frac{1}{4} - \hat{\nu}} ,  \enspace
 \hat{x} = \min\{ \hat{x}_1 , \hat{x}_2 \} .  
\end{align*}
Let $\sample$ be an i.i.d. sample of size $m \geq 1$ taken from $\X$ according to $\probdist$ such that 
\begin{align}
m \geq \sup_{ x \in (0 , \hat{x}] } \brpars{ \frac{ \ln \pars{ \frac{ 2\rho }{ x \delta } }}{ \varx h \pars{ \frac{\varepsilon}{ \varx } } }} .
\label{eq:lowerboundsamplesaboslute}
\end{align}
With probability $\geq 1 - \delta$ over $\sample$, it holds $\supdev \leq \varepsilon$. 
\end{restatable}

To make the bound~\eqref{eq:lowerboundsamplesaboslute} more interpretable, we make the following observations. 
First, while the r.h.s. of~\eqref{eq:lowerboundsamplesaboslute} is implicit, and difficult to express in closed form, it can be easily computed with a numerical procedure 
(e.g., $\hat{x}_1$ can be easily obtained with a binary search in the interval $[ 1/2 - \sqrt{\varepsilon/3 - \varepsilon^2/9} , 1/2 ]$, exploiting the convexity of $\varx h(\varepsilon/\varx)$ in $(0,1)$). 
Then, we remark that for typical values of the parameters (e.g., $\delta , \varepsilon \leq 0.25 \leq \rho $), the maximum  
of~\eqref{eq:lowerboundsamplesaboslute} 
is attained at $x^\star \approx \hat{x}\leq \hat{x}_2 $;
furthermore, we note that $h(x) \geq x^2/2(1+x/3)$ for $x \geq 0$,
and $\hat{x}_2 \geq \hat{\nu} = g(\hat{x}_2)$. 
Combining all these facts, a very accurate approximation of $m$ in \eqref{eq:lowerboundsamplesaboslute} is   
\begin{align*}
m 
\approx \frac{2 \hat{\nu} + \frac{2}{3}\varepsilon }{\varepsilon^2} \pars{ \ln \pars{ \frac{2 \rho}{\hat{\nu}} } + \ln \pars{ \frac{1}{\delta} } } \in \BO{\frac{\hat{\nu} + \varepsilon}{\varepsilon^2} { \ln\pars{ \frac{\rho}{ \delta \hat{\nu}} } }} .
\end{align*}

Since $\rho$ corresponds to (an upper bound to) the \emph{average} number of internals nodes in shortest paths of $G$, it is immediate to conclude that $\rho$ cannot exceed the vertex diameter $D$ (the maximum shortest path length). 
From these observations, it is natural to compare our new bound with the state-of-the-art result based on the VC-dimension presented by~\citet{RiondatoK15}; they show that $\BOi{\ln\pars{ D / \delta }/\varepsilon^2 }$ samples are enough for obtaining an $\varepsilon$-approximation with probability at least $1-\delta$. 
Similarly to their result, our new bound is independent of the size of the graph (e.g., $|V|$ or $|E|$), and 
essentially recovers it since $\hat{\nu} \leq 1/4$ and $\rho \leq D$, 
but provides much tighter results when smaller bounds to $\hat{\nu}$ and $\rho$ are available: it is \emph{distribution-dependent}, rather than \emph{distribution-free}. 
Interestingly, in many real-world graphs the average shortest path length is typically very small (a phenomena observed in small-world networks~\cite{watts1998collective} for which $\rho \in \BOi{\log |V|}$), 
and often much smaller than the diameter. 

Since $\rho$ is usually not known in advance, 
in Section~\ref{secsub:empiricalboundsavgspl} we prove that, given a (not necessarily tight) upper bound to $D$, $\rho$ can be sharply estimated as the \emph{average} number of internal nodes of the shortest paths in a sample $\sample$, resulting in a very efficient data-dependent bound. 
In addition, $\hat{\nu}$ can be sharply upper bounded by Proposition~\ref{prop:varianceupperbounds} (defining $\hat{\nu} \doteq \nu_{\F_j}$ with $t=1$).

As anticipated in Section \ref{sec:algunif}, Theorem~\ref{thm:lowerboundsamplesaboslute} is not only of theoretical interest, but is used in \algname\ to design its sampling schedule (e.g., by upper bounding the number of samples $\hat{m}$ it needs to process). 
\algname\ estimates both $\hat{\nu}$ and $\rho$ using $\sample^{\prime}$, and then plug them in Theorem~\ref{thm:lowerboundsamplesaboslute} to obtain $\hat{m}$; these operations are part of the procedure \texttt{sufficientSamples} (see Algorithm \ref{algo:main}).  
These key results contributes to make \algname\ processing a significantly lower number of samples to obtain an approximation of desired quality in many practical cases, as we will show in our experimental evaluation.  
Moreover, we note that Thereom~\ref{thm:lowerboundsamplesaboslute} applies to all available (unbiased) estimators of the \bc, such as the ones employed by \bavarian~\cite{cousins2021bavarian}, providing a tighter upper bound to the sufficient number of random samples required by different estimators as well.

\subsubsection{Relative Bounds}
\label{secsub:reldevbounds}

In this section we extend Theorem~\ref{thm:lowerboundsamplesaboslute} to obtain new sharp relative deviation bounds; these bounds are very useful to derive sharp confidence intervals on the values of \bc\ $b(v)$ from a random sample, and are particularly very accurate for smaller values of $b(v)$. 
\begin{restatable}{theorem}{relativebounds}
\label{thm:reldeviationbounds}
Let $\F$, $\probdist$, $g$, $\hat{\nu}$, and $\rho$ as in Theorem \ref{thm:lowerboundsamplesaboslute}, and define $n=|V|$. 
Denote $\sample$ as an i.i.d. sample of size $m \geq 1$ from $\probdist$. 
It holds, with probability $\geq 1-\delta$ and for all $v \in V$,
\begin{align*}
| b(v) - \tilde{b}(v) | \leq \sqrt{ \frac{ 2 \min \{ g (b(v)) , \hat{\nu} \} \ln \pars{ \frac{ 4}{\delta} \min \brpars{ \frac{\rho}{b(v)} , n } } }{ m } } + \frac{ \ln \pars{ \frac{ 4}{\delta} \min \brpars{ \frac{\rho}{b(v)} , n } } }{ 3m } 
\doteq d_r(b(v)). 
\end{align*}
\end{restatable}
\Cref{thm:reldeviationbounds} yields sharp upper and lower bounds to $b(v)$ (for all $v \in V$), which are easily obtained (e.g., with a binary search) from their empirical estimates $\tilde{b}(v)$, as evident from the following Corollary. 
\begin{corollary}
Assume the setting of Theorem~\ref{thm:reldeviationbounds}. 
It holds with probability $\geq 1-\delta$
\begin{align*}
\min\brpars{ x \in [0 , \tilde{b}(v) ] : \tilde{b}(v) \leq x + d_r(x) } 
\leq b(v) 
\leq \max\brpars{ x \in [\tilde{b}(v) , 1] : x \leq \tilde{b}(v) + d_r(x) } , \enspace  \forall v \in V .
\end{align*}
\end{corollary}
We remark that the above bound is significantly more accurate than bounds obtained using standard tools (i.e., combining Bernstein's inequality and a union bound over $n$ events) for most interesting values of $b(v)$ (more precisely, when $b(v) \geq 2\rho / n $ since $\ln(4\rho / (b(v) \delta)) \ll \ln(2n/\delta)$).

\subsubsection{Empirical Bounds to the Average Shortest Path Length}
\label{secsub:empiricalboundsavgspl}
In this section we present sharp empirical bounds on the average shortest path length, a key quantity involved in the sample bounds introduced in Section~\ref{sec:upperboundsamplesabs}.
The first result (Proposition~\ref{thm:rhoestimation}) is based on the application of Bernstein's inequality~\cite{boucheron2013concentration}, while the second (Proposition~\ref{thm:rhoestimationempiricalbernstein}) uses the Empirical Bernstein Bound introduced by \citet{maurer2009empirical}.

\begin{restatable}{proposition}{avgsplbernstein}
\label{thm:rhoestimation}
Let $D$ be the vertex diameter of the graph $G$. Let an i.i.d. sample $\sample$ of size $m$, and denote $\tilde{\rho} = \sum_{v \in V} \tilde{b}(v)$. Then, for a fixed $\delta \in (0,1)$, it holds with probability $\geq 1 - \delta$
\begin{align*}
\sum_{v \in V} b(v) \leq \rho \doteq \tilde{\rho} + \sqrt{ \frac{5}{3} \pars{\frac{ D \ln\pars{ \frac{1}{\delta} } }{m}}^2  + \frac{ 2 D \tilde{\rho} \ln\pars{ \frac{1}{\delta} } }{m} } +  \frac{4  D \ln\pars{ \frac{1}{\delta} } }{3m} .
\end{align*}
\end{restatable}

The following gives typically slightly sharper bounds than Proposition~\ref{thm:rhoestimation} since it involves an empirical estimator $\Lambda(\sample)$ of the variance of $\sum_{v \in V} b(v) $. 

\begin{restatable}{proposition}{avgsplbernsteinemp}
\label{thm:rhoestimationempiricalbernstein}
Assume the setting of Proposition~\ref{thm:rhoestimation}
with $\sample = \{ \bagsp_1 , \dots , \bagsp_m \}$. 
Define $\Lambda(\sample)$ as 
\begin{align*}
\Lambda(\sample) = \frac{1}{m (m-1)} \sum_{ 1 \leq i < j \leq m } \pars{ \sum_{v \in V} f_v(\bagsp_i) - \sum_{v \in V} f_v(\bagsp_j)  } ^2 .
\end{align*}
Then, for a fixed $\delta \in (0,1)$, it holds with probability $\geq 1 - \delta$
\begin{align*}
\rho \leq \tilde{\rho} + \sqrt{ \frac{ 2 \Lambda(\sample) \ln\pars{ \frac{2}{\delta} } }{m} } +  \frac{7  D \ln\pars{ \frac{2}{\delta} } }{3m} .
\end{align*}
\end{restatable}

\subsection{Top-$k$ Approximation}
\label{alg:topk}

In this Section we present \algnametopk, an extension of \algname\ to compute high-quality \emph{relative} approximations of the $k$ most central vertices. 

While in some cases additive approximations, for which we guarantee that $|\tilde{b}(v)- b(v)| \le \varepsilon$ for all $v\in V$, are sufficient, in several practical cases \emph{relative} approximations, for which the desired bound to $|\tilde{b}(v)- b(v)|$ depends on the value $b(v)$, may be more informative. 
Such approximations are particularly relevant for the problem of estimating the \emph{$k$ most central nodes}, as the value of their \bc\ is typically highly skewed. 
In such cases, one may prefer to have relative approximations of the type ``$\tilde{b}(v)$ is within $10\%$ of the value of $b(v)$'' than an additive approximation of the type ``$|\tilde{b}(v)- b(v)|\le 0.01$''  that may be either unnecessarily precise for high values of $b(v)$, or uninformative for low values of $b(v)$.
Furthermore, 
the user needs to only fix $k$ and the relative accuracy, a much more natural choice for exploratory analyses, in which the centrality scores of the top-$k$ nodes are unknown. 
We note that, from a \qt{statistical} point of view, obtaining relative approximations is a challenging problem; in statistical learning theory, it is well known that it is not possible to obtain them efficiently using uniform additive approximations as a proxy~\cite{boucheron2005theory}; this motivates the development of specialized techniques for the task~(e.g., \cite{li2001improved,har2011relative,cortes2019relative}). 

In this section we show that empirical peeling, introduced in Section~\ref{sec:empiricalpeeling}, is naturally suited to the problem of computing relative approximations of the set of top-$k$ central nodes, and can do so progressively and adaptively as samples are processed. 
First, let $v_1 , \dots , v_n$ be the nodes sorted according to their \bc, such that $b(v_i) \geq b(v_{i+1})$. The set $TOP(k)$ of top-$k$ nodes is defined as
$TOP(k) = \brpars{ (v_i , b(v_i)) : i \leq k }  $.
We now define the relative approximation we are interested in obtaining.
\begin{definition}
\label{def:topkapproxguarantees}
For $\eta \in (0,1)$ and $k \geq 1$, a set $\tilde{TOP}(k) = \{ (v , \tilde{b}(v)) \} $ with $v \in V , \tilde{b}(v) \in [0,1]$, is a \emph{$\eta$-relative approximation} of $TOP(k)$ if all the following hold: 
\begin{align}
 & \{ v : (v , \tilde{b}(v)) \in \tilde{TOP}(k) \} \supseteq \brpars{ v : (v , b(v)) \in TOP(k) } , \label{eq:guartopkfirst} \\
 & \abs{ b(v) - \tilde{b}(v) } \leq \eta b(v) , \forall (v , \tilde{b}(v)) \in \tilde{TOP}(k) , \label{eq:guartopksecond} \\
 & b(v) \geq b (v_k) \pars{\frac{1-\eta}{1+\eta} }^2 , \forall v: (v , b(v)) \not\in {TOP}(k) , (v , \tilde{b}(v)) \in \tilde{TOP}(k) .  \label{eq:guartopkthird}
\end{align}
\end{definition}

Informally, \eqref{eq:guartopkfirst} ensures that all nodes in $TOP(k)$ are in the approximation; 
\eqref{eq:guartopksecond} ensures that all the estimates in the approximation are close to the true values of the \bc, within relative accuracy given by $\eta$;  
\eqref{eq:guartopkthird} guarantees that nodes $v$ not in the set $TOP(k)$ of the top-$k$ nodes are in the approximation only if their \bc\ $b(v)$ is not too far from $b(v_k)$, the centrality of the $k$-th node. 

We now discuss how to modify \algname\ to obtain an approximation $\tilde{TOP}(k)$ of $TOP(k)$ with the aforementioned guarantees. 
Assume, at the end of some iteration of \algname, that we have confidence intervals $CI_v= [\ell(v) , u(v)]$ for each $v \in V$, such that $b(v) \in CI_v, \forall v$. 
Such confidence intervals are derived from bounds on supremum deviations: for a node $v$ such that $f_v \in \F_j$ for some $j \in [1,t]$, and assuming that $\varepsilon_{\F_j} \geq \supdevj$, we define
\begin{align*}
&\ell(v) \doteq \ebc(v) - \varepsilon_{\F_j} \text{ and } \\
&u(v) \doteq \ebc(v) + \varepsilon_{\F_j}. 
\end{align*}
Naturally, the validity of such confidence intervals is \emph{probabilistic}, and thus we aim to obtain a $\eta$-relative approximation with high probability.
In order to verify that a given set $\tilde{TOP}(k)$ 
is a $\eta$-relative approximation of $TOP(k)$, 
we inspect the confidence intervals $CI_{v}$ for each candidate $v$ to be included in $\tilde{TOP}(k)$.

\begin{restatable}{proposition}{topkapproxstronger}
\label{prop:topkapproxstronger}
For each $v \in V$, denote the intervals 
$CI_v= [\ell(v) , u(v)]$ with $\ell(v) \leq \tilde{b}(v) \leq u(v)$. 
Let $v_1^\ell , \dots , v_n^\ell$ be the sequence of nodes ordered according to $\ell(\cdot)$, such that $\ell(v_i^\ell) \geq \ell(v_{i+1}^\ell)$. 
Define the set $\tilde{TOP}(k)$~as
$\tilde{TOP}(k) = \{ ( v , \tilde{b}(v) ) : u(v) \geq \ell(v_k^\ell) \} $,
and assume that, for all $(v,\tilde{b}(v)) \in \tilde{TOP}(k)$, 
\begin{align}
\frac{\tilde{b}(v)}{1+\eta} \leq \ell(v) \leq b(v) \leq u(v) \leq \frac{\tilde{b}(v)}{1-\eta} . \label{eq:topkapproxchecksstronger}
\end{align}
Then, $\tilde{TOP}(k)$ 
is a $\eta$-relative approximation of $TOP(k)$.
\end{restatable}

Building on this result, Algorithm~\ref{algo:topk} describes our algorithm \algnametopk\ to compute $\eta$-relative approximations of $TOP(k)$. 

\begin{algorithm}[htb]
\SetNoFillComment%
  \KwIn{Graph $G=(V,E)$; $c , k \geq 1$; $\eta,\delta \in (0,1)$.}
  \KwOut{Relative $\eta$-approximation of top-$k$ central nodes with probability $\ge 1 - \delta$}
  \lWhile{not \texttt{stoppingCondFirst()}}{
  $\sample^{\prime} \gets \sample^{\prime} \cup $ \texttt{sampleSPs(1)}
  } \label{alg:firstfirsttopk}
  $\{ \F_{j} , j \in [1,t] \} \gets $ \texttt{empiricalPeeling($\F , \sample^{\prime} $)}\; \label{alg:peelingtopk}
  $\{ m_i \} \gets $ \texttt{samplingSchedule($\F ,\sample^{\prime} $)}\; \label{alg:scheduletopk} 
  $i \gets 0$; 
  $\sample_{0} \gets \emptyset$;
  $\vsigma \gets$ empty matrix\; \label{alg:initvarstopk}
  \While{not \texttt{stoppingCondTopk()}}{ \label{alg:iterationsstarttopk}
	$i \gets i+1$; 
	$d_{m} \gets m_{i} - m_{i-1}$\;
	$\sample_{i} \gets \sample_{i-1} \cup $ \texttt{sampleSPs($d_{m}$)}\; \label{alg:newsamplestopk}
	$\vsigma \gets $ add columns $\{$\texttt{sampleRrvs($d_{m} , c$)}$\}$ to $\vsigma$\; \label{alg:newsamplessigmatopk}
	$\tilde{b} , \tilde{r} , \{ \nu_{\F_j} \} \gets $ \texttt{updateEstimates($\sample_{i} , \vsigma , \{ \F_{j}\}$)}\; \label{alg:updateestimatestopk}
	\ForAll{$j \in [1,t]$}{
	$\erade^{c}_{m_{i}}(\F_{j}, \sample_{i}, \vsigma) \gets \frac{1}{c} \sum_{ x=1 }^{c} \max_{v \in V , f_{v} \in \F_{j}} \{ \tilde{r}(v , x) \} $\; \label{alg:computemceratopk}
	$\varepsilon_{\F_{j}} \gets $ \texttt{epsBound($\erade^{c}_{m_{i}}(\F_{j}, \sample_{i}, \vsigma) , \nu_{\F_j} , \delta / 2^{i}$)}\; \label{alg:epsilonstopk}
		\lForAll{$ v: f_v \in \F_j $}{
			$[ \ell(v) , u(v) ] \gets [ \ebc(v) - \varepsilon_{\F_{j}} ,  \ebc(v) + \varepsilon_{\F_{j}} ]$ \label{alg:confintervalstopk}
		}
	}
	$\tilde{TOP}(k) \gets \{ (v , \ebc(v) ) : u(v) \geq \ell(v_k^\ell)  \} $\; \label{alg:tosetapproxdef}
  }
  \textbf{return} $\tilde{TOP}(k)$\; \label{alg:topkreturn}
  \caption{\algnametopk}\label{algo:topk}
\end{algorithm}

As \algname, \algnametopk\ is divided in two phases: in the first phase (lines \ref{alg:firstfirsttopk}-\ref{alg:scheduletopk})
it samples $\sample^\prime$, uses it for empirical peeling (line~\ref{alg:peelingtopk}), and defines the sampling schedule (line~\ref{alg:scheduletopk}). 
Then, it obtains the $\eta$-relative approximations using progressive sampling in the second phase (lines~\ref{alg:initvarstopk}-\ref{alg:topkreturn}). 

The first phase of \algnametopk, instead of considering a fixed number of samples $m^{\prime}$ for $\sample^\prime$ (as in Algorithm~1), continues to  draw shortest paths taken at random until at least $k$ distinct nodes have been observed at least a constant number of times (therefore after~$\approx 1/b(v_{k})$ samples); 
when this is verified, the function \texttt{stoppingCondFirst} returns true (line~\ref{alg:firstfirsttopk}) and the generation of $\sample^\prime$ stops.
Following this scheme, the first phase \emph{adapts} to the (unknown) value of $b(v_k)$. 

The second phase of \algnametopk\ is similar to \algname. 
At iteration $i$, after obtaining bounds $\varepsilon_{\F_j}$ on supremum deviations $\sd(\F_j, \sample_i)$ from the sample $\sample_i$, the algorithm defines the confidence intervals $[ \ell(v) , u(v) ]$ w.r.t. $b(v)$ (line \ref{alg:confintervalstopk}). 
Then, it creates the set $\tilde{TOP}(k)$ including all vertices with upper bound $u(v)$ at least $\ell(v_k^\ell) $, where $\ell(v_k^\ell) $ is the $k$-th lower bound (line~\ref{alg:tosetapproxdef}), as defined in Proposition~\ref{prop:topkapproxstronger}. 
To obtain the approximation described by \Cref{def:topkapproxguarantees}, \algnametopk\ outputs $\tilde{TOP}(k)$ when its stopping condition \texttt{stoppingCondTopk} verifies that \eqref{eq:topkapproxchecksstronger} holds for all $(v,\tilde{b}(v)) \in \tilde{TOP}(k)$. 
Note that the algorithm does not need to know $b(v_{k})$ (or $b(v)$ for any $v$), as the left and rightmost inequalities in \eqref{eq:topkapproxchecksstronger} only depend on empirical quantities. 
From the probabilistic guarantees implied by Theorem~\ref{thm:bound_dev} and from Proposition~\ref{prop:topkapproxstronger}, the following result easily follows.

\begin{proposition}
The output of \algnametopk\ is a $\eta$-relative approximation of $TOP(k)$ with probability $\geq 1-\delta$.
\end{proposition}

We remark that this general approach can be adapted easily to other definitions of relative approximations (e.g.,~\cite{li2001improved,cortes2019relative}). 
As we will show in our experimental evaluation, empirical peeling is essential to achieve $\eta$-relative approximations efficiently.

\section{Experiments}
\label{sec:experiments}

We implemented \algname\ and tested it on several real-world graphs. 
In our experimental evaluations we assess the effectiveness of the progressive sampling approach of \algname\ to approximate the \bc\ of all nodes, and
evaluate the performance of \algnametopk\ in approximating the top-$k$ most central nodes.

\paragraph{Experimental Setup}
We implemented \algname\ by extending the \texttt{C++} implementation of \kadabra\ made available from its authors\footnote{\url{https://github.com/natema/kadabra}}. 
All the code was compiled with \texttt{GCC} 8 and run on a machine with 2.30 GHz Intel Xeon CPU, 512 GB of RAM, on Ubuntu 20.04, with a total of 64 cores. 
All experiments were performed using multithreading on all threads. 
Our implementation of \algname, with automated scripts to reproduce all experiments, is available 
online\footnote{\url{https://github.com/VandinLab/SILVAN}}. 
We compare \algname\ with \kadabra, that has been shown~\cite{borassi2019kadabra} to uniformly and significantly outperform previous methods, and with \bavarian~\cite{cousins2021bavarian}, the most recent method for \bc\ approximation. 
When referring to \bavarian, we consider its variant based on progressive sampling (denoted \bavarianp, see Alg.~2 and Sect.~4.2 of~\cite{cousins2021bavarian}) which addresses the same problem solved by \algname\ and \kadabra, and we tested it using all different estimators for the \bc\ presented in~\cite{cousins2021bavarian} 
(called \texttt{rk}, \texttt{ab}, and \texttt{bp}).

~\emph{Graphs.}
We tested \algname\ on $7$ undirected and $11$ directed real-world graphs from SNAP\footnote{\url{http://snap.stanford.edu/data/index.html}} and KONECT\footnote{\url{http://konect.cc/networks/}}, most of them previously analysed by \kadabra\ \cite{borassi2019kadabra} and other previous methods~\cite{RiondatoK15,RiondatoU18,cousins2021bavarian}.
The characteristics of the graphs are described in detail in Table~\ref{tab:graphs}. 

\begin{table}[ht]
  \caption{Statistics of undirected (top section) and directed (bottom section) graphs. ${D}$ is the vertex diameter, $\rho$ is an upper bound of the average shortest path lenth, and $\maxabsf$ is an upper bound of $\max_v \{ b(v) \}$.}
\label{tab:graphs}
  \begin{tabular}{lrrrrr}
    \toprule
    $G$              & $|V|$      & $|E|$      & ${D}$  & $\rho$  & $\maxabsf$ \\
    \midrule
    actor-collaboration     & 3.82e5  & 3.31e7 & 13 & 2.87 & 0.0090  \\
    ca-AstroPh &  1.87e4 & 1.98e5 & 14 & 3.20 & 0.0285 \\
    ca-GrQc &  5.24e3 & 1.44e4 & 17 &  3.51 &  0.0450 \\
	com-amazon     & 3.34e5 & 9.25e5 & 44 & 11.97  & 0.0450  \\
	com-dblp     & 3.17e5 & 1.04e6 & 21 & 6.27  & 0.0162  \\
	com-youtube     & 1.13e6 & 2.98e6 & 20 & 4.68  & 0.2573  \\
	email-Enron & 3.66e4 & 1.83e5 & 11 & 2.78  & 0.0749  \\
	\midrule
	cit-HepPh      & 3.45e4 & 4.21e5 & 12 & 5.35 & 0.1817  \\
        cit-HepTh      & 2.77e4 & 3.52e5 & 13 & 2.10 & 0.1237  \\
        email-EuAll    & 2.65e5 & 4.20e5 & 14 & 0.56  & 0.0121  \\
        p2p-Gnutella31 & 6.25e4 & 1.47e5 & 11 & 2.16  & 0.0071  \\
        soc-Epinions1 & 7.58e4 & 5.08e5 & 14 & 2.11 & 0.0210  \\
        soc-LiveJournal1 & 4.84e6 & 6.90e7 & 16 & 4.58 & 0.0270  \\
        soc-pokec & 1.63e6 & 3.06e7 & 16 & 3.94 & 0.0802  \\
        wiki-Talk & 2.39e6 & 5.02e6 & 9 & 0.26  & 0.0037  \\	
        wiki-topcats & 1.79e6 & 2.85e7 & 9 & 5.87 & 0.0985  \\	
	wiki-Vote & 7.11e3 & 1.03e5 & 7 & 0.66  & 0.0240  \\	
	wikipedia-link-en & 1.35e7 & 4.37e8 & 10 & 3.21 & 0.0300  \\	
  \bottomrule
\end{tabular}
\end{table}

\subsection{Absolute Approximation}
\label{sec:experimentsabsolute}

We first consider the task of computing an $\varepsilon$ absolute approximation to the \bc\ of all nodes.

For every graph, we ran all algorithms with parameter $\varepsilon \in \{ 0.01 , 0.005 , 0.0025 , 0.001 , 0.0005 \}$, chosen to have comparable magnitude to the \bc\ of the most central nodes (i.e., see col. $\hat{\xi}$ of Table~\ref{tab:graphs}); this is required to compute meaningful approximations (i.e., an $\varepsilon$ absolute approximation is useless when the centralities of the most central nodes are much smaller than $\varepsilon$). 
We fix $\delta = 0.05$, and use $c$ Monte Carlo Rademacher vectors with $c=25$ for \algname\ and \bavarian\ (note that $c=k$ in \cite{cousins2021bavarian}).  
We do not show results for other values of $\delta$, as this parameter has minimal impact on the results, due to the use of exponential tail bounds (see $\delta$ in \eqref{eq:epsrade} and \eqref{eq:lowerboundsamplesaboslute}). 
Regarding $c$, we follow \cite{pellegrina2020mcrapper} and \cite{cousins2021bavarian}, that have shown that sharp bounds are obtained even with a low number of Monte Carlo trials, and that there are minimal improvements using $c>30$. 
We ran all algorithms $10$ times and report averages $\pm$ stds. We limit the execution time of each run to $6$ hours; we terminate the algorithm when exceeding this threshold. 

For the empirical peeling scheme of \algname, 
we sample $m^{\prime} = \log(1/\delta)/\varepsilon$ shortest paths 
to generate $\sample^{\prime}$; we note that $m^{\prime}$
always results in a very small fraction of the overall samples analysed by \algname. 
Regarding the sampling schedule followed in the second phase, we 
use $\sample^{\prime}$ 
to identify a minimum number $m_1$ of samples before starting to evaluate the stopping condition. 
To do so, we perform a binary search to identify the minimum $m_1$ such that \eqref{eq:epsrade} (with $R_j = 0$) is not larger than $\varepsilon$; 
this gives an optimistic first guess of the number of samples to process for obtaining an $\varepsilon$-aproximation. 
We then increase each $m_i$ with a geometric progression, such that $m_i = 1.2 \cdot m_{i-1}$. 
While a geometric progression is considered to be optimal~\cite{provost1999efficient}, we note that the procedure \texttt{samplingSchedule} can be implemented with general schedules. 

As described in Section~\ref{sec:upperboundsamplesabs}, \algname\ uses the procedure \texttt{sufficientSamples} to obtain an upper bound $\hat{m}$ to the number of samples to process to obtain an $\varepsilon$ absolute approximation: it does so using $\sample^{\prime}$, computing an upper bound $\rho$ to the average shortest path length (Theorem~\ref{thm:rhoestimationempiricalbernstein}) and an upper bound $\hat{\nu}$ to the suprem variance of the estimators $\sup_{f_v \in \F} Var(f_v)$ ($\hat{\nu} = \nu_{\F_j}$ with $t=1$ in \Cref{prop:varianceupperbounds}). 
Then, \texttt{sufficientSamples} plugs these estimates in Theorem~\ref{thm:lowerboundsamplesaboslute} to compute $\hat{m}$. 
The \texttt{empiricalPeeling} procedure of \algname\ follows the scheme described at the end of Section~\ref{sec:algunif} using $a = 2$.

For the progressive sampling schedule of \bavarian, we use the same geometric progression parameter of \algname\ (equal to $1.2$, analogous to the parameter $\theta$ in \cite{cousins2021bavarian}). 

Figure~\ref{fig:absapprox} shows the results for this set of experiments comparing \algname\ to \kadabra, while Figure~\ref{fig:absapproxbav} shows the results comparing \algname\ to \bavarian\ for the estimator \texttt{ab} (more results in Figure~\ref{fig:absapproxappendixbavab}, and analogous plots for \texttt{rk} and \texttt{bp} in Figures~\ref{fig:absapproxappendixbavrk} and~\ref{fig:absapproxappendixbavbp}, all in Appendix).

\begin{figure*}[ht]
\centering
\begin{subfigure}{.49\textwidth}
  \centering
  \includegraphics[width=\textwidth]{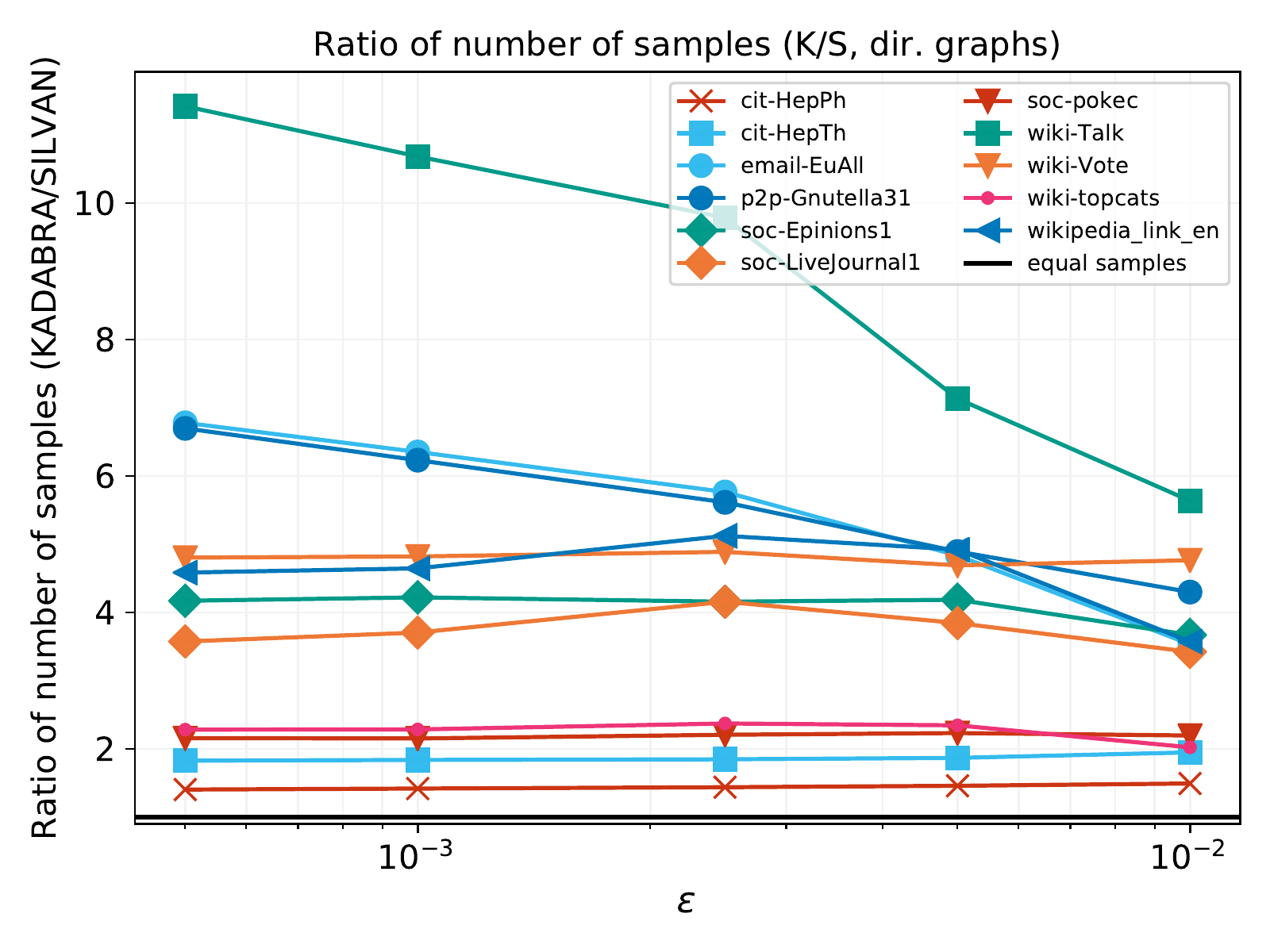}
  \caption{}
\end{subfigure}
\begin{subfigure}{.49\textwidth}
  \centering
  \includegraphics[width=\textwidth]{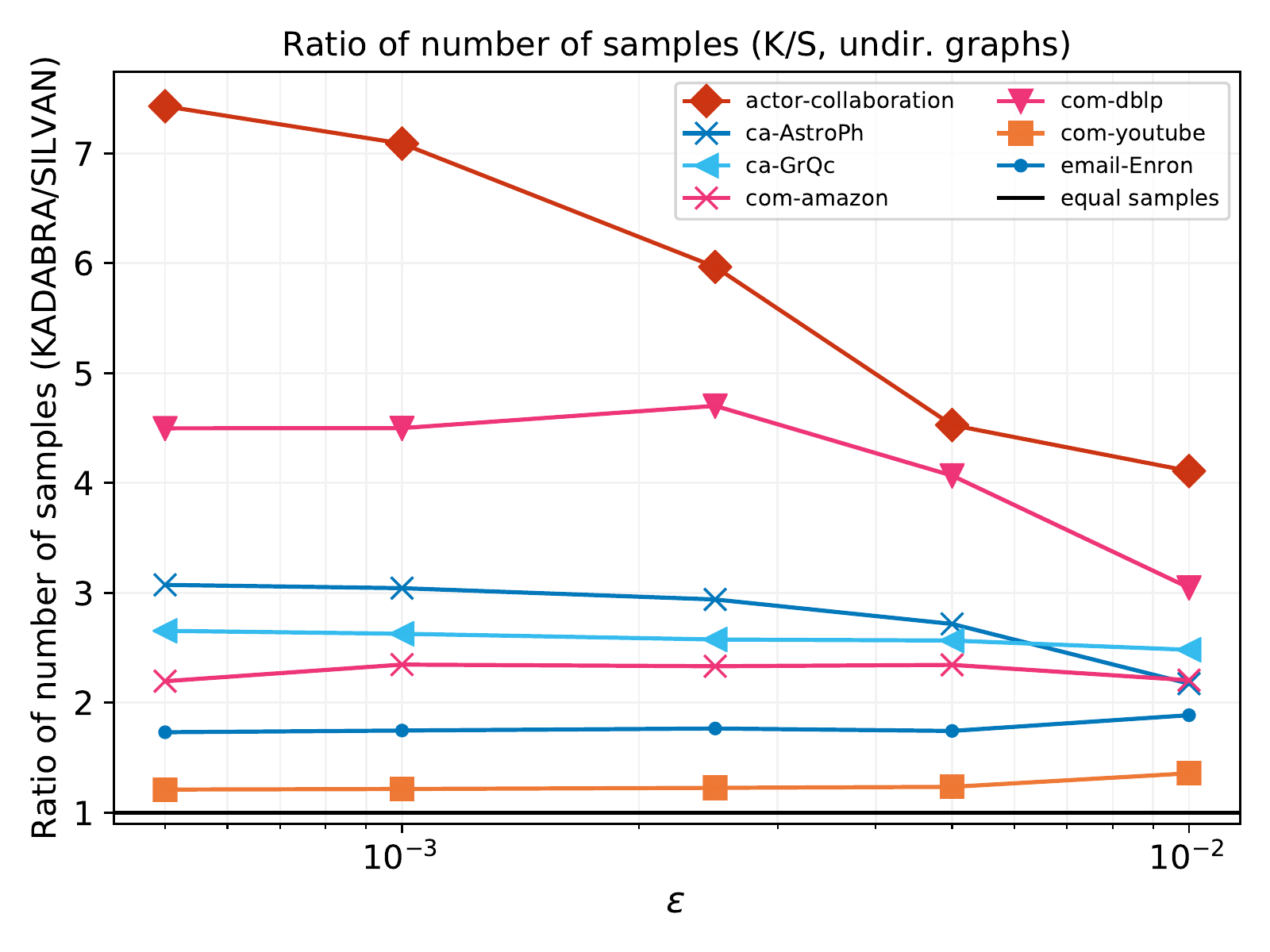}
  \caption{}
\end{subfigure}
\begin{subfigure}{.49\textwidth}
  \centering
  \includegraphics[width=\textwidth]{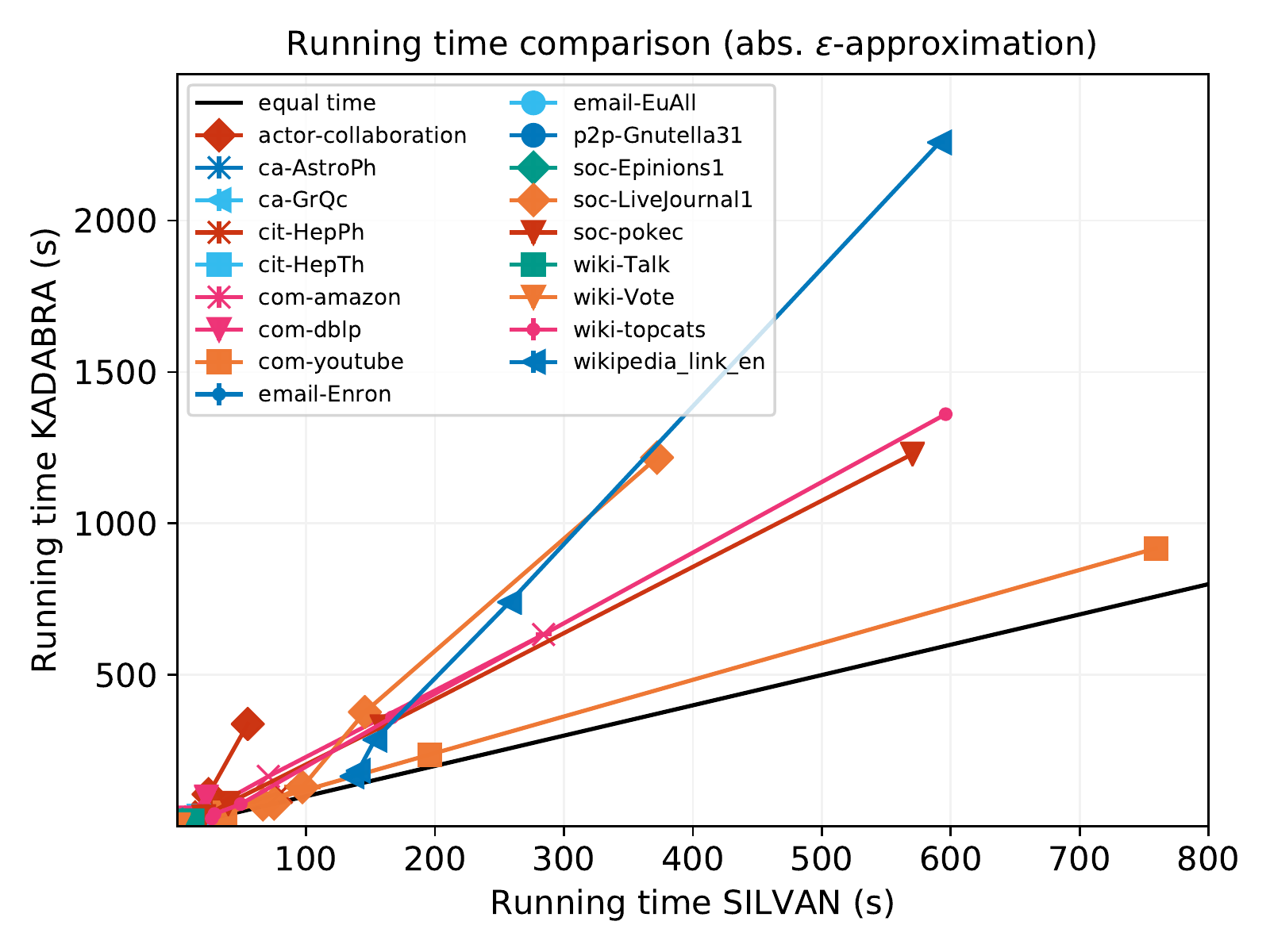}
  \caption{}
\end{subfigure}
\caption{Comparison between the 
performance of \kadabra\ and \algname\ for obtaining absolute 
approximations. 
(a): ratios of the number of samples required by \kadabra\ and the number of samples required by \algname\ for directed graphs.
(b): as (a) for undirected graphs. 
(c): comparison of the running times of \kadabra\ ($y$ axis) and \algname\ ($x$ axis) for all graphs. Additional plots in Figure~\ref{fig:absapproxappendixkad}. 
}
\label{fig:absapprox}
\end{figure*}

\begin{figure*}[ht]
\centering
\begin{subfigure}{.49\textwidth}
  \centering
  \includegraphics[width=\textwidth]{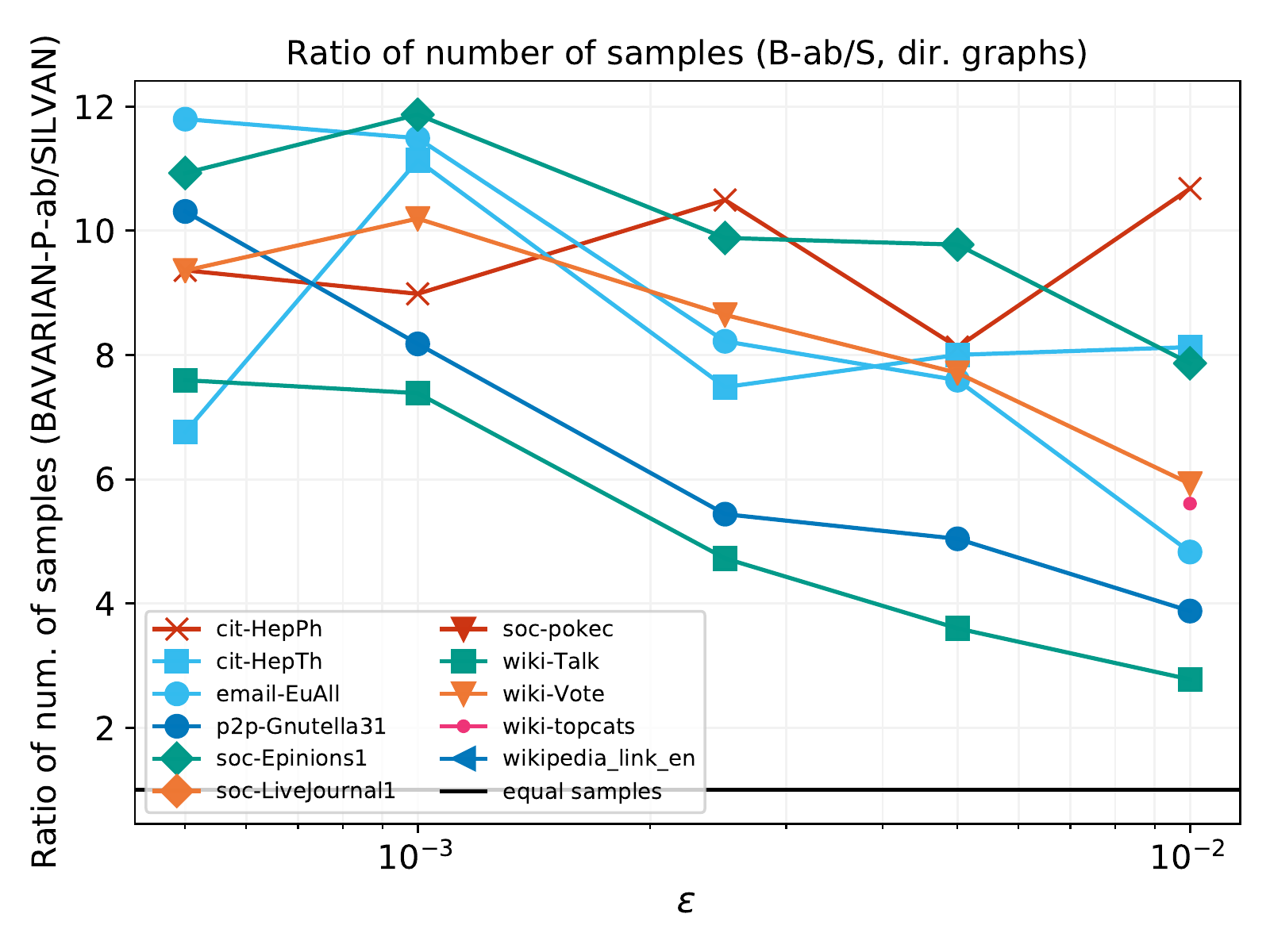}
  \caption{}
\end{subfigure}
\begin{subfigure}{.49\textwidth}
  \centering
  \includegraphics[width=\textwidth]{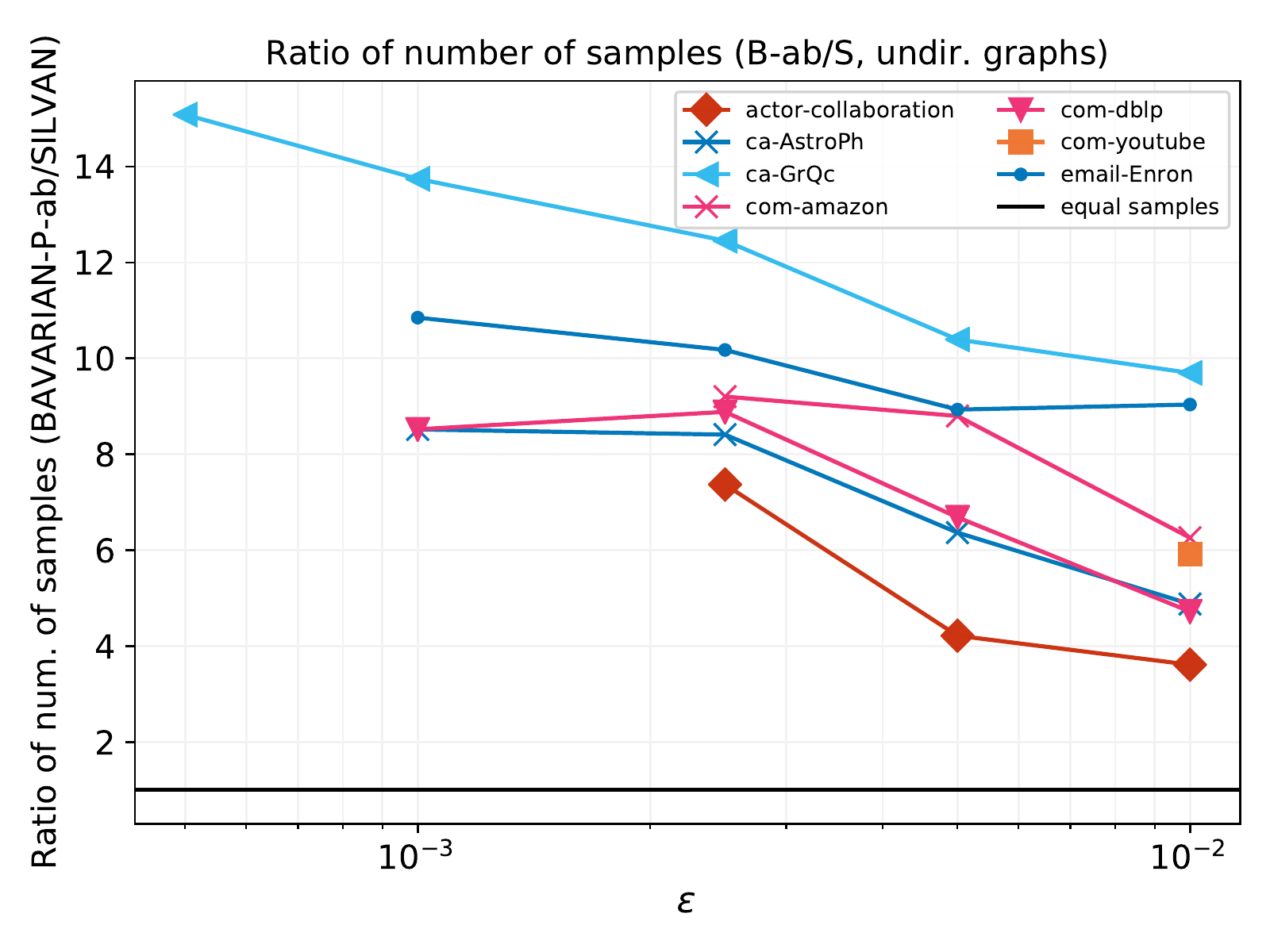}
  \caption{}
\end{subfigure}
\begin{subfigure}{.49\textwidth}
  \centering
  \includegraphics[width=\textwidth]{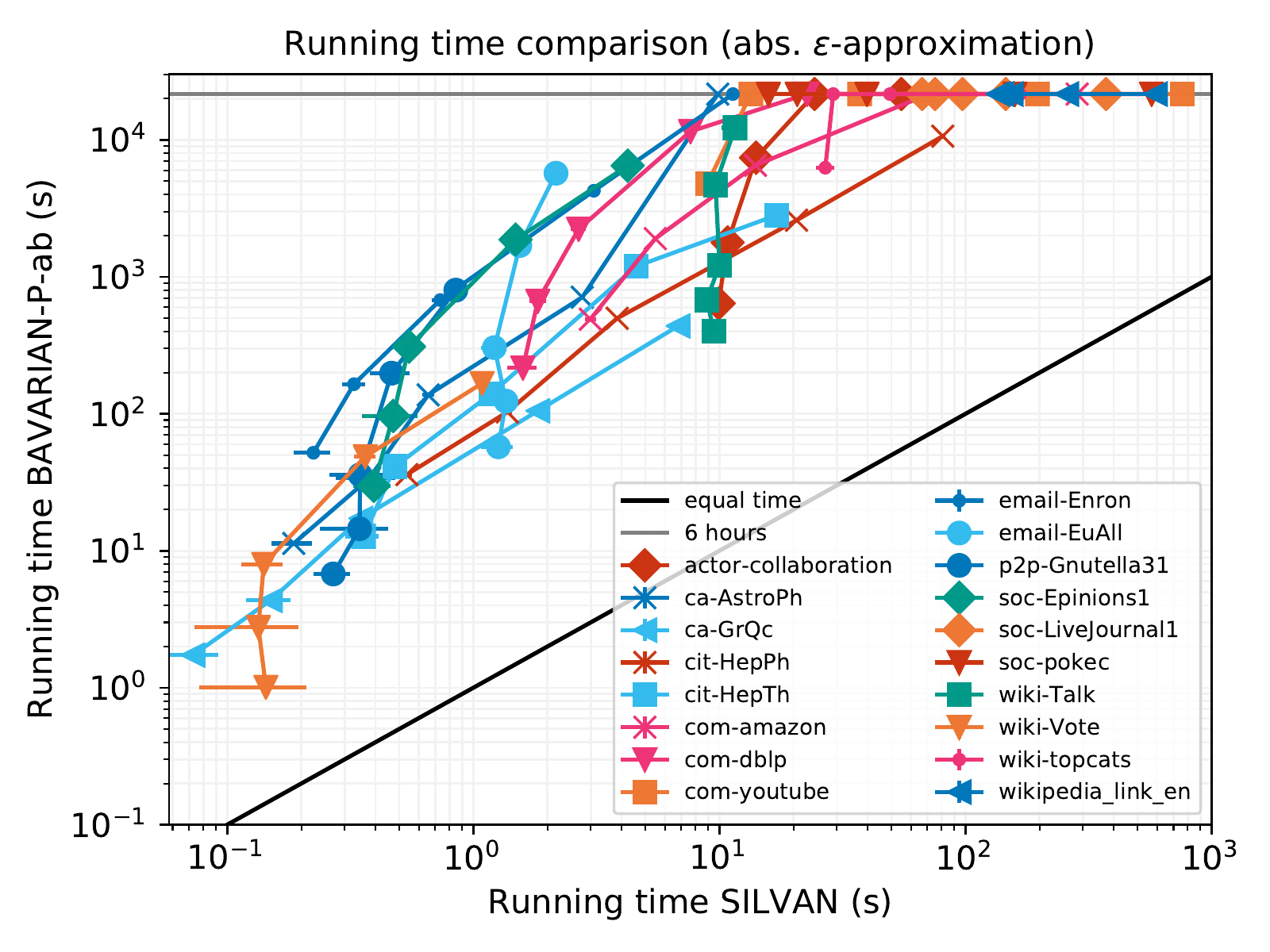}
  \caption{}
\end{subfigure}
\caption{Comparison between the 
performance of \bavarian\ (\texttt{ab} estimator) and \algname\ for obtaining absolute 
approximations. 
(a): ratios of the number of samples required by \bavarian\ and the number of samples required by \algname\ for directed graphs.
(b): as (a) for undirected graphs. 
(c): comparison of the running times of \bavarian\ ($y$ axis) and \algname\ ($x$ axis) for all graphs (axes in logarithmic scale). Additional plots in Figures~\ref{fig:absapproxappendixbavab}-\ref{fig:absapproxappendixbavbp}.
}
\label{fig:absapproxbav}
\end{figure*}

\subsubsection{Sample sizes} 

In Figures~\ref{fig:absapprox}~(a) 
and~(b) 
we show the ratios between the number of samples required by \kadabra\ and \algname\ to converge (we sum the number of samples of both phases, for both algorithms) for directed and undirected graphs. 
We can see that the number of samples needed by \algname\ is always smaller than \kadabra, by at least $20 \% $; for $14$ out of $18$ graphs, \algname\ finished after processing \emph{less than half} of the samples considered by \kadabra, and may require up to an order of magnitude less samples. 
By inspecting the graphs' statistics (Table~\ref{tab:graphs}), 
the largest improvements 
are obtained for graphs with smallest $\sup_{v \in V} b(v) \leq \maxabsf$. 
In fact, the number of samples required by \algname\ (Figure~\ref{fig:silvanvsplots}~(a)) 
varies significantly among graphs, with strong dependence on $\maxabsf$. 
Notice that $\sup_{v \in V} b(v)$ upper bounds the maximum variance $\sup_{v \in V} Var(f_v)$. 
A potential cause of the gap between \algname\ and \kadabra\ may depend on the use of the VC-dimension based bound in the adaptive sampling analysis of \kadabra; 
such bound is indeed required for its correctness, but it is agnostic to any property of the underlying graph (apart from the vertex diameter) and thus results in overly conservative guarantees in such cases. 
This confirms the significance of \algname's sharp \emph{variance-adaptive bounds}. 
In addition, the fact that \algname\ obtains simultaneous and non-uniform data-dependent approximations for \emph{sets of nodes}, exploting correlations among nodes through the use of the $c$-\mcera, leads to refined guarantees.

We now compare \algname\ with \bavarian\ in terms of sample sizes. 
We remark that the plots for sample sizes only show the results for cases in which \bavarian\ terminates in reasonable time (i.e, in less than $6$ hours), while figures for running times show a lower bound for such cases. 
From Figures~\ref{fig:absapproxbav}~(a)~and~(b), we can see that \algname\ always requires a fraction of the samples needed by \bavarian: at most \emph{half} of the samples for all graphs, and at most $1/4$ of the samples for $17$ out of $18$ graphs, with an improvement of up to one order of magnitude. We observed analogous results for the \texttt{rk} estimator (Figures~\ref{fig:absapproxappendixbavrk}~(c)~and~(d)). 
The \texttt{bp} estimator resulted to be the most efficient version of \bavarian\ in terms of number of samples; this is not suprising, since one \texttt{bp}'s sample considers all shortest paths starting from a single node, rather than shortest paths between two nodes (on the other hand, each sample is potentially much more expensive to compute). In any case, the number of samples needed by \algname\ is always smaller than \bavarian-\texttt{bp} by at least $10 \% $, up to a factor $5$. 

Overall, \algname\ obtains high-quality approximations at a fraction of the samples required by state-of-the-art methods; this highlights the significance of \algname's non-uniform approximation approach via empirical peeling and its novel improved bounds on the number of sufficient samples presented in \Cref{sec:upperboundsamplesabs}. 

\subsubsection{Running times} 
We now discuss how the reduction in the number of samples impacts the overall running times.  
We observed that, generally, 
the running time
roughly increases linearly with the sample size (Figure~\ref{fig:silvanvsplots}~(b)  shows that the relationship between the sample sizes and the running times of \algname\ is essentialy linear). In fact, the time spent on sampling shortest paths is usually the dominating cost of the algorithms.

In Figure~\ref{fig:absapprox}~(c) we compare the running times of \algname\ ($x$ axis) and \kadabra\ ($y$ axis) (we show ratios and axis in log. scale in Figure~\ref{fig:absapproxappendixkad}). 
While for smaller graphs both \algname\ and \kadabra\ terminate very quickly (e.g., in $< 10$ seconds), 
for the largest and most demanding graphs the reduction on the number of samples achieved by \algname\ has a sensible and significant impact on the running times, as clearly shown in Figure~\ref{fig:absapprox}~(c). 
For instance, \algname\ analyses the most demanding graph (\texttt{wikipedia-link-en}) in less than $1/3$ of the time required by \kadabra\ when $\varepsilon \leq 10^{-3}$ (see Figure~\ref{fig:absapproxappendixkad}~(f)). 
This is a consequence of significantly reducing the required samples, and also reflects the capability of \algname\ to 
compute the \nmcera\ \emph{incrementally} 
as shortest paths are sampled, incurring in 
a negligible computational overhead. 

In Figure~\ref{fig:absapproxbav}~(c) we compare the running times of \algname\ ($x$ axis) and \bavarian\ using the \texttt{ab} estimator ($y$ axis) (additional plots and other estimators in Figures~\ref{fig:absapproxappendixbavab}-\ref{fig:absapproxappendixbavbp}). 
Note that we report a lower bound to the running time of \bavarian\ when exceeding $6$~hours ($=2.16 \cdot 10^4$~seconds); \bavarian\ exceeded this threshold on most large graphs and for smaller values of $\varepsilon$, while \algname\ never required more than $17$~minutes ($=10^3$~seconds). 
Overall, we observed \algname\ to be at least one order of magnitude faster than \bavarian, up to $3$ orders of magnitude. 
We observed very similar results for the $\texttt{rk}$ estimator (Figure~\ref{fig:absapproxappendixbavrk}).
\algname\ is also at least one order of magnitude faster than \bavarian\ using the \texttt{bp} estimator for all but for the \texttt{wiki-Vote} graph, for which it is $> 3$ times faster (Figure~\ref{fig:absapproxappendixbavbp}). 
\algname's improvements are due to both the significant reduction in the number of samples (as discussed previously) thanks to its non-uniform approximation scheme, and from the fact that \algname\ leverages a more efficient algorithm for sampling shortest paths, based on the balanced bidirectional BFS, drastically reducing the computational requirement for the task.

We conclude that \algname\ requires much fewer resources to obtain rigorous approximations of the \bc\ of all nodes of the same quality, or, equivalently, sharper guarantees for the same amount of work. 

\subsubsection{Quality of \algname's approximations} 
Finally, we investigated the accuracy of the approximations reported by \algname\ by computing the exact \bc\ of all the nodes of $6$ graphs ($3$ undirected and $3$ directed, representative of other instances) and measuring $\supdev$ over all runs.
We show these results in Figure~\ref{fig:accuracyapprox} (see Appendix). 
As expected from our theoretical analysis, we always observed $\supdev \leq \varepsilon$, thus \algname\ is more accurate than guaranteed. 
However, the gap between $\supdev$ and $\varepsilon$ is not large, confirming the sharpness of the guarantees provided by \algname. 
We remark that the exact approach requires several hours on the larger graphs we considered for this set of experiments (e.g., for the \texttt{com-dblp} graph, the exact approach implemented in Networkit~\cite{networkit2016} requires $>$ 1 hour to terminate using all $64$ cores), and does not complete in reasonable time for the largest instances (e.g., \cite{borassi2019kadabra} reports that $\approx 1$ week is necessary for graphs of size similar to the largest of our test set). Instead, \algname\ finishes in at most few minutes for the lowest value of $\varepsilon$ (e.g., always less than $20$ seconds for \texttt{com-dblp}), and it is much faster for other cases.

\subsection{Top-$k$ Approximation}
\label{exp:topk}
We now present experiments on the task of computing relative approximation of the set of top-$k$ most central nodes. 
\algname\ is the first method that allows to approximate the top-$k$ most central nodes with \emph{relative} and \emph{non-uniform} bounds via empirical peeling, differently from previous methods that focus on additive approximations (or rely on uniform additive 
approximations as a proxy)~\cite{RiondatoU18,borassi2019kadabra}; 
therefore, we compare \algname~\cite{borassi2019kadabra} with 
\kadabra, the best performing approach that allows to obtain approximations of comparable quality. 
We recall that the top-$k$ approximation proposed in~\cite{borassi2019kadabra}, for a given $k$ and $\varepsilon$, guarantees an $\varepsilon$ additive approximation of the 
top-$k$ nodes;
however, the confidence intervals for some of the nodes can be
relaxed (i.e., be wider than $2\varepsilon$)  
if they can be ranked correctly with looser accuracy. 
Instead, \algname\ guarantees that all nodes are well estimated within the relative accuracy $\eta$. 
For given $k$ and $\eta$, we first run \algname\ on all graphs; when finished, we store the maximum absolute deviation $\varepsilon_{k}$ required to guarantee all properties of \Cref{def:topkapproxguarantees}, and 
we run \kadabra\ with $k$ and $\varepsilon_{k}$. 

We considered $k \in \{5 , 10 , 25\}$ and $\eta \in \{ 0.25 , 0.1 , 0.05 \}$. 
As in previous experiments, \algname's empirical peeling follows the procedure described in Section~\ref{sec:algunif} with $a=2$. 
We report avgs. $\pm$ stds over $10$ runs.

\begin{figure*}[ht]
\centering
\begin{subfigure}{.35\textwidth}
  \centering
  \includegraphics[width=\textwidth]{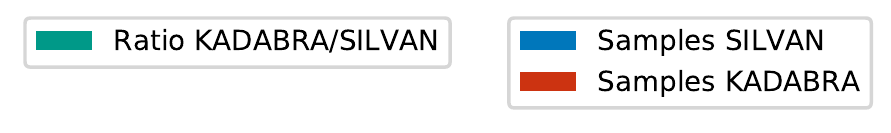}
\end{subfigure}
\hspace{50px}
\begin{subfigure}{.35\textwidth}
  \centering
  \includegraphics[width=\textwidth]{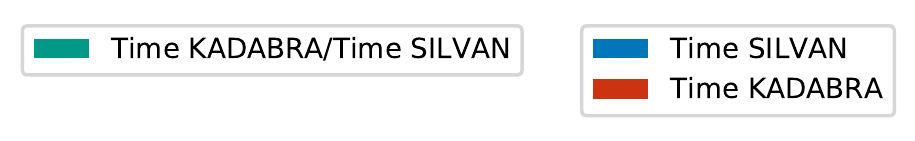}
\end{subfigure} \\
\begin{subfigure}{.49\textwidth}
  \centering
  \includegraphics[width=\textwidth]{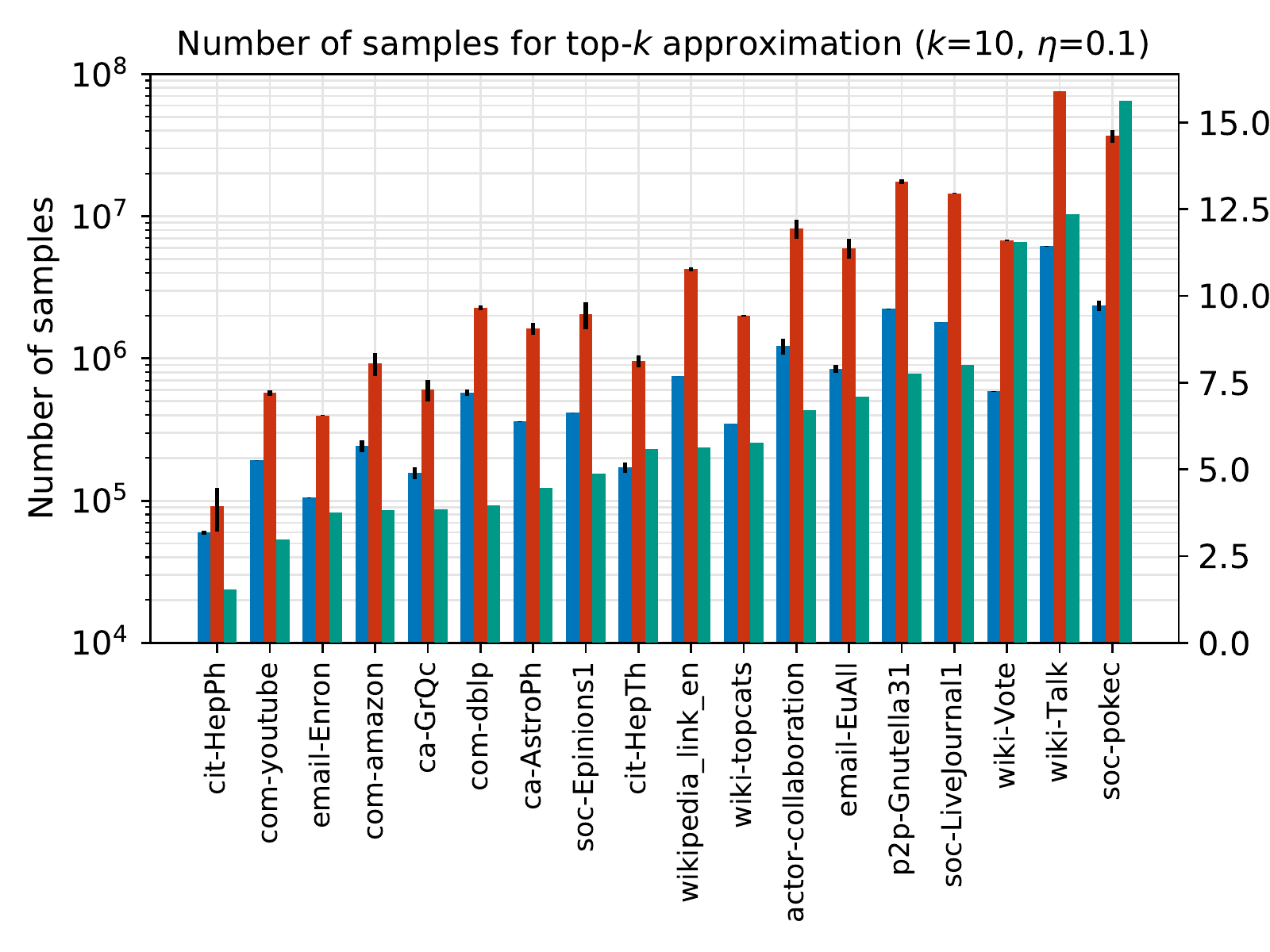}
 \caption{}
\end{subfigure}
\begin{subfigure}{.49\textwidth}
  \centering
  \includegraphics[width=\textwidth]{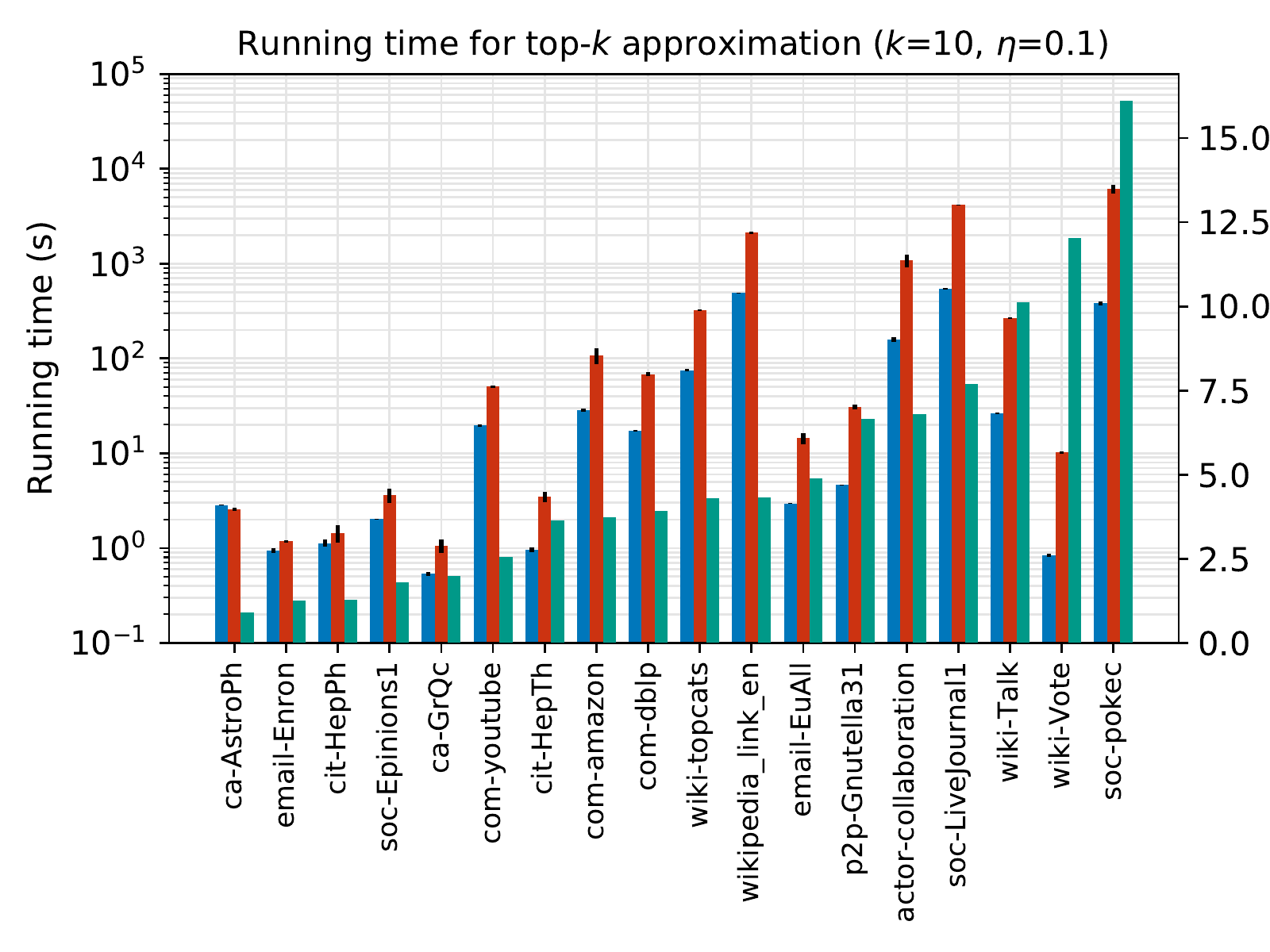}
  \caption{}
\end{subfigure}
\caption{Comparison of number of samples (a) and running times (b) of \algnametopk\ with \kadabra\ for obtaining top-$k$ approximations for $k=10$ and $\eta=0.1$ (all other combinations shown in \Cref{fig:topkappxtimesappx}).}
\label{fig:topkappx}
\end{figure*}

\Cref{fig:topkappx} shows the number of samples and running time of both algorithms to obtain their respective top-$k$ approximations for $k=10$ and $\eta=0.1$, representative of other cases (we show all combinations of $k$ and $\eta$ in \Cref{fig:topkappxtimesappx} in Appendix). 

From \Cref{fig:topkappx}~(a) we see that \algnametopk\ requires a fraction of the samples of \kadabra, even if offering stronger guarantees (no confidence intervals are relaxed according to the ranking): \algnametopk\ requires at most $2/3$ of \kadabra's samples, and finishes after processing less than $1/3$ of the samples for $17$ out of $18$ graphs; for $10$ graphs, it needs less than $1/5$ of the samples, and less than $1/10$ for $3$ graphs. 

The reduction of the number of samples, similarly to the previous setting, significantly impacts the running times.
From \Cref{fig:topkappx}~(b) we conclude that, while in cases in which both algorithms conclude very quickly (e.g., in less than $3$ seconds on small graphs) they obtain comparable performances,  on larger graphs \algname\ significantly outperforms \kadabra. 
In fact, for $14$ graphs \algnametopk\ finished in less than half of the time of \kadabra, and it is at least $5$ times faster in $6$ cases. 
For three of the most demanding graphs, 
\algname\ is more than $10$ times faster. 
Overall, \algname\ analyses all graphs in $< 0.5$ hours, while \kadabra\ needs $> 4$ hours.
We conclude that \algname\ significantly reduces the computational requirements for the task, potentially allowing much more interactive exploratory analyses.

Additionally, we compared the quality of the output of \kadabra\ w.r.t. \algnametopk\ when using the same number of samples:  
we stop \kadabra\ 
at the number of samples required by \algnametopk. 
\Cref{fig:topkaccuracyapprox} in Appendix shows runs taken at random for $3$ graphs, using $k=10$ and $\eta = 0.1$.
From \Cref{fig:topkaccuracyapprox}
we can see that 
\algnametopk\ (in blue) provides much tighter upper and lower bounds than \kadabra\ (in red); 
obtaining sharper confidence intervals on the top-$k$ nodes has a drastic effect on the capability of the algorithms to rank nodes correctly. 
Consequently, for the same work, \algnametopk\ reports much less false positives (e.g., $16$ vs $45$ results for \texttt{com-dblp})
and can clearly identify the rank of the most central nodes.

\section{Conclusions}
We introduced \algname, a novel progressive sampling algorithm to estimate the betweenness centrality of all nodes in a graph. 
\algname\ relies on new bounds on supremum deviation of functions, based on the \nmcera\ and non-uniform approximation scheme via empirical peeling.
We present variants of \algname\ to obtain additive approximations, and relative approximations for the top-$k$ betweenness centrality. 
Our experimental results show that \algname\ significantly outperforms state-of-the-art approaches for approximating betweenness centrality 
with the same guarantees.

There are multiple interesting directions for future work.
While in this work we considered various approximations of the \bc\ in a static setting, recent works considered extending the problem to dynamic~\cite{bergamini2014approximating,bergamini2015fully,hayashi2015fully}, temporal~\cite{santoro2022onbra}, and uncertain graphs~\cite{saha2021shortest}, 
or different types of centralities~\cite{de2020estimating}, 
all settings in which we believe the ideas behind our algorithm \algname\ could lead to improved approximations. 

Furthermore, the empirical peeling scheme we introduced in this work is general: it can be applied to sets of functions with arbitrary domains, so it can potentially benefit randomized approximation algorithms in other settings, such as interesting~\cite{RiondatoV18,santoro2020mining,sarpe2021oden} and significant pattern mining~\cite{PellegrinaRV19a}, and sequential hypothesis testing~\cite{de2019rademacher}.

\paragraph{Acknowledgments}
This work was supported, in part, by MIUR of Italy, 
under PRIN Project n.~20174LF3T8 AHeAD and grant L.~232 (Dipartimenti di Eccellenza), 
and by the University of Padova under project ``SID~2020: RATED-X''.

\bibliographystyle{ACM-Reference-Format}
\bibliography{bibliography}


\begin{thebibliography}{64}


\ifx \showCODEN    \undefined \def \showCODEN     #1{\unskip}     \fi
\ifx \showDOI      \undefined \def \showDOI       #1{#1}\fi
\ifx \showISBNx    \undefined \def \showISBNx     #1{\unskip}     \fi
\ifx \showISBNxiii \undefined \def \showISBNxiii  #1{\unskip}     \fi
\ifx \showISSN     \undefined \def \showISSN      #1{\unskip}     \fi
\ifx \showLCCN     \undefined \def \showLCCN      #1{\unskip}     \fi
\ifx \shownote     \undefined \def \shownote      #1{#1}          \fi
\ifx \showarticletitle \undefined \def \showarticletitle #1{#1}   \fi
\ifx \showURL      \undefined \def \showURL       {\relax}        \fi
\providecommand\bibfield[2]{#2}
\providecommand\bibinfo[2]{#2}
\providecommand\natexlab[1]{#1}
\providecommand\showeprint[2][]{arXiv:#2}

\bibitem[\protect\citeauthoryear{Anthonisse}{Anthonisse}{1971}]%
        {anthonisse1971rush}
\bibfield{author}{\bibinfo{person}{Jac~M Anthonisse}.}
  \bibinfo{year}{1971}\natexlab{}.
\newblock \showarticletitle{The rush in a directed graph}.
\newblock \bibinfo{journal}{\emph{Stichting Mathematisch Centrum. Mathematische
  Besliskunde}} \bibinfo{number}{BN 9/71} (\bibinfo{year}{1971}).
\newblock


\bibitem[\protect\citeauthoryear{Bartlett, Bousquet, Mendelson,
  et~al\mbox{.}}{Bartlett et~al\mbox{.}}{2005}]%
        {bartlett2005local}
\bibfield{author}{\bibinfo{person}{Peter~L Bartlett}, \bibinfo{person}{Olivier
  Bousquet}, \bibinfo{person}{Shahar Mendelson}, {et~al\mbox{.}}}
  \bibinfo{year}{2005}\natexlab{}.
\newblock \showarticletitle{Local rademacher complexities}.
\newblock \bibinfo{journal}{\emph{The Annals of Statistics}}
  \bibinfo{volume}{33}, \bibinfo{number}{4} (\bibinfo{year}{2005}),
  \bibinfo{pages}{1497--1537}.
\newblock


\bibitem[\protect\citeauthoryear{Bartlett and Mendelson}{Bartlett and
  Mendelson}{2002}]%
        {BartlettM02}
\bibfield{author}{\bibinfo{person}{Peter~L. Bartlett} {and}
  \bibinfo{person}{Shahar Mendelson}.} \bibinfo{year}{2002}\natexlab{}.
\newblock \showarticletitle{{R}ademacher and {G}aussian complexities: Risk
  bounds and structural results}.
\newblock \bibinfo{journal}{\emph{Journal of Machine Learning Research}}
  \bibinfo{volume}{3}, \bibinfo{number}{Nov} (\bibinfo{year}{2002}),
  \bibinfo{pages}{463--482}.
\newblock


\bibitem[\protect\citeauthoryear{Bergamini, Borassi, Crescenzi, Marino, and
  Meyerhenke}{Bergamini et~al\mbox{.}}{2019}]%
        {bergamini2019computing}
\bibfield{author}{\bibinfo{person}{Elisabetta Bergamini},
  \bibinfo{person}{Michele Borassi}, \bibinfo{person}{Pierluigi Crescenzi},
  \bibinfo{person}{Andrea Marino}, {and} \bibinfo{person}{Henning Meyerhenke}.}
  \bibinfo{year}{2019}\natexlab{}.
\newblock \showarticletitle{Computing top-k closeness centrality faster in
  unweighted graphs}.
\newblock \bibinfo{journal}{\emph{ACM Transactions on Knowledge Discovery from
  Data (TKDD)}} \bibinfo{volume}{13}, \bibinfo{number}{5}
  (\bibinfo{year}{2019}), \bibinfo{pages}{1--40}.
\newblock


\bibitem[\protect\citeauthoryear{Bergamini, Crescenzi, D'angelo, Meyerhenke,
  Severini, and Velaj}{Bergamini et~al\mbox{.}}{2018}]%
        {bergamini2018improving}
\bibfield{author}{\bibinfo{person}{Elisabetta Bergamini},
  \bibinfo{person}{Pierluigi Crescenzi}, \bibinfo{person}{Gianlorenzo
  D'angelo}, \bibinfo{person}{Henning Meyerhenke}, \bibinfo{person}{Lorenzo
  Severini}, {and} \bibinfo{person}{Yllka Velaj}.}
  \bibinfo{year}{2018}\natexlab{}.
\newblock \showarticletitle{Improving the betweenness centrality of a node by
  adding links}.
\newblock \bibinfo{journal}{\emph{Journal of Experimental Algorithmics (JEA)}}
  \bibinfo{volume}{23} (\bibinfo{year}{2018}), \bibinfo{pages}{1--32}.
\newblock


\bibitem[\protect\citeauthoryear{Bergamini and Meyerhenke}{Bergamini and
  Meyerhenke}{2015}]%
        {bergamini2015fully}
\bibfield{author}{\bibinfo{person}{Elisabetta Bergamini} {and}
  \bibinfo{person}{Henning Meyerhenke}.} \bibinfo{year}{2015}\natexlab{}.
\newblock \showarticletitle{Fully-dynamic approximation of betweenness
  centrality}.
\newblock In \bibinfo{booktitle}{\emph{Algorithms-ESA 2015}}.
  \bibinfo{publisher}{Springer}, \bibinfo{pages}{155--166}.
\newblock


\bibitem[\protect\citeauthoryear{Bergamini, Meyerhenke, and Staudt}{Bergamini
  et~al\mbox{.}}{2014}]%
        {bergamini2014approximating}
\bibfield{author}{\bibinfo{person}{Elisabetta Bergamini},
  \bibinfo{person}{Henning Meyerhenke}, {and} \bibinfo{person}{Christian~L
  Staudt}.} \bibinfo{year}{2014}\natexlab{}.
\newblock \showarticletitle{Approximating betweenness centrality in large
  evolving networks}. In \bibinfo{booktitle}{\emph{2015 Proceedings of the
  Seventeenth Workshop on Algorithm Engineering and Experiments (ALENEX)}}.
  SIAM, \bibinfo{pages}{133--146}.
\newblock


\bibitem[\protect\citeauthoryear{Bhatia and Davis}{Bhatia and Davis}{2000}]%
        {bhatia2000better}
\bibfield{author}{\bibinfo{person}{Rajendra Bhatia} {and}
  \bibinfo{person}{Chandler Davis}.} \bibinfo{year}{2000}\natexlab{}.
\newblock \showarticletitle{A better bound on the variance}.
\newblock \bibinfo{journal}{\emph{The American Mathematical Monthly}}
  \bibinfo{volume}{107}, \bibinfo{number}{4} (\bibinfo{year}{2000}),
  \bibinfo{pages}{353--357}.
\newblock


\bibitem[\protect\citeauthoryear{Boldi, Rosa, and Vigna}{Boldi
  et~al\mbox{.}}{2011}]%
        {boldi2011hyperanf}
\bibfield{author}{\bibinfo{person}{Paolo Boldi}, \bibinfo{person}{Marco Rosa},
  {and} \bibinfo{person}{Sebastiano Vigna}.} \bibinfo{year}{2011}\natexlab{}.
\newblock \showarticletitle{HyperANF: Approximating the neighbourhood function
  of very large graphs on a budget}. In \bibinfo{booktitle}{\emph{Proceedings
  of the 20th international conference on World Wide Web}}.
  \bibinfo{pages}{625--634}.
\newblock


\bibitem[\protect\citeauthoryear{Boldi and Vigna}{Boldi and Vigna}{2013}]%
        {boldi2013core}
\bibfield{author}{\bibinfo{person}{Paolo Boldi} {and}
  \bibinfo{person}{Sebastiano Vigna}.} \bibinfo{year}{2013}\natexlab{}.
\newblock \showarticletitle{In-core computation of geometric centralities with
  hyperball: A hundred billion nodes and beyond}. In
  \bibinfo{booktitle}{\emph{2013 IEEE 13th International Conference on Data
  Mining Workshops}}. IEEE, \bibinfo{pages}{621--628}.
\newblock


\bibitem[\protect\citeauthoryear{Borassi, Crescenzi, and Habib}{Borassi
  et~al\mbox{.}}{2016}]%
        {borassi2016into}
\bibfield{author}{\bibinfo{person}{Michele Borassi}, \bibinfo{person}{Pierluigi
  Crescenzi}, {and} \bibinfo{person}{Michel Habib}.}
  \bibinfo{year}{2016}\natexlab{}.
\newblock \showarticletitle{Into the square: On the complexity of some
  quadratic-time solvable problems}.
\newblock \bibinfo{journal}{\emph{Electronic Notes in Theoretical Computer
  Science}}  \bibinfo{volume}{322} (\bibinfo{year}{2016}),
  \bibinfo{pages}{51--67}.
\newblock


\bibitem[\protect\citeauthoryear{Borassi and Natale}{Borassi and
  Natale}{2019}]%
        {borassi2019kadabra}
\bibfield{author}{\bibinfo{person}{Michele Borassi} {and}
  \bibinfo{person}{Emanuele Natale}.} \bibinfo{year}{2019}\natexlab{}.
\newblock \showarticletitle{KADABRA is an adaptive algorithm for betweenness
  via random approximation}.
\newblock \bibinfo{journal}{\emph{Journal of Experimental Algorithmics (JEA)}}
  \bibinfo{volume}{24} (\bibinfo{year}{2019}), \bibinfo{pages}{1--35}.
\newblock


\bibitem[\protect\citeauthoryear{Boucheron, Bousquet, and Lugosi}{Boucheron
  et~al\mbox{.}}{2005}]%
        {boucheron2005theory}
\bibfield{author}{\bibinfo{person}{St{\'e}phane Boucheron},
  \bibinfo{person}{Olivier Bousquet}, {and} \bibinfo{person}{G{\'a}bor
  Lugosi}.} \bibinfo{year}{2005}\natexlab{}.
\newblock \showarticletitle{Theory of classification: A survey of some recent
  advances}.
\newblock \bibinfo{journal}{\emph{ESAIM: probability and statistics}}
  \bibinfo{volume}{9} (\bibinfo{year}{2005}), \bibinfo{pages}{323--375}.
\newblock


\bibitem[\protect\citeauthoryear{Boucheron, Lugosi, and Massart}{Boucheron
  et~al\mbox{.}}{2013}]%
        {boucheron2013concentration}
\bibfield{author}{\bibinfo{person}{St{\'e}phane Boucheron},
  \bibinfo{person}{G{\'a}bor Lugosi}, {and} \bibinfo{person}{Pascal Massart}.}
  \bibinfo{year}{2013}\natexlab{}.
\newblock \bibinfo{booktitle}{\emph{Concentration inequalities: A nonasymptotic
  theory of independence}}.
\newblock \bibinfo{publisher}{Oxford university press}.
\newblock


\bibitem[\protect\citeauthoryear{Bousquet}{Bousquet}{2002}]%
        {bousquet2002bennett}
\bibfield{author}{\bibinfo{person}{Olivier Bousquet}.}
  \bibinfo{year}{2002}\natexlab{}.
\newblock \showarticletitle{A {B}ennett concentration inequality and its
  application to suprema of empirical processes}.
\newblock \bibinfo{journal}{\emph{Comptes Rendus Mathematique}}
  \bibinfo{volume}{334}, \bibinfo{number}{6} (\bibinfo{year}{2002}),
  \bibinfo{pages}{495--500}.
\newblock


\bibitem[\protect\citeauthoryear{Brandes}{Brandes}{2001}]%
        {brandes2001faster}
\bibfield{author}{\bibinfo{person}{Ulrik Brandes}.}
  \bibinfo{year}{2001}\natexlab{}.
\newblock \showarticletitle{A faster algorithm for betweenness centrality}.
\newblock \bibinfo{journal}{\emph{Journal of mathematical sociology}}
  \bibinfo{volume}{25}, \bibinfo{number}{2} (\bibinfo{year}{2001}),
  \bibinfo{pages}{163--177}.
\newblock


\bibitem[\protect\citeauthoryear{Ceccarello, Pietracaprina, Pucci, and
  Upfal}{Ceccarello et~al\mbox{.}}{2020}]%
        {ceccarello2020distributed}
\bibfield{author}{\bibinfo{person}{Matteo Ceccarello}, \bibinfo{person}{Andrea
  Pietracaprina}, \bibinfo{person}{Geppino Pucci}, {and} \bibinfo{person}{Eli
  Upfal}.} \bibinfo{year}{2020}\natexlab{}.
\newblock \showarticletitle{Distributed Graph Diameter Approximation}.
\newblock \bibinfo{journal}{\emph{Algorithms}} \bibinfo{volume}{13},
  \bibinfo{number}{9} (\bibinfo{year}{2020}), \bibinfo{pages}{216}.
\newblock


\bibitem[\protect\citeauthoryear{Chechik, Cohen, and Kaplan}{Chechik
  et~al\mbox{.}}{2015}]%
        {Chechik2015}
\bibfield{author}{\bibinfo{person}{Shiri Chechik}, \bibinfo{person}{Edith
  Cohen}, {and} \bibinfo{person}{Haim Kaplan}.}
  \bibinfo{year}{2015}\natexlab{}.
\newblock \showarticletitle{Average distance queries through weighted samples
  in graphs and metric spaces: high scalability with tight statistical
  guarantees.}. In \bibinfo{booktitle}{\emph{APPROX/RANDOM}}.
\newblock


\bibitem[\protect\citeauthoryear{Cortes, Greenberg, and Mohri}{Cortes
  et~al\mbox{.}}{2019}]%
        {cortes2019relative}
\bibfield{author}{\bibinfo{person}{Corinna Cortes}, \bibinfo{person}{Spencer
  Greenberg}, {and} \bibinfo{person}{Mehryar Mohri}.}
  \bibinfo{year}{2019}\natexlab{}.
\newblock \showarticletitle{Relative deviation learning bounds and
  generalization with unbounded loss functions}.
\newblock \bibinfo{journal}{\emph{Annals of Mathematics and Artificial
  Intelligence}} \bibinfo{volume}{85}, \bibinfo{number}{1}
  (\bibinfo{year}{2019}), \bibinfo{pages}{45--70}.
\newblock


\bibitem[\protect\citeauthoryear{Cousins and Riondato}{Cousins and
  Riondato}{2020}]%
        {cousins2020sharp}
\bibfield{author}{\bibinfo{person}{Cyrus Cousins} {and} \bibinfo{person}{Matteo
  Riondato}.} \bibinfo{year}{2020}\natexlab{}.
\newblock \showarticletitle{Sharp uniform convergence bounds through empirical
  centralization}.
\newblock \bibinfo{journal}{\emph{Advances in Neural Information Processing
  Systems}}  \bibinfo{volume}{33} (\bibinfo{year}{2020}),
  \bibinfo{pages}{15123--15132}.
\newblock


\bibitem[\protect\citeauthoryear{Cousins, Wohlgemuth, and Riondato}{Cousins
  et~al\mbox{.}}{2021}]%
        {cousins2021bavarian}
\bibfield{author}{\bibinfo{person}{Cyrus Cousins}, \bibinfo{person}{Chloe
  Wohlgemuth}, {and} \bibinfo{person}{Matteo Riondato}.}
  \bibinfo{year}{2021}\natexlab{}.
\newblock \showarticletitle{Bavarian: Betweenness Centrality Approximation with
  Variance-Aware Rademacher Averages}. In \bibinfo{booktitle}{\emph{Proceedings
  of the 27th ACM SIGKDD Conference on Knowledge Discovery \& Data Mining}}.
  \bibinfo{pages}{196--206}.
\newblock


\bibitem[\protect\citeauthoryear{Crescenzi, Grossi, Habib, Lanzi, and
  Marino}{Crescenzi et~al\mbox{.}}{2013}]%
        {crescenzi2013computing}
\bibfield{author}{\bibinfo{person}{Pilu Crescenzi}, \bibinfo{person}{Roberto
  Grossi}, \bibinfo{person}{Michel Habib}, \bibinfo{person}{Leonardo Lanzi},
  {and} \bibinfo{person}{Andrea Marino}.} \bibinfo{year}{2013}\natexlab{}.
\newblock \showarticletitle{On computing the diameter of real-world undirected
  graphs}.
\newblock \bibinfo{journal}{\emph{Theoretical Computer Science}}
  \bibinfo{volume}{514} (\bibinfo{year}{2013}), \bibinfo{pages}{84--95}.
\newblock


\bibitem[\protect\citeauthoryear{Crescenzi, Grossi, Lanzi, and
  Marino}{Crescenzi et~al\mbox{.}}{2012}]%
        {crescenzi2012computing}
\bibfield{author}{\bibinfo{person}{Pierluigi Crescenzi},
  \bibinfo{person}{Roberto Grossi}, \bibinfo{person}{Leonardo Lanzi}, {and}
  \bibinfo{person}{Andrea Marino}.} \bibinfo{year}{2012}\natexlab{}.
\newblock \showarticletitle{On computing the diameter of real-world directed
  (weighted) graphs}. In \bibinfo{booktitle}{\emph{International Symposium on
  Experimental Algorithms}}. Springer, \bibinfo{pages}{99--110}.
\newblock


\bibitem[\protect\citeauthoryear{de~Lima, da~Silva, and Vignatti}{de~Lima
  et~al\mbox{.}}{2020}]%
        {de2020estimating}
\bibfield{author}{\bibinfo{person}{Alane~M de Lima}, \bibinfo{person}{Murilo~VG
  da Silva}, {and} \bibinfo{person}{Andr{\'e}~L Vignatti}.}
  \bibinfo{year}{2020}\natexlab{}.
\newblock \showarticletitle{Estimating the Percolation Centrality of Large
  Networks through Pseudo-dimension Theory}. In
  \bibinfo{booktitle}{\emph{Proceedings of the 26th ACM SIGKDD International
  Conference on Knowledge Discovery \& Data Mining}}.
  \bibinfo{pages}{1839--1847}.
\newblock


\bibitem[\protect\citeauthoryear{De~Stefani and Upfal}{De~Stefani and
  Upfal}{2019}]%
        {de2019rademacher}
\bibfield{author}{\bibinfo{person}{Lorenzo De~Stefani} {and}
  \bibinfo{person}{Eli Upfal}.} \bibinfo{year}{2019}\natexlab{}.
\newblock \showarticletitle{A Rademacher Complexity Based Method for
  Controlling Power and Confidence Level in Adaptive Statistical Analysis}. In
  \bibinfo{booktitle}{\emph{2019 IEEE International Conference on Data Science
  and Advanced Analytics (DSAA)}}. IEEE, \bibinfo{pages}{71--80}.
\newblock


\bibitem[\protect\citeauthoryear{Erd{\H{o}}s, Ishakian, Bestavros, and
  Terzi}{Erd{\H{o}}s et~al\mbox{.}}{2015}]%
        {erdHos2015divide}
\bibfield{author}{\bibinfo{person}{D{\'o}ra Erd{\H{o}}s},
  \bibinfo{person}{Vatche Ishakian}, \bibinfo{person}{Azer Bestavros}, {and}
  \bibinfo{person}{Evimaria Terzi}.} \bibinfo{year}{2015}\natexlab{}.
\newblock \showarticletitle{A divide-and-conquer algorithm for betweenness
  centrality}. In \bibinfo{booktitle}{\emph{Proceedings of the 2015 SIAM
  International Conference on Data Mining}}. SIAM, \bibinfo{pages}{433--441}.
\newblock


\bibitem[\protect\citeauthoryear{Flajolet and Martin}{Flajolet and
  Martin}{1983}]%
        {flajolet1983probabilistic}
\bibfield{author}{\bibinfo{person}{Philippe Flajolet} {and}
  \bibinfo{person}{G~Nigel Martin}.} \bibinfo{year}{1983}\natexlab{}.
\newblock \showarticletitle{Probabilistic counting}. In
  \bibinfo{booktitle}{\emph{24th Annual Symposium on Foundations of Computer
  Science (sfcs 1983)}}. IEEE, \bibinfo{pages}{76--82}.
\newblock


\bibitem[\protect\citeauthoryear{Freeman}{Freeman}{1977}]%
        {freeman1977set}
\bibfield{author}{\bibinfo{person}{Linton~C Freeman}.}
  \bibinfo{year}{1977}\natexlab{}.
\newblock \showarticletitle{A set of measures of centrality based on
  betweenness}.
\newblock \bibinfo{journal}{\emph{Sociometry}} (\bibinfo{year}{1977}),
  \bibinfo{pages}{35--41}.
\newblock


\bibitem[\protect\citeauthoryear{Geer and van~de Geer}{Geer and van~de
  Geer}{2000}]%
        {geer2000empirical}
\bibfield{author}{\bibinfo{person}{Sara~A Geer} {and} \bibinfo{person}{Sara
  van~de Geer}.} \bibinfo{year}{2000}\natexlab{}.
\newblock \bibinfo{booktitle}{\emph{Empirical Processes in M-estimation}}.
  Vol.~\bibinfo{volume}{6}.
\newblock \bibinfo{publisher}{Cambridge university press}.
\newblock


\bibitem[\protect\citeauthoryear{Har-Peled and Sharir}{Har-Peled and
  Sharir}{2011}]%
        {har2011relative}
\bibfield{author}{\bibinfo{person}{Sariel Har-Peled} {and}
  \bibinfo{person}{Micha Sharir}.} \bibinfo{year}{2011}\natexlab{}.
\newblock \showarticletitle{Relative (p, $\varepsilon$)-approximations in
  geometry}.
\newblock \bibinfo{journal}{\emph{Discrete \& Computational Geometry}}
  \bibinfo{volume}{45}, \bibinfo{number}{3} (\bibinfo{year}{2011}),
  \bibinfo{pages}{462--496}.
\newblock


\bibitem[\protect\citeauthoryear{Hayashi, Akiba, and Yoshida}{Hayashi
  et~al\mbox{.}}{2015}]%
        {hayashi2015fully}
\bibfield{author}{\bibinfo{person}{Takanori Hayashi}, \bibinfo{person}{Takuya
  Akiba}, {and} \bibinfo{person}{Yuichi Yoshida}.}
  \bibinfo{year}{2015}\natexlab{}.
\newblock \showarticletitle{Fully dynamic betweenness centrality maintenance on
  massive networks}.
\newblock \bibinfo{journal}{\emph{Proceedings of the VLDB Endowment}}
  \bibinfo{volume}{9}, \bibinfo{number}{2} (\bibinfo{year}{2015}),
  \bibinfo{pages}{48--59}.
\newblock


\bibitem[\protect\citeauthoryear{Ishakian, Erd{\"o}s, Terzi, and
  Bestavros}{Ishakian et~al\mbox{.}}{2012}]%
        {ishakian2012framework}
\bibfield{author}{\bibinfo{person}{Vatche Ishakian}, \bibinfo{person}{D{\'o}ra
  Erd{\"o}s}, \bibinfo{person}{Evimaria Terzi}, {and} \bibinfo{person}{Azer
  Bestavros}.} \bibinfo{year}{2012}\natexlab{}.
\newblock \showarticletitle{A framework for the evaluation and management of
  network centrality}. In \bibinfo{booktitle}{\emph{Proceedings of the 2012
  SIAM International Conference on Data Mining}}. SIAM,
  \bibinfo{pages}{427--438}.
\newblock


\bibitem[\protect\citeauthoryear{Koltchinskii and Panchenko}{Koltchinskii and
  Panchenko}{2000}]%
        {KoltchinskiiP00}
\bibfield{author}{\bibinfo{person}{Vladimir Koltchinskii} {and}
  \bibinfo{person}{Dmitriy Panchenko}.} \bibinfo{year}{2000}\natexlab{}.
\newblock \showarticletitle{{R}ademacher processes and bounding the risk of
  function learning}.
\newblock In \bibinfo{booktitle}{\emph{High dimensional probability II}}.
  \bibinfo{publisher}{Springer}, \bibinfo{pages}{443--457}.
\newblock


\bibitem[\protect\citeauthoryear{Li, Long, and Srinivasan}{Li
  et~al\mbox{.}}{2001}]%
        {li2001improved}
\bibfield{author}{\bibinfo{person}{Yi Li}, \bibinfo{person}{Philip~M Long},
  {and} \bibinfo{person}{Aravind Srinivasan}.} \bibinfo{year}{2001}\natexlab{}.
\newblock \showarticletitle{Improved bounds on the sample complexity of
  learning}.
\newblock \bibinfo{journal}{\emph{J. Comput. System Sci.}}
  \bibinfo{volume}{62}, \bibinfo{number}{3} (\bibinfo{year}{2001}),
  \bibinfo{pages}{516--527}.
\newblock


\bibitem[\protect\citeauthoryear{Magnien, Latapy, and Habib}{Magnien
  et~al\mbox{.}}{2009}]%
        {magnien2009fast}
\bibfield{author}{\bibinfo{person}{Cl{\'e}mence Magnien},
  \bibinfo{person}{Matthieu Latapy}, {and} \bibinfo{person}{Michel Habib}.}
  \bibinfo{year}{2009}\natexlab{}.
\newblock \showarticletitle{Fast computation of empirically tight bounds for
  the diameter of massive graphs}.
\newblock \bibinfo{journal}{\emph{Journal of Experimental Algorithmics (JEA)}}
  \bibinfo{volume}{13} (\bibinfo{year}{2009}), \bibinfo{pages}{1--10}.
\newblock


\bibitem[\protect\citeauthoryear{Mahmoody, Tsourakakis, and Upfal}{Mahmoody
  et~al\mbox{.}}{2016}]%
        {mahmoody2016scalable}
\bibfield{author}{\bibinfo{person}{Ahmad Mahmoody},
  \bibinfo{person}{Charalampos~E Tsourakakis}, {and} \bibinfo{person}{Eli
  Upfal}.} \bibinfo{year}{2016}\natexlab{}.
\newblock \showarticletitle{Scalable betweenness centrality maximization via
  sampling}. In \bibinfo{booktitle}{\emph{Proceedings of the 22nd ACM SIGKDD
  International Conference on Knowledge Discovery and Data Mining}}.
  \bibinfo{pages}{1765--1773}.
\newblock


\bibitem[\protect\citeauthoryear{Martello and Toth}{Martello and Toth}{1990}]%
        {toth1990knapsack}
\bibfield{author}{\bibinfo{person}{Silvano Martello} {and}
  \bibinfo{person}{Paolo Toth}.} \bibinfo{year}{1990}\natexlab{}.
\newblock \bibinfo{booktitle}{\emph{Knapsack Problems: Algorithms and Computer
  Implementations}}.
\newblock \bibinfo{publisher}{John Wiley \& Sons, Inc.},
  \bibinfo{address}{USA}.
\newblock
\showISBNx{0471924202}


\bibitem[\protect\citeauthoryear{Massart}{Massart}{2000}]%
        {Massart00}
\bibfield{author}{\bibinfo{person}{Pascal Massart}.}
  \bibinfo{year}{2000}\natexlab{}.
\newblock \showarticletitle{Some applications of concentration inequalities to
  statistics}.
\newblock \bibinfo{journal}{\emph{Annales de la Facult{\'e} des sciences de
  Toulouse: Math{\'e}matiques}} \bibinfo{volume}{9}, \bibinfo{number}{2}
  (\bibinfo{year}{2000}), \bibinfo{pages}{245--303}.
\newblock


\bibitem[\protect\citeauthoryear{Maurer and Pontil}{Maurer and Pontil}{2009}]%
        {maurer2009empirical}
\bibfield{author}{\bibinfo{person}{Andreas Maurer} {and}
  \bibinfo{person}{Massimiliano Pontil}.} \bibinfo{year}{2009}\natexlab{}.
\newblock \showarticletitle{Empirical {B}ernstein bounds and sample variance
  penalization}.
\newblock \bibinfo{journal}{\emph{arXiv preprint arXiv:0907.3740}}
  (\bibinfo{year}{2009}).
\newblock


\bibitem[\protect\citeauthoryear{Mitzenmacher and Upfal}{Mitzenmacher and
  Upfal}{2017}]%
        {mitzenmacher2017probability}
\bibfield{author}{\bibinfo{person}{Michael Mitzenmacher} {and}
  \bibinfo{person}{Eli Upfal}.} \bibinfo{year}{2017}\natexlab{}.
\newblock \bibinfo{booktitle}{\emph{Probability and computing: Randomization
  and probabilistic techniques in algorithms and data analysis}}.
\newblock \bibinfo{publisher}{Cambridge university press}.
\newblock


\bibitem[\protect\citeauthoryear{Newman}{Newman}{2018}]%
        {newman2018networks}
\bibfield{author}{\bibinfo{person}{Mark Newman}.}
  \bibinfo{year}{2018}\natexlab{}.
\newblock \bibinfo{booktitle}{\emph{Networks}}.
\newblock \bibinfo{publisher}{Oxford university press}.
\newblock


\bibitem[\protect\citeauthoryear{Oneto, Ghio, Anguita, and Ridella}{Oneto
  et~al\mbox{.}}{2013}]%
        {oneto2013improved}
\bibfield{author}{\bibinfo{person}{Luca Oneto}, \bibinfo{person}{Alessandro
  Ghio}, \bibinfo{person}{Davide Anguita}, {and} \bibinfo{person}{Sandro
  Ridella}.} \bibinfo{year}{2013}\natexlab{}.
\newblock \showarticletitle{An improved analysis of the Rademacher
  data-dependent bound using its self bounding property}.
\newblock \bibinfo{journal}{\emph{Neural Networks}}  \bibinfo{volume}{44}
  (\bibinfo{year}{2013}), \bibinfo{pages}{107--111}.
\newblock


\bibitem[\protect\citeauthoryear{Pellegrina}{Pellegrina}{2020}]%
        {pellegrina2020sharper}
\bibfield{author}{\bibinfo{person}{Leonardo Pellegrina}.}
  \bibinfo{year}{2020}\natexlab{}.
\newblock \showarticletitle{Sharper convergence bounds of Monte Carlo
  Rademacher Averages through Self-Bounding functions}.
\newblock \bibinfo{journal}{\emph{arXiv preprint arXiv:2010.12103}}
  (\bibinfo{year}{2020}).
\newblock


\bibitem[\protect\citeauthoryear{Pellegrina}{Pellegrina}{2021}]%
        {leogrinaphdthesis}
\bibfield{author}{\bibinfo{person}{Leonardo Pellegrina}.}
  \bibinfo{year}{2021}\natexlab{}.
\newblock \bibinfo{booktitle}{\emph{Rigorous and Efficient Algorithms for
  Significant and Approximate Pattern Mining}}.
\newblock \bibinfo{publisher}{Ph.D. Thesis.
  \url{http://www.dei.unipd.it/~pellegri/thesis/leonardo_pellegrina_tesi.pdf}}.
\newblock


\bibitem[\protect\citeauthoryear{Pellegrina, Cousins, Vandin, and
  Riondato}{Pellegrina et~al\mbox{.}}{2020}]%
        {pellegrina2020mcrapper}
\bibfield{author}{\bibinfo{person}{Leonardo Pellegrina}, \bibinfo{person}{Cyrus
  Cousins}, \bibinfo{person}{Fabio Vandin}, {and} \bibinfo{person}{Matteo
  Riondato}.} \bibinfo{year}{2020}\natexlab{}.
\newblock \showarticletitle{MCRapper: Monte-Carlo Rademacher Averages for Poset
  Families and Approximate Pattern Mining}. In
  \bibinfo{booktitle}{\emph{Proceedings of the 26th ACM SIGKDD International
  Conference on Knowledge Discovery \& Data Mining}}.
  \bibinfo{pages}{2165--2174}.
\newblock


\bibitem[\protect\citeauthoryear{Pellegrina, Riondato, and Vandin}{Pellegrina
  et~al\mbox{.}}{2019}]%
        {PellegrinaRV19a}
\bibfield{author}{\bibinfo{person}{Leonardo Pellegrina},
  \bibinfo{person}{Matteo Riondato}, {and} \bibinfo{person}{Fabio Vandin}.}
  \bibinfo{year}{2019}\natexlab{}.
\newblock \showarticletitle{{SPuManTE}: Significant Pattern Mining with
  Unconditional Testing}. In \bibinfo{booktitle}{\emph{Proceedings of the 25th
  ACM SIGKDD International Conference on Knowledge Discovery \& Data Mining}}
  (Anchorage, AK, USA) \emph{(\bibinfo{series}{KDD '19})}.
  \bibinfo{publisher}{ACM}, \bibinfo{address}{New York, NY, USA},
  \bibinfo{pages}{1528--1538}.
\newblock
\showISBNx{978-1-4503-6201-6}
\urldef\tempurl%
\url{https://doi.org/10.1145/3292500.3330978}
\showDOI{\tempurl}


\bibitem[\protect\citeauthoryear{Pollard}{Pollard}{2012}]%
        {pollard2012convergence}
\bibfield{author}{\bibinfo{person}{David Pollard}.}
  \bibinfo{year}{2012}\natexlab{}.
\newblock \bibinfo{booktitle}{\emph{Convergence of stochastic processes}}.
\newblock \bibinfo{publisher}{Springer Science \& Business Media}.
\newblock


\bibitem[\protect\citeauthoryear{Provost, Jensen, and Oates}{Provost
  et~al\mbox{.}}{1999}]%
        {provost1999efficient}
\bibfield{author}{\bibinfo{person}{Foster Provost}, \bibinfo{person}{David
  Jensen}, {and} \bibinfo{person}{Tim Oates}.} \bibinfo{year}{1999}\natexlab{}.
\newblock \showarticletitle{Efficient progressive sampling}. In
  \bibinfo{booktitle}{\emph{Proceedings of the fifth ACM SIGKDD international
  conference on Knowledge discovery and data mining}}. \bibinfo{pages}{23--32}.
\newblock


\bibitem[\protect\citeauthoryear{Riondato and Kornaropoulos}{Riondato and
  Kornaropoulos}{2016}]%
        {RiondatoK15}
\bibfield{author}{\bibinfo{person}{Matteo Riondato} {and}
  \bibinfo{person}{Evgenios~M. Kornaropoulos}.}
  \bibinfo{year}{2016}\natexlab{}.
\newblock \showarticletitle{Fast approximation of betweenness centrality
  through sampling}.
\newblock \bibinfo{journal}{\emph{Data Mining and Knowledge Discovery}}
  \bibinfo{volume}{30}, \bibinfo{number}{2} (\bibinfo{year}{2016}),
  \bibinfo{pages}{438--475}.
\newblock


\bibitem[\protect\citeauthoryear{Riondato and Upfal}{Riondato and
  Upfal}{2014}]%
        {RiondatoU14}
\bibfield{author}{\bibinfo{person}{Matteo Riondato} {and} \bibinfo{person}{Eli
  Upfal}.} \bibinfo{year}{2014}\natexlab{}.
\newblock \showarticletitle{Efficient Discovery of Association Rules and
  Frequent Itemsets through Sampling with Tight Performance Guarantees}.
\newblock \bibinfo{journal}{\emph{ACM Trans. Knowl. Disc. from Data}}
  \bibinfo{volume}{8}, \bibinfo{number}{4} (\bibinfo{year}{2014}),
  \bibinfo{pages}{20}.
\newblock
\urldef\tempurl%
\url{https://doi.org/10.1145/2629586}
\showDOI{\tempurl}


\bibitem[\protect\citeauthoryear{Riondato and Upfal}{Riondato and
  Upfal}{2015}]%
        {RiondatoU15}
\bibfield{author}{\bibinfo{person}{Matteo Riondato} {and} \bibinfo{person}{Eli
  Upfal}.} \bibinfo{year}{2015}\natexlab{}.
\newblock \showarticletitle{Mining Frequent Itemsets through Progressive
  Sampling with {R}ademacher Averages}. In
  \bibinfo{booktitle}{\emph{Proceedings of the 21st ACM SIGKDD International
  Conference on Knowledge Discovery and Data Mining}}
  \emph{(\bibinfo{series}{KDD '15})}. \bibinfo{publisher}{ACM},
  \bibinfo{pages}{1005--1014}.
\newblock


\bibitem[\protect\citeauthoryear{Riondato and Upfal}{Riondato and
  Upfal}{2018}]%
        {RiondatoU18}
\bibfield{author}{\bibinfo{person}{Matteo Riondato} {and} \bibinfo{person}{Eli
  Upfal}.} \bibinfo{year}{2018}\natexlab{}.
\newblock \showarticletitle{{ABRA}: Approximating Betweenness Centrality in
  Static and Dynamic Graphs with {R}ademacher Averages}.
\newblock \bibinfo{journal}{\emph{ACM Trans. Knowl. Disc. from Data}}
  \bibinfo{volume}{12}, \bibinfo{number}{5} (\bibinfo{year}{2018}),
  \bibinfo{pages}{61}.
\newblock


\bibitem[\protect\citeauthoryear{Riondato and Vandin}{Riondato and
  Vandin}{2018}]%
        {RiondatoV18}
\bibfield{author}{\bibinfo{person}{Matteo Riondato} {and}
  \bibinfo{person}{Fabio Vandin}.} \bibinfo{year}{2018}\natexlab{}.
\newblock \showarticletitle{{MiSoSouP}: Mining Interesting Subgroups with
  Sampling and Pseudodimension}. In \bibinfo{booktitle}{\emph{Proc.~24th ACM
  SIGKDD Int.~Conf.~Knowl.~Disc.~and Data Mining}} \emph{(\bibinfo{series}{KDD
  '18})}. \bibinfo{publisher}{ACM}, \bibinfo{pages}{2130--2139}.
\newblock


\bibitem[\protect\citeauthoryear{Saha, Brokkelkamp, Velaj, Khan, and
  Bonchi}{Saha et~al\mbox{.}}{2021}]%
        {saha2021shortest}
\bibfield{author}{\bibinfo{person}{Arkaprava Saha}, \bibinfo{person}{Ruben
  Brokkelkamp}, \bibinfo{person}{Yllka Velaj}, \bibinfo{person}{Arijit Khan},
  {and} \bibinfo{person}{Francesco Bonchi}.} \bibinfo{year}{2021}\natexlab{}.
\newblock \showarticletitle{Shortest paths and centrality in uncertain
  networks}.
\newblock \bibinfo{journal}{\emph{Proceedings of the VLDB Endowment}}
  \bibinfo{volume}{14}, \bibinfo{number}{7} (\bibinfo{year}{2021}),
  \bibinfo{pages}{1188--1201}.
\newblock


\bibitem[\protect\citeauthoryear{Santoro and Sarpe}{Santoro and Sarpe}{2022}]%
        {santoro2022onbra}
\bibfield{author}{\bibinfo{person}{Diego Santoro} {and} \bibinfo{person}{Ilie
  Sarpe}.} \bibinfo{year}{2022}\natexlab{}.
\newblock \showarticletitle{ONBRA: Rigorous Estimation of the Temporal
  Betweenness Centrality in Temporal Networks}.
\newblock \bibinfo{journal}{\emph{arXiv preprint arXiv:2203.00653}}
  (\bibinfo{year}{2022}).
\newblock


\bibitem[\protect\citeauthoryear{Santoro, Tonon, and Vandin}{Santoro
  et~al\mbox{.}}{2020}]%
        {santoro2020mining}
\bibfield{author}{\bibinfo{person}{Diego Santoro}, \bibinfo{person}{Andrea
  Tonon}, {and} \bibinfo{person}{Fabio Vandin}.}
  \bibinfo{year}{2020}\natexlab{}.
\newblock \showarticletitle{Mining Sequential Patterns with VC-Dimension and
  Rademacher Complexity}.
\newblock \bibinfo{journal}{\emph{Algorithms}} \bibinfo{volume}{13},
  \bibinfo{number}{5} (\bibinfo{year}{2020}), \bibinfo{pages}{123}.
\newblock


\bibitem[\protect\citeauthoryear{Sariy{\"u}ce, Saule, Kaya, and
  {\c{C}}ataly{\"u}rek}{Sariy{\"u}ce et~al\mbox{.}}{2013}]%
        {sariyuce2013shattering}
\bibfield{author}{\bibinfo{person}{Ahmet~Erdem Sariy{\"u}ce},
  \bibinfo{person}{Erik Saule}, \bibinfo{person}{Kamer Kaya}, {and}
  \bibinfo{person}{{\"U}mit~V {\c{C}}ataly{\"u}rek}.}
  \bibinfo{year}{2013}\natexlab{}.
\newblock \showarticletitle{Shattering and compressing networks for betweenness
  centrality}. In \bibinfo{booktitle}{\emph{Proceedings of the 2013 SIAM
  International Conference on Data Mining}}. SIAM, \bibinfo{pages}{686--694}.
\newblock


\bibitem[\protect\citeauthoryear{Sarpe and Vandin}{Sarpe and Vandin}{2021}]%
        {sarpe2021oden}
\bibfield{author}{\bibinfo{person}{Ilie Sarpe} {and} \bibinfo{person}{Fabio
  Vandin}.} \bibinfo{year}{2021}\natexlab{}.
\newblock \showarticletitle{odeN: Simultaneous Approximation of Multiple Motif
  Counts in Large Temporal Networks}. In \bibinfo{booktitle}{\emph{Proceedings
  of the 30th ACM International Conference on Information \& Knowledge
  Management}}. \bibinfo{pages}{1568--1577}.
\newblock


\bibitem[\protect\citeauthoryear{Shalev-Shwartz and Ben-David}{Shalev-Shwartz
  and Ben-David}{2014}]%
        {ShalevSBD14}
\bibfield{author}{\bibinfo{person}{Shai Shalev-Shwartz} {and}
  \bibinfo{person}{Shai Ben-David}.} \bibinfo{year}{2014}\natexlab{}.
\newblock \bibinfo{booktitle}{\emph{Understanding Machine Learning: From Theory
  to Algorithms}}.
\newblock \bibinfo{publisher}{Cambridge University Press}.
\newblock


\bibitem[\protect\citeauthoryear{Staudt, Sazonovs, and Meyerhenke}{Staudt
  et~al\mbox{.}}{2016}]%
        {networkit2016}
\bibfield{author}{\bibinfo{person}{Christian~L. Staudt},
  \bibinfo{person}{Aleksejs Sazonovs}, {and} \bibinfo{person}{Henning
  Meyerhenke}.} \bibinfo{year}{2016}\natexlab{}.
\newblock \showarticletitle{NetworKit: A tool suite for large-scale complex
  network analysis}.
\newblock \bibinfo{journal}{\emph{Network Science}} \bibinfo{volume}{4},
  \bibinfo{number}{4} (\bibinfo{year}{2016}), \bibinfo{pages}{508–530}.
\newblock
\urldef\tempurl%
\url{https://doi.org/10.1017/nws.2016.20}
\showDOI{\tempurl}


\bibitem[\protect\citeauthoryear{Van Der~Vaart, van~der Vaart, van~der Vaart,
  and Wellner}{Van Der~Vaart et~al\mbox{.}}{1996}]%
        {van1996weak}
\bibfield{author}{\bibinfo{person}{Aad~W Van Der~Vaart}, \bibinfo{person}{Aad
  van~der Vaart}, \bibinfo{person}{Adrianus~Willem van~der Vaart}, {and}
  \bibinfo{person}{Jon Wellner}.} \bibinfo{year}{1996}\natexlab{}.
\newblock \bibinfo{booktitle}{\emph{Weak convergence and empirical processes:
  with applications to statistics}}.
\newblock \bibinfo{publisher}{Springer Science \& Business Media}.
\newblock


\bibitem[\protect\citeauthoryear{Vapnik and Chervonenkis}{Vapnik and
  Chervonenkis}{1971}]%
        {Vapnik:1971aa}
\bibfield{author}{\bibinfo{person}{V.~N. Vapnik} {and} \bibinfo{person}{A.~Ya.
  Chervonenkis}.} \bibinfo{year}{1971}\natexlab{}.
\newblock \showarticletitle{On the Uniform Convergence of Relative Frequencies
  of Events to Their Probabilities}.
\newblock \bibinfo{journal}{\emph{Theory of Probability \& Its Applications}}
  \bibinfo{volume}{16}, \bibinfo{number}{2} (\bibinfo{year}{1971}),
  \bibinfo{pages}{264}.
\newblock
\urldef\tempurl%
\url{https://doi.org/10.1137/1116025}
\showDOI{\tempurl}


\bibitem[\protect\citeauthoryear{Watts and Strogatz}{Watts and
  Strogatz}{1998}]%
        {watts1998collective}
\bibfield{author}{\bibinfo{person}{Duncan~J Watts} {and}
  \bibinfo{person}{Steven~H Strogatz}.} \bibinfo{year}{1998}\natexlab{}.
\newblock \showarticletitle{Collective dynamics of ‘small-world’networks}.
\newblock \bibinfo{journal}{\emph{nature}} \bibinfo{volume}{393},
  \bibinfo{number}{6684} (\bibinfo{year}{1998}), \bibinfo{pages}{440--442}.
\newblock


\bibitem[\protect\citeauthoryear{Yoshida}{Yoshida}{2014}]%
        {yoshida2014almost}
\bibfield{author}{\bibinfo{person}{Yuichi Yoshida}.}
  \bibinfo{year}{2014}\natexlab{}.
\newblock \showarticletitle{Almost linear-time algorithms for adaptive
  betweenness centrality using hypergraph sketches}. In
  \bibinfo{booktitle}{\emph{Proceedings of the 20th ACM SIGKDD international
  conference on Knowledge discovery and data mining}}.
  \bibinfo{pages}{1416--1425}.
\newblock


\end{thebibliography}


\appendix
\section{Appendix}\label{sec:appendix}

\subsection{Proofs of Section~\ref{sec:empiricalpeeling}}
\label{sec:proofs}

We prove \Cref{thm:neweraboundhypercube}, our new concentration bound of the $c$-MCERA towards the ERA, one of our main technical contributions. 

\neweraboundhypercube*

\begin{proof}
We first invoke an important result on the concentration of functions uniformly distributed on the binary hypercube.
\begin{theorem}[Thm. 5.3, \cite{boucheron2013concentration}]
\label{thm:concentrationhyperc}
For $k > 0$, let a function $g : \{ -1 , 1 \}^k \rightarrow \mathbb{R}$ and assume that $X$ is uniformly distributed on $\{ -1 , 1 \}^k$. 
Let $v > 0$ be such that 
\begin{align*}
\sum_{i=1}^{k} \pars{ g\pars{x} - g \pars{ \overline{x}^i }  }^2_{+} \leq v 
\end{align*}
for all $x = ( x_1 , \dots , x_k ) \in \{-1 , 1\}^k$, where 
\[
\overline{x}^i = ( x_1 , \dots , x_{i-1} , -x_i , x_{i+1} , \dots , x_k )
\] 
is a copy of $x$ with the $i$-th component multiplied by $-1$, and $( b )_+ = \max \{ b , 0 \}$ is the positive part of $b \in \mathbb{R}$.
Then, the random variable $Z \doteq g(X)$ satisfies, for all $t > 0$,
\begin{align*}
\Pr \pars{ Z > \E \sqpars{ Z } + t } , \Pr \pars{ Z < \E \sqpars{ Z } - t } \leq \exp\pars{ -t^2 /v }.
\end{align*}
\end{theorem}
We first observe the equivalence 
\begin{align}
\pars{ g\pars{x} - g \pars{ \overline{x}^i }  }_+ = \pars{ g\pars{x} - g \pars{ \overline{x}^i } } \ind{ g\pars{x} > g \pars{ \overline{x}^i }}  
\label{eq:proofequivalence}
\end{align}
for the function $g$ as stated in Thm.~\ref{thm:concentrationhyperc}; 
consequently, to apply the result we only have to consider the cases in which $g\pars{x} > g ( \overline{x}^i )$, as otherwise \eqref{eq:proofequivalence} is equal to $0$. 
We define the function $g : \{ -1 , 1 \}^{c m} \rightarrow \mathbb{R}$ as $g(\vsigma) \doteq nm \mera$. For a given $\vsigma$, denote with $f^\star_j$ one of the functions attaining the supremum of $\sup_{f \in \F} \{ \sum_{i=1}^m \vsigma_{j,i} f(s_i) \}$. 
Note that we can safely assume the supremum to be always achieved as we assume functions $\in \F$ to be bounded (or, at least, we assume $\sup_{f\in \F} \max_{s \in \sample} | f(s) |$ is bounded). 
We define $\overline{\vsigma}_{j,i}$ as a copy of $\vsigma$ with the $(j,i)$ component $\vsigma_{j,i}$ of $\vsigma$ multiplied by $-1$.
Thus, $g(\overline{\vsigma}_{j,i})$ is evaluated as $g(\vsigma)$ but with $\vsigma_{j,i}$ having flipped sign. 
We first prove that $g ( \overline{\vsigma}_{j,i} ) \geq g\pars{\vsigma} - 2 \vsigma_{j,i} f^\star_j(s_i)$ for all $i \in [1,m]$ and for all $j \in [1,c]$,
following ideas developed in \cite{leogrinaphdthesis,pellegrina2020sharper}. 
From the definition of $g$ given above, we have
\begin{align*}
 g ( \overline{\vsigma}_{j,i} ) 
&= c m \erade^{c}_{m}(\F, \sample, \overline{\vsigma}_{j,i}) 
=  \sum_{ \substack{ h=1 \\ h \neq j }}^c \sup_{f\in\F} \brpars{  \sum_{t=1}^m\vsigma_{h,t}f(s_t) } + \sup_{f\in\F} \brpars{  \sum_{\substack{t=1 \\ t \neq i } }^m\vsigma_{j,t}f(s_t) - \vsigma_{j,i}f(s_i) } \\
&=  \sum_{ \substack{ h=1 \\ h \neq j }}^c \sup_{f\in\F} \brpars{  \sum_{t=1}^m\vsigma_{h,t}f(s_t) } + \sup_{f\in\F} \brpars{  \sum_{ t=1 }^m\vsigma_{j,t}f(s_t) - 2\vsigma_{j,i}f(s_i) }\\
&\geq  \sum_{ \substack{ h=1 \\ h \neq j }}^c \sup_{f\in\F} \brpars{  \sum_{t=1}^m\vsigma_{h,t}f(s_t) } + \sum_{ t=1 }^m\vsigma_{j,t}f^{\star}_{j}(s_t) - 2\vsigma_{j,i}f^{\star}_{j}(s_i)\\
& =  \sum_{ h=1 }^c \sup_{f\in\F} \brpars{  \sum_{t=1}^m\vsigma_{h,t}f(s_t) } - 2\vsigma_{j,i}f^{\star}_{j}(s_i)  
=  g(\vsigma) - 2\vsigma_{j,i}f^{\star}_{j}(s_i) .
\end{align*}
Therefore, we have that
\begin{align*}
\sum_{i=1}^{m} \sum_{j=1}^{c} \pars{ g\pars{\vsigma} - g \pars{ \overline{\vsigma}_{j,i} }  }^2_{+}  
&\leq  \sum_{i=1}^{m} \sum_{j=1}^{c} \pars{ g\pars{\vsigma} - \pars{ g\pars{\vsigma} - 2 \vsigma_{j,i} f^\star_j(s_i) }  }^2_{+} \\
& =  \sum_{i=1}^{m} \sum_{j=1}^{c} \pars{ 2 \vsigma_{j,i} f^\star_j(s_i)  }^2_{+}  
\leq \sum_{i=1}^{m} \sum_{j=1}^{c} 4 f^\star_j(s_i)^2 
 \leq 4cm \ewvar_{\F} .
\end{align*}
We apply Theorem~\ref{thm:concentrationhyperc} to $Z = g(\vsigma)$ with $v = 4cm\ewvar_{\F}$, obtaining
\begin{align*}
&\Pr \pars{ cm\mera > cm\era + t } \leq \exp\pars{ \frac{-t^2}{4cm\ewvar_{\F}} } , \numberthis \label{eq:proofsubstitution} \\
&\Pr \pars{ cm\mera < cm\era - t } \leq \exp\pars{ \frac{-t^2}{4cm\ewvar_{\F}} } .
\end{align*}
The substitution $\varepsilon = t/(cm)$ in \eqref{eq:proofsubstitution} 
yields the inequality
 \begin{align}
\Pr \left( \era > \mera + \varepsilon \right) \leq \exp \left( - \frac{c m \varepsilon^2}{4 \ewvar_{\F} }  \right)  . \label{eq:mceraerahypercube}
\end{align}
The statement follows from imposing the r.h.s. of \eqref{eq:mceraerahypercube} to be $\leq \delta$ and solving for~$\varepsilon$. 
\end{proof}


The proof of our main result (Theorem~\ref{thm:bound_dev}) 
builds on the combination of the most refined concentration inequalities relating Rademacher averages to the supremum deviation~\cite{boucheron2013concentration}, that we now introduce. 
Let the \emph{Rademacher complexity} $\rc$ of a set of functions $\F$ be defined as $\rc = \E_{\sample} \sqpars{ \era }$. The following central result relates $\rc$ to the \emph{expected} supremum deviation.
\begin{lemma}{Symmetrization Lemma \citep{ShalevSBD14,mitzenmacher2017probability}}
\label{symlemma}
Let $Z$ be either 
\begin{align*}
\sup_{f\in \F} \brpars{ \mu_{\sample}(f) - \mu_{\probdist}(f) }  
\text{ or }  
\sup_{f\in \F} \brpars{ \mu_{\probdist}(f) - \mu_{\sample}(f) } .
\end{align*}
It holds
$ \E_{\sample} \sqpars{ Z } \leq 2 \rc $.
\end{lemma}

The following gives a variance-dependent bound to the supremum deviation above its expectation, and it is due to Bousquet~\cite{bousquet2002bennett}.
\begin{theorem}{(Thm. 2.3 \cite{bousquet2002bennett})}\label{thm:sdbousquetbound}
  Let $\hat{\nu}_{\F} \geq \sup_{f\in \F} \left\lbrace Var(f) \right\rbrace$,
  and $Z$ be either 
$\sup_{f\in \F} \brpars{ \mu_{\sample}(f) - \mu_{\probdist}(f) } $ 
or 
$\sup_{f\in \F} \brpars{ \mu_{\probdist}(f) - \mu_{\sample}(f) } $. 
  Then, with probability at
  least $1 - \lambda$ over $\sample$, it holds 
  \begin{align*}
    Z \le \E_\sample\left[ Z \right] + \sqrt{\frac{2 \ln \left( \frac{1}{\lambda} \right)
    \left( \hat{\nu}_{\F} + 2 \E[ Z ] \right)}{m}}
        + \frac{ \ln \left( \frac{1}{\lambda} \right)}{3m} .
  \end{align*}
\end{theorem}

The next result bounds $\rc$ above its estimate $\era$.
\begin{theorem}{\cite{boucheron2013concentration,oneto2013improved}}
\label{thm:rcboundselfbounding}
With probability $\geq 1-\lambda$ over $\sample$, it holds
\begin{align*}
\rc \leq \era  + \sqrt{ \pars{\frac{\ln \pars{ \frac{1}{\lambda}}}{m}}^2 + \frac{2\ln\pars{\frac{1}{\lambda}} \era }{ m } }   + \frac{\ln \pars{ \frac{1}{\lambda}}}{m} .
\end{align*}
\end{theorem}

We are now ready to prove \Cref{thm:bound_dev}, the main technical contribution of \Cref{sec:algunif}.
\thmboundsupdev*
\begin{proof}
Define the events $E_{i , j}$, $E_j$, and $E$ as
\begin{align*}
&E_{1 , j} = \text{``} \erade\left(\F_j, \sample\right) > \tilde{\rade}_j \text{''}  , \enspace
E_{2 , j} = \text{``} \rade(\F_j, m) > \rade_j \text{''} , \\
&E_{3 , j} = \text{``} \sup_{f\in \F_j} \brpars{ \mu_{\sample}(f) - \mu_{\probdist}(f) }  > \varepsilon_{\F_j} \text{''} , \enspace
E_{4} = \text{``} \sup_{f\in \F_j} \brpars{ \mu_{\probdist}(f) - \mu_{\sample}(f) }  > \varepsilon_{\F_j} \text{''} , \\
&E_j = \text{``} \supdevj > \varepsilon_{\F_j} \text{''}  , \enspace 
E = \text{``} \exists j : \supdevj > \varepsilon_{\F_j} \text{''} .
\end{align*}
Note that the statement is proved if $\Pr(E) \leq \delta$, and that 
the event $E$ is contained in the event $\bigcup_j E_j$, meaning that $\Pr(E) \leq \Pr(\bigcup_j E_j) \leq \sum_j \Pr(E_j)$.
Moreover, each event $E_j$ is contained in the event $\bigcup_i E_{i, j}$, therefore 
$\Pr(E_j) \leq \sum_i \Pr(E_{i,j})$. 
To upper bound the probabilities of $E_{1 , j}$, $E_{2 , j}$, and $E_{3 , j}$, we apply 
\Cref{symlemma}, and 
Theorems~\ref{thm:sdbousquetbound}-\ref{thm:rcboundselfbounding}, replacing $\F$ by $\F_j$ and $\lambda$ by $\delta /(4t)$: 
\begin{enumerate}
\item $\Pr(E_{1,j}) \leq \delta /(4t)$ follows from \Cref{thm:neweraboundhypercube} (replacing $\delta$ by $\delta /(4t)$);
\item $\Pr(E_{2,j}) \leq \delta /(4t)$ follows from \Cref{thm:rcboundselfbounding};
\item From \Cref{symlemma} and twice the application of \Cref{thm:sdbousquetbound}, $\Pr(E_{3,j})$ and $\Pr(E_{4,j})$ are bounded below $\lambda = \delta /(4t)$.
\end{enumerate} 
The event $E$ is true with probability at most $\delta$, since 
$\Pr \pars{ E } 
\leq \sum_{j}\Pr(E_{j}) 
\leq \sum_{j}\sum_{i}\Pr(E_{i,j}) 
\leq \delta$. 
\end{proof}

We now prove \Cref{prop:varianceupperbounds}, which shows 
that $\sup_{f \in \F} Var(f)$ can be tightly estimated using the empirical wimpy variance $\ewvar_{\F}(\sample)$. 
We prove it using the self-bounding properties of the function $\ewvar_{\F}(\sample)$, proved in~\cite{leogrinaphdthesis}. 

\varianceupperbounds*
\begin{proof}
We use the fact that 
\begin{align*}
\sup_{f \in \F_j} Var(f) 
= \sup_{f \in \F_j} \brpars{ \E[f^2] - \E[f]^2 } 
\leq \sup_{f \in \F_j} \E[f^2]  ,
\end{align*}
so we focus on bounding the wimpy variance $\sup_{f \in \F_j} \E[f^2]$ of $\F_j$. We use following result. 
\begin{theorem}{(Thm.~7.5.8~\cite{leogrinaphdthesis})}\label{thm:wimpyvarupperbound}
With probability $\geq 1 - \lambda$ over $\sample$ it holds
\begin{align*}
\sup_{f \in \F} \E \sqpars{f^2} \leq \ewvar_{\F}(\sample) + \frac{\ln\pars{ \frac{1}{\lambda} } }{ m } + \sqrt{ \pars{\frac{\ln\pars{ \frac{1}{\lambda} } }{ m }}^2 + \frac{2 \ewvar_{\F}(\sample) \ln \pars{ \frac{1}{\lambda}} }{ m } } .
\end{align*}
\end{theorem}
We apply \Cref{thm:wimpyvarupperbound} to each set $\F_j$ (using $\lambda = \delta/t$), obtaining the statement. 
\end{proof}

\subsection{Proofs of Section~\ref{sec:algunif}}
\label{secappx:mainalgproofs}
We now prove \Cref{prop:main} which provides the probabilistic guarantees of \algname. 

\silvancorrect*

\begin{proof}
We prove that the statement is a consequence of 
the following facts: 
\begin{enumerate}
\item the $\nmcera$ for all $\F_j$ is computed correctly at line~\ref{alg:computemcera}; \label{item:mcera}
\item $\hat{\nu}_{j}$, at the end of iteration $i$, are computed, for all $j$, using the bound of \Cref{prop:varianceupperbounds} with probability 
$\delta/(2^{i+1}5)$
(i.e., in \Cref{prop:varianceupperbounds}, $\delta$ is replaced by 
$\delta/(2^{i+1}5)$); \label{item:supvars}
\item $\varepsilon_{\F_{j}}$, at the end of iteration $i$, are computed, for all $j$, using the bound of \Cref{thm:bound_dev} with probability 
$\delta/(2^{i+1}5)$
(i.e., in \Cref{thm:bound_dev} $\delta/4$ is replaced by 
$\delta/(2^{i+1}5)$); \label{item:supdevs}
\item $\hat{m}$ is computed s.t. $\sd(\F, \sample) \leq \varepsilon$ (with $|\sample| \geq \hat{m}$) with probability $\geq 1-\delta/2$; \label{item:mhat}
\item \algname\ stops at an iteration $i$ when $\max_j \varepsilon_j \leq \varepsilon$ or at the first $i$ s.t. $m_i \geq \hat{m}$. 
\end{enumerate}

We now prove the statement. 
Let $X$ be a random variable equal to the index of the iteration in which the algorithm stops. 
Let the events $A_i$ and $\hat{A}$ be defined as 
\begin{align*}
A_i &= \text{ \qt{at $i$-th iteration, $\exists v : f_{v} \in \F_{j} ,  \varepsilon_{\F_{j}} < \abs{b(v) - \tilde{b}(v)}$} } , \\
\hat{A} &= \text{ \qt{at the first $i$ s.t. $m_i \geq \hat{m}$, $\exists v : f_{v} \in \F_{j} ,  \varepsilon_{\F_{j}} < \abs{b(v) - \tilde{b}(v)}$} } .
\end{align*}
The algorithm is correct if both $\hat{A}$ and $A_X$ are false, thus we want to prove that $\Pr ( A_X \cup \hat{A} ) \leq \delta$. 
(Following analogous derivations discussed in~\cite{borassi2019kadabra}, we remark that $A_X$ depends on the random variable $X$, and it should not be confused with $A_i$ for a fixed $i$.) 
It is enough to prove that 
$\Pr ( A_X ) + \Pr( \hat{A} ) \leq \delta$; since we assume that the fact \eqref{item:mhat} is true, we already have $\Pr( \hat{A} ) \leq \delta/2$. 
We proceed to show that $\Pr ( A_X ) \leq \delta/2$ to prove the statement. 
Assuming facts \eqref{item:mcera}, \eqref{item:supvars}, and \eqref{item:supdevs} are all true, 
at the end of each iteration $i$, for $i=1,2,\dots$, it holds that $\sd(\F_{j} , \sample) \le \varepsilon_{j}$ for all $j$ with probability $\ge 1 - \delta/(2^{i+1})$, 
combining \Cref{prop:varianceupperbounds} and \Cref{thm:bound_dev}.
Consequently, $\Pr(A_i) \le \delta/(2^{i+1}) , \forall i \geq 1$. 
We have that
\begin{align*}
\Pr \pars{ A_X } 
= \sum_{i} \Pr \pars{ A_i \cap \qtm{ X = i } }
\leq \sum_{i} \Pr \pars{ A_i } 
\leq \sum_{i} \frac{\delta}{2^{i+1}}
\leq \delta/2 .
\end{align*}
\end{proof}

\subsection{Proofs of Section~\ref{sec:upperboundsamplesabs}}
\label{secappx:upperbounds}

We first prove our novel refined upper limit on the number of samples required to achieve a $\varepsilon$ absolute approximation of the \bc\ of all nodes in a graph.

\boundonnumberofsamples*

\begin{proof}[Proof of Thm.~\ref{thm:lowerboundsamplesaboslute}]
For a sample $\sample$ of size $m$, define the events $A$ and $A_v$ as
\begin{align*}
A &= \qtm{ \exists v \in V : \abs{ b(v) - \tilde{b}(v)} > \varepsilon } ,  \\
A_v &= \qtm{ \abs{ b(v) - \tilde{b}(v)} > \varepsilon } . 
\end{align*}
From a union bound, we have that 
\begin{align*}
\Pr \pars{ A } = \Pr \pars{ \bigcup_{v \in V} A_v } \leq \sum_{v \in V} \Pr \pars{A_v} . 
\end{align*}
Then, through the application of Hoeffding's and Bennet's inequalities~\cite{boucheron2013concentration}, Bathia and Davis inequality on variance~\cite{bhatia2000better} (which implies $Var(f_v) \leq g(b(v))$), 
and from the fact that $sup_{f_v \in \F} Var(f_v) \leq \hat{\nu}$, 
it holds, for all $v \in V$, 
\begin{align*}
\Pr \pars{A_v} &\leq 2 \min\brpars{ \exp\pars{ -2m \varepsilon^2 } , \exp\pars{ -m \min\{ \hat{\nu} ,  g(b(v)) \} h \pars{ \frac{\varepsilon}{\min\{ \hat{\nu} ,  g(b(v)) \}} } } }  . 
\end{align*}
By defining the functions $H(m,\varepsilon) = \exp\pars{ -2m \varepsilon^2 } $, $B(x,m,\varepsilon) = \exp\pars{ -m x h \pars{ \frac{\varepsilon}{x} } }$, and $\psi(x)$ (see below), we rewrite
\begin{align}
\Pr \pars{ A } 
\leq \sum_{v \in V} 2 \min\brpars{ H(m,\varepsilon) , B(\min\{ \hat{\nu} ,  g(b(v)) \} ,m,\varepsilon) } 
\doteq \sum_{v \in V} \psi(b(v)) . \label{eq:appxunionbound}
\end{align}
Since the values of $b(v)$ are not known a priori, it is not possible to directly compute the r.h.s. of~\eqref{eq:appxunionbound}. 
However, we show how to obtain a sharp upper bound by leveraging constraints on the possible values of $b(v)$ imposed by $\hat{\nu}$ and $\rho$.
To do so, we define an appropriate optimization problem w.r.t. the (unknown) values of $b(v)$. 
Denote with $m_x $ the number of nodes of $V$ that we assume have $b(v) = x$, for $x \in \Q \cap (0 , 1)$ (we can safely ignore nodes $v$ with $b(v) = 0$ or $b(v) = 1$, since $f_v$ is constant, and $\Pr(A_v)=0$); 
then, we define the following constrained optimization problem over the variables $m_x$:
\begin{align*}
\max & \sum_{x \in (0 , 1) , m_x > 0} m_x \psi(x) , \\
\text{subject to } & \sum_{x \in (0 , 1) , m_x > 0} x m_x \leq \rho , \\
& 0 \leq m_x \leq \frac{\rho}{x} , m_x \in \N . 
\end{align*}
The first contraint follows from $\sum_{v \in V} b(v) \leq \rho$, while the second set of constraints imposes that $m_x$ are positive integers and that there cannot be more than $\rho/x$ nodes with $b(v) = x$ by definition of $\rho$. 
Therefore, from \eqref{eq:appxunionbound}, the value of the objective function of the optimal solution of this problem upper bounds $\Pr(A)$, as we consider a worst-case configuration of the admissible values of $b(v)$ (i.e., the graph $G$ belongs to the space of all possible graphs with the above mentioned constraints). 
We recognize this formulation as a specific instance of the Bounded Knapsack Problem (BKP)~\citep{toth1990knapsack} over the variables $m_x$, where items with label $x$ are selected $m_x$ times, with unitary profit $\psi(x)$ and weight $x$; each item can be selected at most $\rho / x$ times, while the total knapsack capacity is $\rho$. 
We are not interested in the optimal solution of the integer problem, but rather in its upper bound given by the optimal solution of the continous relaxation, in which we let $m_x \in \R$.
Informally, such solution is obtained by choosing at maximum possible capacity every item in decreasing order of profit-weight ratio $\psi(x)/x$ until the total capacity is filled (Chapter~3~of~\citep{toth1990knapsack}). 
In our case, from the particular definition of the constraints, it is enough to fully select the item with higher profit-weight ratio to fill the entire knapsack. 
More formally, let $x^\star$ be
\begin{align*}
x^\star = \argmax_{ x \in (0,1) } \brpars{ \frac{\psi(x)}{x} } ;
\end{align*}
the optimal solution to the continous relaxation is $m_{x^\star} = \rho / x^\star $, $m_x = 0 , \forall x \neq x^\star$, while the optimal objective is equal to
\begin{align*}
\frac{\rho \psi(x^\star)}{x^\star} \geq \Pr\pars{A}.
\end{align*} 
We note that $x^\star$ always exists, as $\psi(x)/x$ is a positive, bounded and continous function in $(0 , 1)$. 
We now simplify the search of $x^\star$ limiting the range of $x$. 
First, we prove that $x^\star \in (0 , \hat{x}_2]$. 
Observe that the function $B(\min\{ \hat{\nu} ,  g(x) \} ,m,\varepsilon)$ (thus $\psi(x)$) is symmetric w.r.t. $1/2$ (since $g(x) = g(1-x)$); 
also, we note that 
\begin{align*}
\min \{ \hat{\nu} , g(x) \} =  
\begin{cases}
\hat{\nu} \text{, } \frac{1}{2} - \sqrt{\frac{1}{4} - \hat{\nu}} \leq x \leq \frac{1}{2} + \sqrt{\frac{1}{4} - \hat{\nu}} ,  \\
g(x) \text{, } \text{ otherwise,}
\end{cases}
\end{align*}
which means that $\psi(x)$ is constant for $x \in [\frac{1}{2} - \sqrt{\frac{1}{4} - \hat{\nu}}  ,  \frac{1}{2} + \sqrt{\frac{1}{4} - \hat{\nu}}]$. 
From these observation, we prove by contradiction that it holds that $x^\star \leq \frac{1}{2} - \sqrt{\frac{1}{4} - \hat{\nu}} = \hat{x}_2$. 
Assume that $x^\star > \frac{1}{2} + \sqrt{\frac{1}{4} - \hat{\nu}}$; then defining $x^\prime = 1-x^\star$ we have 
\begin{align*}
\frac{\psi(x^\star)}{x^\star} = \frac{\psi(x^\prime)}{1- x^\prime} < \frac{\psi(x^\prime)}{x^\prime} ,
\end{align*}
which contradicts the definition of $x^\star$.
Now, assume that $x^\star \in (\frac{1}{2} - \sqrt{\frac{1}{4} - \hat{\nu}}  ,  \frac{1}{2} + \sqrt{\frac{1}{4} - \hat{\nu}}]$; we have 
\begin{align*}
\frac{\psi(x^\star)}{x^\star} = \frac{\psi(\hat{x}_2)}{x^\star} < \frac{\psi(\hat{x}_2)}{\hat{x}_2} ,
\end{align*}
which is another contradiction; therefore, we conclude that $x^\star \in (0 , \hat{x}_2]$. 
We can write 
\begin{align*}
x^\star = \argmax_{ x \in (0,1) } \brpars{ \frac{\psi(x)}{x} }
= \argmax_{ x \in (0,\hat{x}_2] } \brpars{ \frac{2 \min\brpars{ H(m,\varepsilon) , B(g(x) ,m,\varepsilon) }}{x} }
\doteq \argmax_{ x \in (0,\hat{x}_2] } \brpars{ \frac{\psi_1(x)}{x} } .
\end{align*}
We now prove that $x^\star \in (0 , \hat{x}_1]$. 
Note that $\psi_1(x)$ is symmetric around $1/2$; furthermore,
\begin{align*}
\min\brpars{ H(m,\varepsilon) , B(g(x),m,\varepsilon) }  
 = 
\begin{cases}
B(g(x),m,\varepsilon) \text{, } g(x) h\pars{\frac{\varepsilon}{g(x)}} \geq 2 \varepsilon^2 ,  \\
H(m,\varepsilon) \text{, } \text{ otherwise.}
\end{cases}
\end{align*}
The function $g(x) h (\varepsilon/g(x))$ is monotonically increasing for $0 < x < 1/2$ and decreasing for $1/2 < x < 1$. 
Denote $h_1(x) = 1 + x - \sqrt{1 + 2 x}$ for $x \geq 0$. 
We use the facts that
$9h_1(x/3) \leq h(x) \leq x^2/2$ for all $x \geq 0$, so we have that $g(x) h\pars{{\varepsilon}/{g(x)}} = 2 \varepsilon^2$ holds for a pair $x_1$ and $x_2$ such that $ 1/2 - \sqrt{\varepsilon/3 - \varepsilon^2/9} \leq x_1 \leq 1/2$, $1/2 \leq x_2 \leq 1/2 + \sqrt{\varepsilon/3 - \varepsilon^2/9} $, and $1 = x_1 + x_2$ (this follows easily from knowing that the inverse of $h_1$ over $[0 , +\infty)$ is $h_1^{-1}(x) = x+\sqrt{2x}$).
We can write
\begin{align*}
\min\brpars{ H(m,\varepsilon) , B(g(x),m,\varepsilon) }  
 &= 
\begin{cases}
B(g(x),m,\varepsilon) \text{, } x \in (0 , x_1] \cup [x_2 , 1) , \\
H(m,\varepsilon) \text{, } \text{ otherwise,}
\end{cases} 
\end{align*}
Exploting the symmetry of $\psi_1(x)$, it is easy to prove that $x^\star \in (0 , x_1]$ by following a similar argument used above.
Then, note that $x_1 \leq \hat{x}_1$, $\psi_1(x) \leq 2 B(g(x),m,\varepsilon)$, and $\hat{x}= \min \{\hat{x}_1 , \hat{x}_2\}$; we obtain 
\begin{align*}
\Pr\pars{A} \leq 
\sup_{x \in (0,\min \{x_1 , \hat{x}_2\} ] } \brpars{ \frac{\rho \psi_1(x)}{x} }
\leq \sup_{x \in (0,\min \{\hat{x}_1 , \hat{x}_2\}]} \brpars{ \frac{\rho \psi_1(x)}{x} }
\leq \sup_{x \in (0, \hat{x}] } \brpars{ \frac{\rho 2 B(g(x),m,\varepsilon) }{x} } .
\end{align*}

We now show that if $m$ is chosen as assumed in the statement, it holds $\Pr(A) \leq \delta$, thus proving the Theorem. 
We have, for $x \in (0,\hat{x}]$, 
\begin{align*}
\frac{\rho 2 B(g(x),m,\varepsilon) }{x} \leq \delta 
~~~~\text{ if }~~~~~ 
m \geq  \frac{ \ln \pars{ \frac{ 2\rho }{ x \delta } }}{ \varx h \pars{ \frac{\varepsilon}{ \varx } } } ,
\end{align*}
which follows from \eqref{eq:lowerboundsamplesaboslute}. 
\end{proof}

\relativebounds*
\begin{proof}

Denote the event $A = \qtm{ \exists v : \abs{ b(v) - \ebc(v) } > d_r(b(v)) }$, and the events $A_v = \qtm{ \abs{ b(v) - \ebc(v) } > d_r(b(v)) }$. 
Our goal is to show that $\Pr (A) \leq \delta$, which proves the statement. 
First, through a union bound and Bennet's inequality we obtain that 
\begin{align}
\Pr \pars{ A } = \Pr \pars{ \bigcup_{v \in V} A_v } \leq \sum_{v \in V} \Pr \pars{A_v} \leq \sum_{v \in V} 2 B(\min\{\hat{\nu} , g(x)\} , m , d_r(b(v))) , \label{eq:toproverelbound}
\end{align}
with $B(x,m,y) = \exp(-mxh(y/x))$ 
and $h(x) = (1+x) \ln (1+x)-x$ for $x \geq 0$.
We define an optimization problem to upper bound~\eqref{eq:toproverelbound}:
 \begin{align*}
\max & \sum_{x \in (0,1)} 2 m_x  B(\min\{\hat{\nu} , g(x)\} , m , d_r(x)) , \\
\text{subject to } & \sum_{x \in (0,1)} x m_x \leq \rho , \\
& \sum_{x \in (0,1)} m_x \leq n \doteq |V| , \\
& 0 \leq m_x \leq \frac{\rho}{x} , m_x \in \R . 
\end{align*}
This problem is similar to the one introduced in the proof of Theorem~\ref{thm:lowerboundsamplesaboslute}, but with an additional constraint which imposes that the number of vertices $\sum_{x \in (0,1)} m_x$ with $b(v) \in (0,1)$ is at most $n$, where $n$ is the number of vertices $|V|$ of the graph $G$ 
(note that we ignore vertices with $b(v)=0$ or $b(v)=1$ since the corresponding estimator $f_v$ is constant, therefore $\Pr(A_v)=0$). 
We obtain an upper bound to the optimal solution by upper bounding the objective function and through a Lagrangean relaxation. 
Using the fact that $h(x) \geq 9h_1(x/3) , \forall x \geq 0$, with $h_1(x) = 1 + x - \sqrt{1+2x}$, we define the function $\psi(x)$ 
\begin{align*}
\psi(x) \doteq 2 \exp\pars{ -m \min \{ g(x) , \hat{\nu} \} \pars{ \frac{d_r(x)}{3 \min \{ g(x) , \hat{\nu} \} } } } \geq 2 B( \min \{ g(x) , \hat{\nu} \} ,m,d_r(x)).
\end{align*}
After straightforward simplifications, we observe from the definitions of $\psi(x)$ and $d_r(x)$ that 
\begin{align*}
\psi(x) = \frac{\delta}{2 \min \brpars{ \frac{\rho}{x} , n }}
= 
\begin{cases}
\frac{x \delta}{2 \rho} , x \geq \frac{\rho}{n}, \\
\frac{\delta}{2 n} , x < \frac{\rho}{n} . 
\end{cases}
\end{align*}
Fix any $\lambda \geq0$; 
the optimal solution of the following problem upper bounds the optimal solution  of the problem above:
\begin{align*}
\max & \sum_{x \in (0,1)} m_x \psi(x) + \lambda \pars{ n- \sum_{x \in (0,1)} m_x } , \\
\text{subject to } & \sum_{x \in (0,1)} x m_x \leq \rho , \\
& 0 \leq m_x \leq  \frac{\rho}{x} , m_x \in \R . 
\end{align*}
Rewriting the objective function, we have 
\begin{align*}
\sum_{x \in (0,1)} m_x \psi(x) + \lambda \pars{ n- \sum_{x \in (0,1)} m_x } = \sum_{x \in (0,1)} m_x \pars{\psi(x) - \lambda} +\lambda n .
\end{align*}
Ignoring the constant $\lambda n$, we expressed the problem as another Bounded Knapsack Problem formulation with profit $(\psi(x) - \lambda)$ for items with label $x$; we compute its optimal solution as done in the proof of Theorem~\ref{thm:lowerboundsamplesaboslute}. 
We fix $\lambda = \psi(\rho/n) = \delta/(2n)$, and remark that $\psi(x)$ is increasing with $x$: it holds $\psi(x) \leq \psi(\rho/n) , \forall x \in (0, \rho/n]$, and $\psi(x) \geq \psi(\rho/n) , \forall x \in [ \rho/n , 1)$; therefore,
define $x^\star$ as 
\begin{align*}
x^\star = \argmax_{ x \in (0,1) } \brpars{ \frac{\psi(x) - \psi(\rho/n)}{x} } ;
\end{align*}
it follows that $x^\star \geq \rho/n$ and $\psi(x^\star) \geq \psi(\rho/n)$ from observations made above. 
The optimal solution is given by $m_{x^\star} = \rho / x^\star , m_x = 0 , \forall x \neq x^\star$, with 
objective 
\begin{align*}
\frac{\rho}{x^\star}\psi(x^\star) + n \psi(\rho/n) -\frac{\rho}{x^\star}\psi(\rho/n) 
\leq \delta , 
\end{align*}
proving the statement. 
\end{proof}

Before proving data-dependent upper bounds to $\rho$, we show the following straightforward fact.

\begin{restatable}{lemma}{vertexdiameterrho}
Let $D$ be the vertex diameter of a graph $G$. Then 
$\rho \leq D$.
\end{restatable}

\begin{proof}
From the definition of $\rho$ and from linearity of expectation, we have
\begin{align*}
\rho 
& = \sum_{v \in V} b(v) 
= \sum_{v \in V} \E_{\sample} [ \tilde{b}(v) ] 
=  \sum_{v \in V} \E_{\sample} \sqpars{  \frac{1}{m} \sum_{\pi \in \sample} \tilde{b}_v(\pi)  } \\
& = \frac{1}{m} \E_{\sample} \sqpars{  \sum_{\pi \in \sample} \sum_{v \in V} \tilde{b}_v(\pi)  } 
\leq \frac{1}{m} \E_{\sample} \sqpars{  \sum_{\pi \in \sample} D  } \\
&= \frac{1}{m} \E_{\sample} \sqpars{  mD  } = D .
\end{align*}
\end{proof}

The result below shows that given a (not necessarily tight) upper bound to $D$, the average shortest path length $\rho$ can be sharply estimated as the \emph{average} number of internal nodes of the shortest paths in a sample $\sample$, resulting in a very efficient data-dependent bound. 

\avgsplbernstein*
\begin{proof}
From the definition of $\tilde{\rho}$, we have
\begin{align*}
\tilde{\rho} 
= \sum_{v \in V} \frac{1}{m} \sum_{\pi \in \sample} \tilde{b}_v(\pi)
= \frac{1}{m} \sum_{\pi \in \sample} \sum_{v \in V} \tilde{b}_v(\pi) ,
\end{align*}
with $\E_{\sample} [ \tilde{\rho}  ] = \rho$.
We recognize $\tilde{\rho}$ as an average of $m$ independent (bounded) random variables; 
it holds, for all $\pi$, that  
$0 \leq \sum_{v \in V} \tilde{b}_v(\pi) \leq D $,
and, from~\cite{bhatia2000better},  
$Var \pars{ \sum_{v \in V} \tilde{b}_v(\pi) } \leq (D - \rho ) \rho \leq D \rho $.
We define the random variables $X_{i}$, for $i \in [1,m]$, as
\begin{align*}
X_{i} = \rho - \sum_{v \in V} \tilde{b}_v(\pi_{i}) ,
\end{align*}
noting that $\E[X_{i}] = 0$, $X_{i} \leq \rho \leq D$, and $\E[X_{i}^{2}] = Var(X_{i}) \leq D \rho$. 
Therefore, we apply Bernstein's inequality~(Thm.~2.10~\cite{boucheron2013concentration}) to the sum $ \sum_{i=1}^{m} X_{i} = m(\rho - \tilde{\rho})$, obtaining
\begin{align*}
\Pr \pars{ \rho \geq \tilde{\rho} + \sqrt{\frac{2 \ln \pars{\frac{1}{\delta}} \rho D }{m}} + \frac{\ln\pars{\frac{1}{\delta}} D }{3m} } \leq \delta .
\end{align*}
To conclude the proof, we need the following straightforward Lemma.
\begin{lemma}
\label{lemma:fixedpoint}
Let $u,v,y \geq 0$. The fixed point of 
\[
r(x) = u + \sqrt{v + yx}
\]
is at
\[
x = u + \frac{y}{2} + \sqrt{\frac{y^{2}}{4} + uy + v}  .
\]
\end{lemma}
The statement follows by applying Lemma~\ref{lemma:fixedpoint} to the inequality 
\begin{align*}
\rho \leq \tilde{\rho} + \sqrt{\frac{2 \ln \pars{\frac{1}{\delta}} \rho D }{m}} + \frac{\ln\pars{\frac{1}{\delta}} D }{3m} .
\end{align*}
\end{proof}

\avgsplbernsteinemp*
\begin{proof}
As discussed in the proof of \Cref{thm:rhoestimation}, $\tilde{\rho}$ is an average of $m$ i.i.d. (bounded below by $D$) random variables.
The result follows from Corollary~5~of~\cite{maurer2009empirical} after scaling $\tilde{\rho}$ by $1/D$. 
\end{proof}

\subsection{Proofs of Section~\ref{alg:topk}}
\ifextended
\topkapproxstronger*
\fi

\begin{proof}[Proof of Proposition~\ref{prop:topkapproxstronger}]
We first note that, as the $k$-th most central node $v_{k}$ and its centrality $b(v_{k})$ are both unknown, we need a principled way to identify \emph{bounds} to $b(v_{k})$. 
Let $v_1^\ell , \dots , v_n^\ell$ be the sequence of nodes ordered according to $\ell(\cdot)$, such that $\ell(v_i^\ell) \geq \ell(v_{i+1}^\ell)$. 
Then, let $v_1^u , \dots , v_n^u$ be the sequence of nodes ordered according to $u(\cdot)$, such that $u(v_i^u) \geq u(v_{i+1}^u)$.
We have the following relations between $b(v_k)$, $u(v_k^u)$, and $\ell(v_k^\ell)$.
\begin{restatable}{lemma}{lowerboundvk}
\label{prop:boundsbvk}
It holds $\ell(v_k^\ell) \leq b(v_k) \leq u(v_k^u)$.
\end{restatable}
\begin{proof}
We prove the statement by contradiction. Assume that it holds $b(v_k) < \ell(v_k^\ell)$. From the definition of $\ell(v_k^\ell)$, that there are $k$ nodes $\{ v_i^\ell , i \leq k  \}$ such that $l(v_i^\ell) \geq \ell(v_k^\ell) > b(v_k)$, and it holds, for all of them, that $b(v_i^\ell) \geq \ell(v_i^\ell)$. This implies that there are $k$ nodes such that $b(v_i^\ell) > b(v_k)$, that is in contradiction with the definition of $v_k$. 
\end{proof}
We now continue proving the Proposition. 
The first guarantee is immediate from \Cref{prop:boundsbvk}: since it holds for all $v \in V$, that $b(v) \in CI_v$, we have that $u(v) \geq b(v) \geq \ell(v)$. Therefore, if $v \in TOP(k)$, then $b(v) \geq b(v_k)$, thus we have $u(v) \geq b(v) \geq b(v_k) \geq \ell(v_k^{\ell})$; 
this means that $v$ is in $\tilde{TOP}(k)$ as $u(v) \geq \ell(v_k^{\ell})$.
The second follows directly from \eqref{eq:topkapproxchecksstronger}. We now focus on the third. 
Let $\tilde{v}_1 , \dots , \tilde{v}_k$ be the sequence of nodes in order of $\tilde{b}(\cdot)$, such that $\tilde{b}(\tilde{v}_i) \geq \tilde{b}(\tilde{v}_{i+1})$. From \eqref{eq:topkapproxchecksstronger} we have that, for all $\tilde{v}_i$,
\begin{align*}
\frac{\tilde{b}(\tilde{v}_i)}{1+\eta} \leq \ell(\tilde{v}_i) \leq u(\tilde{v}_i) \leq \frac{\tilde{b}(\tilde{v}_i)}{1-\eta} ,
\end{align*}
but also that 
\begin{align*}
\frac{\tilde{b}(\tilde{v}_{i+1})}{1+\eta} \leq \frac{\tilde{b}(\tilde{v}_{i})}{1+\eta} , \text{ and } \frac{\tilde{b}(\tilde{v}_{i+1})}{1-\eta} \leq \frac{\tilde{b}(\tilde{v}_{i})}{1-\eta} .
\end{align*}
Notice that $v_k^\ell \not\eq v_k^u$ in general; nevertheless, we show it is possible to bound $u(v_k^u) - \ell(v_k^\ell)$. 
Considering $\tilde{v}_k$, we have that
\begin{align*}
\frac{\tilde{b}(\tilde{v}_k)}{1+\eta} \leq \ell(v_k^\ell) \leq u(v_k^u) \leq \frac{\tilde{b}(\tilde{v}_{k})}{1-\eta} ;
\end{align*}
the lower bound to $\ell(v_k^\ell)$ follows from the definition of $v_k^\ell$, since there are $k$ values $\ell(\tilde{v}_i) \geq \tilde{b}(\tilde{v}_i)/(1+\eta), i \geq k$; the upper bound to $u(v_k^u)$ follows from the fact that there are $k$ values $u(\tilde{v}_i) \leq \tilde{b}(\tilde{v}_i)/(1-\eta) , i \geq k$. 
We combine previous inequalities to obtain 
\begin{align*}
b(v) \geq l(v) 
 \geq \frac{\tilde{b}(v)}{1+\eta} 
 \geq u(v)\frac{1-\eta}{1+\eta} 
\geq \ell (v^\ell_k) \frac{1-\eta}{1+\eta}
 \geq b (v_k) \pars{\frac{1-\eta}{1+\eta} }^2 .
\end{align*}
\end{proof}

\subsection{Additional Experimental Results}

This Section presents additional experimental results.

\begin{figure}[ht]
\centering
\begin{subfigure}{.9\textwidth}
  \centering
  \includegraphics[width=\textwidth]{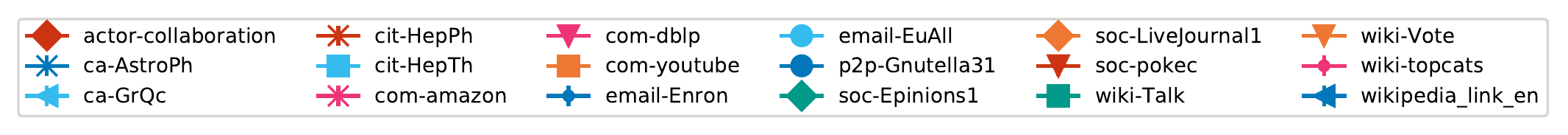}
\end{subfigure}
\begin{subfigure}{.49\textwidth}
  \centering
  \includegraphics[width=\textwidth]{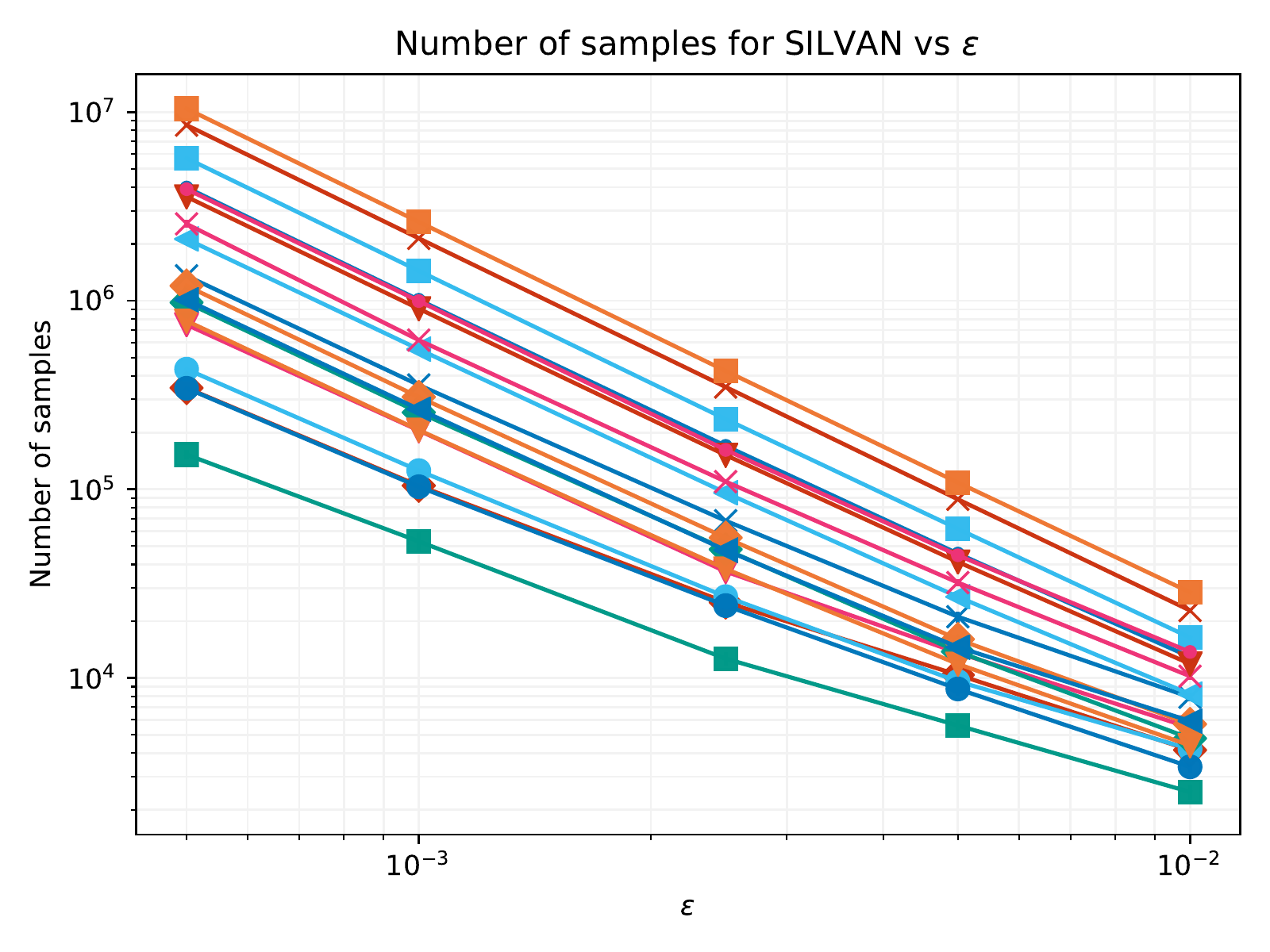}
  \caption{}
  \label{fig:samplesvsepssilvan}
\end{subfigure}
\begin{subfigure}{.49\textwidth}
  \centering
  \includegraphics[width=\textwidth]{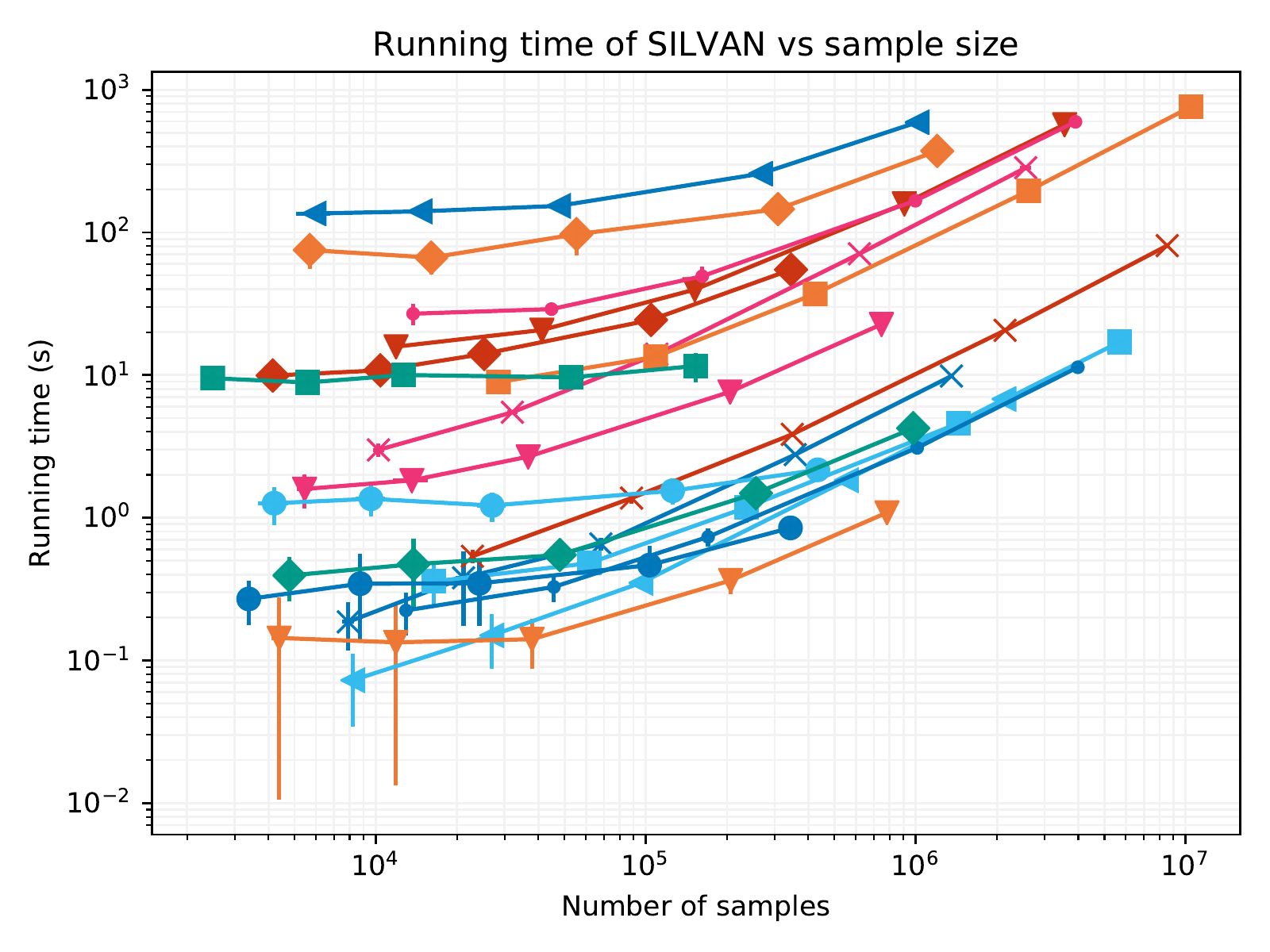}
  \caption{}
  \label{fig:samplesvstimesilvan}
\end{subfigure}
\begin{subfigure}{.49\textwidth}
  \centering
  \includegraphics[width=\textwidth]{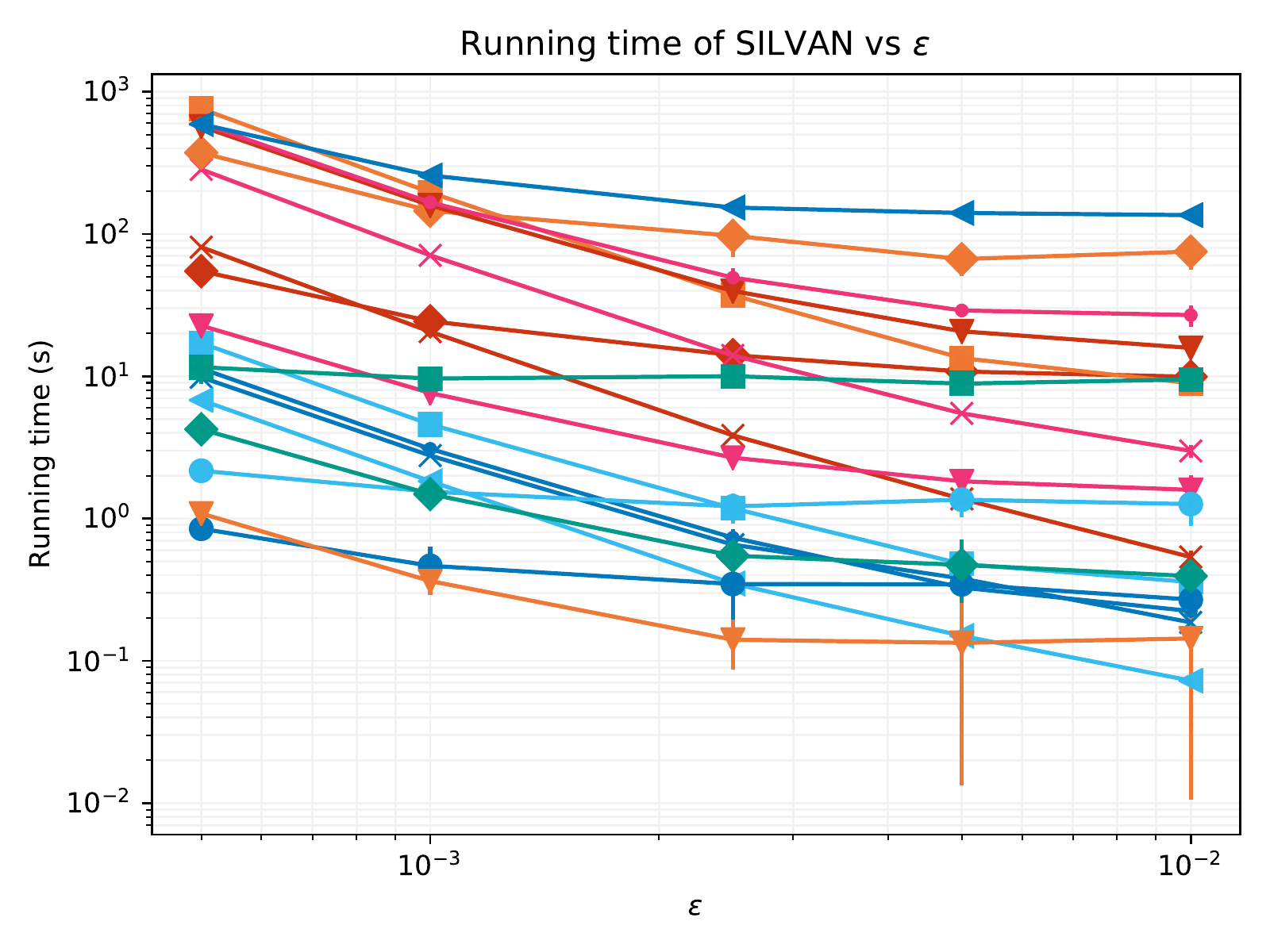}
  \caption{}
  \label{fig:timevsepssilvan}
\end{subfigure}
\caption{
Resources required by \algname\ for obtaining absolute $\varepsilon$ approximations. 
(a): Number of samples for \algname\ vs. $\varepsilon$. 
(b): Running times of \algname\ vs. Number of samples. 
(c): Running times of \algname\ vs. $\varepsilon$. 
}
\label{fig:silvanvsplots}
\end{figure}

\begin{figure}[ht]
\centering
\begin{subfigure}{.49\textwidth}
  \centering
  \includegraphics[width=\textwidth]{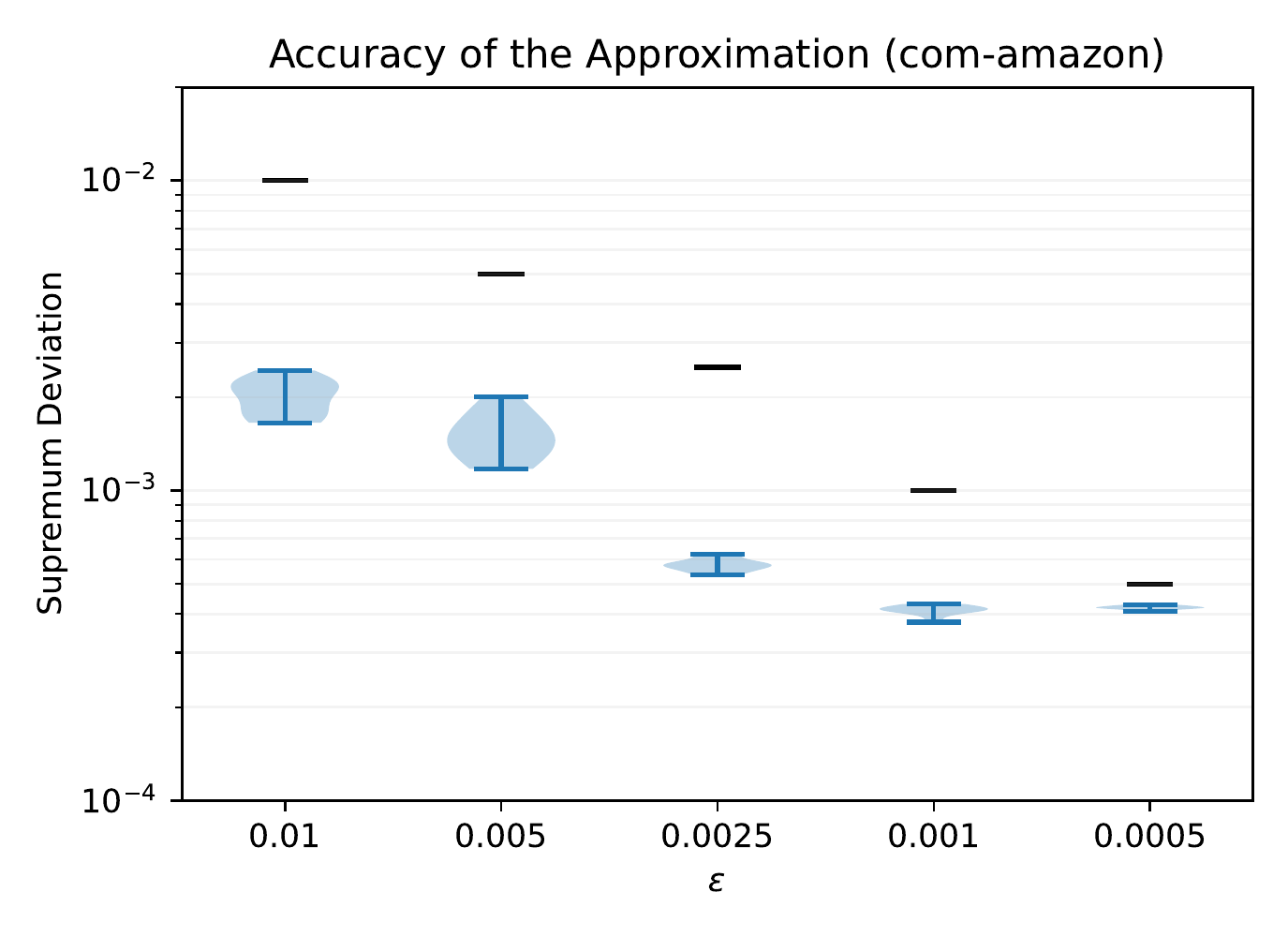}
\end{subfigure}
\begin{subfigure}{.49\textwidth}
  \centering
  \includegraphics[width=\textwidth]{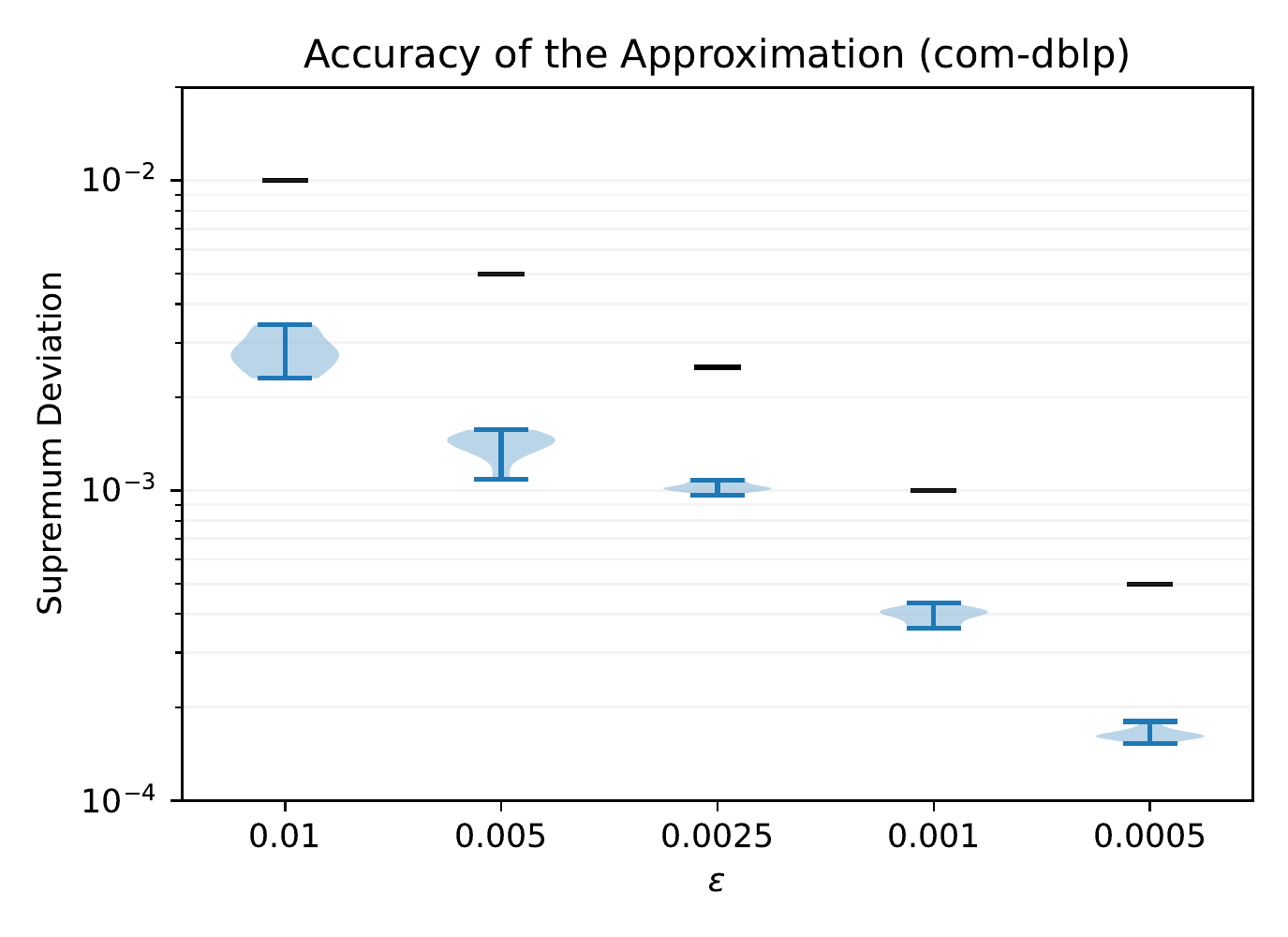}
\end{subfigure}
\begin{subfigure}{.49\textwidth}
  \centering
  \includegraphics[width=\textwidth]{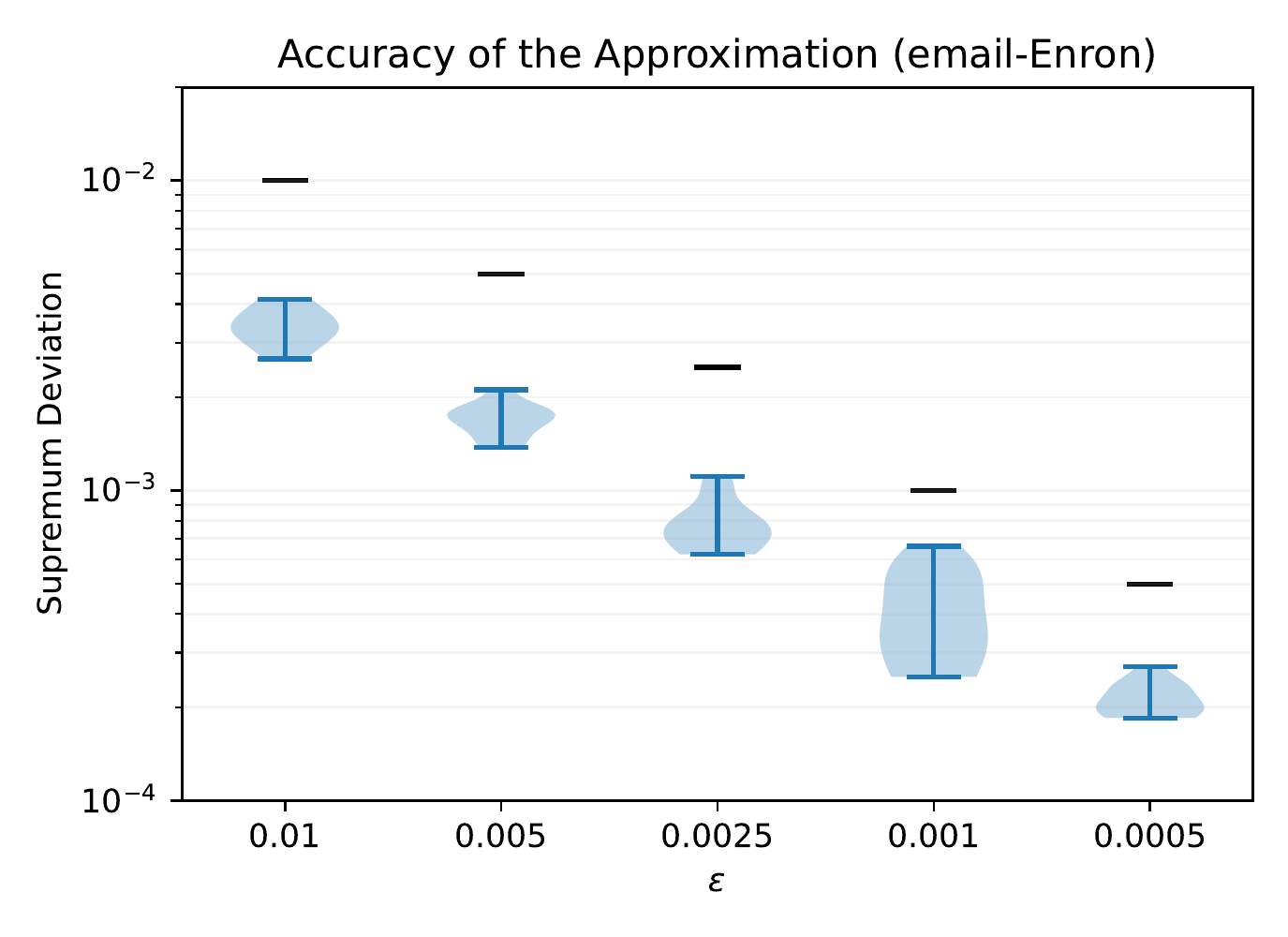}
\end{subfigure}
\begin{subfigure}{.49\textwidth}
  \centering
  \includegraphics[width=\textwidth]{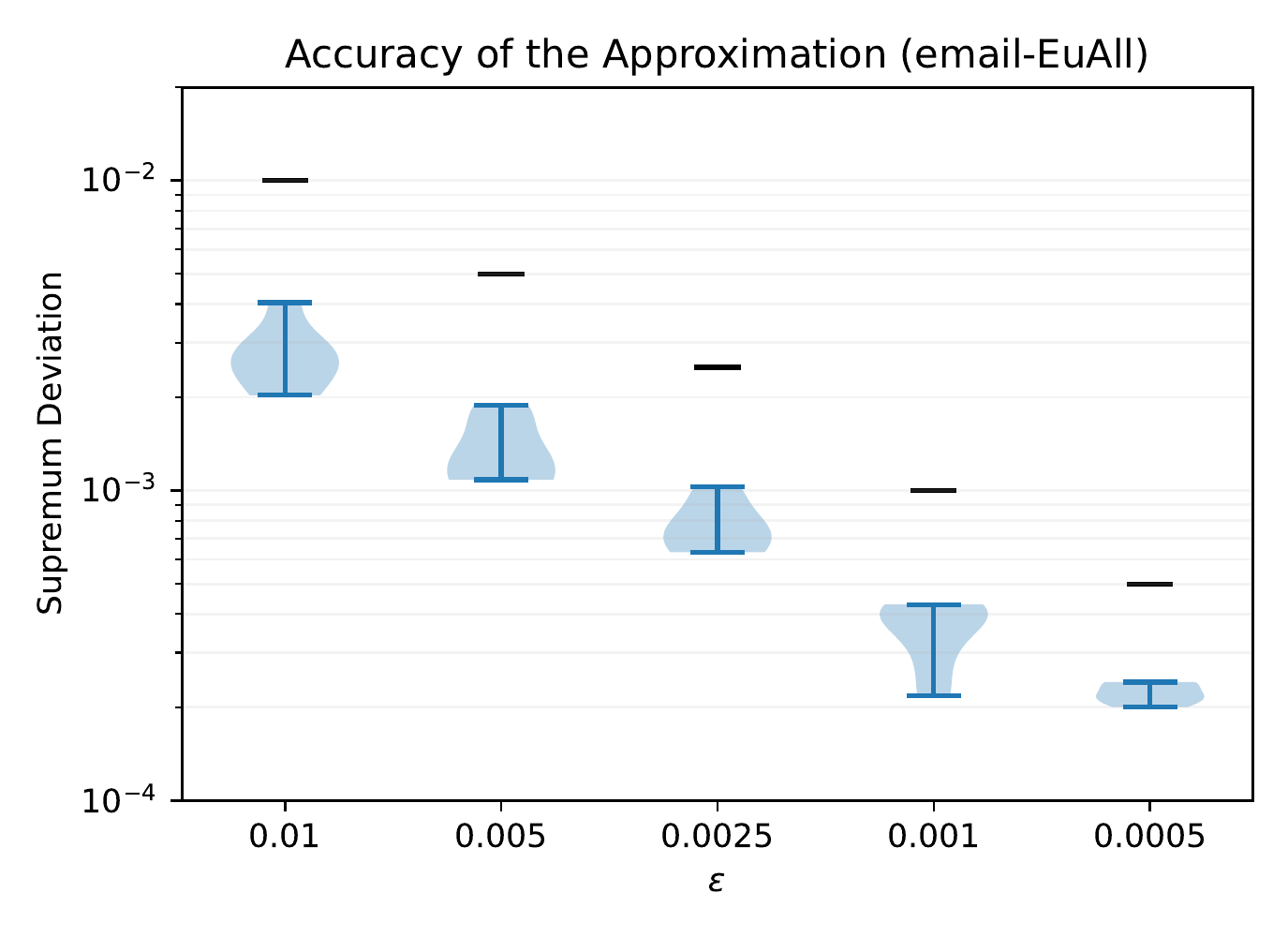}
\end{subfigure}
\begin{subfigure}{.49\textwidth}
  \centering
  \includegraphics[width=\textwidth]{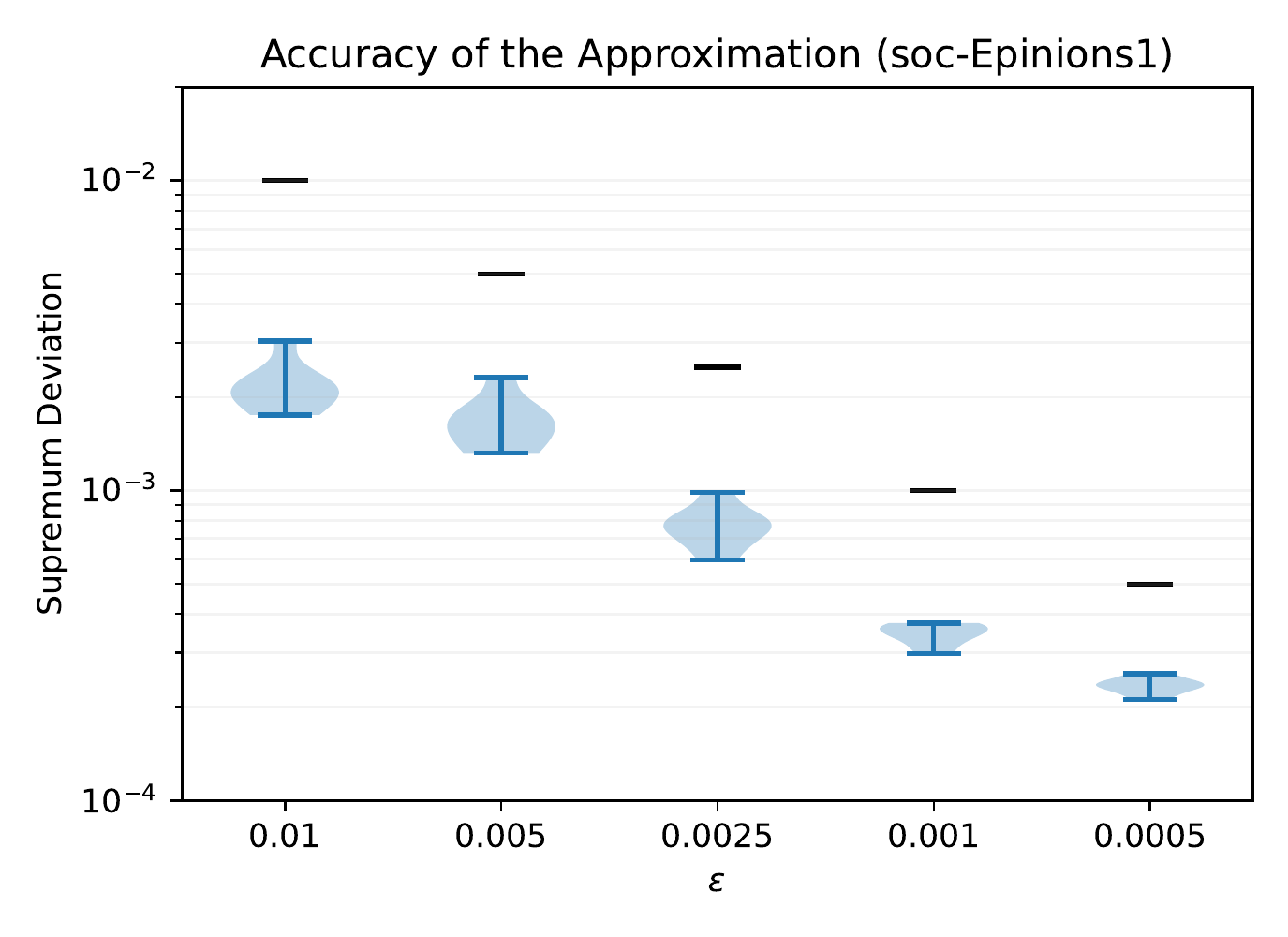}
\end{subfigure}
\begin{subfigure}{.49\textwidth}
  \centering
  \includegraphics[width=\textwidth]{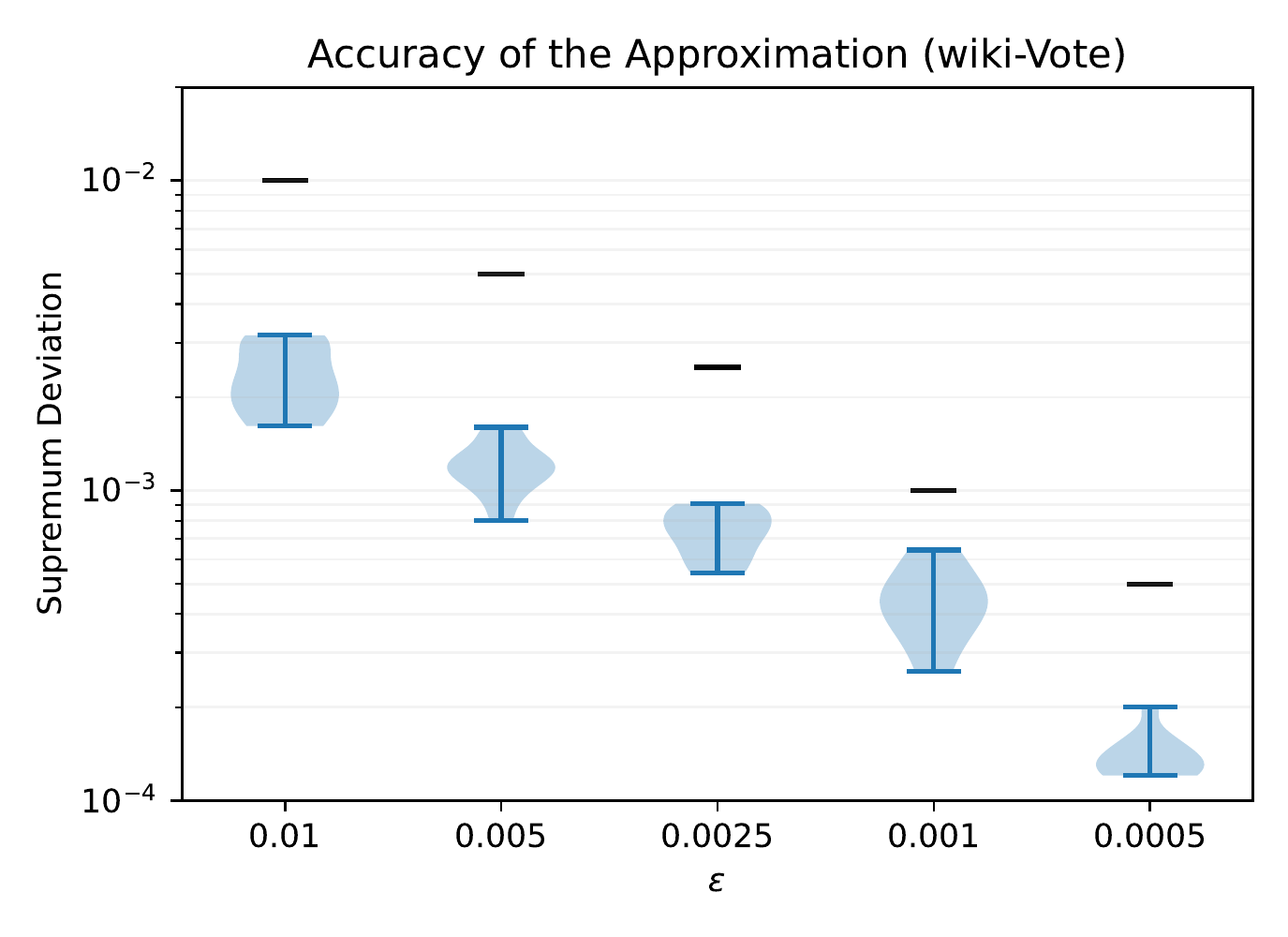}
\end{subfigure}
\caption{
Violin plot showing the empirical distribution of the supremum deviation $\supdev$ (blue violins) as function of $\varepsilon$ ($x$ axis, and black horizontal bars) for $3$ undirected and $3$ directed graphs over $10$ runs. 
}
\label{fig:accuracyapprox}
\end{figure}

\begin{figure}[ht]
\centering
\begin{subfigure}{.49\textwidth}
  \centering
  \includegraphics[width=\textwidth]{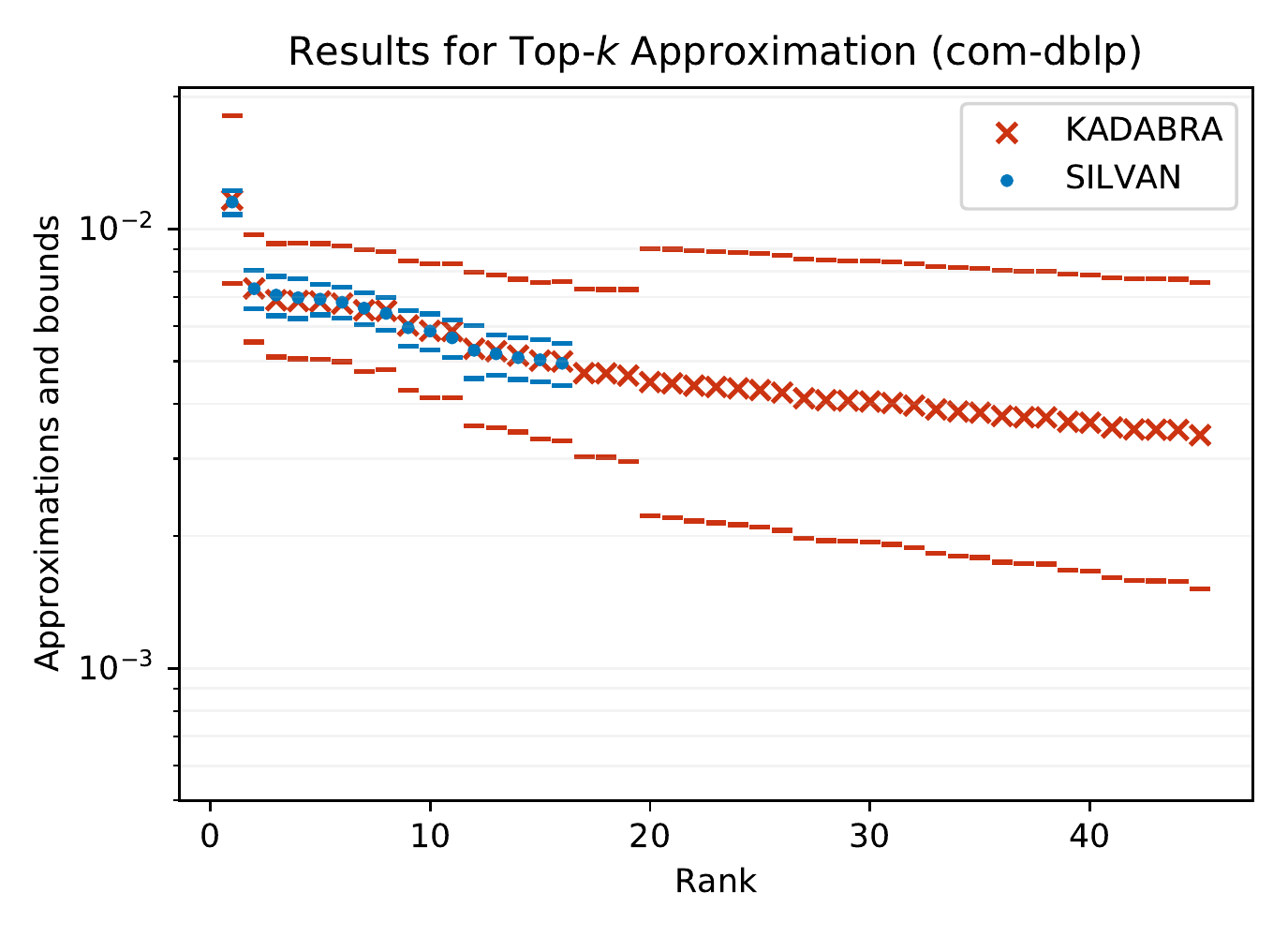}
\end{subfigure}
\begin{subfigure}{.49\textwidth}
  \centering
  \includegraphics[width=\textwidth]{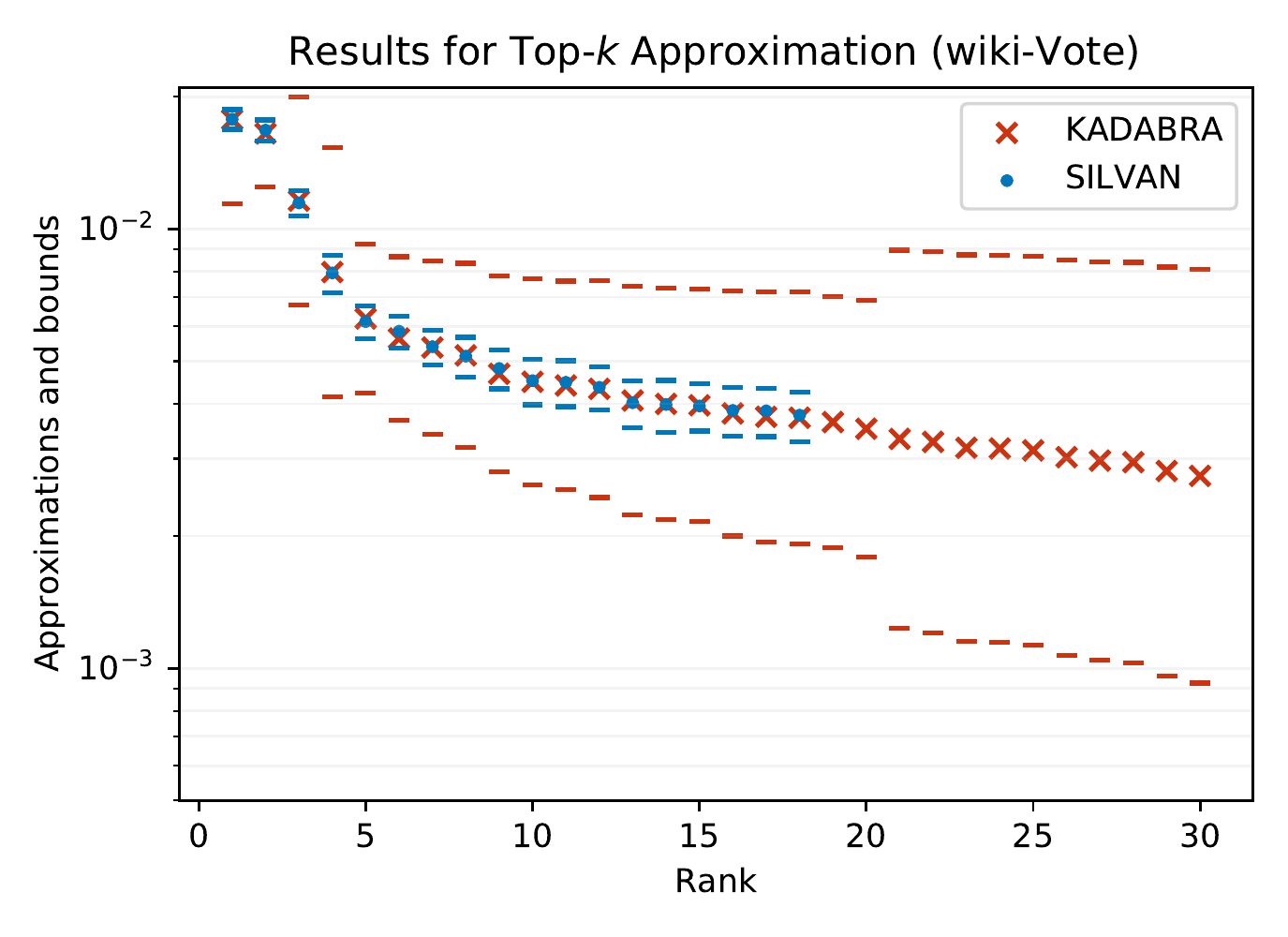}
\end{subfigure}
\begin{subfigure}{.49\textwidth}
  \centering
  \includegraphics[width=\textwidth]{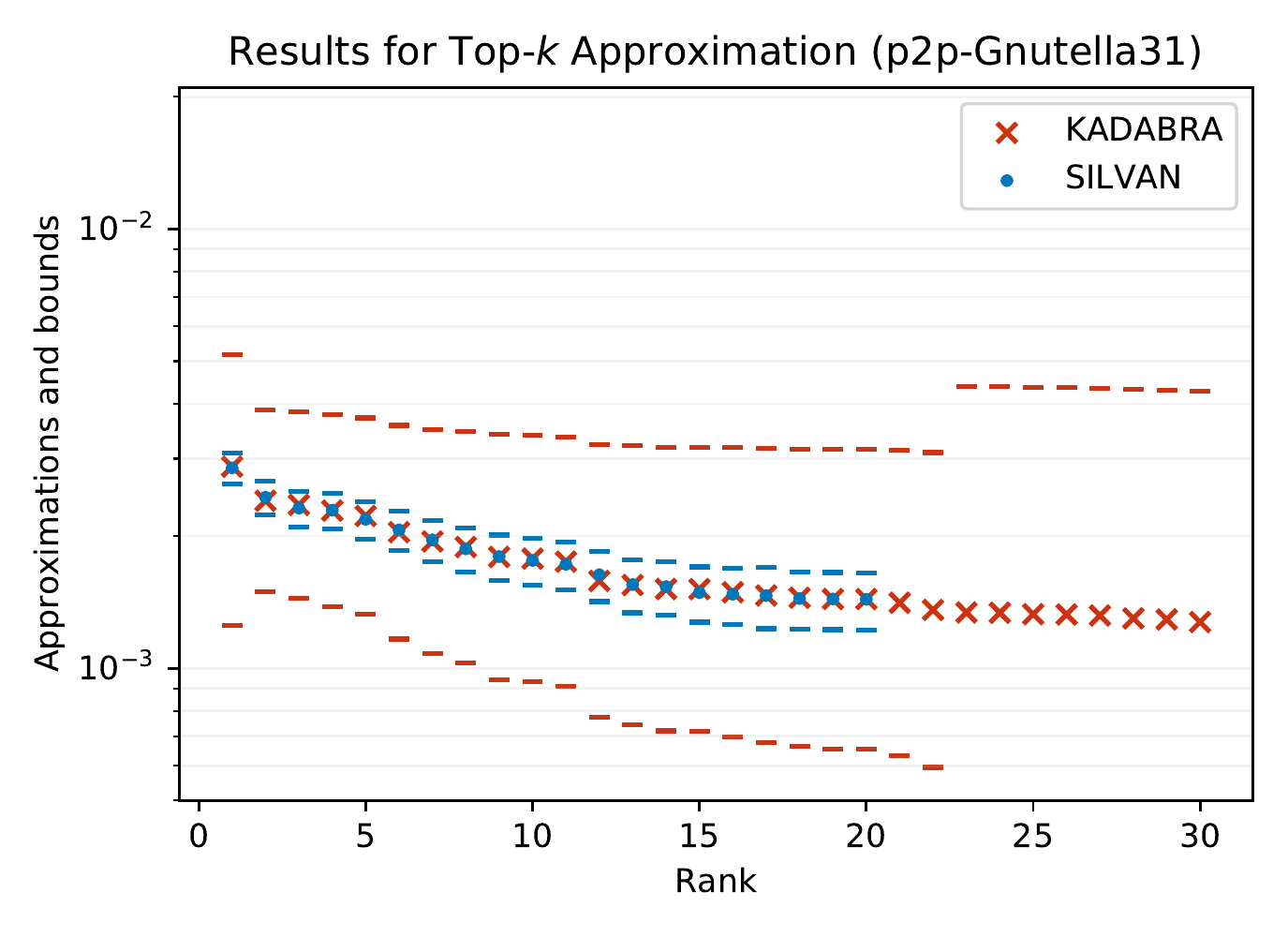}
\end{subfigure}
\caption{
Comparison of the quality of the output of \algnametopk\ and \kadabra\ for top-$k$ approximation for $3$ graphs (with $k=10$ and $\eta = 0.1$) when using the same resources: \kadabra\ is stopped after processing the same number of samples required by \algnametopk\ to stop. The $x$ axis is the rank of the node reported in output by both algorithms, the $y$ axis shows the estimated values $\tilde{b}$ of the \bc\ (dots and crosses) and upper and lower bounds w.r.t. the exact values $b$ (horizontal bars).
}
\label{fig:topkaccuracyapprox}
\end{figure}

\begin{figure*}[ht]
\centering
\begin{subfigure}{.49\textwidth}
  \centering
  \includegraphics[width=\textwidth]{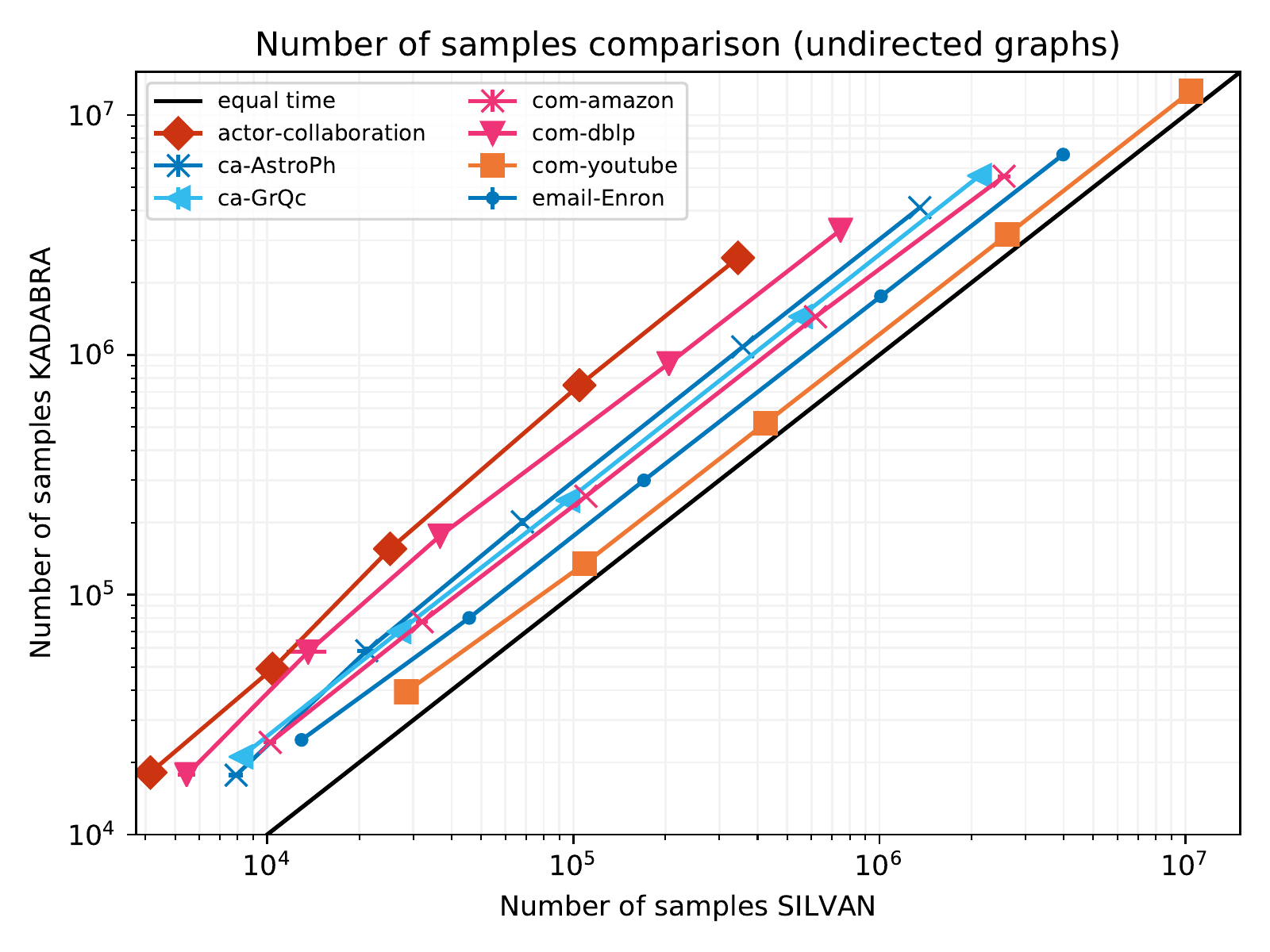}
  \caption{}
\end{subfigure}
\begin{subfigure}{.49\textwidth}
  \centering
  \includegraphics[width=\textwidth]{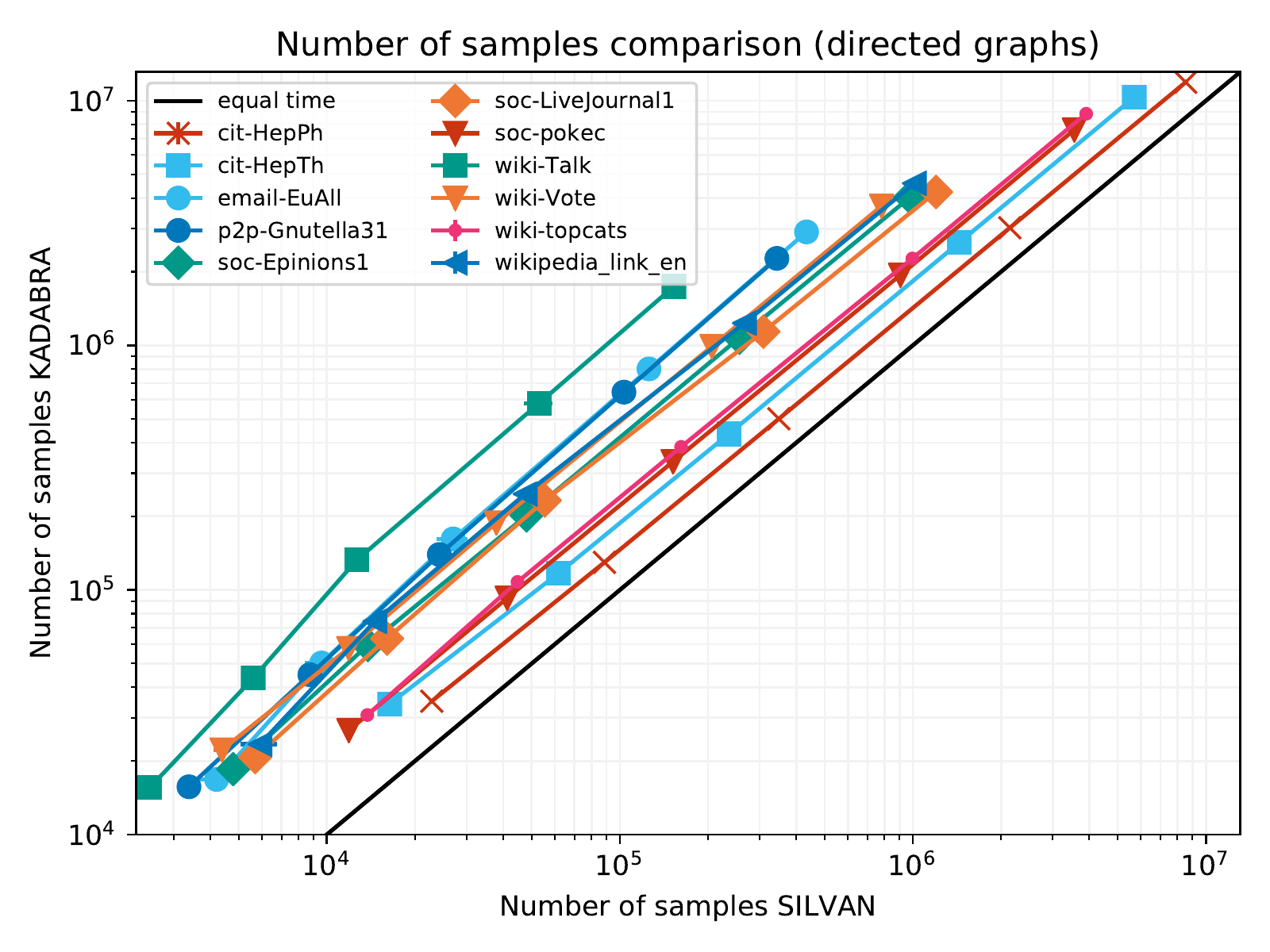}
\caption{}
\end{subfigure}
\begin{subfigure}{.49\textwidth}
  \centering
  \includegraphics[width=\textwidth]{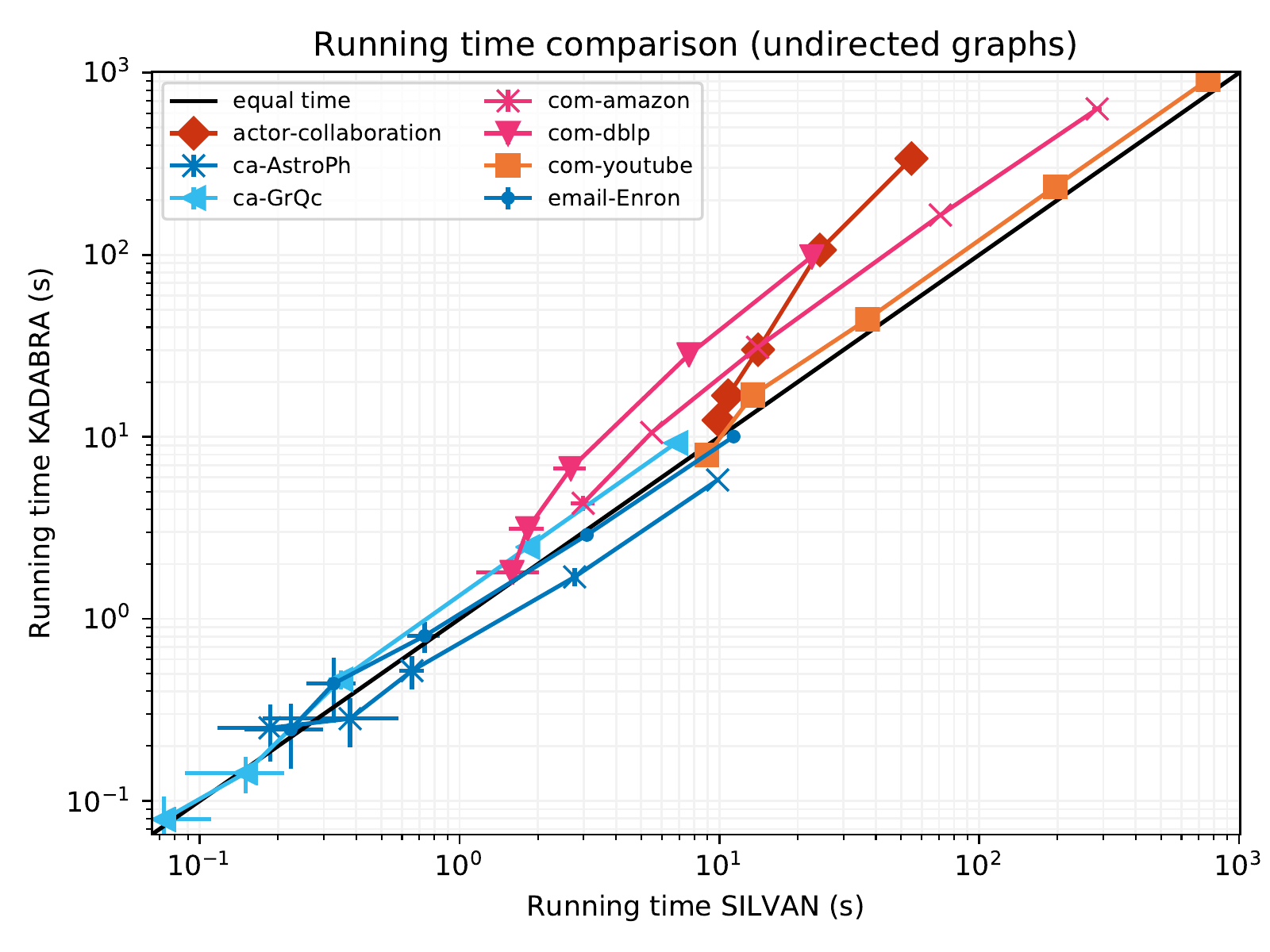}
  \caption{}
\end{subfigure}
\begin{subfigure}{.49\textwidth}
  \centering
  \includegraphics[width=\textwidth]{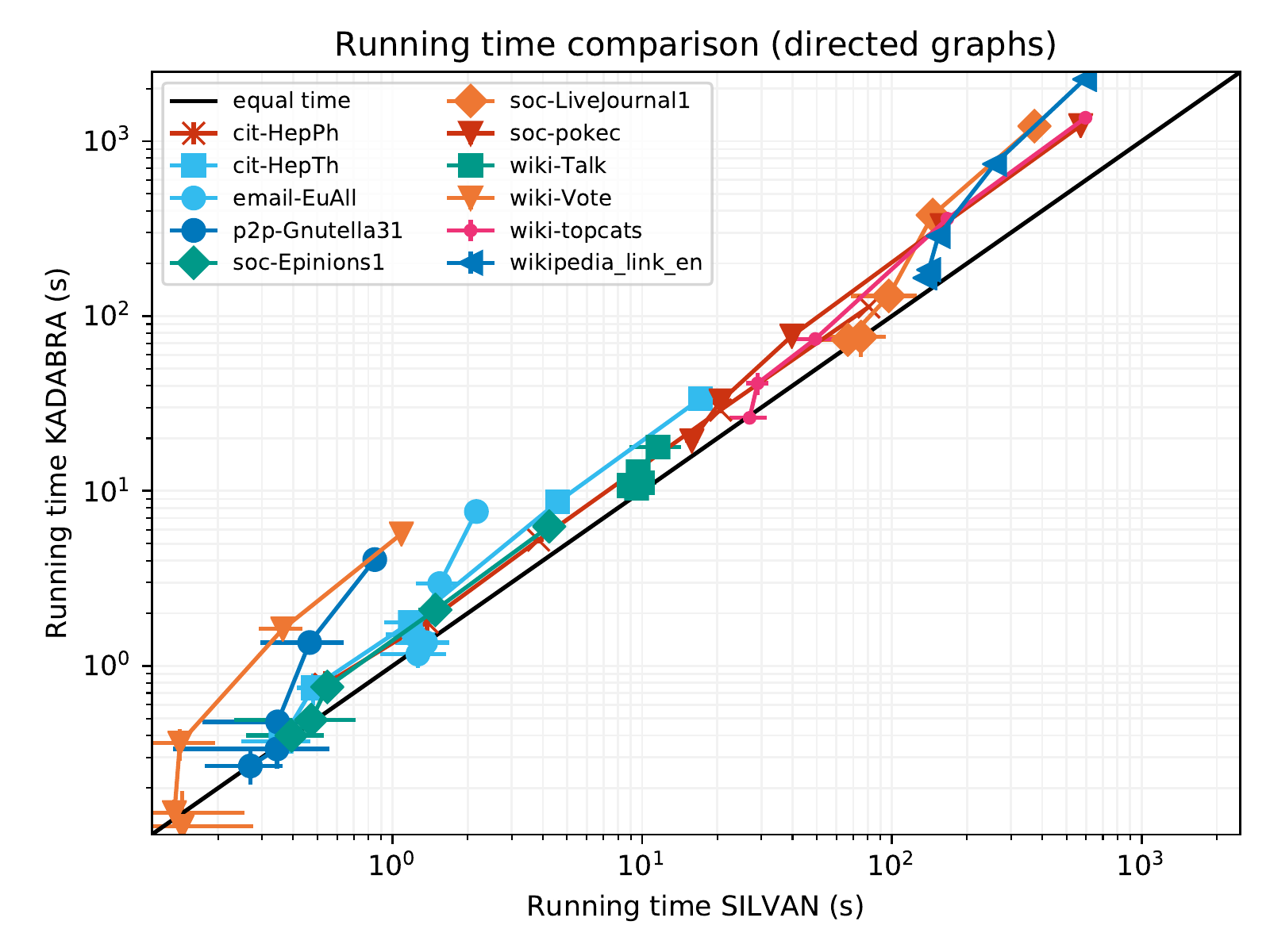}
\caption{}
\end{subfigure}
\begin{subfigure}{.49\textwidth}
  \centering
  \includegraphics[width=\textwidth]{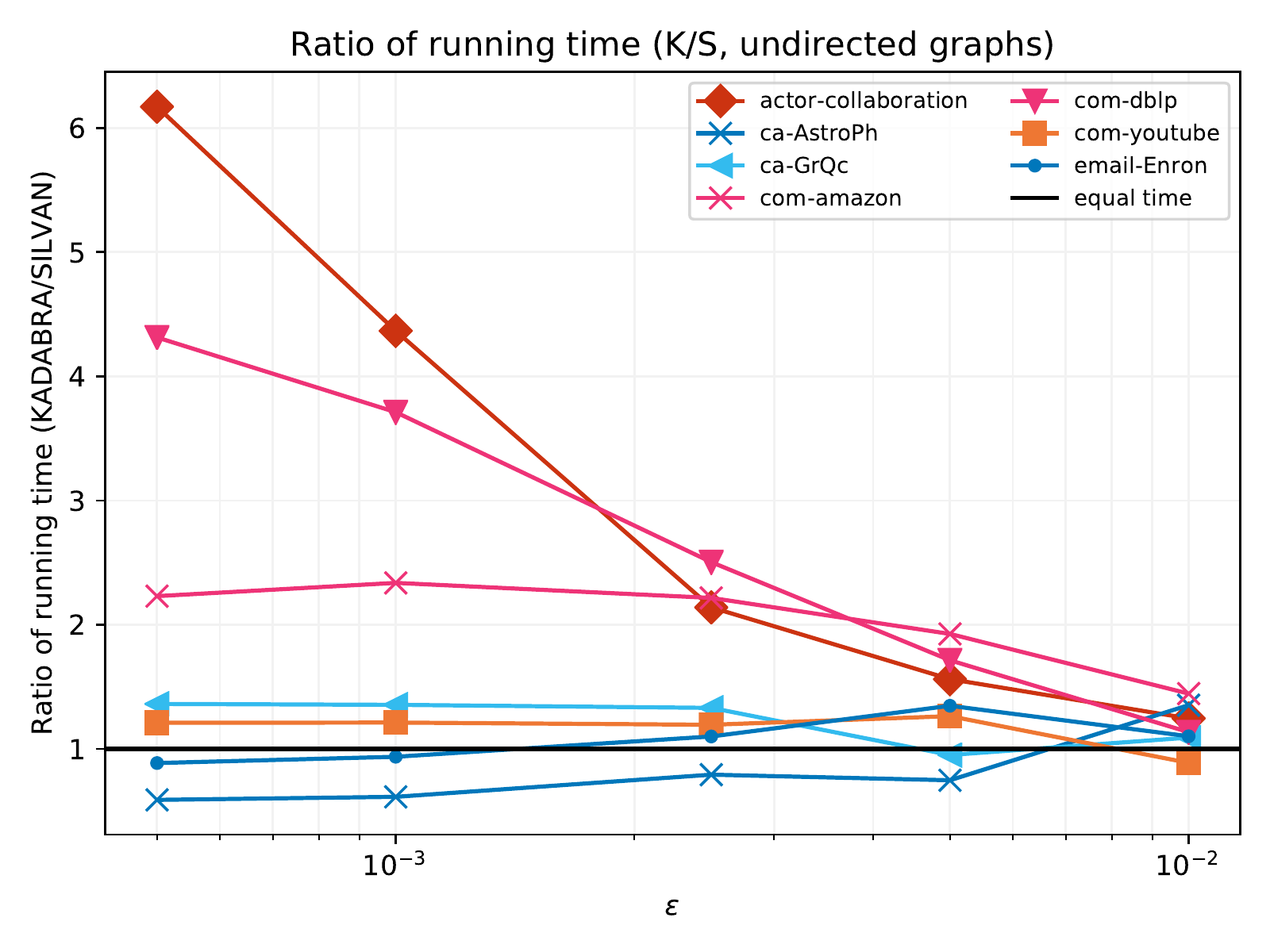}
  \caption{}
  \label{fig:ratiotimesundir}
\end{subfigure}
\begin{subfigure}{.49\textwidth}
  \centering
  \includegraphics[width=\textwidth]{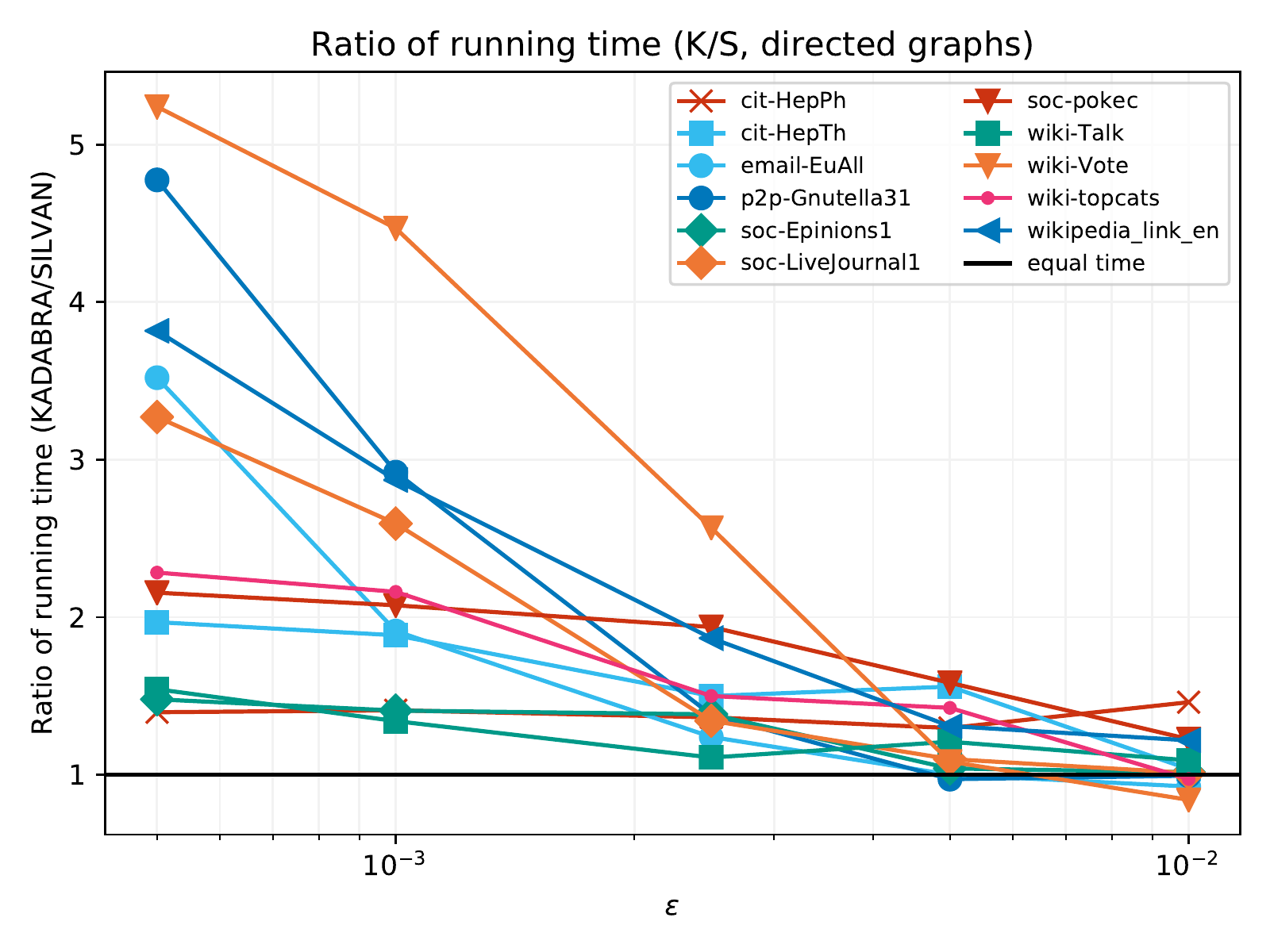}
\caption{}
\label{fig:ratiotimesdir}
\end{subfigure}
\caption{Additional figures for comparing the 
performance of \kadabra\ and \algname\ for obtaining an absolute $\varepsilon$ approximation. 
(a): comparison of the number of samples for \kadabra\ ($y$ axis) and \algname\ ($x$ axis) for undirected graphs.
(b): analogous of (a) for directed graphs.
(c): comparison of the running times of \kadabra\ ($y$ axis) and \algname\ ($x$ axis) for undirected graphs (axes in logarithmic scales).
(e): analogous of (d) for directed graphs.
(c): ratios of the running times of \kadabra\ and \algname\ for undirected graphs.
(d): analogous of (c) for directed graphs.
}
\label{fig:absapproxappendixkad}
\end{figure*}

\begin{figure*}[ht]
\centering
\begin{subfigure}{.49\textwidth}
  \centering
  \includegraphics[width=\textwidth]{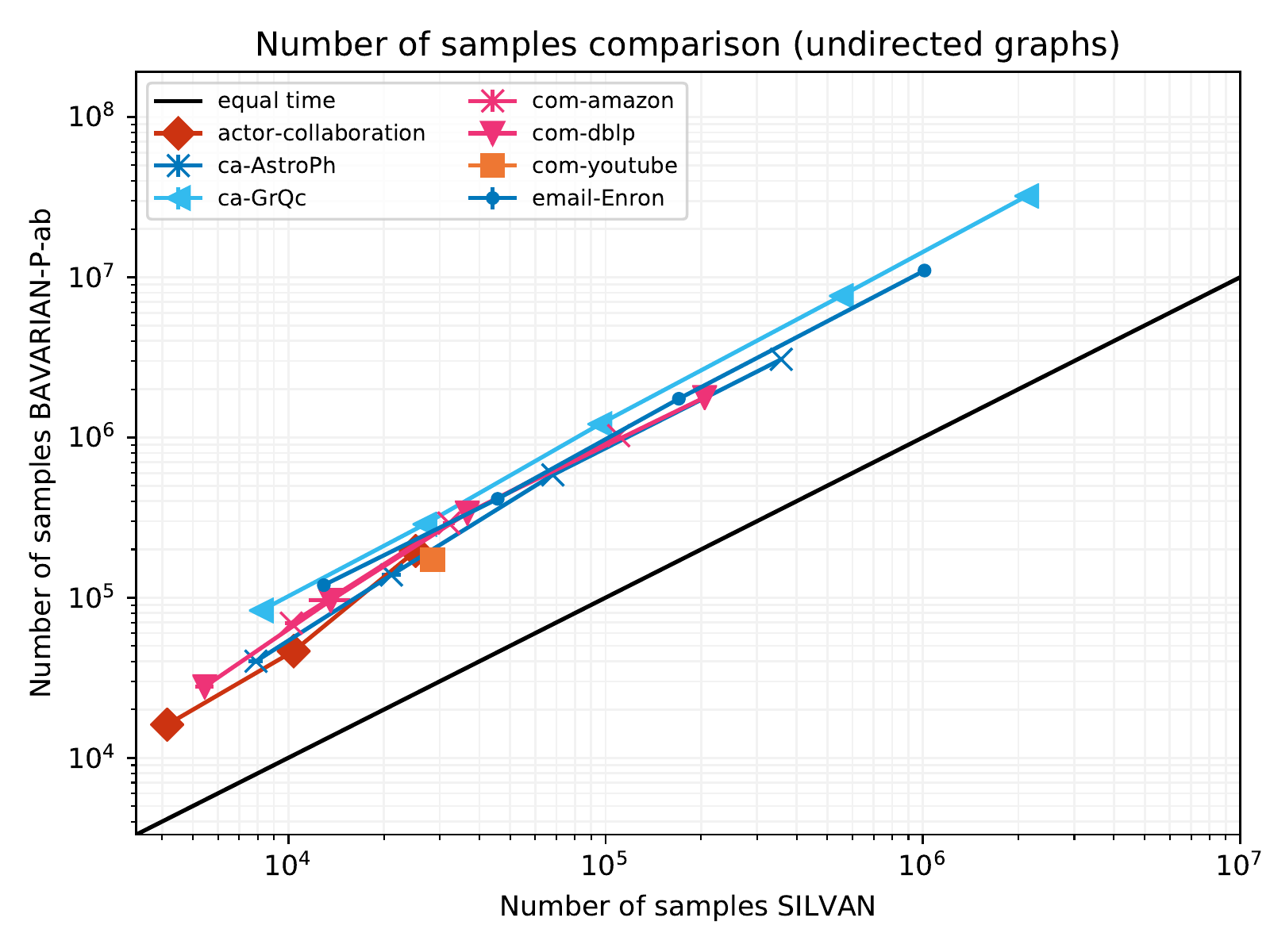}
  \caption{}
\end{subfigure}
\begin{subfigure}{.49\textwidth}
  \centering
  \includegraphics[width=\textwidth]{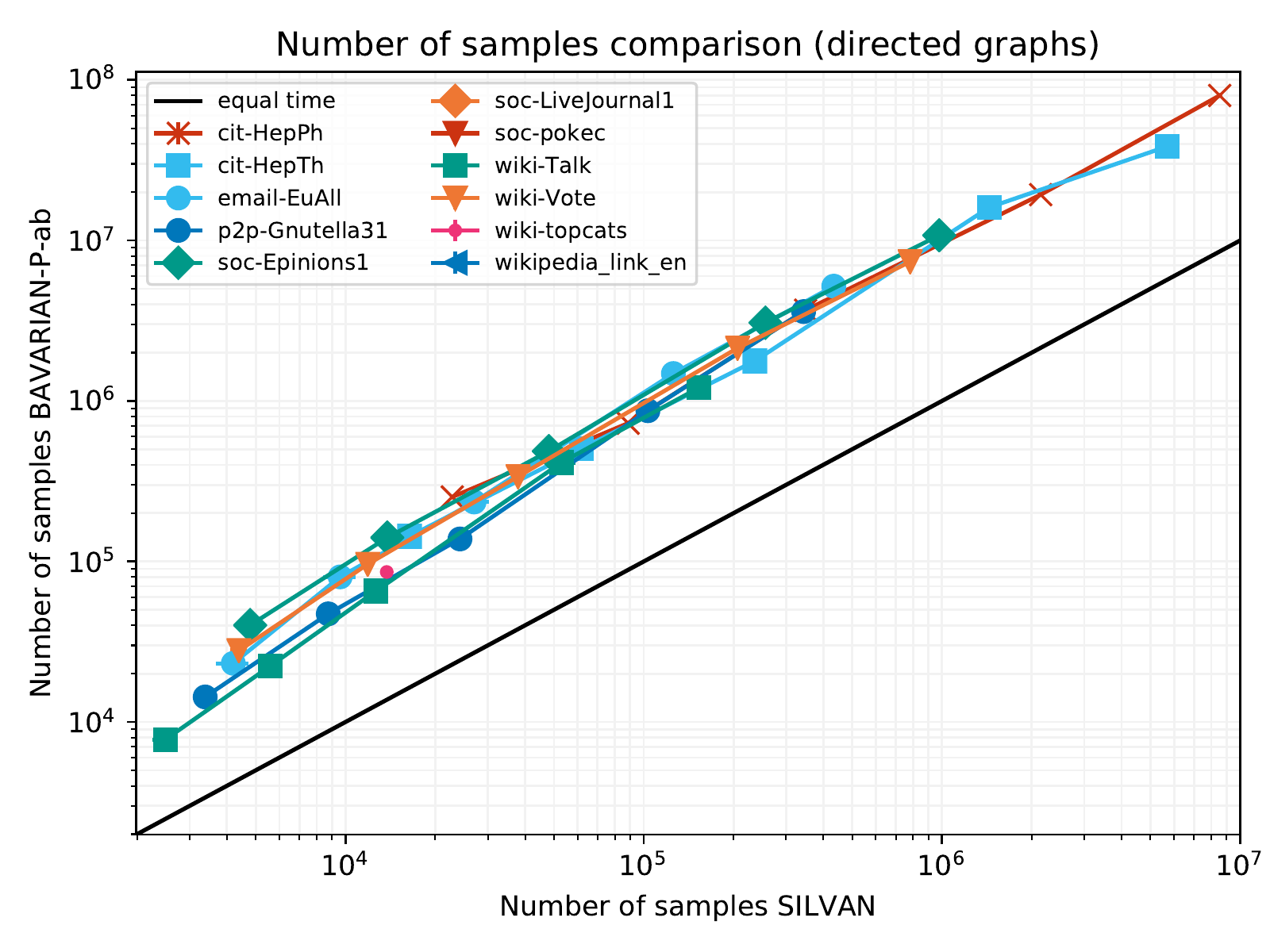}
\caption{}
\end{subfigure}
\begin{subfigure}{.49\textwidth}
  \centering
  \includegraphics[width=\textwidth]{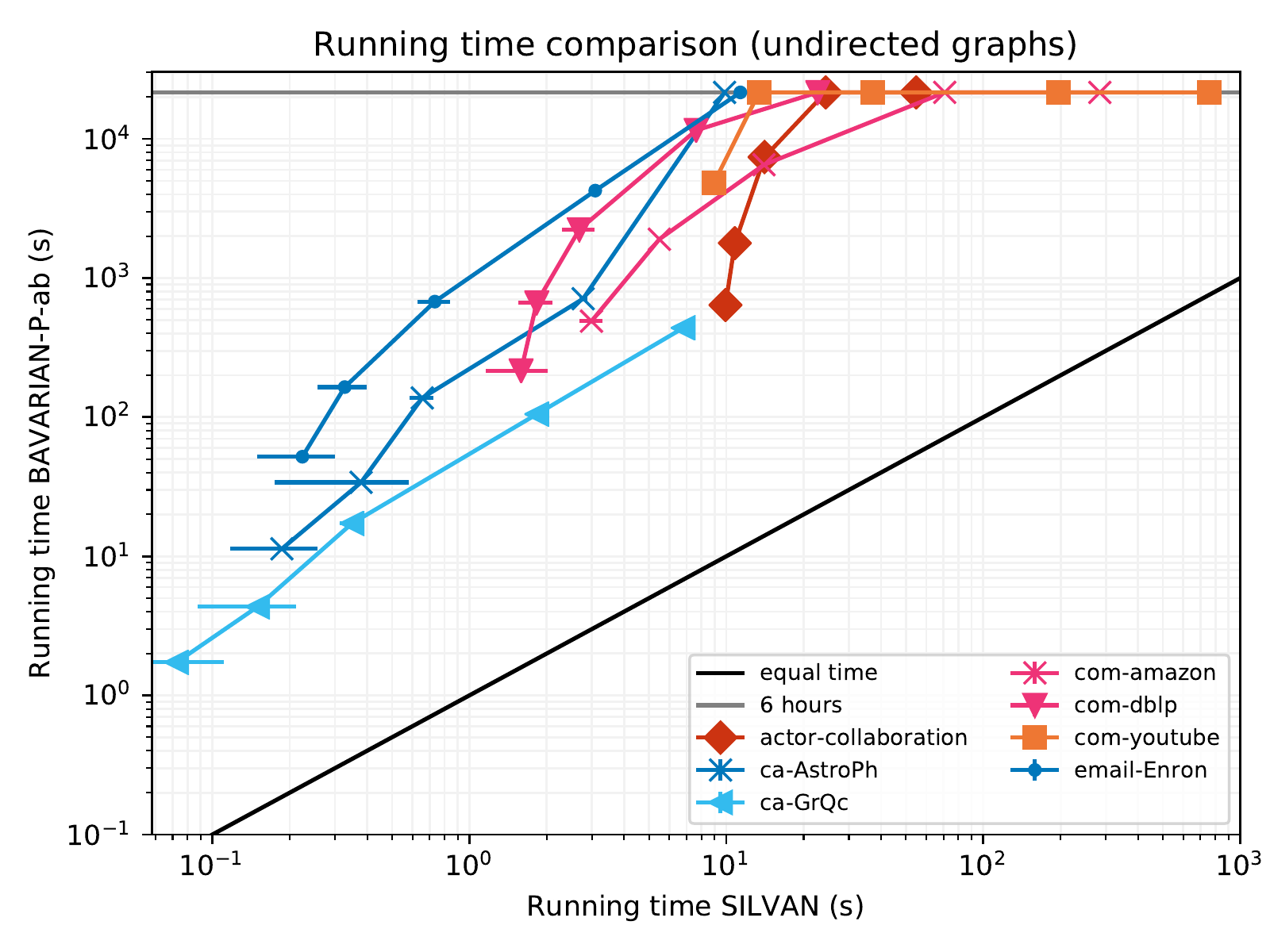}
  \caption{}
\end{subfigure}
\begin{subfigure}{.49\textwidth}
  \centering
  \includegraphics[width=\textwidth]{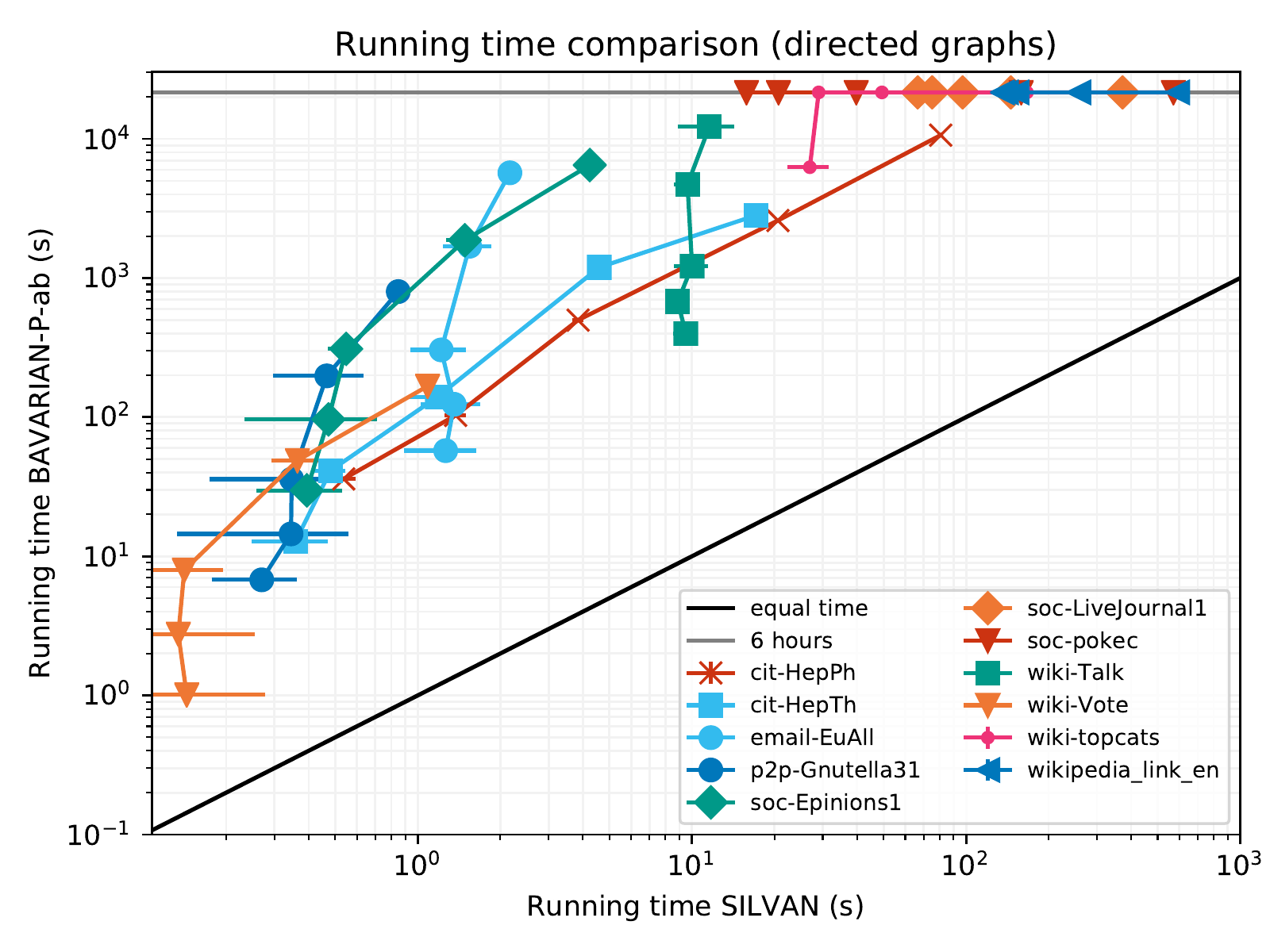}
\caption{}
\end{subfigure}
\begin{subfigure}{.49\textwidth}
  \centering
  \includegraphics[width=\textwidth]{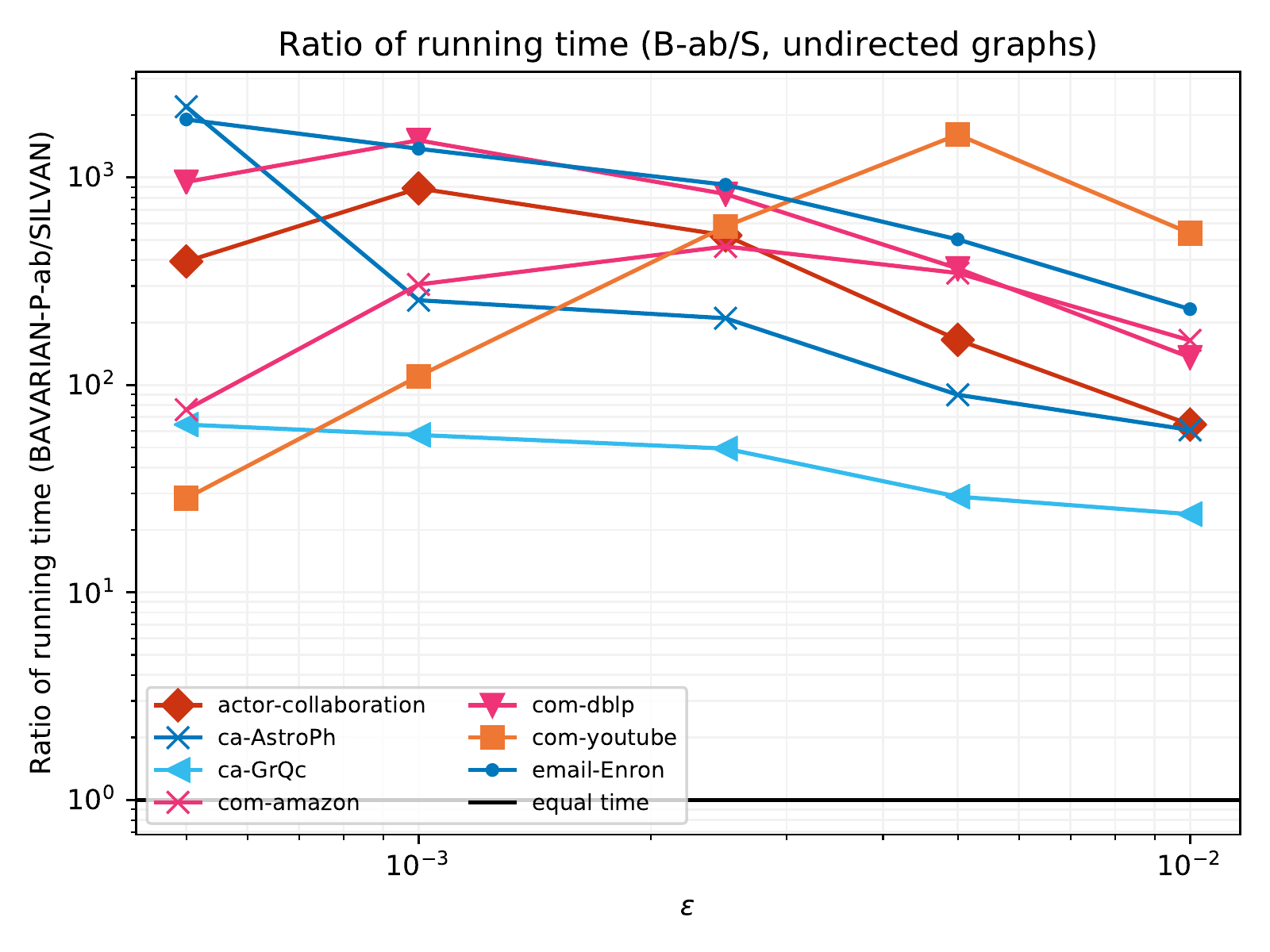}
  \caption{}
\end{subfigure}
\begin{subfigure}{.49\textwidth}
  \centering
  \includegraphics[width=\textwidth]{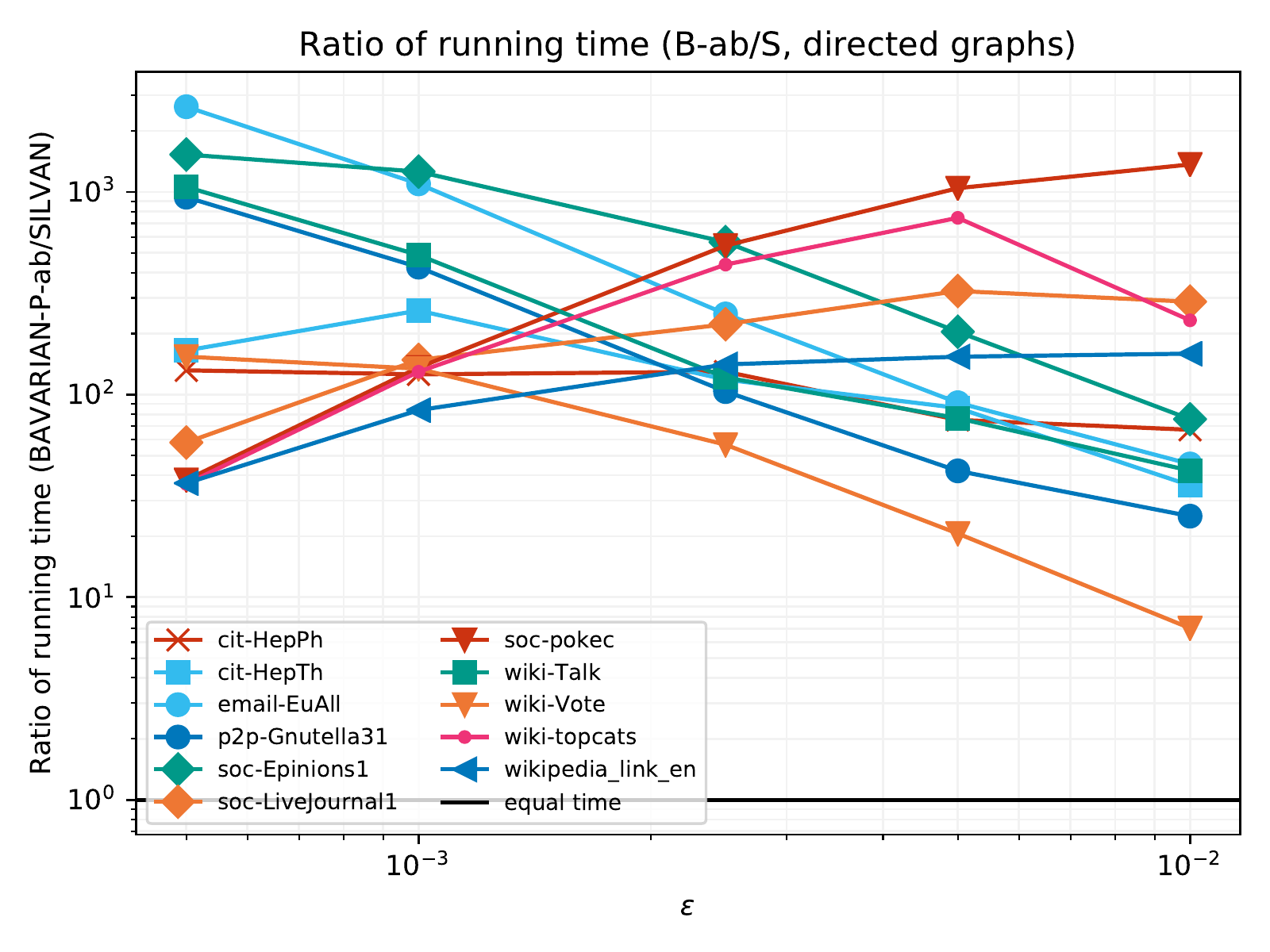}
\caption{}
\end{subfigure}
\caption{Additional figures for comparing the 
performance of \kadabra\ and \bavarianp\ (\texttt{ab} estimator) for obtaining an absolute $\varepsilon$ approximation. 
(a): comparison of the number of samples for \bavarian\ ($y$ axis) and \algname\ ($x$ axis) for undirected graphs (axes in logarithmic scales).
(b): analogous of (a) for directed graphs.
(c): comparison of the running times of \bavarian\ ($y$ axis) and \algname\ ($x$ axis) for undirected graphs (axes in logarithmic scales).
(d): analogous of (c) for directed graphs.
(e): ratios of the running times of \bavarian\ and \algname\ for undirected graphs.
(f): analogous of (e) for directed graphs.
}
\label{fig:absapproxappendixbavab}
\end{figure*}

\begin{figure*}[ht]
\centering
\begin{subfigure}{.35\textwidth}
  \centering
  \includegraphics[width=\textwidth]{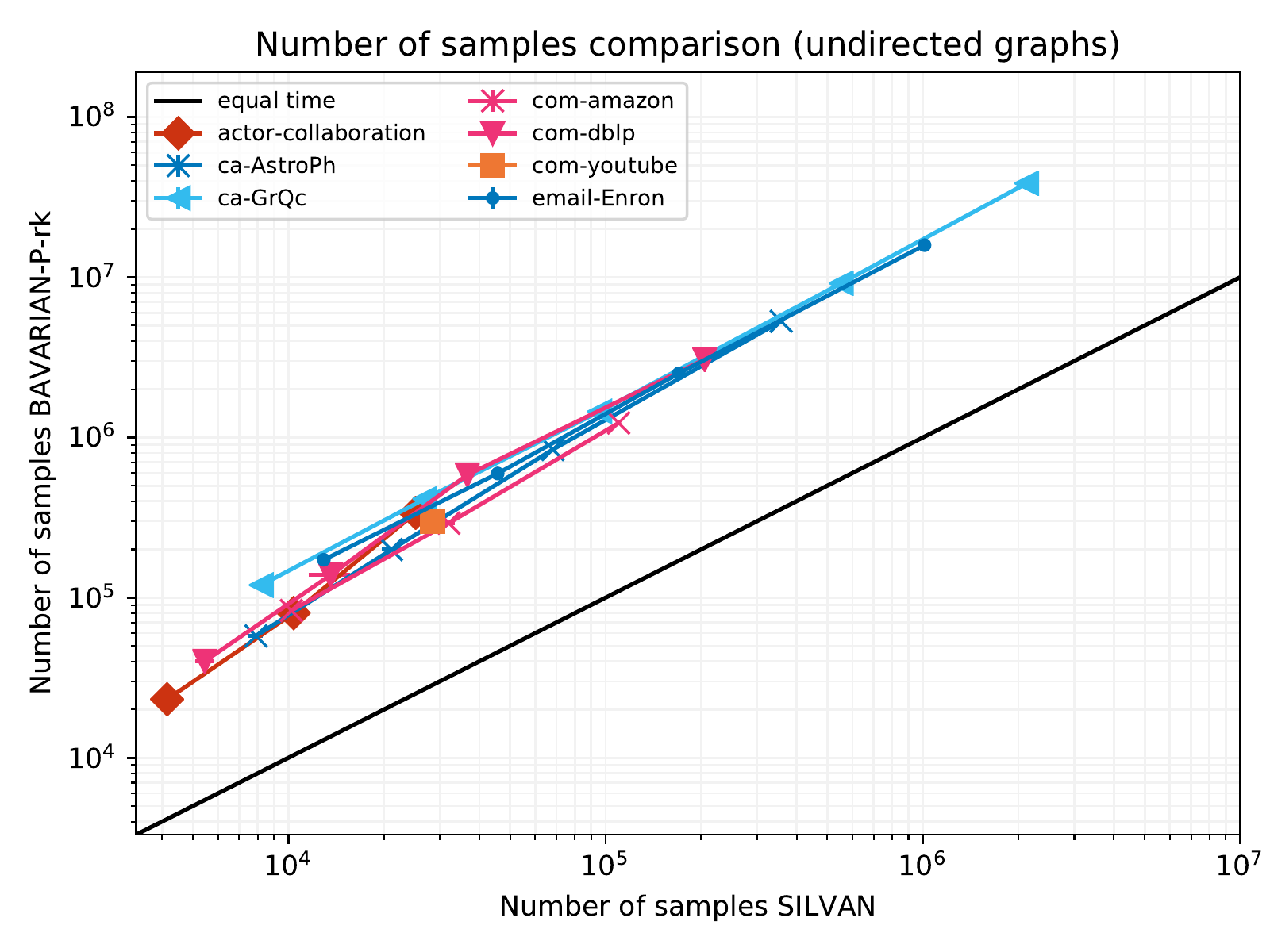}
  \caption{}
\end{subfigure}
\begin{subfigure}{.35\textwidth}
  \centering
  \includegraphics[width=\textwidth]{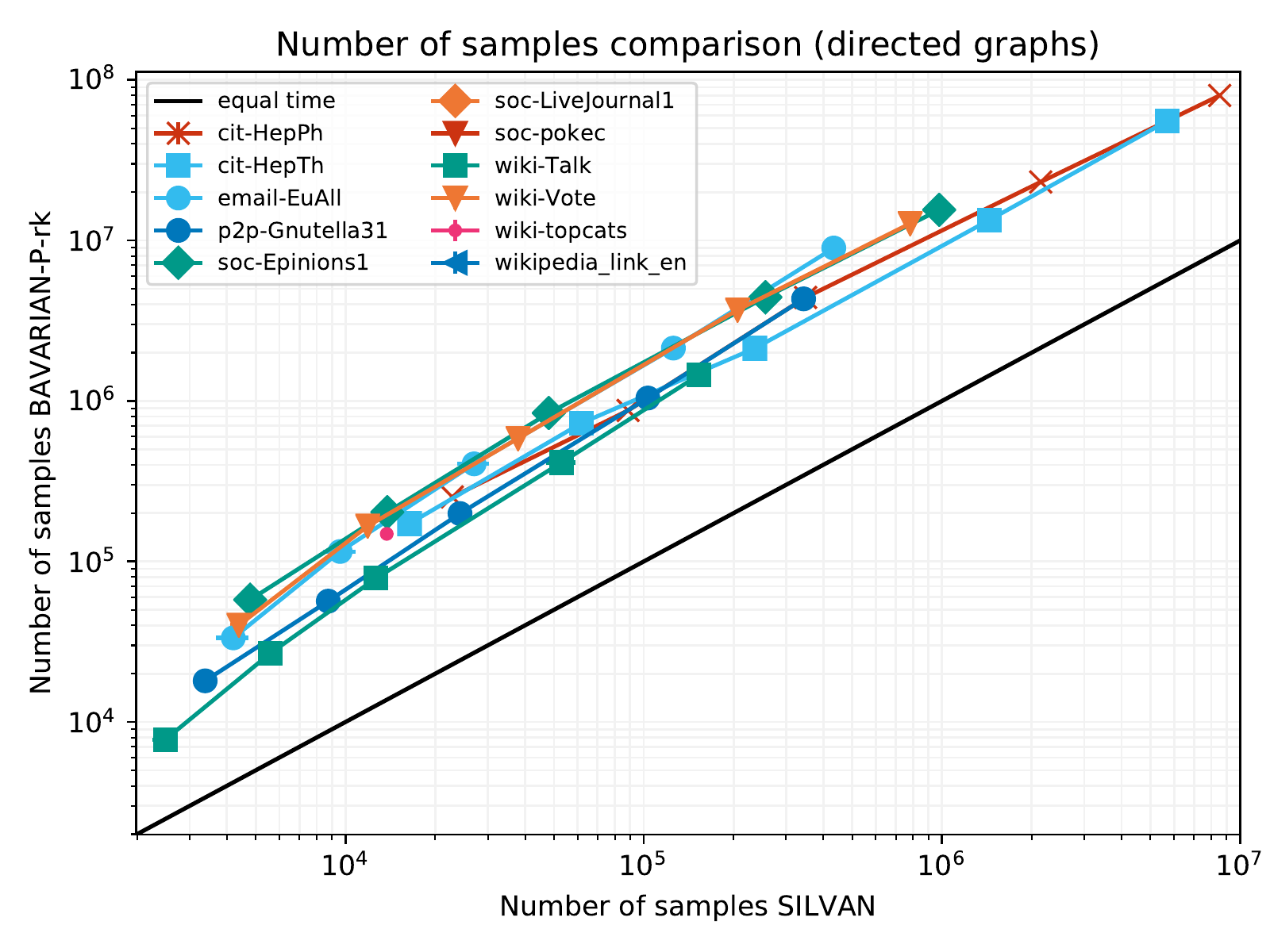}
\caption{}
\end{subfigure}
\begin{subfigure}{.35\textwidth}
  \centering
  \includegraphics[width=\textwidth]{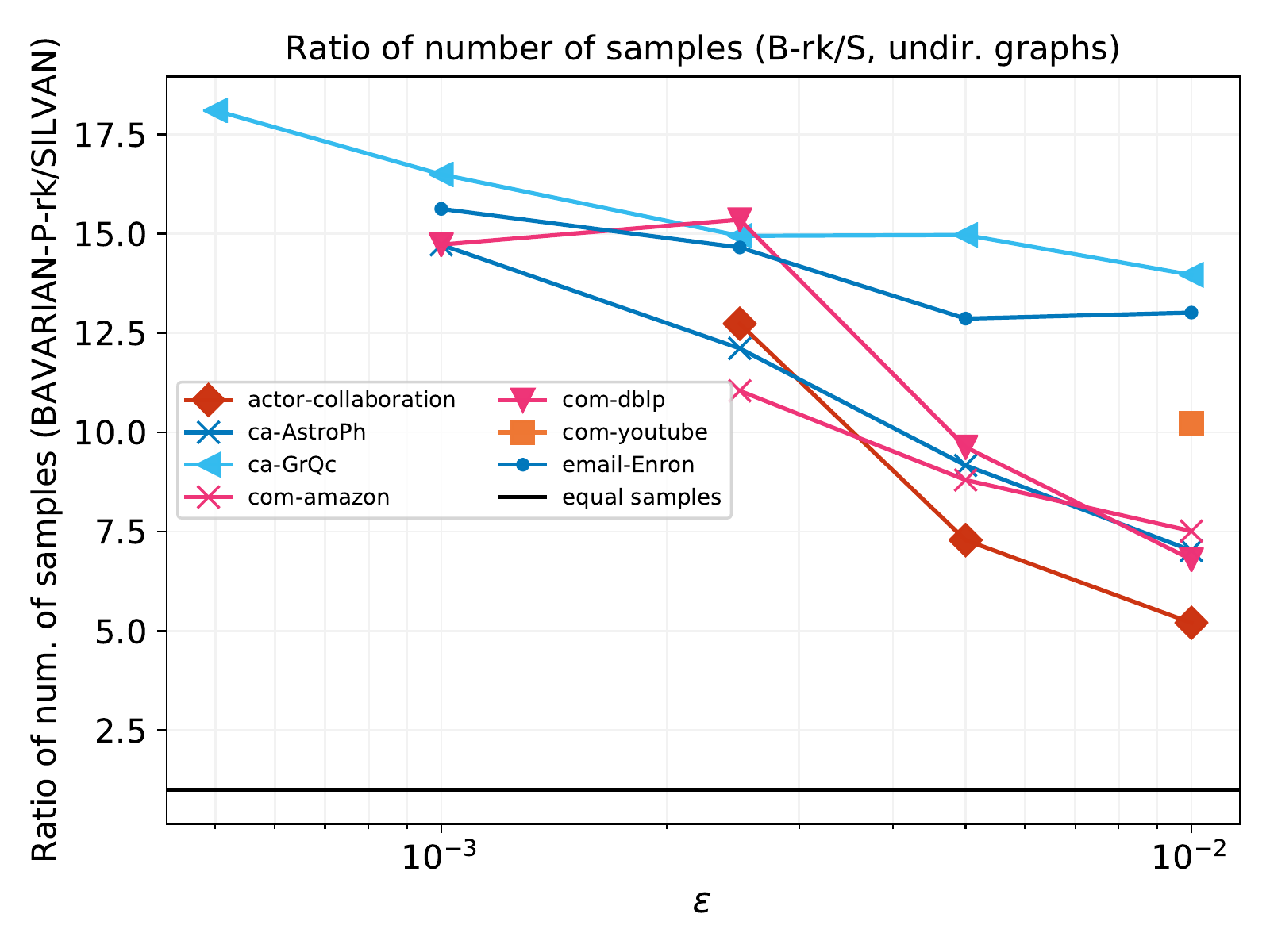}
  \caption{}
\end{subfigure}
\begin{subfigure}{.35\textwidth}
  \centering
  \includegraphics[width=\textwidth]{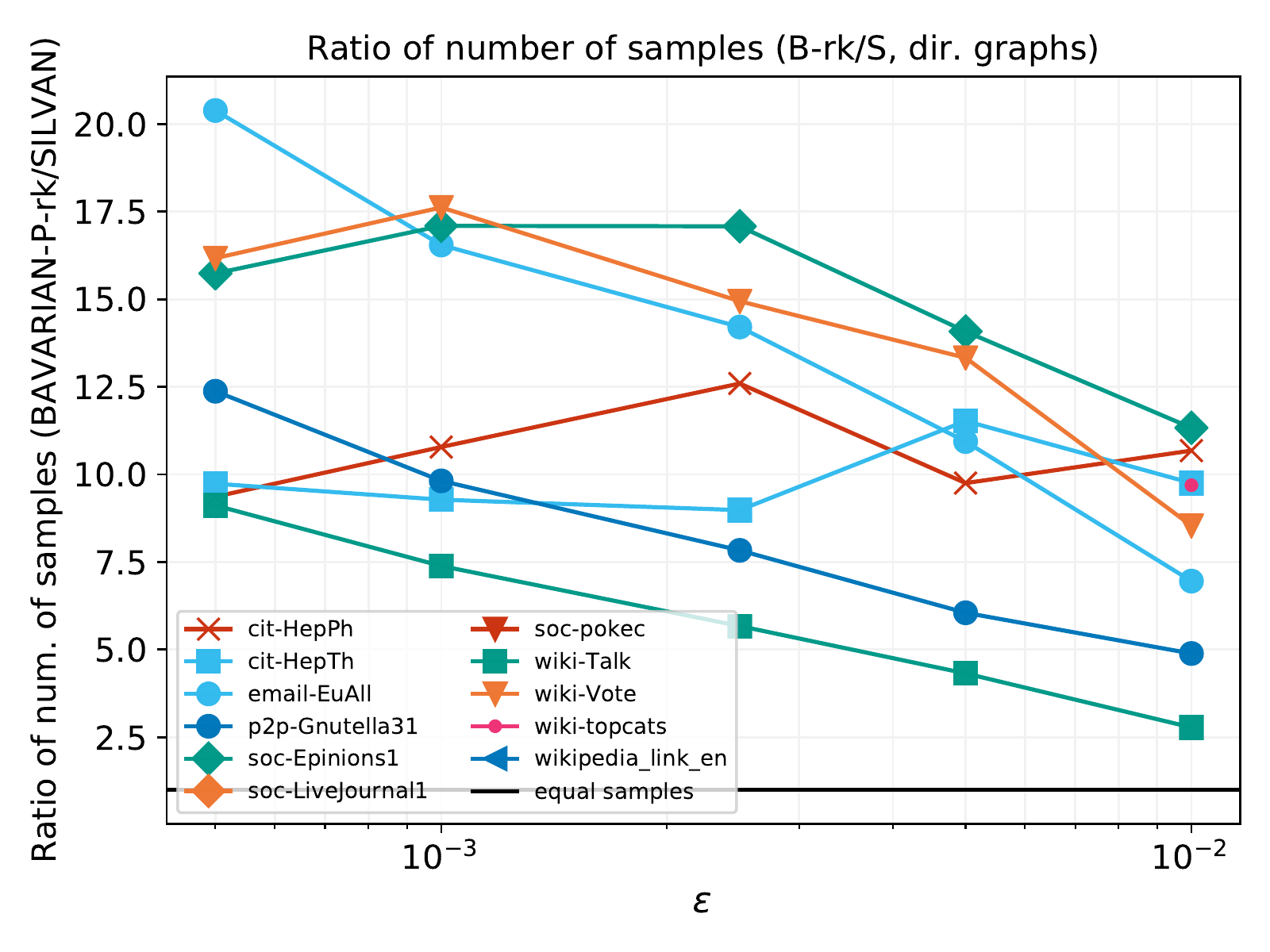}
\caption{}
\end{subfigure}
\begin{subfigure}{.35\textwidth}
  \centering
  \includegraphics[width=\textwidth]{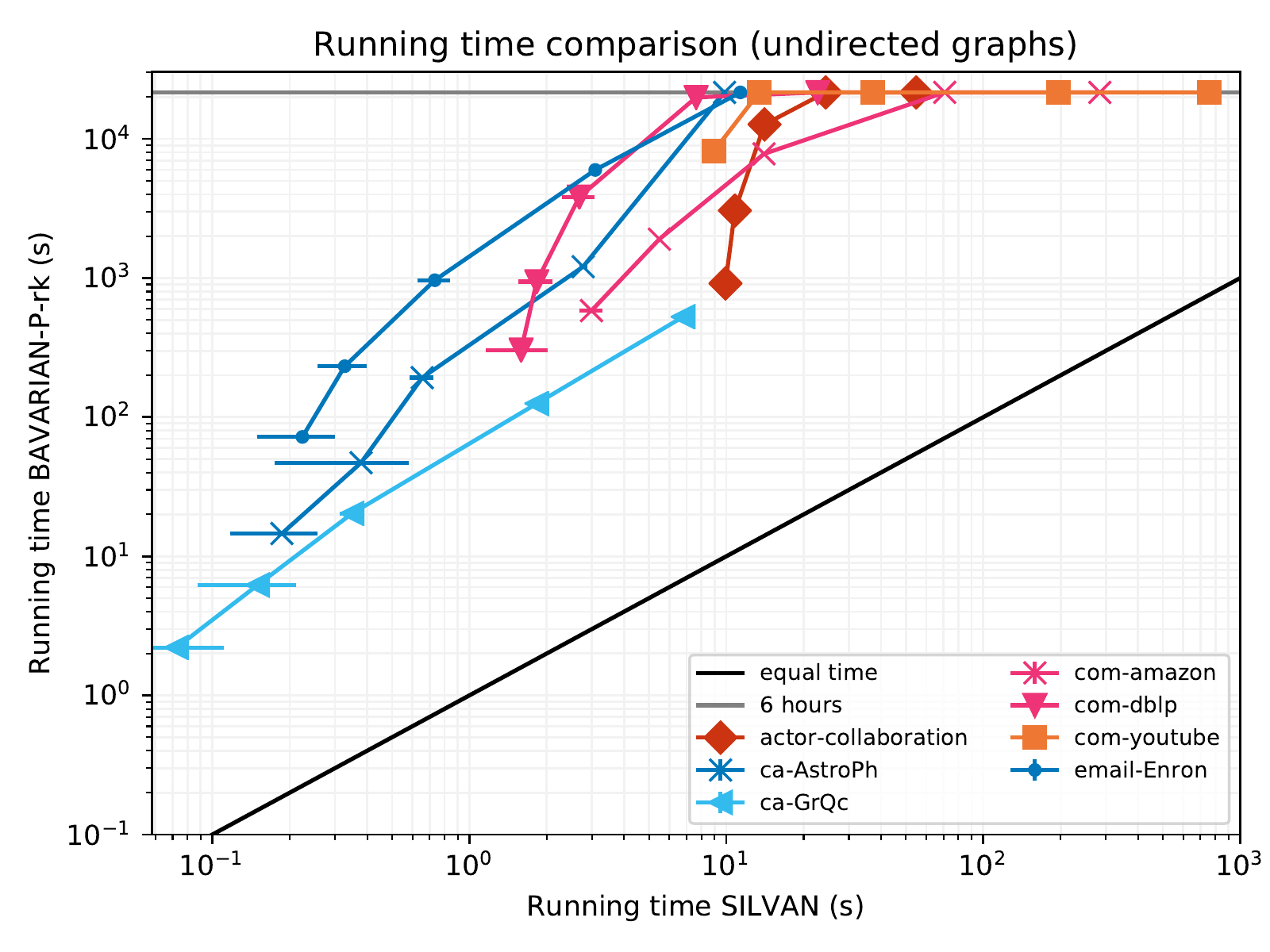}
  \caption{}
\end{subfigure}
\begin{subfigure}{.35\textwidth}
  \centering
  \includegraphics[width=\textwidth]{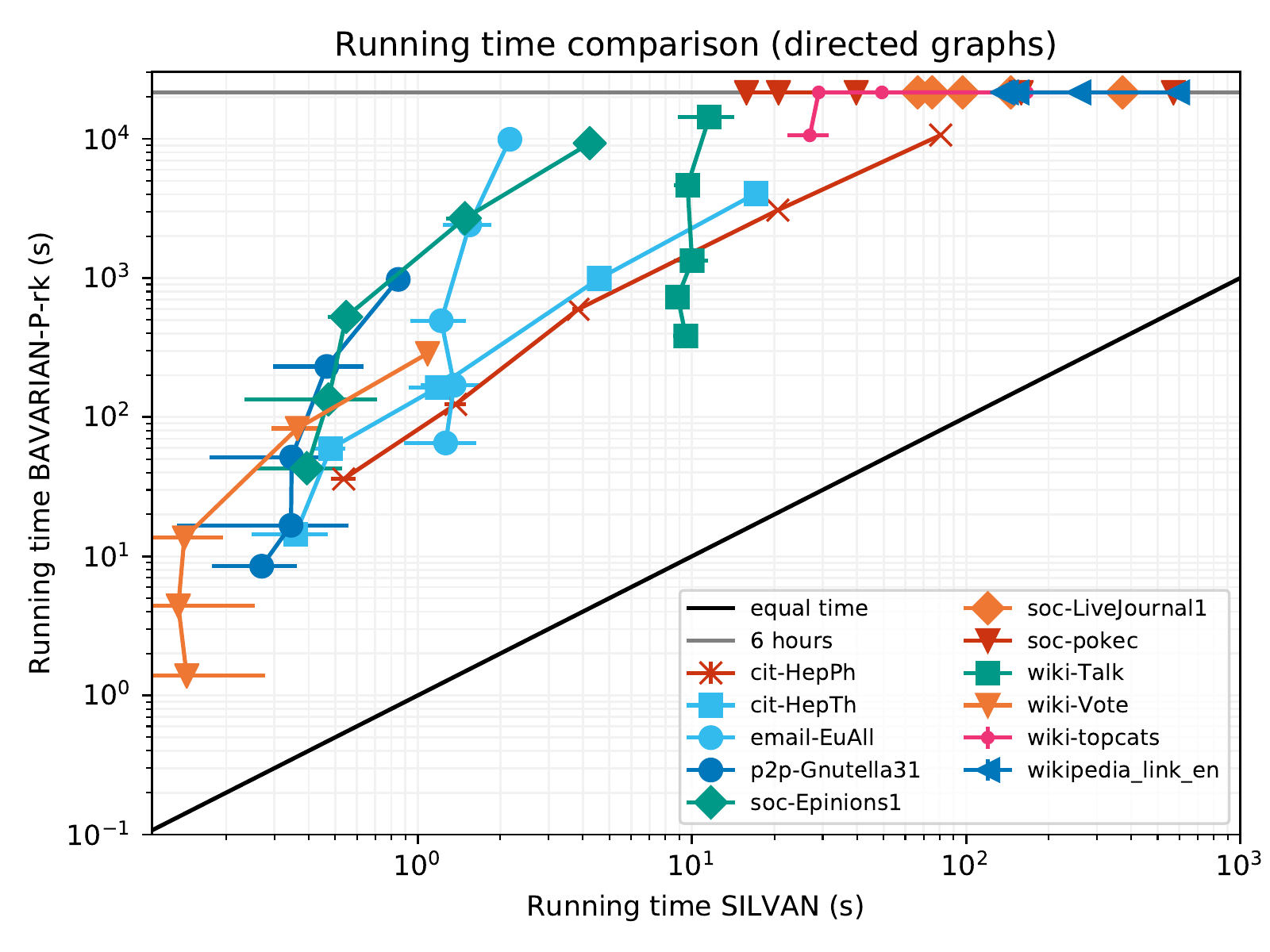}
\caption{}
\end{subfigure}
\begin{subfigure}{.35\textwidth}
  \centering
  \includegraphics[width=\textwidth]{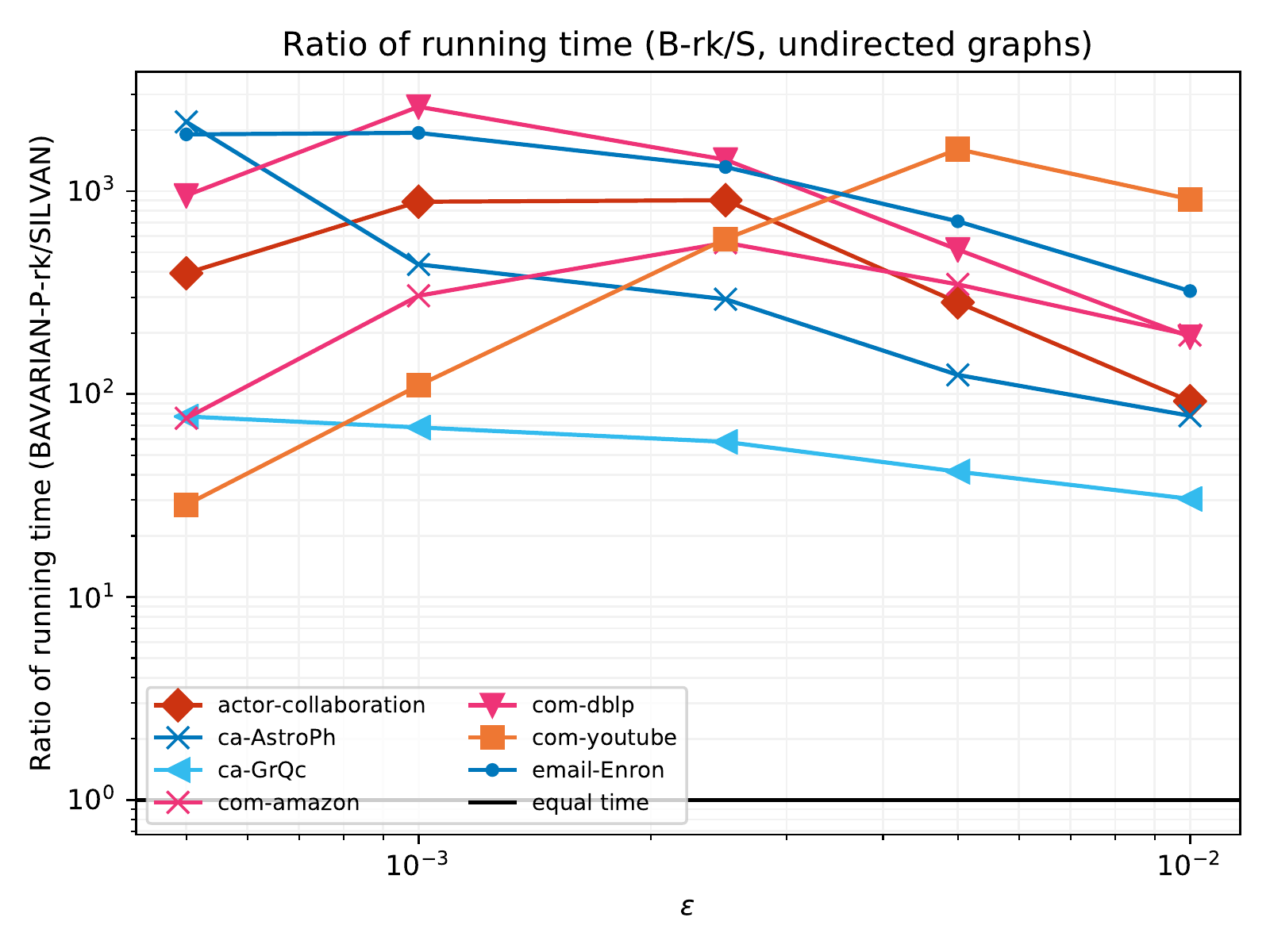}
  \caption{}
\end{subfigure}
\begin{subfigure}{.35\textwidth}
  \centering
  \includegraphics[width=\textwidth]{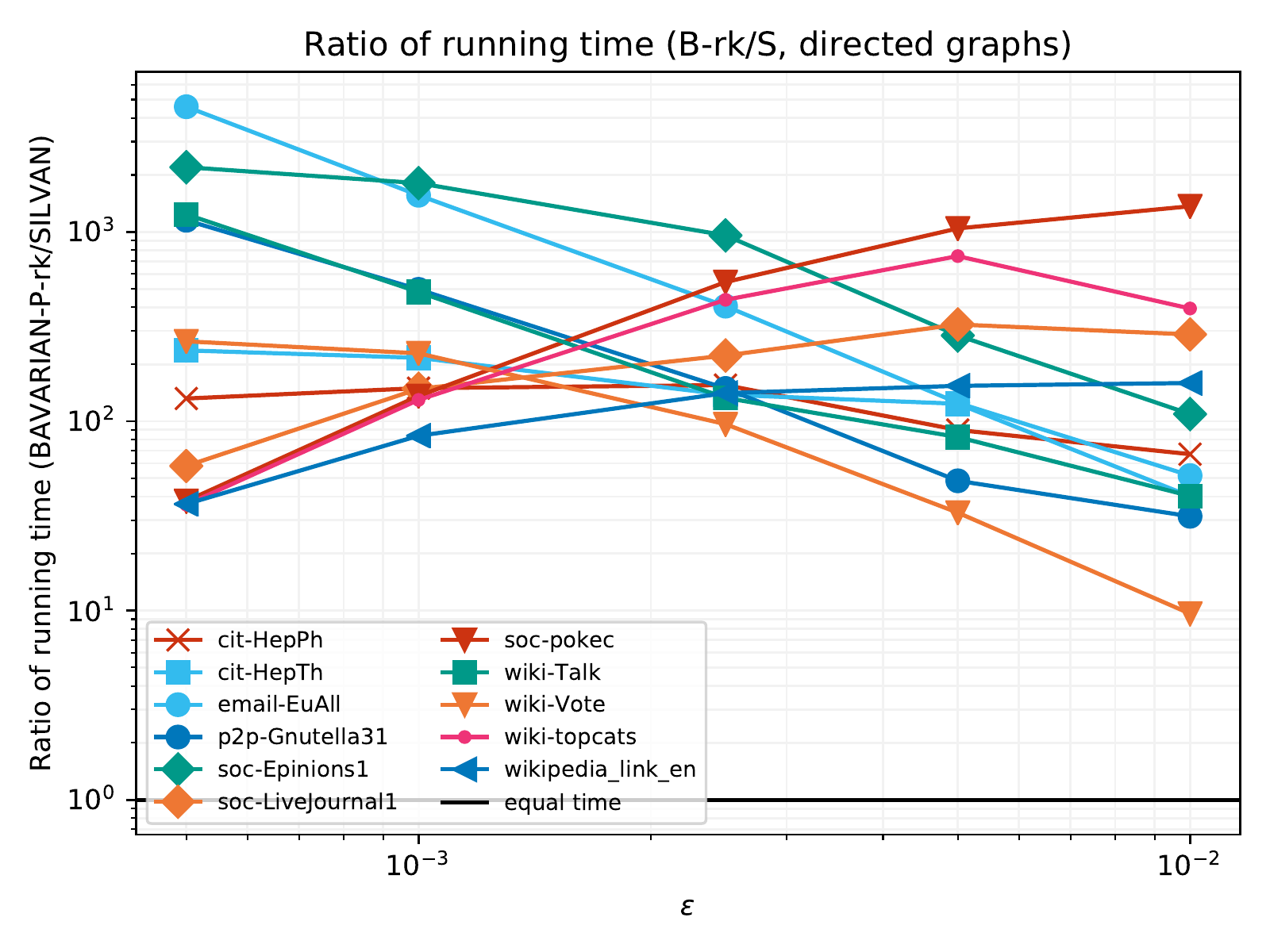}
\caption{}
\end{subfigure}
\caption{Figures comparing the 
performance of \kadabra\ and \bavarianp\ (\texttt{rk} estimator) for obtaining an absolute $\varepsilon$ approximation. 
(a): comparison of the number of samples for \bavarian\ ($y$ axis) and \algname\ ($x$ axis) for undirected graphs (axes in logarithmic scales).
(b): analogous of (a) for directed graphs.
(c): ratios of the number of samples for \bavarian\ and \algname\ for undirected graphs.
(d): analogous of (c) for directed graphs.
(e): comparison of the running times of \bavarian\ ($y$ axis) and \algname\ ($x$ axis) for undirected graphs (axes in logarithmic scales).
(f): analogous of (e) for directed graphs.
(g): ratios of the running times of \bavarian\ and \algname\ for undirected graphs.
(h): analogous of (g) for directed graphs.
}
\label{fig:absapproxappendixbavrk}
\end{figure*}

\begin{figure*}[ht]
\centering
\begin{subfigure}{.35\textwidth}
  \centering
  \includegraphics[width=\textwidth]{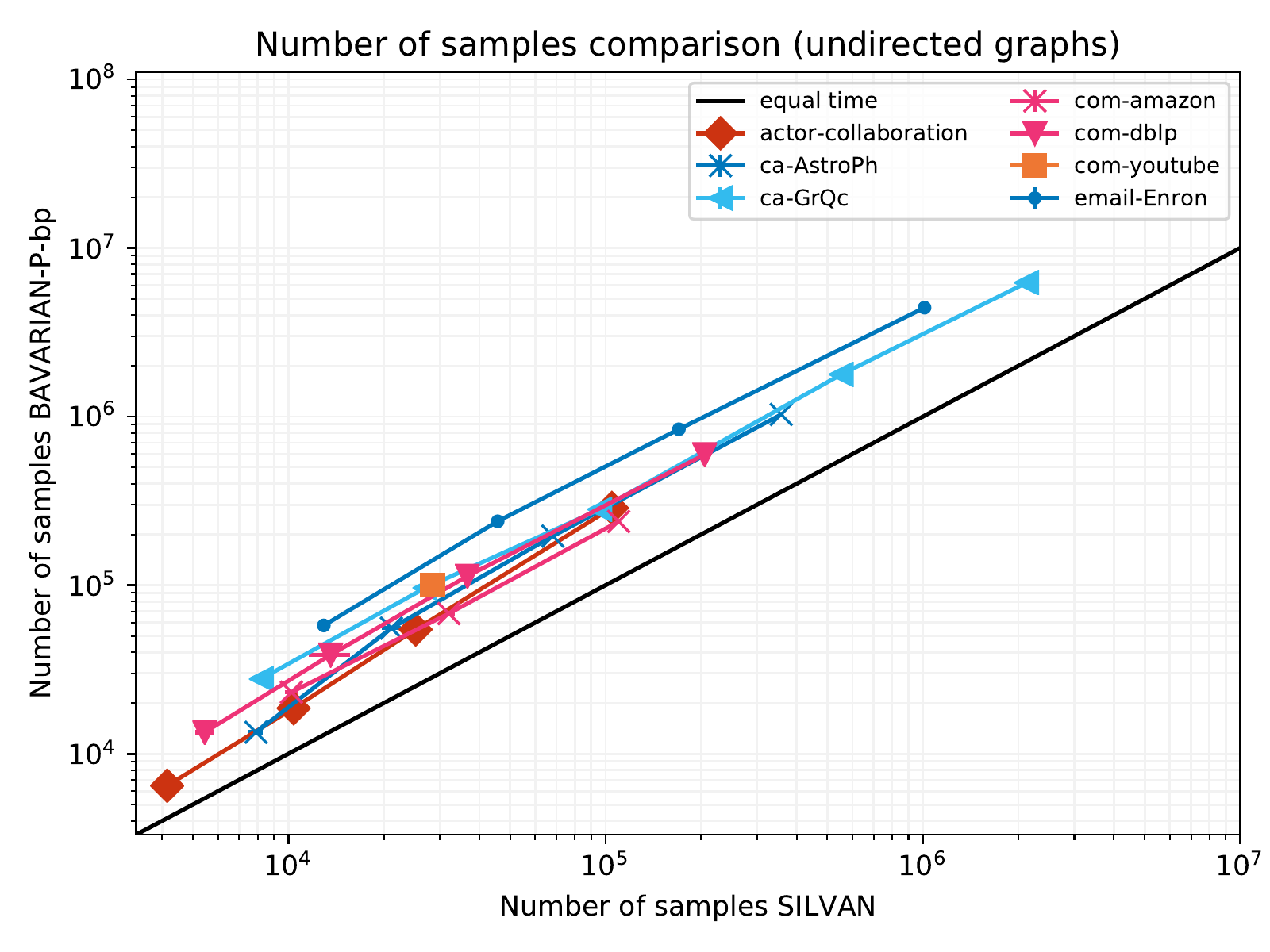}
  \caption{}
\end{subfigure}
\begin{subfigure}{.35\textwidth}
  \centering
  \includegraphics[width=\textwidth]{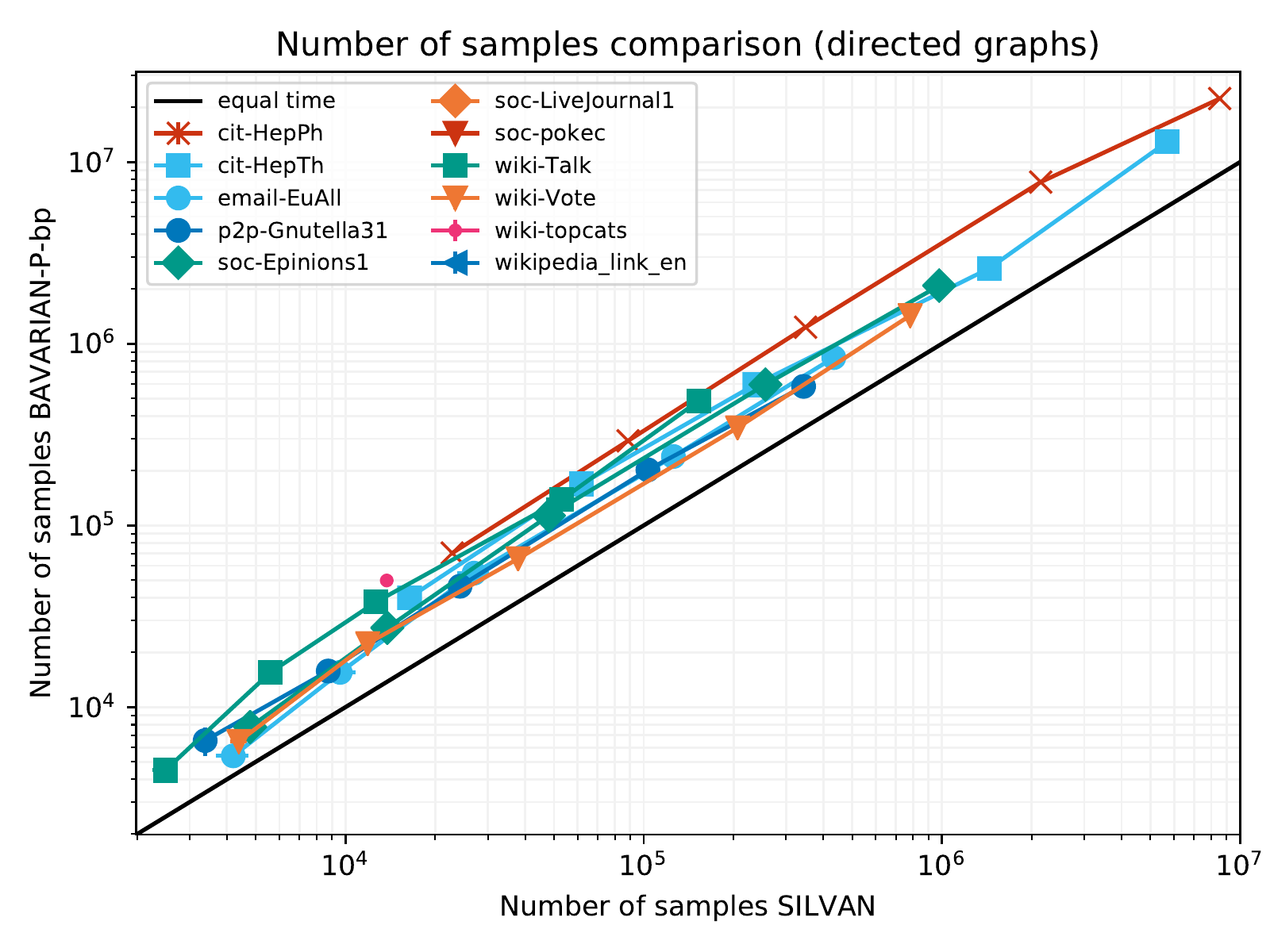}
\caption{}
\end{subfigure}
\begin{subfigure}{.35\textwidth}
  \centering
  \includegraphics[width=\textwidth]{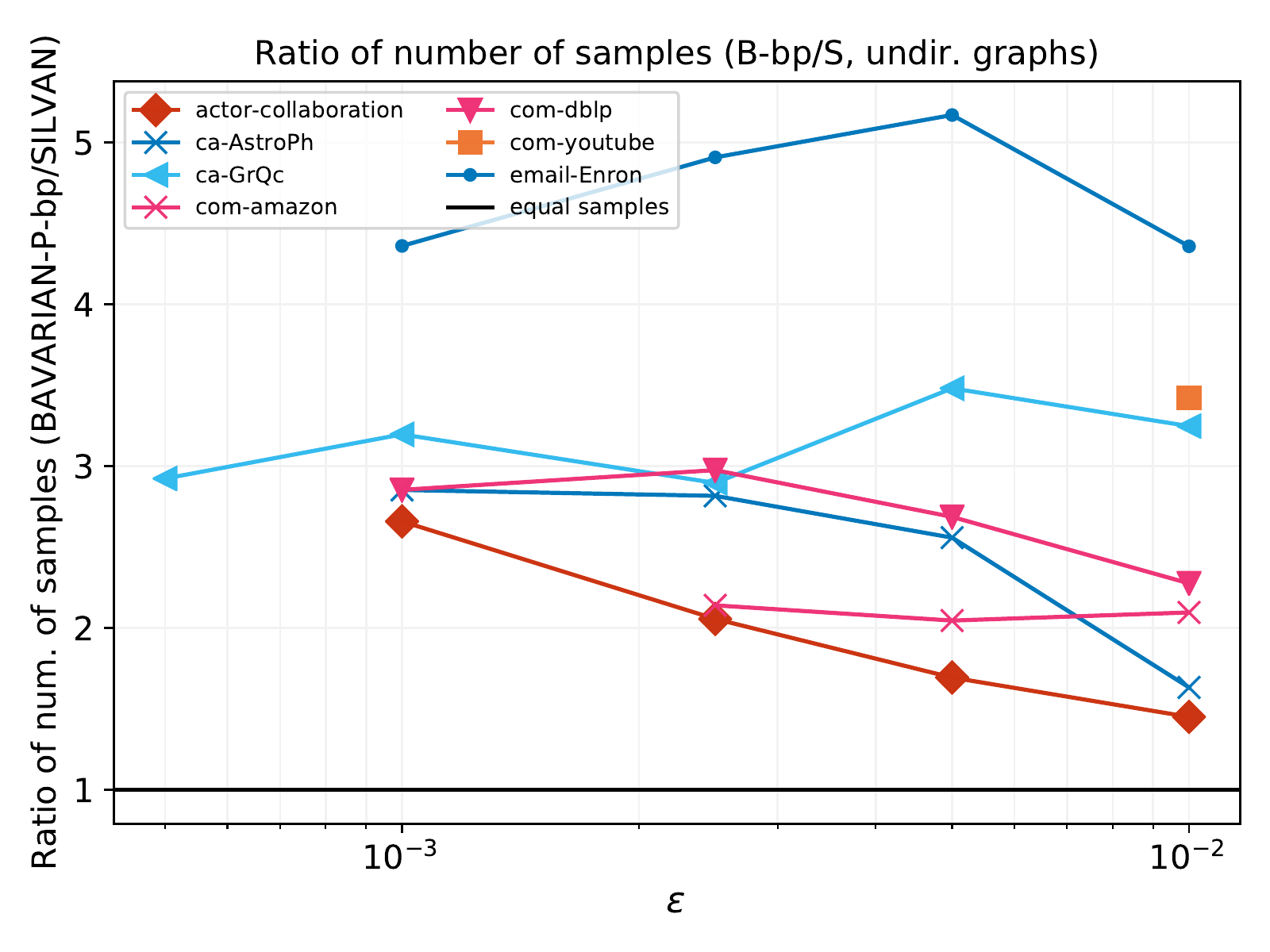}
  \caption{}
\end{subfigure}
\begin{subfigure}{.35\textwidth}
  \centering
  \includegraphics[width=\textwidth]{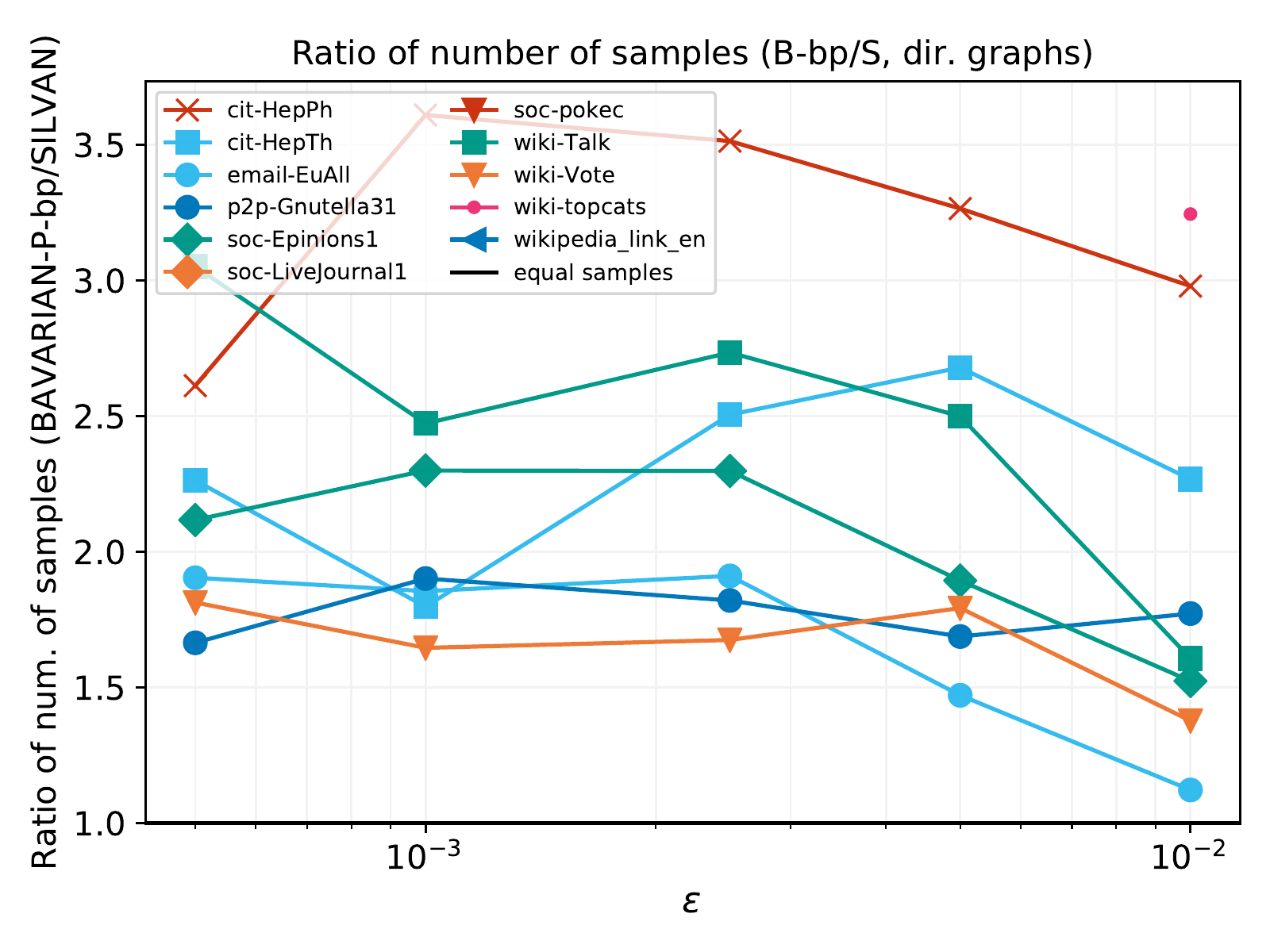}
\caption{}
\end{subfigure}
\begin{subfigure}{.35\textwidth}
  \centering
  \includegraphics[width=\textwidth]{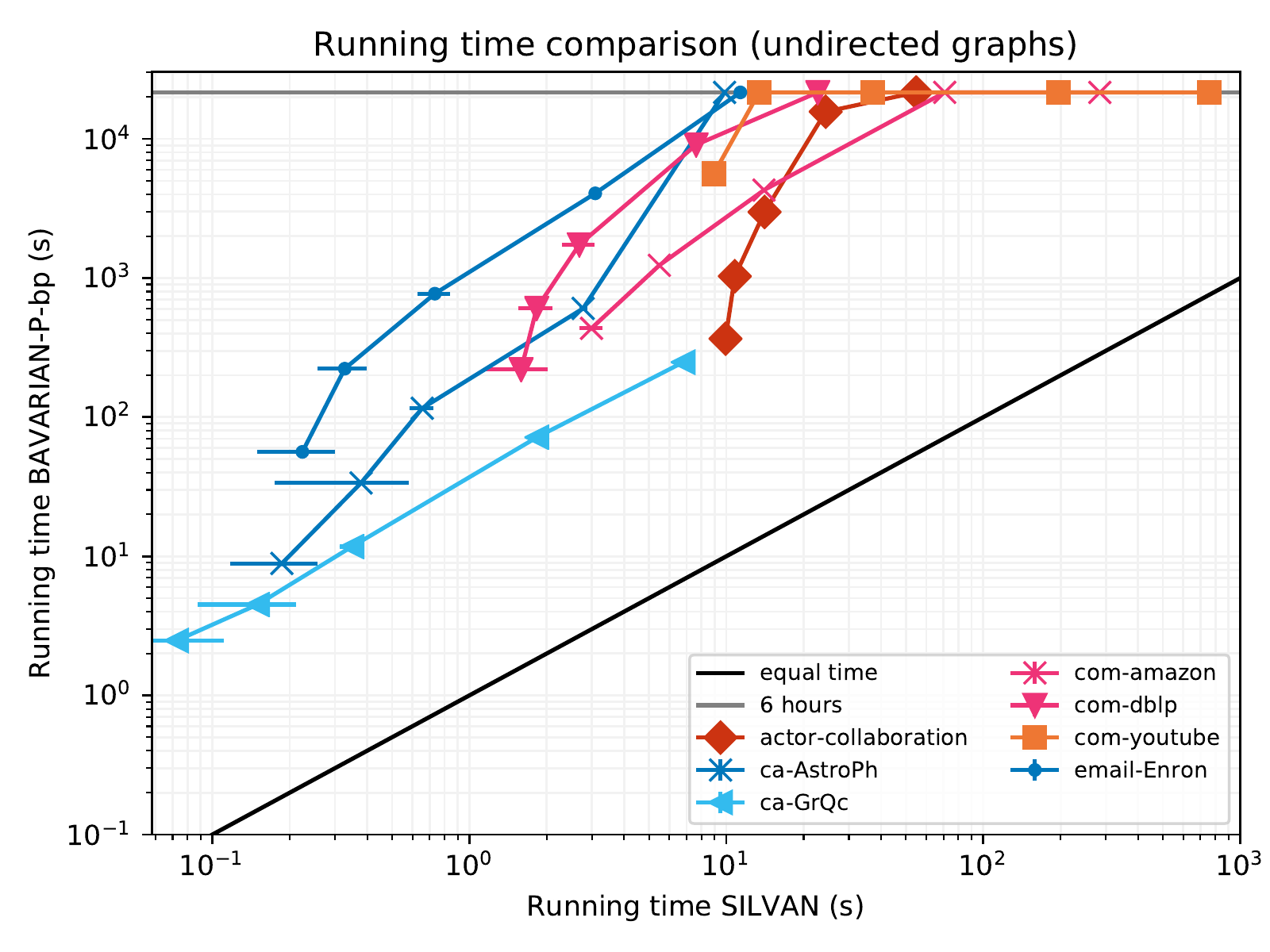}
  \caption{}
\end{subfigure}
\begin{subfigure}{.35\textwidth}
  \centering
  \includegraphics[width=\textwidth]{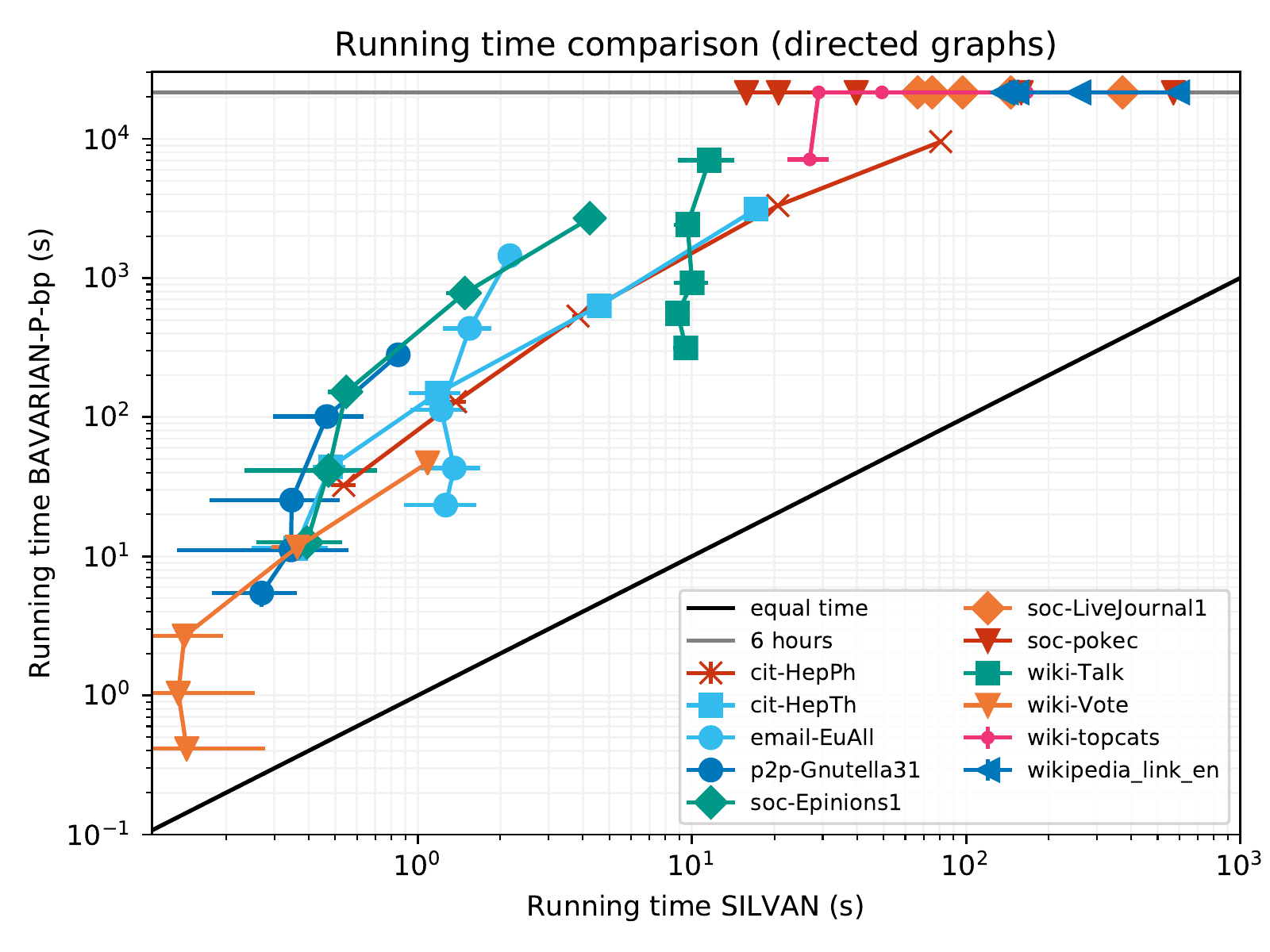}
\caption{}
\end{subfigure}
\begin{subfigure}{.35\textwidth}
  \centering
  \includegraphics[width=\textwidth]{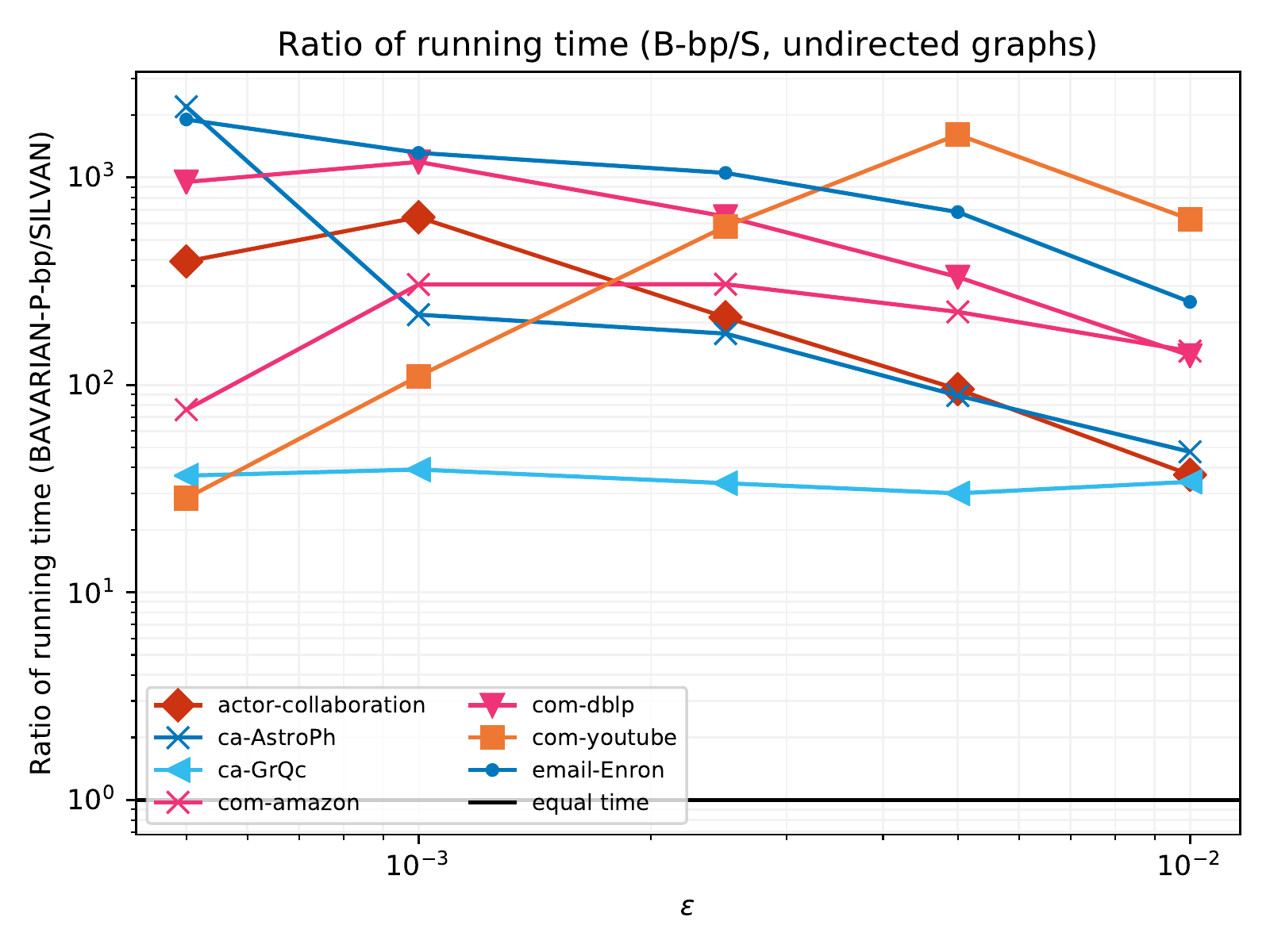}
  \caption{}
\end{subfigure}
\begin{subfigure}{.35\textwidth}
  \centering
  \includegraphics[width=\textwidth]{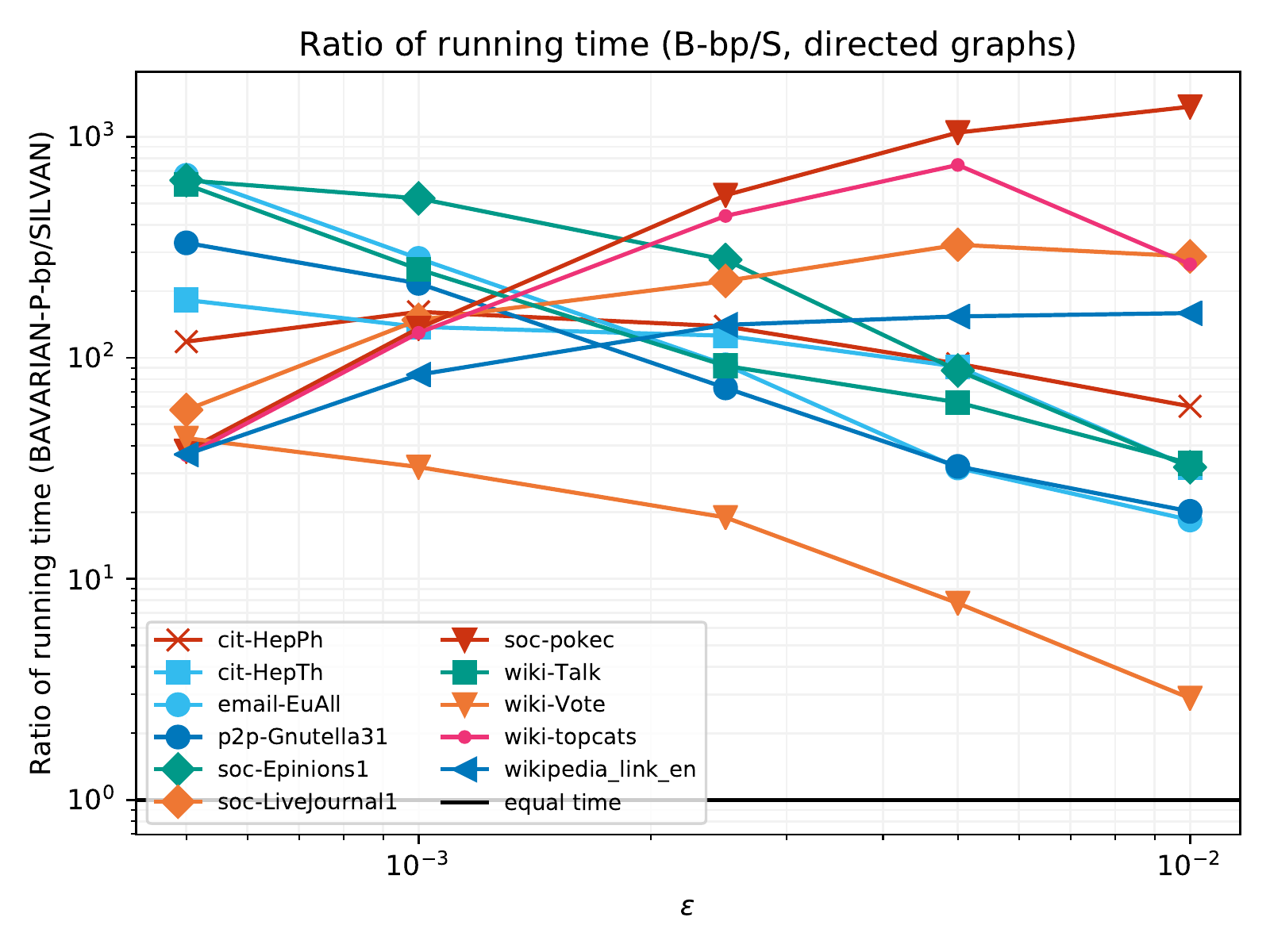}
\caption{}
\end{subfigure}
\caption{Figures comparing the 
performance of \kadabra\ and \bavarianp\ (\texttt{bp} estimator) for obtaining an absolute $\varepsilon$ approximation. 
(a): comparison of the number of samples for \bavarian\ ($y$ axis) and \algname\ ($x$ axis) for undirected graphs (axes in logarithmic scales).
(b): analogous of (a) for directed graphs.
(c): ratios of the number of samples for \bavarian\ and \algname\ for undirected graphs.
(d): analogous of (c) for directed graphs.
(e): comparison of the running times of \bavarian\ ($y$ axis) and \algname\ ($x$ axis) for undirected graphs (axes in logarithmic scales).
(f): analogous of (e) for directed graphs.
(g): ratios of the running times of \bavarian\ and \algname\ for undirected graphs.
(h): analogous of (g) for directed graphs.
}
\label{fig:absapproxappendixbavbp}
\end{figure*}

\begin{figure*}[ht]
\centering
\begin{subfigure}{.27\textwidth}
  \centering
  \includegraphics[width=\textwidth]{./figures/plot_silvan_topktimes_legend-cropped.pdf}
\end{subfigure} \\
\begin{subfigure}{.32\textwidth}
  \centering
  \includegraphics[width=\textwidth]{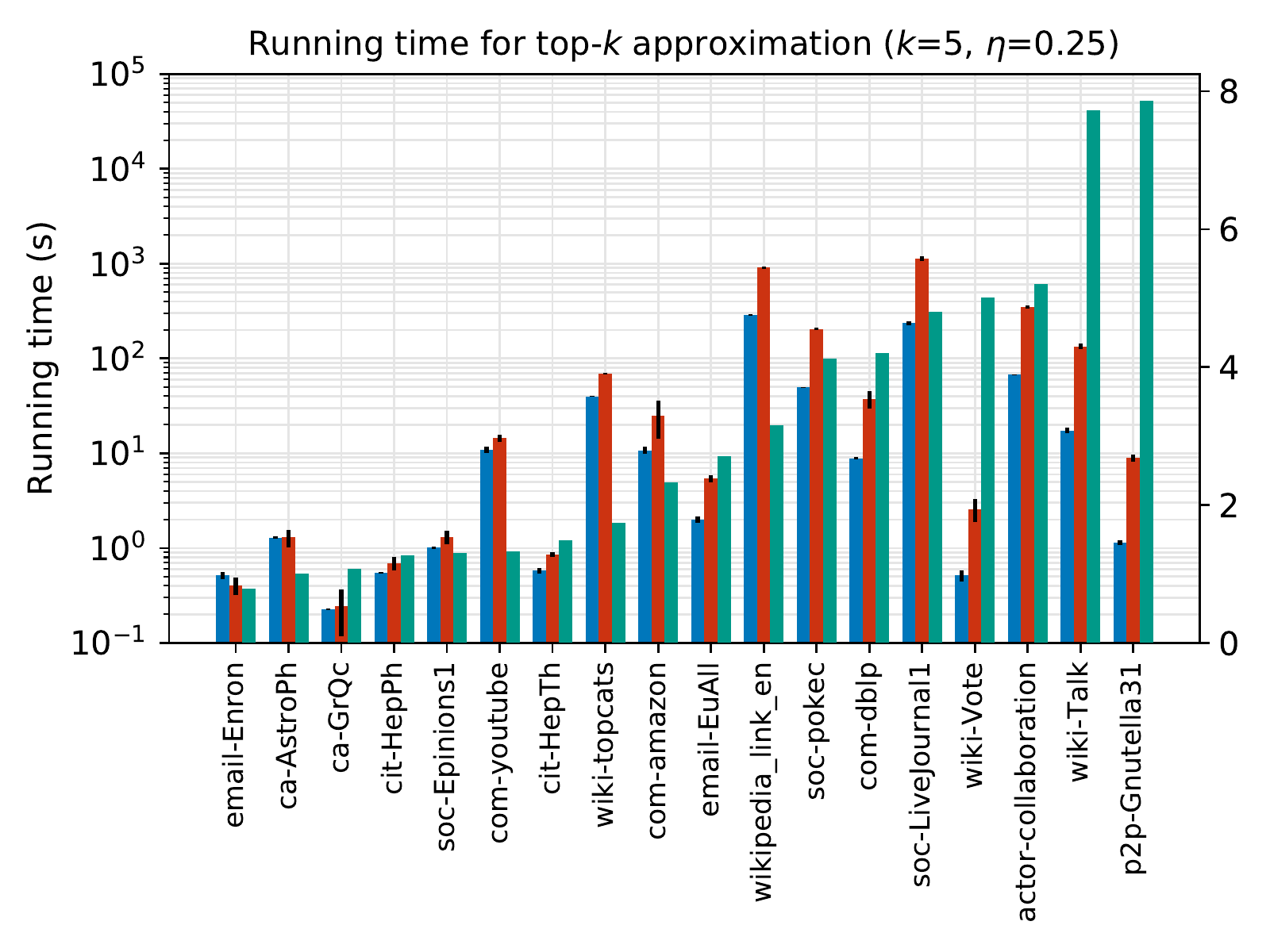}
\end{subfigure}
\begin{subfigure}{.32\textwidth}
  \centering
  \includegraphics[width=\textwidth]{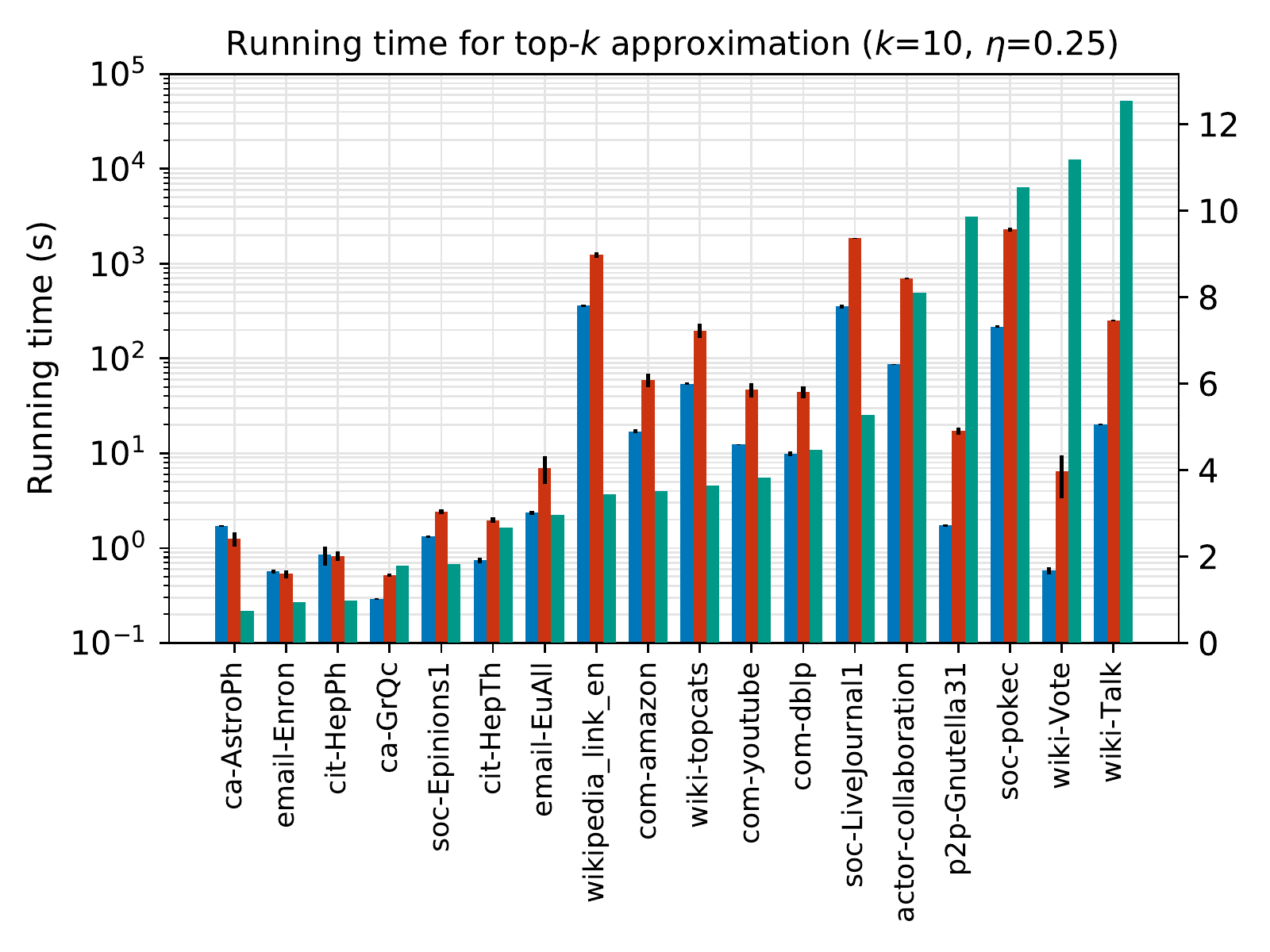}
\end{subfigure}
\begin{subfigure}{.32\textwidth}
  \centering
  \includegraphics[width=\textwidth]{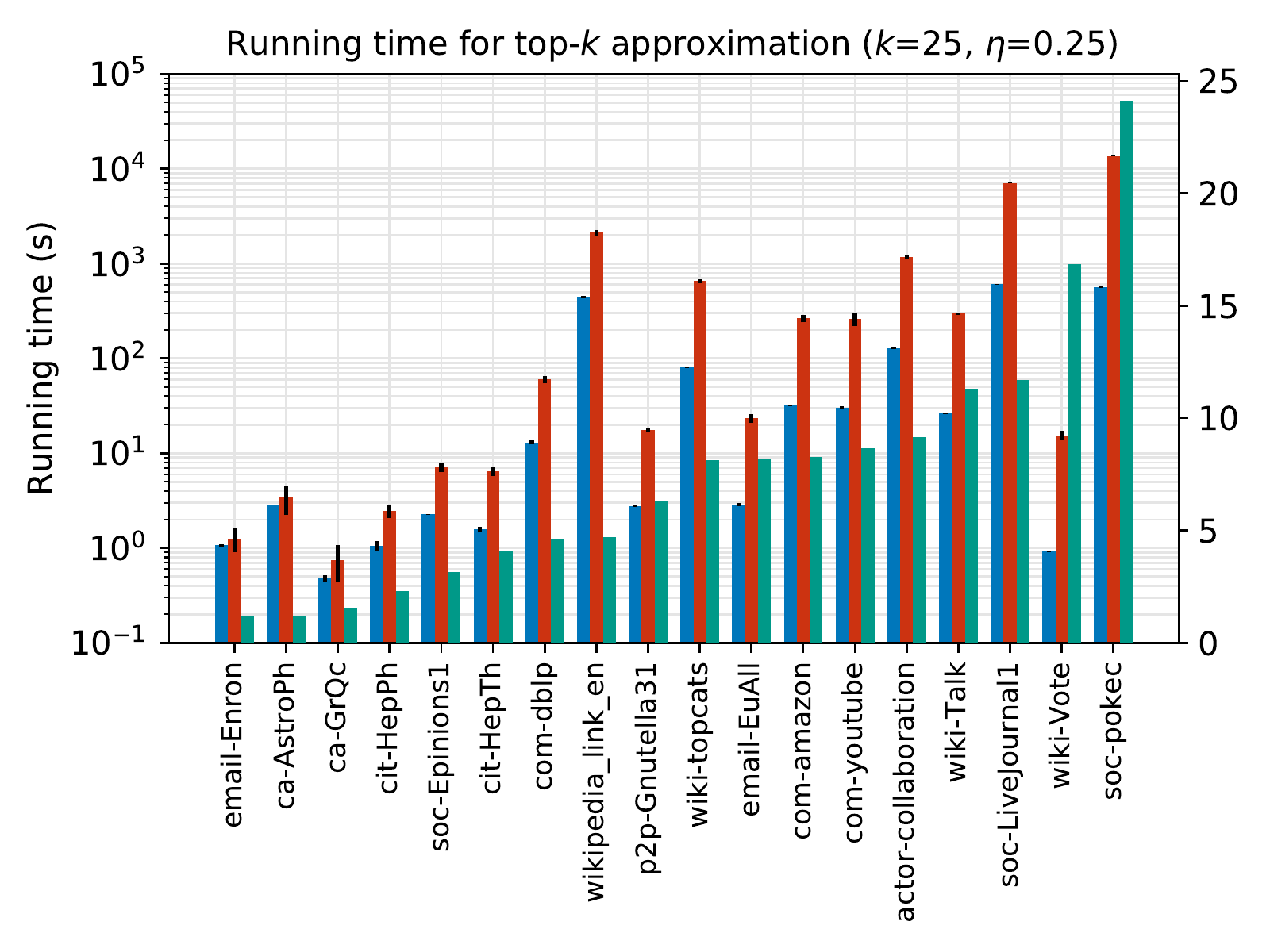}
\end{subfigure}
\begin{subfigure}{.32\textwidth}
  \centering
  \includegraphics[width=\textwidth]{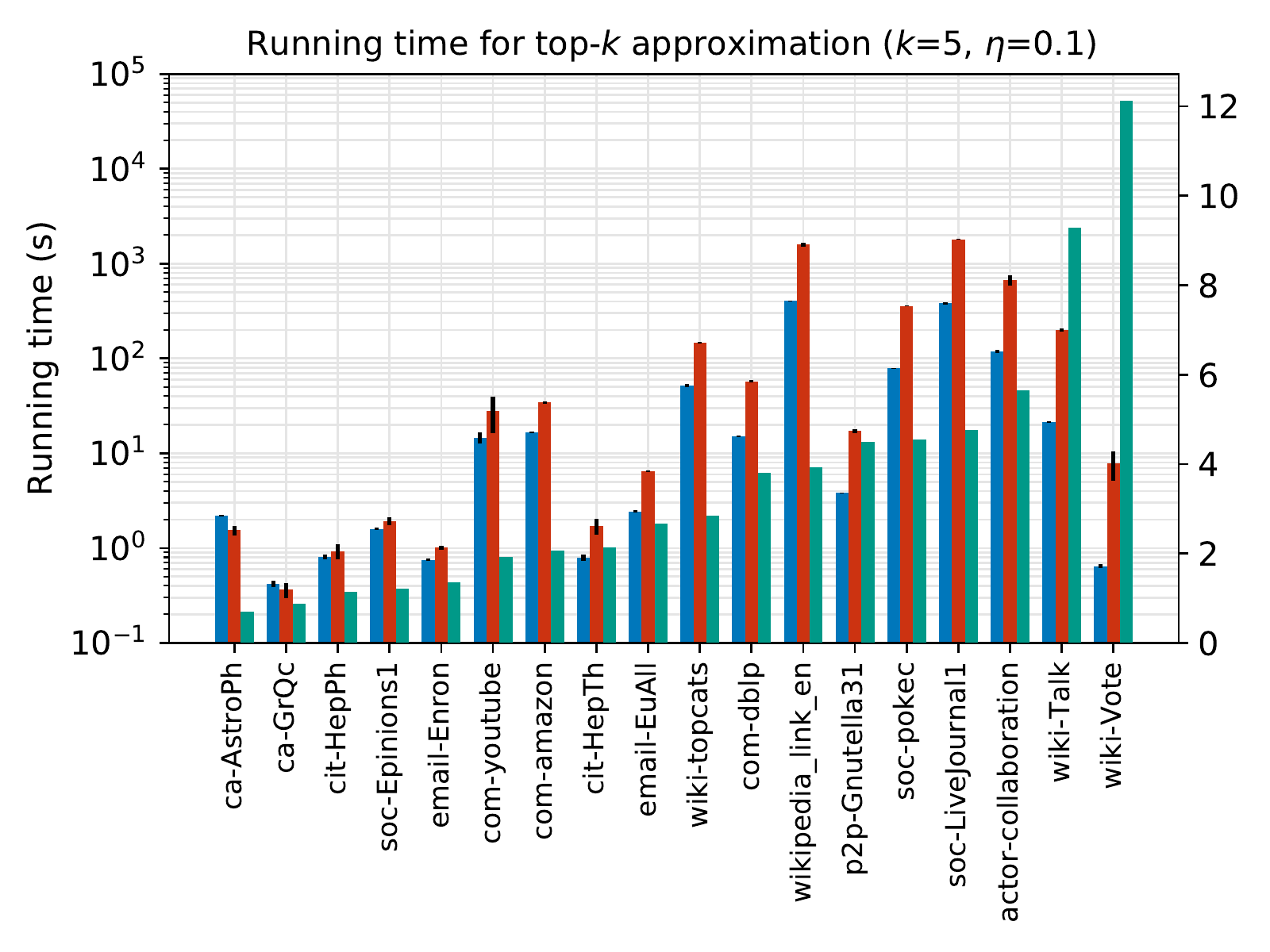}
\end{subfigure}
\begin{subfigure}{.32\textwidth}
  \centering
  \includegraphics[width=\textwidth]{./figures/plot_top_k_10_10.pdf}
\end{subfigure}
\begin{subfigure}{.32\textwidth}
  \centering
  \includegraphics[width=\textwidth]{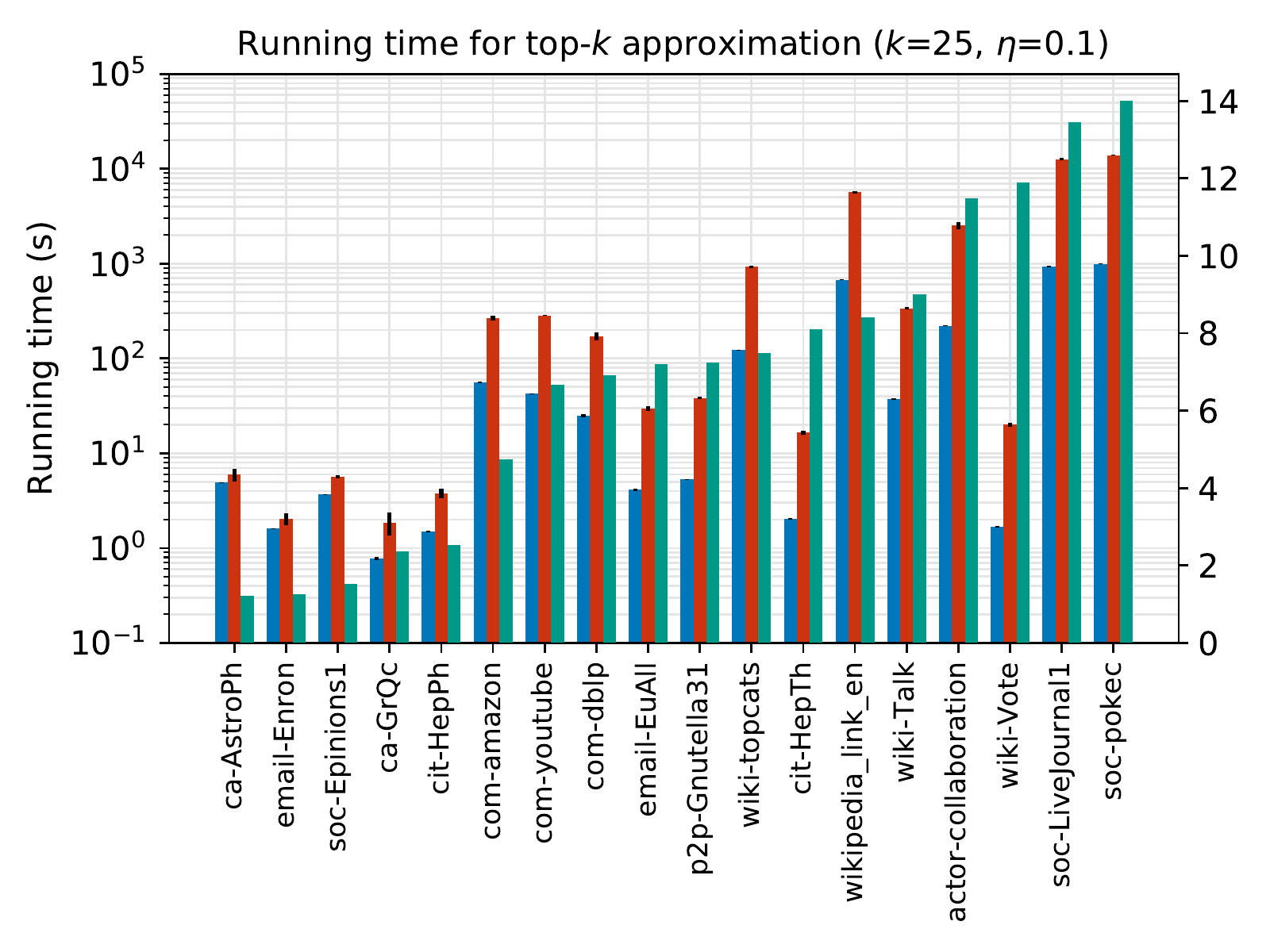}
\end{subfigure}
\begin{subfigure}{.32\textwidth}
  \centering
  \includegraphics[width=\textwidth]{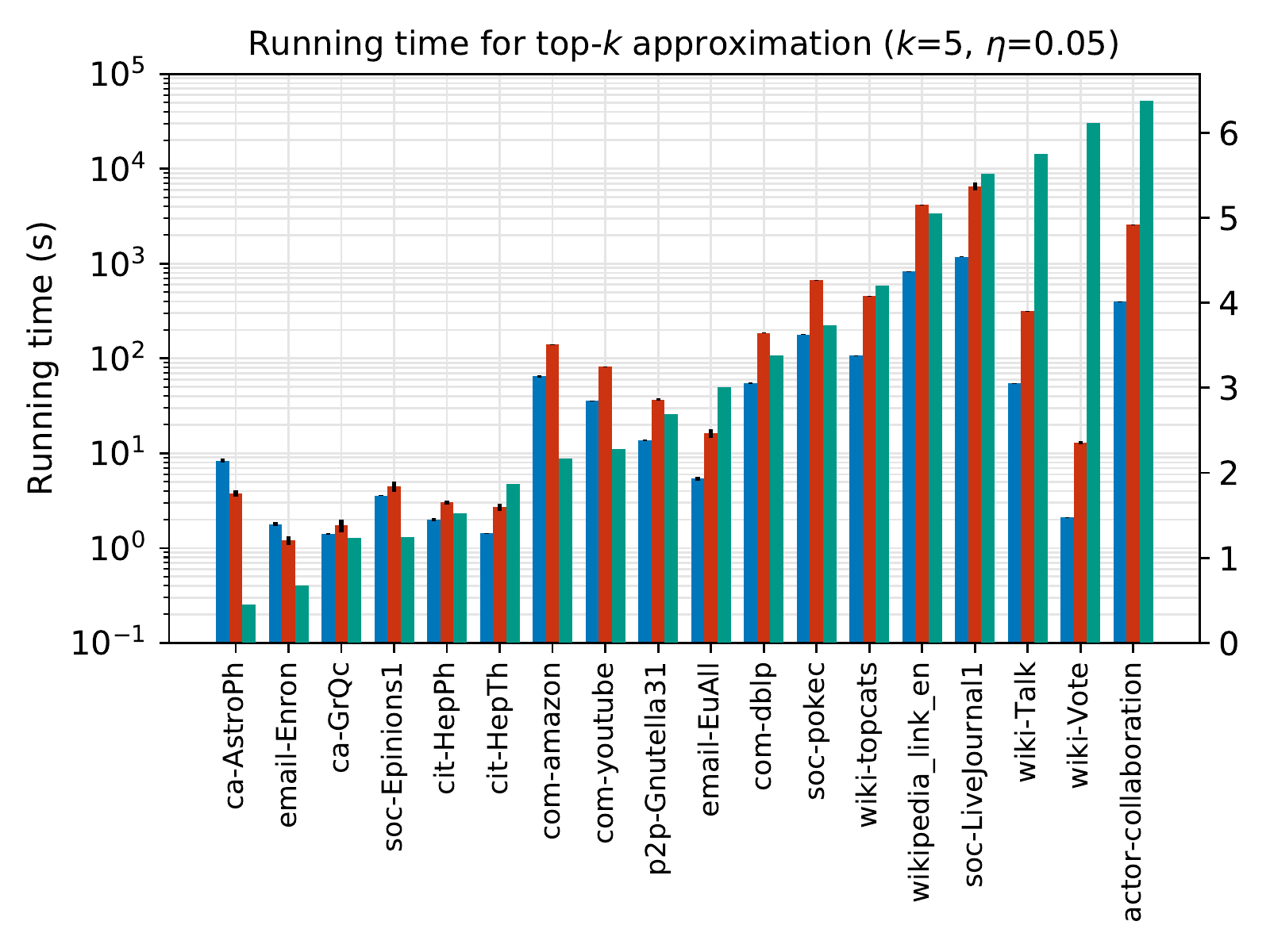}
\end{subfigure}
\begin{subfigure}{.32\textwidth}
  \centering
  \includegraphics[width=\textwidth]{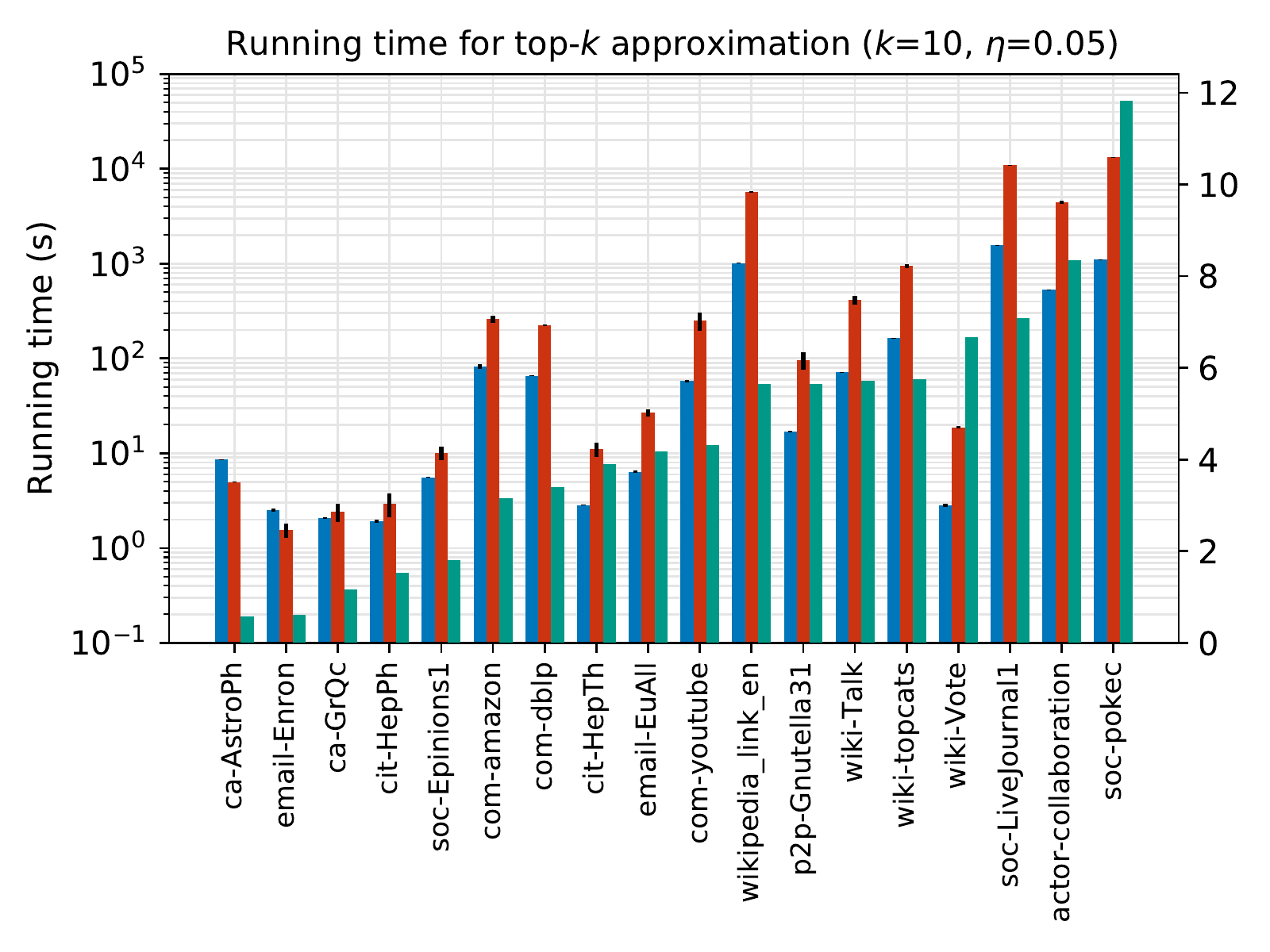}
\end{subfigure}
\begin{subfigure}{.32\textwidth}
  \centering
  \includegraphics[width=\textwidth]{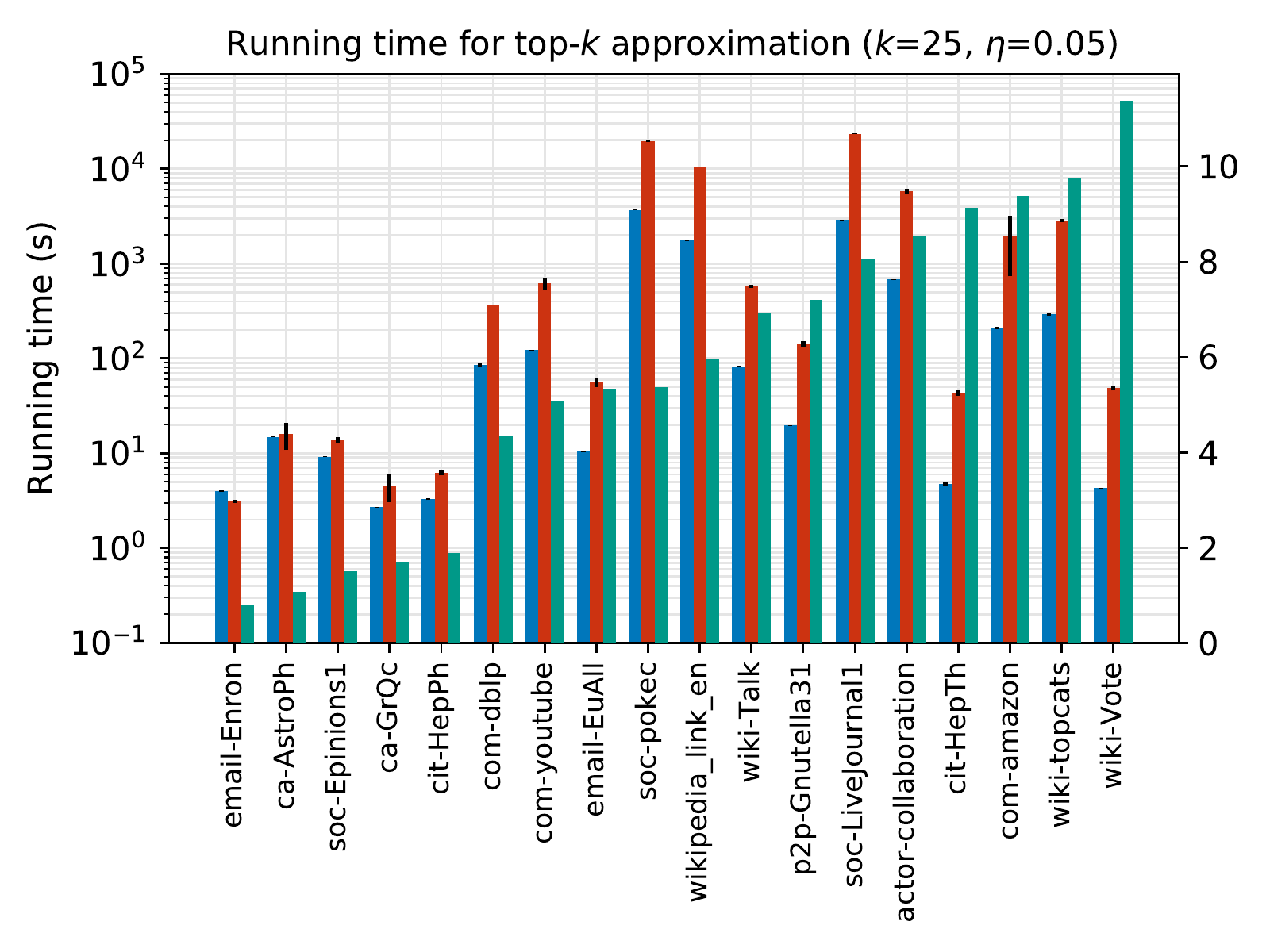}
\end{subfigure}
\caption{Comparison of running times of \algnametopk\ with \kadabra\ for obtaining top-$k$ approximations for all combinations of $k$ and $\eta$.}
\label{fig:topkappxtimesappx}
\end{figure*}

\begin{figure*}[ht]
\centering
\begin{subfigure}{.27\textwidth}
  \centering
  \includegraphics[width=\textwidth]{./figures/plot_silvan_topksamples_legend-cropped.pdf}
\end{subfigure} \\
\begin{subfigure}{.32\textwidth}
  \centering
  \includegraphics[width=\textwidth]{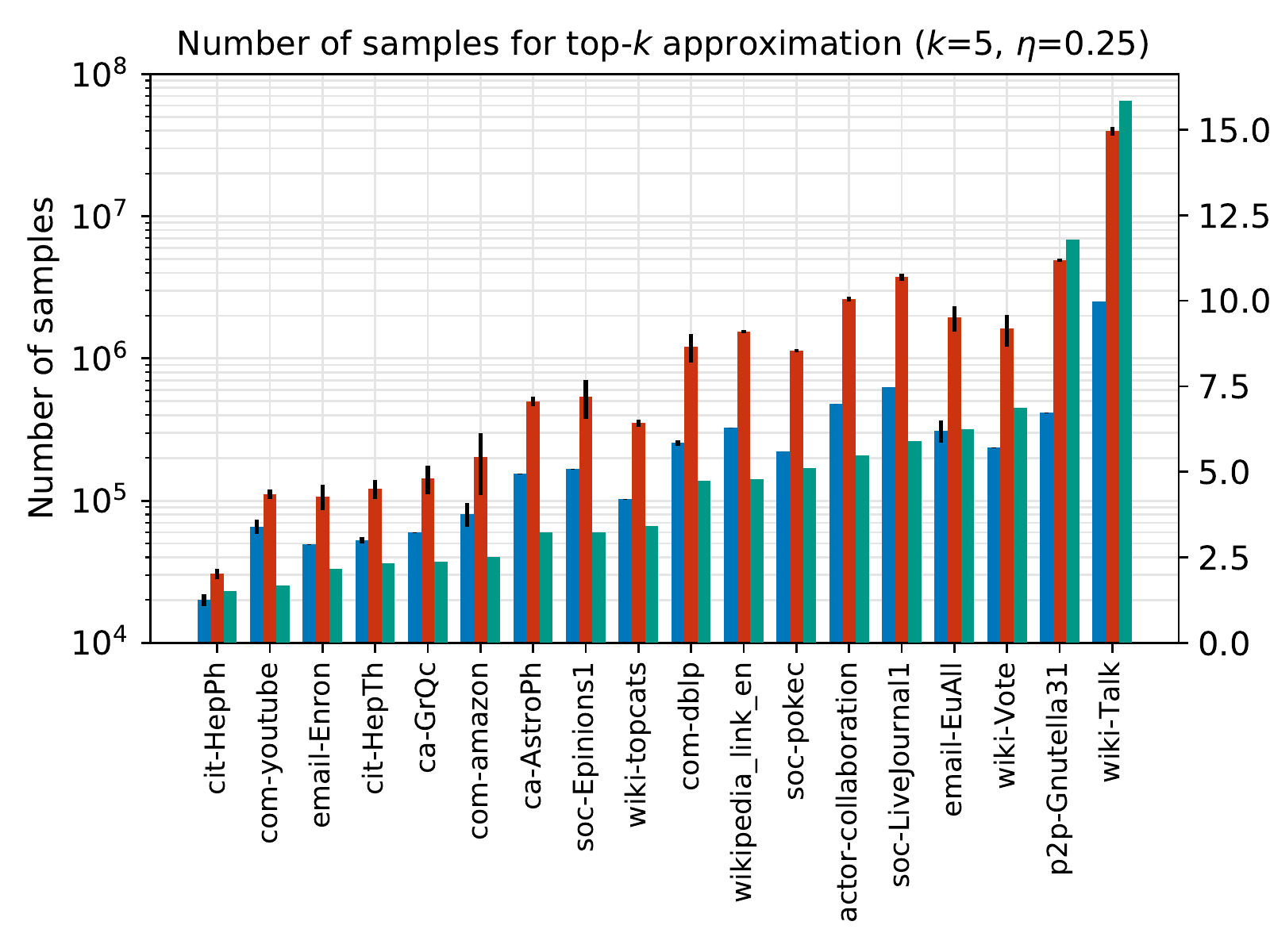}
\end{subfigure}
\begin{subfigure}{.32\textwidth}
  \centering
  \includegraphics[width=\textwidth]{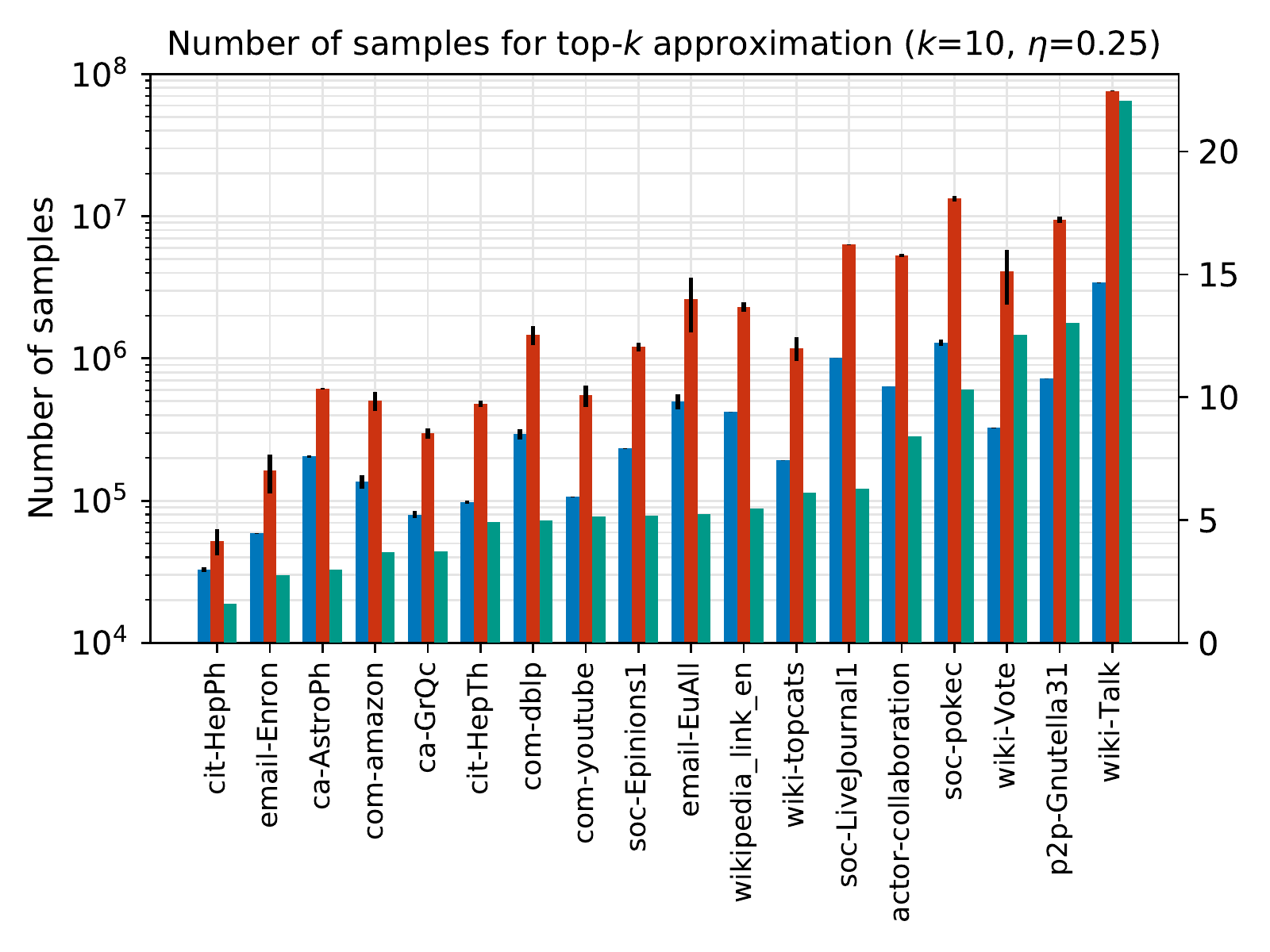}
\end{subfigure}
\begin{subfigure}{.32\textwidth}
  \centering
  \includegraphics[width=\textwidth]{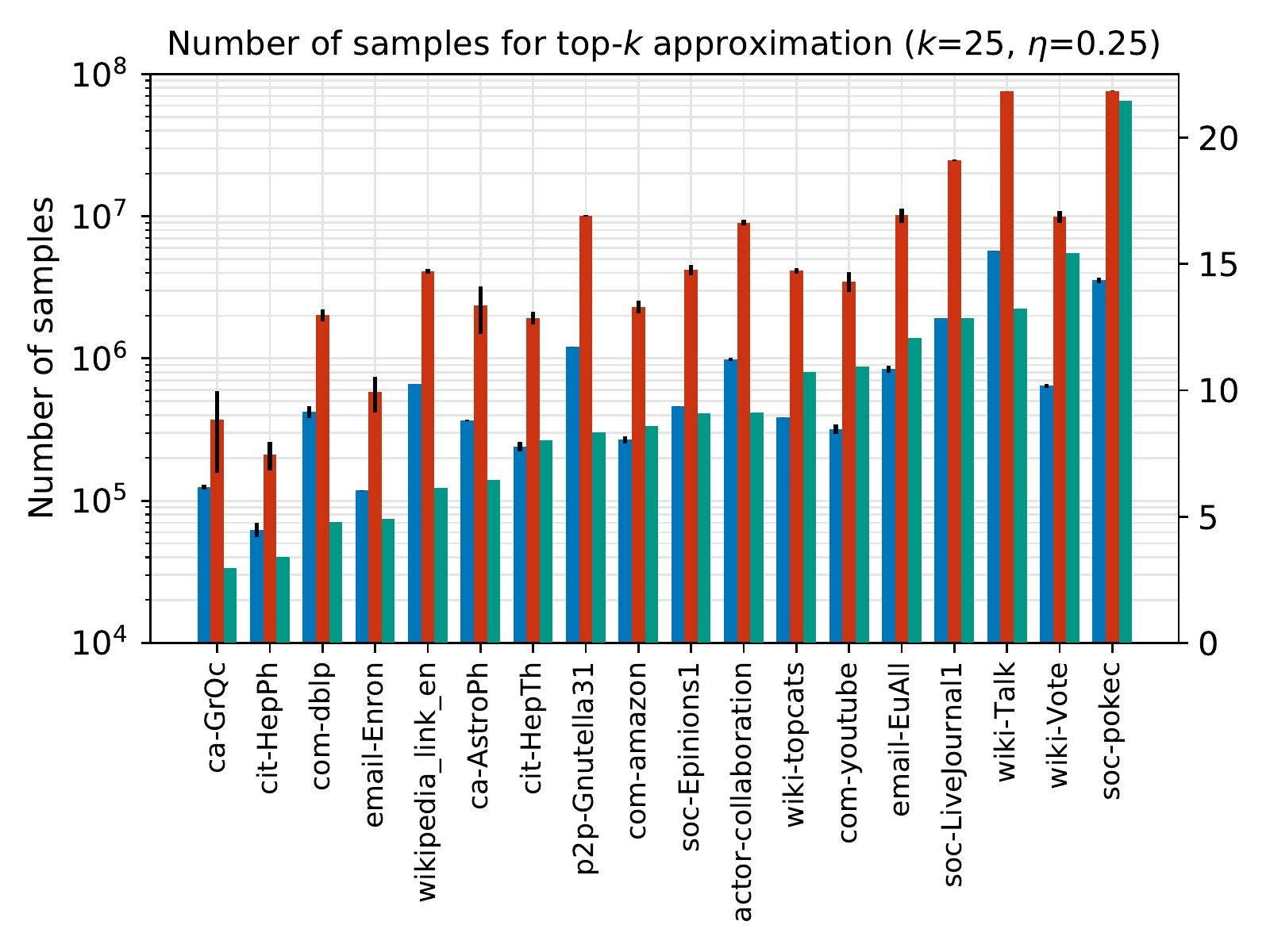}
\end{subfigure}
\begin{subfigure}{.32\textwidth}
  \centering
  \includegraphics[width=\textwidth]{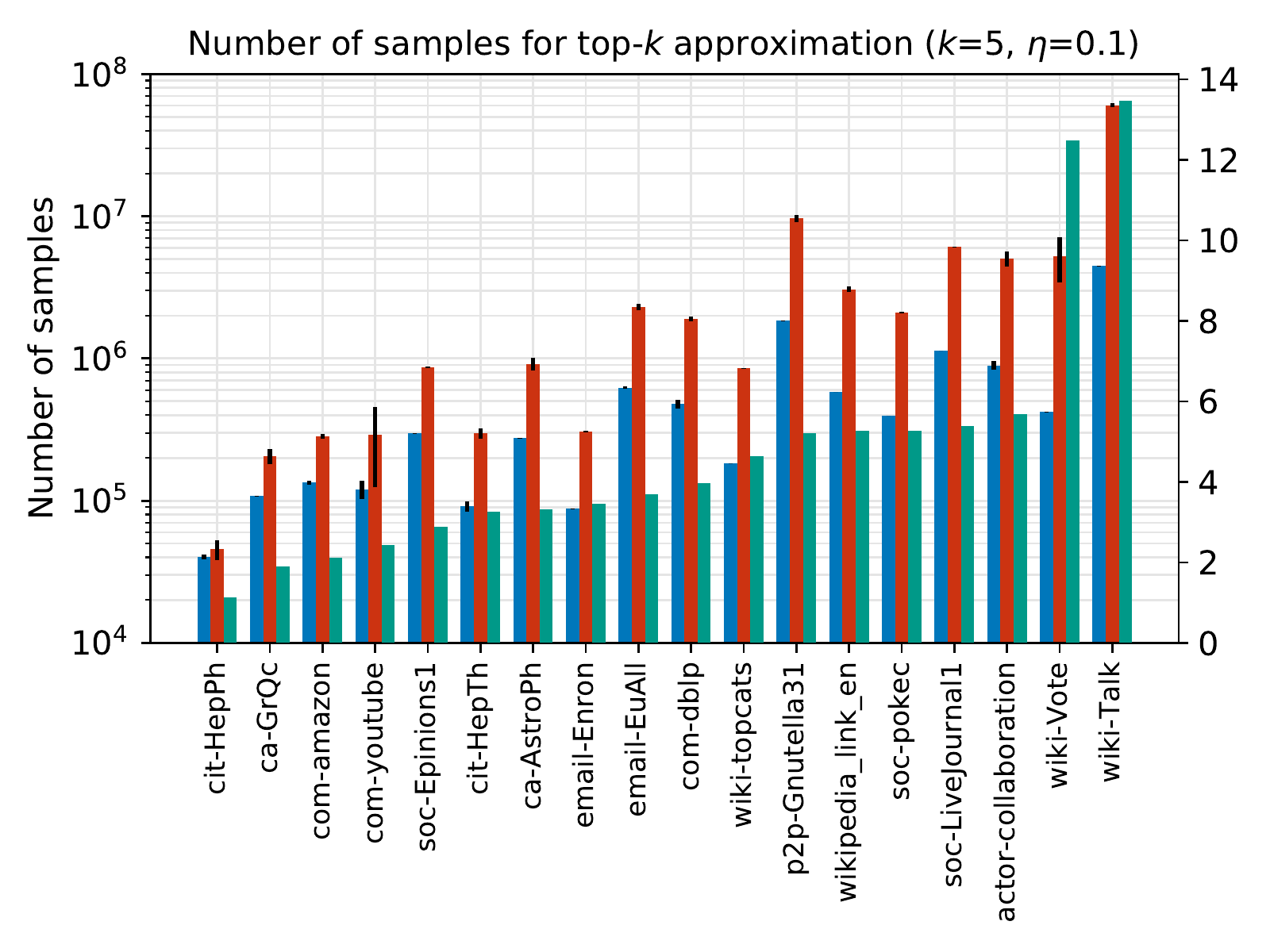}
\end{subfigure}
\begin{subfigure}{.32\textwidth}
  \centering
  \includegraphics[width=\textwidth]{./figures/plot_top_k_10_10_samples.pdf}
\end{subfigure}
\begin{subfigure}{.32\textwidth}
  \centering
  \includegraphics[width=\textwidth]{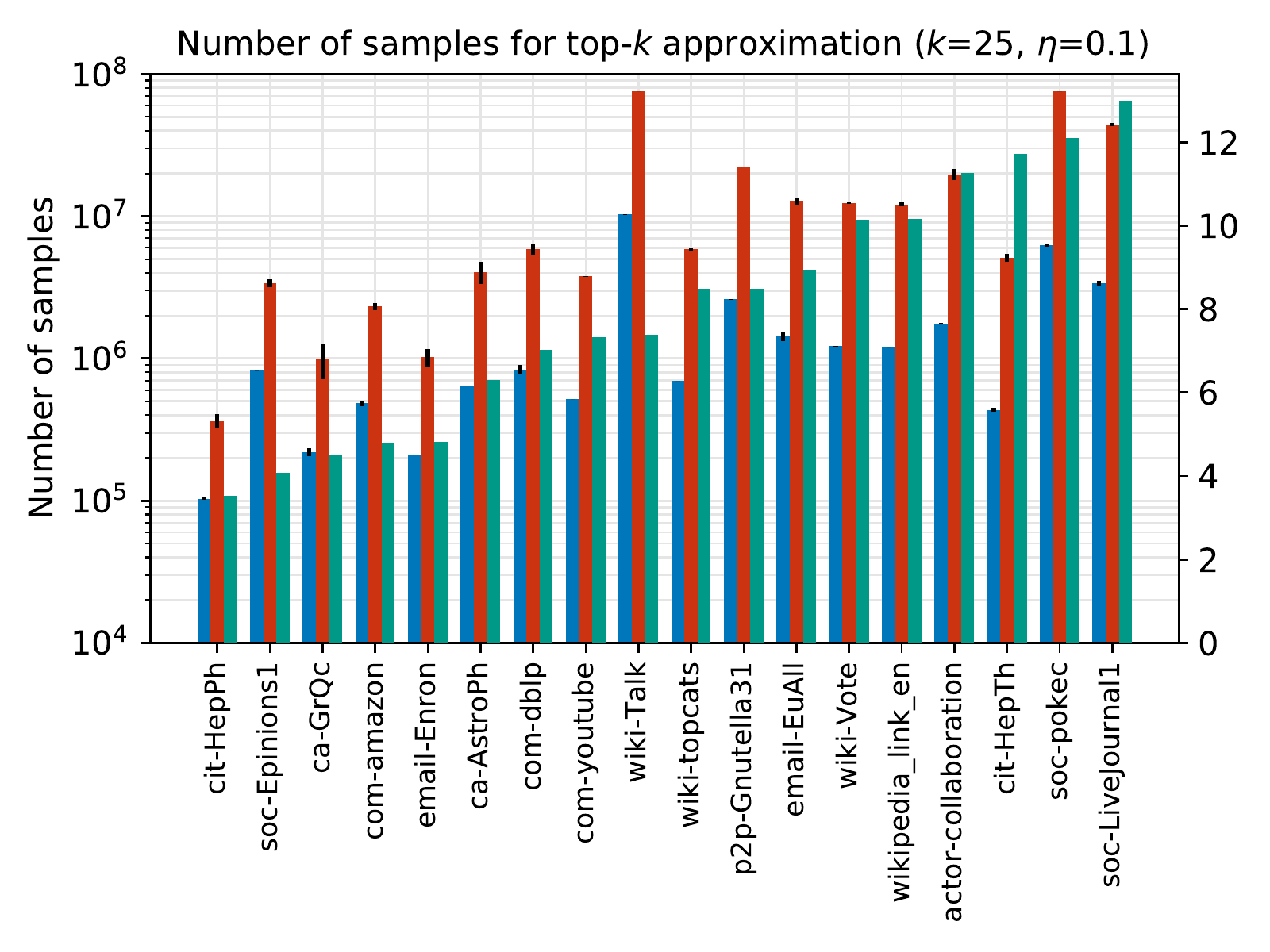}
\end{subfigure}
\begin{subfigure}{.32\textwidth}
  \centering
  \includegraphics[width=\textwidth]{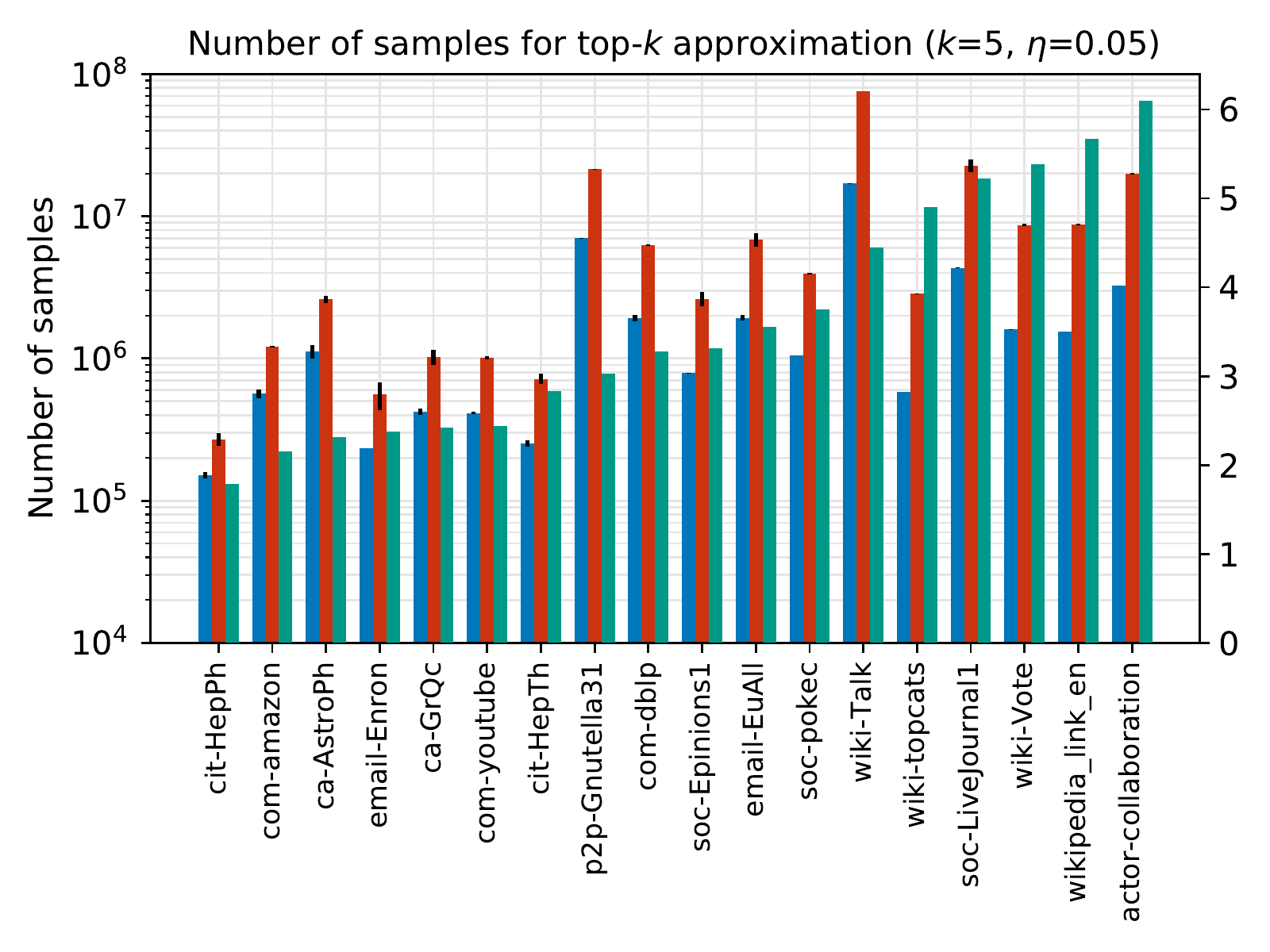}
\end{subfigure}
\begin{subfigure}{.32\textwidth}
  \centering
  \includegraphics[width=\textwidth]{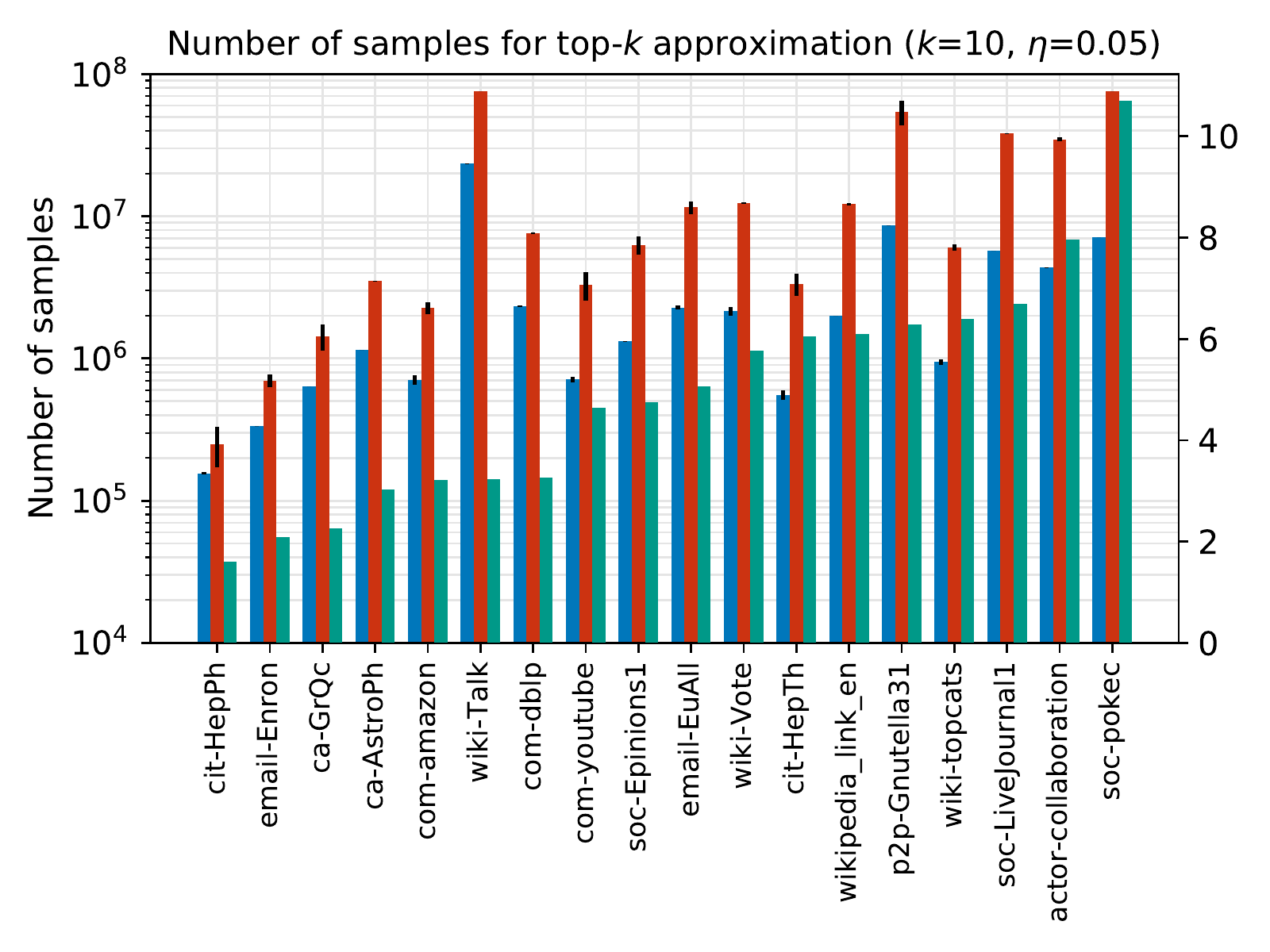}
\end{subfigure}
\begin{subfigure}{.32\textwidth}
  \centering
  \includegraphics[width=\textwidth]{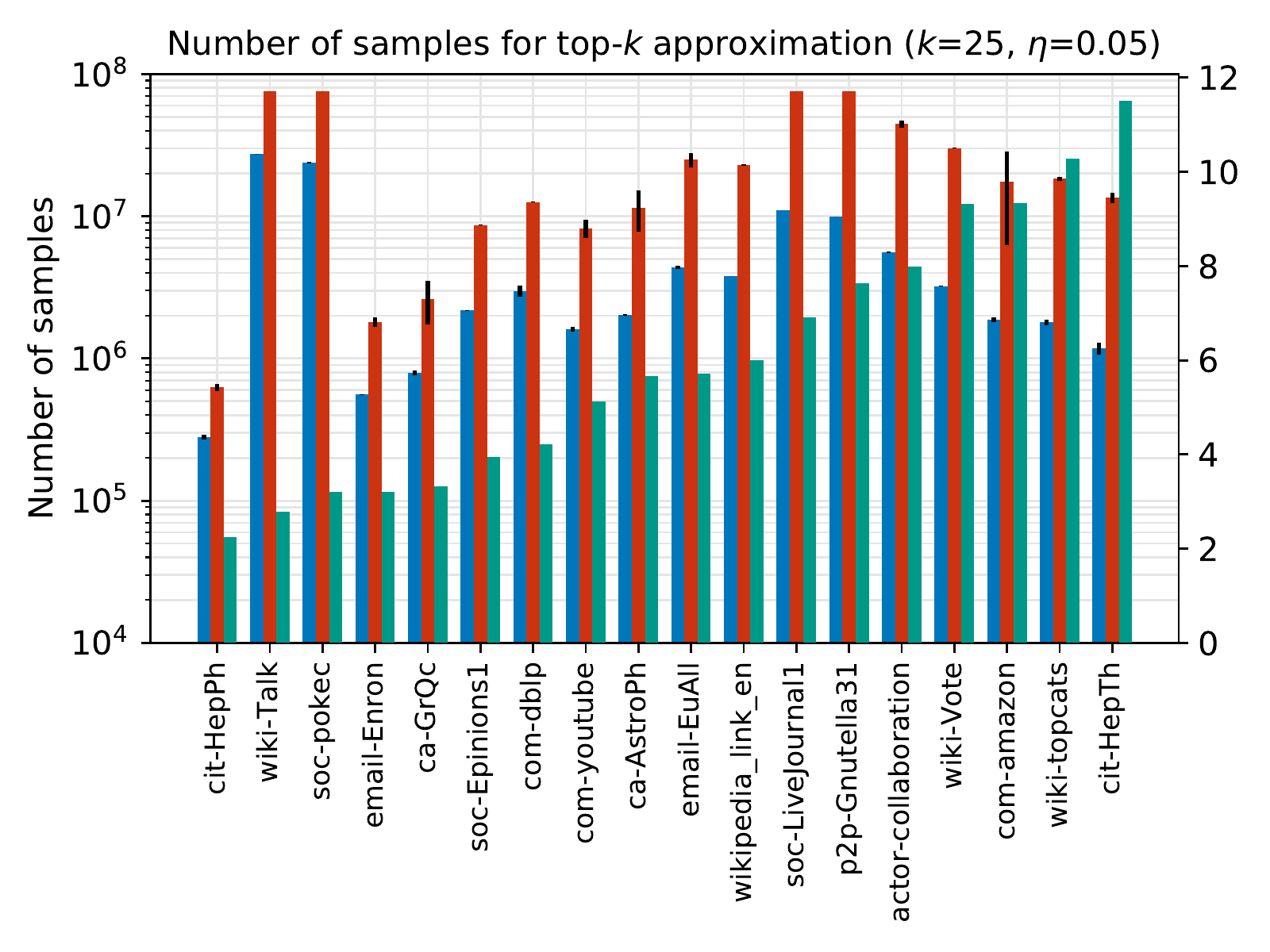}
\end{subfigure}
\caption{Comparison of number of samples of \algnametopk\ with \kadabra\ for obtaining top-$k$ approximations for all combinations of $k$ and $\eta$.}
\label{fig:topkappxsamplesappx}
\end{figure*}

\end{document}